\documentclass[11pt,a4paper]{amsart}
\usepackage[backend=biber, style=alphabetic, date=year, url=false, doi=false, isbn=false, eprint=true, firstinits=true, maxalphanames=4, sorting=anyt, maxcitenames=10, maxbibnames=10]{biblatex}
\AtEveryBibitem{%
  \ifentrytype{online}{%
  }{%
    \clearfield{eprint}%
    \clearfield{urldate}%
  }%
}

\usepackage{tikz}
\usepackage{tkz-euclide}
\usetikzlibrary{positioning,arrows,patterns}
\usetikzlibrary{decorations.markings}
\usetikzlibrary{decorations.pathmorphing}
\usetikzlibrary{calc}
\usetikzlibrary{patterns}
\tikzset{
  photon/.style={decorate, decoration={snake}, draw=black},
  fermion/.style={draw=black, postaction={decorate},decoration={markings,mark=at position .55 with {\arrow{>}}}},
  ghost/.style={dashed, postaction={decorate},decoration={markings,mark=at position .55 with {\arrow{>}}}},
    antighost/.style={dashed, postaction={decorate},decoration={markings,mark=at position .55 with {\arrow{<}}}},
  vertex/.style={draw,shape=circle,fill=black,minimum size=5pt,inner sep=0pt},
  black/.style={draw,shape=circle,fill=black,minimum size=5pt,inner sep=0pt},
  striped/.style={draw,shape=circle,fill=white,pattern=north east lines,minimum size=5pt,inner sep=0pt},
  white/.style={draw,shape=circle,fill=white,minimum size=5pt,inner sep=0pt},
  grey/.style={draw,shape=circle,fill=black!20,minimum size=5pt,inner sep=0pt},
particle/.style={thick,draw=black},
particle2/.style={thick,draw=blue},
avector/.style={thick,draw=black, postaction={decorate},
    decoration={markings,mark=at position 1 with {\arrow[black]{triangle 45}}}},
gluon/.style={ draw=black,decorate,
    decoration={coil,aspect=0.9,amplitude=0.2cm}},
leaf/.style={draw,shape=circle, inner sep=2pt}
 }

\usepackage{amsmath}
\usepackage{amssymb}
\usepackage{amsthm}
\usepackage{paralist}
\usepackage{hyperref}  % turn it off for arXiv
\usepackage{tikz-cd}
\usepackage{todonotes}
\usepackage{array}
\usepackage{subcaption}
\usepackage{comment}
\usepackage{xcolor}
\usepackage{footnote}
\usepackage{cancel}
\usepackage[normalem]{ulem} % for \sout (strikeout text) command

\makesavenoteenv{tabular}
\makesavenoteenv{table}
\theoremstyle{plain} % 'default
\newtheorem{thm}{Theorem}[section]
\newtheorem{lem}[thm]{Lemma}
\newtheorem{cor}[thm]{Corollary}
\newtheorem{prop}[thm]{Proposition}

\newtheorem{conj}[thm]{Conjecture}
\theoremstyle{definition}
\newtheorem{defn}[thm]{Definition}
\newtheorem{expl}[thm]{Example}

\newtheorem*{defn*}{Definition}
\theoremstyle{remark}
\newtheorem{rem}[thm]{Remark}
\newtheorem{ass}[thm]{Assumption}

\newcommand{\N}{\mathbb{N}}

\newcommand{\Z}{\mathbb{Z}}
\newcommand{\R}{\mathbb{R}}
\newcommand{\HH}{\mathbb{H}}
\newcommand{\g}{\mathfrak{g}}

\newcommand{\dd}{d}

\newcommand{\Aut}{\operatorname{Aut}}

\newcommand{\sign}{\operatorname{sign}}
\newcommand{\Det}{\operatorname{Det}}

\newcommand{\tr}{\operatorname{tr}}

\newcommand{\im}{\operatorname{im}}
\newcommand{\sdet}{\operatorname{Sdet}}
\newcommand{\Sym}{\operatorname{Sym}}
\newcommand{\mr}{\mathrm}

\newcommand{\calL}{\mathcal{L}}
\newcommand{\sfa}{\mathsf{a}}
\newcommand{\sfb}{\mathsf{b}}
\newcommand{\sfc}{\mathsf{c}}

\newcommand{\bt}{\bullet}
\newcommand{\ra}{\rightarrow}
\newcommand{\sh}{\mathsf{h}}
\newcommand{\M}{\mathcal{M}}

\newcommand{\wh}{\widehat}
\newcommand{\til}{\widetilde}
\newcommand{\mc}{\mathcal}
\newcommand{\ad}{\mathrm{ad}}
\newcommand{\FC}{\mathrm{FC}}
\newcommand{\HHH}{\mathcal{H}}
\newcommand{\xra}{\xrightarrow}
\newcommand{\la}{\leftarrow}
\newcommand{\BB}{\mathfrak{B}}

\newcommand{\UU}{\underline{\mc{U}}}
\newcommand{\ul}{\underline}
\newcommand{\sfg}{\mathrm{g}} % group elt
\newcommand{\Kb}{\overline{K}}
\newcommand{\ib}{\overline{i}}
\newcommand{\pb}{\overline{p}}
\newcommand{\Qb}{\overline{Q}}
\newcommand{\Gb}{\overline{G}}
\newcommand{\mm}{\mathrm{m}}
\newcommand{\ddd}{\mathfrak{d}}
\newcommand{\Conf}[1][V]{\overline{\mathrm{Conf}}_#1(M)}
\newcommand{\II}{\mathbb{I}}
\newcommand{\PP}{\mathbb{P}}

\usepackage[automake,toc = false]{glossaries-extra}
\glsxtrRevertMarks
\newglossarystyle{myNoHeaderStyle}{
    \setglossarystyle{longheader} % Use the list style as base 
}

\makeglossaries

\newglossaryentry{space}
{
  name={},
  sort={},
  description={}
}
\newglossaryentry{M}
{
  name={$M$},
  sort={a},
  description={compact 3-manifold}
}
\newglossaryentry{G}
{
  name={$G$},
  sort={b},
  description={compact, simple, simply connected matrix group, e.g. $G = SU(N)$}
}
\newglossaryentry{P}
{
  name={$P$},
  sort={c},
  description={trivialized principal $G$-bundle, $P = M\times G$}
}
\newglossaryentry{A0}
{
  name={$A_0$},
  sort={d},
  description={A flat connection on $P$}
}
\newglossaryentry{Omega}
{
name = {$\Omega^{\bullet}_{A_0},\Omega^k_{A_0}$},
sort = {d1},
description={The dgla $(\Omega^{\bullet}(M,\g),\dd_{A_0})$}
}
\newglossaryentry{H}
{
  name={$H_{A_0}^\bullet, H_{A_0}^k$},
  sort={e},
  description={Cohomology groups of $d_{A_0}$}
}
\newglossaryentry{abc}
{
  name={$\sfa,\sfb,\sfc,\ldots$},
  sort={f},
  description={Cohomology classes $\sfa,\sfb,\sfc \in H^\bullet_{A_0}$}
}

\newcommand\restr[2]{{% we make the whole thing an ordinary symbol
  \left.\kern-\nulldelimiterspace % automatically resize the bar with \right
  #1 % the function
  \vphantom{\big|} % pretend it's a little taller at normal size
  \right|_{#2} % this is the delimiter
  }}

\allowdisplaybreaks

\title[Globalization in Chern-Simons theory]{
Globalization of perturbative Chern-Simons theory on the moduli space of flat connections in the BV formalism
}

\author[P. Mnev]{Pavel Mnev}
\address{University of Notre Dame, Notre Dame, IN 46556, USA}
\address{Institut f\"ur Mathematik, Universit\"at Z\"urich, Winterthurerstrasse 190, CH-8057 Z\"urich,
Switzerland}
\email{pmnev@nd.edu}

\author[K. Wernli]{Konstantin Wernli}
\address{Centre for Quantum Mathematics, IMADA, University of Southern Denmark, Campusvej 55, 5230 Odense M, Denmark}
\email{kwernli@imada.sdu.dk}

\thanks{
P.M. acknowledges partial support %of the SNF Grant No. 200021\_227719, 
of the Simons Collaboration on Global Categorical Symmetries, of the Simons Foundation Travel Support Grant, and of the FIM at ETH Zurich. P.M.'s research was (partly) supported by the NCCR SwissMAP, funded
by the Swiss National Science Foundation.
K. W. was supported by the ERC-SyG project, Recursive and Exact New Quantum Theory (ReNewQuantum) with funding from the European Research Council under the European
Union's Horizon 2020 research and innovation programme, grant agreement no.
810573, and by a Sapere Aude Grant from the Independent Research Fund Denmark, grant ID 10.46540/4251-00015B.}
\date{\today}
\begin{document}
\begin{abstract}
%\sout{Using the Batalin-Vilkovisky formalism, one can define the perturbative partition function of Chern-Simons theory on all compact oriented 3-manifolds around any reference flat connection. We recall this definition in detail and analyze the dependence of the partition function on the reference flat connection. We show that the BV cohomology class of the Chern-Simons partition function descends to the moduli space. On smooth components of the moduli space, this cohomology class is the Taylor expansion of a density with values in formal power series. YADA YADA YADA}

 %\marginpar{\bl another version}
We study the perturbative path integral of Chern-Simons theory (the effective BV action on zero-modes)  in Lorenz gauge, expanded around a (possibly non-acyclic) flat connection, as a family over the smooth irreducible stratum $\M’\subset \M$ of the moduli space of flat connections.
We prove that it is horizontal with respect to the Grothendieck connection up to a BV-exact term. From it, 
we construct a volume form on $\M’$ — the ``global partition function’’ — whose cohomology class is independent of the metric, and so is a 3-manifold invariant. 

As an element of the construction, we construct an extension of the perturbative partition function to a nonhomogeneous form on the space of triples $(A,A’,g)$ consisting of (1) a ``kinetic’’ flat connection $A$ around which Chern-Simons action is expanded, (2) a ``gauge-fixing’’ flat connection $A’$, (3) a metric $g$. This extension is horizontal with respect to an appropriate Gauss-Manin superconnection (which involves the BV operator as a degree zero component).

\end{abstract}
\maketitle
\setcounter{tocdepth}{3}
\tableofcontents
%\marginpar{Remove pert everywhere}
\section{Introduction}
The Chern-Simons field theory has been a major focus of interest of the mathematical physics community since the discovery of its close links to invariants of knots and 3-manifolds,
%knot invariants, 
both in non-perturbative \cite{Witten1989}, \cite{Reshetikhin1991} and perturbative \cite{Froehlich1989}, \cite{Kontsevich1993} treatments of the theory. 
An early breakthrough by Axelrod and Singer \cite{Axelrod1991}, \cite{Axelrod1994} was the result that the perturbative series at acyclic flat connections is well defined and yields topological invariants of the spacetime three-manifold equipped with framing. In this work, we generalize this result to smooth irreducible components of the moduli space: We show that these carry a volume form (valued in formal power series) %that 
whose cohomology class (total volume of the component)
is a topological invariant of the (framed) spacetime three-manifold. 
%all of this is very important (hence the title) and one should be rewarded with infinite grant money for it. 
Somewhat surprisingly, in the construction of this volume form, extra corrections beyond the usual Feynman diagrams are needed.

%\subsection{Main Results}

\subsection{Perturbative Chern-Simons partition function at a non-acyclic flat connection}
Fix a closed oriented 3-manifold $M$ and a compact simply-connected matrix Lie group $G$ with Lie algebra $\g$. %Let $\mc{P}$ be the trivial $G$-bundle over $M$.

We consider Chern-Simons theory, defined classically by the action functional
\begin{equation}\label{intro S_CS}
    S_{CS}(A)=\int_M \mr{tr} \left( \frac12 A\wedge dA +\frac16 A\wedge [A,A]  \right)
\end{equation}
on the space of connections $A$ in the trivial $G$-bundle $P$ on $M$, which are identified with $\g$-valued 1-forms on $M$. The critical points of $S_{CS}$ are \textit{flat} connections.

In Batalin-Vilkovisky (BV) formalism, one replaces (\ref{intro S_CS}) by the ``master action'' given by the same formula, but with field 
%$\mc{A}$ is now a nonhomogeneous $\g$-valued form on $M$, see \cite{Alexandrov1997}.
$A$ replaced with a nonhomogeneous $\g$-valued form  $\mc{A}$ on $M$, see \cite{Alexandrov1997}.

\subsubsection{Path integral heuristics}
We are interested in the Chern-Simons path integral over gauge equivalence classes of connections
\begin{equation}\label{intro CS path integral non-BV}
    \int_{\mr{Conn}(P)/\mr{Gauge}} \mc{D}A\; e^{\frac{i}{\hbar}S_{CS}(A)}.
\end{equation}
%which in BV formalism is replaced by the integral over a 
The perturbative (stationary phase) contribution of an acyclic flat connection to the $\hbar\ra 0$ asymptotics of  (\ref{intro CS path integral non-BV}) was studied in \cite{Witten1989} (one-loop approximation) and \cite{Axelrod1991}, \cite{Axelrod1994} (higher-loop contributions).

Given a non-acyclic flat connection $A_0$, one can decompose fields in the neighborhood of $A_0$ as $\mc{A}=A_0+\sfa+\alpha_\mr{fl}$, with $\sfa$ a $d_{A_0}$-cohomology class (represented by a harmonic form) and $\alpha_\mr{fl}$ a fluctuation. Then, one considers the path integral 
\begin{equation}\label{intro: CS path integral}
    Z_{A_0}(\sfa)= \int \mc{D}\alpha_\mr{fl}\; e^{\frac{i}{\hbar}S_\mr{CS}(A_0+\sfa+\alpha_\mr{fl})}
\end{equation}
where the integration is over field fluctuations -- $d^*_{A_0}$-exact forms $\alpha_\mr{fl}\in \Omega^\bt(M,\g)$.\footnote{Here $d^*_{A_0}$-exactness is the $A_0$-twisted Lorenz gauge condition and the fact that $\alpha_\mr{fl}$ is allowed to be a nonhomogeneous differential form is the AKSZ-BV gauge-fixing mechanism. In terms of Faddeev-Popov ghosts $c,\bar{c}$, the degree zero component of $\alpha_\mr{fl}$ is the ghost $c$ and the degree two component is $d^*_{A_0}\bar{c}$.}
In the case $A_0=0$, perturbative expansion for the path integral (\ref{intro: CS path integral}) was constructed and studied in \cite{Cattaneo2008}, as an effective BV action in $\sfa$.  For nonzero $A_0$ the construction is spelled out in \cite{Wernli2022}.

%\marginpar{Sep 26: This paragraph is very vague {\bl I think it's ok for the intro}}
A related idea is that in the path integral (\ref{intro CS path integral non-BV}) one might want to consider a tubular neighborhood of the moduli space of flat connections and integrate over fibers, producing a volume form of the moduli space.

%\marginpar{Sep 26: This paragraph is no longer up to date {\bl I updated it a bit}}
In this paper we will be denoting the perturbative evaluation of (\ref{intro: CS path integral}) $Z_{A_0}(\sfa)$ and the volume form on the moduli space as above $Z^\mr{glob}$ -- we will define both objects mathematically, without reference to heuristic path integral expressions. They are linked by $$Z^\mr{glob}=Z|_{\sfa=0}+\mr{correction\; terms},$$ 
see Section \ref{sss intro Zglob}.

\subsubsection{Mathematical definition of the perturbative partition function}
For $A_0$ any flat connection %in the trivial $G$-bundle 
on $M$, adapting the construction of \cite{Cattaneo2008}, one defines 
%(following the ideas of Axelrod-Singer \cite{Axelrod1991} and \cite{Cattaneo2008}) 
the perturbative Chern-Simons partition function as 
%\marginpar{should we mention that this is the BV partition function alreday here?}
\begin{multline}\label{Zpert}
    Z_{A_0}(\sfa)=e^{\frac{i}{\hbar}S_{CS}(A_0)} \tau(A_0)^\frac{1}{2} e^{\frac{\pi i}{4}\psi(A_0)} \exp \sum_\Gamma \frac{(-i\hbar)^{-\chi(\Gamma)}}{|\mr{Aut}(\Gamma)|}\Phi_{\Gamma,A_0}(\sfa) \\ 
    \in \mr{Dens}^{\frac12,\mr{formal}}(H^\bt_{A_0}[1])=\mr{Det}^\frac12(H^\bt_{A_0})\otimes \wh{\mr{Sym}}(H^\bt_{A_0}[1])^*
\end{multline}
%\marginpar{To get the motivating example correct, we should include the correct prefactors of $\hbar$, and probably also the framing correction.}
-- a formal half-density on de Rham cohomology twisted by $A_0$. Here:
\begin{itemize}
\item $H_{A_0}^\bt$ is the cohomology of the complex of $\g$-valued differential forms on $M$ with differential $d_{A_0}=d+\mr{ad}_{A_0}$. One calls the variable $\sfa\in H_{A_0}^\bt[1]$ the \emph{zero-mode}.
%\item $S_{CS}(A)=\int_M \mr{tr} \left( \frac12 A\wedge dA +\frac16 A\wedge [A,A]  \right)$ is the Chern-Simons action functional.
\item $\tau(A_0)$ is the Ray-Singer torsion. For $A_0$ non-acyclic, rather than being a number, it is an element of the determinant line of the cohomology $H^\bt_{A_0}$, see Remark \ref{rem: torsion in Gaussian integrals} below.
\item $\psi(A_0)$ is the Atiyah-Patodi-Singer eta-invariant of the %Dirac 
    operator \\${L_-\colon=*d_{A_0}+d_{A_0}*}$ acting on forms of odd degree.
\item The sum ranges over connected 3-valent graphs (``Feynman graphs'') $\Gamma$ with leaves (loose half-edges) allowed. $\chi(\Gamma)$ is the Euler characteristic of the graph. %and $\mr{Aut}(\Gamma)$ is the automorphism group. 
The weight $\Phi_{\Gamma,A_0}(\sfa)$ of a graph $\Gamma$ is a polynomial in $\sfa$ with coefficients given by certain integrals over the compactified configuration space of points on $M$, of a form defined in terms of Hodge decomposition data on $M$, defined by a metric on $M$ and twisted by the local system $A_0$.
\end{itemize}
We refer to Section \ref{ss perturbative partition function} for full details, in particular for the formula for Feynman weights $\Phi_{\Gamma,A_0}(\sfa)$.

Some elements of the formula (\ref{Zpert}) depend on the choice of a Riemann metric on $M$ (namely, the eta-invariant and Feynman weights). The dependence of the full object $Z$ on metric -- 
%modulo the framing anomaly 
with an appropriate renormalization factor included
-- turns out to be BV-exact, see Section \ref{sec:intro metric}.
%Theorem \ref{thm:metric dependence intro}.
%\marginpar{\bl wrong reference: that thm talks about $Z^{glob}$}

\subsubsection{Aside:  BV pushforward perspective
%Why is it natural for $Z$ to be a half-density on the space of zero-modes
}

We briefly recall the BV pushforward construction which in particular elucidates: 
\begin{enumerate}[(i)]
    \item why one should expect $Z$ to be a half-density on the space of zero-modes and
    \item why one should expect $Z$ to change by a BV-exact term when the metric on $M$ is deformed.
\end{enumerate}

%\begin{rem}
    Recall that in the BV formalism, one has a construction of a BV pushforward, or fiber BV integral:\footnote{See, e.g., \cite{mnev2019quantum}.} Let 
    \begin{equation}\label{intro V=V'xV''}
    V=V'\times V''
    \end{equation}
    be a degree $(-1)$-symplectic manifold (``space of fields'') presented as a product of degree $(-1)$-symplectic manifolds (``slow/infrared fields'' and ``fast/ultraviolet fields'') and $\mc{L}\subset V''$ be a Lagrangian submanifold. Then one has a BV pushforward map from half-densities on all fields to half-densities on slow fields
    \begin{equation}
        P_* \stackrel{\mr{def}}{=} \mr{id}_{V'}\otimes \int_{\mc{L}}\colon\quad  \mr{Dens}^{\frac12}(V) \ra \mr{Dens}^{\frac12}(V').
    \end{equation}
 By BV version of Stokes' theorem \cite{Schwarz1993}, one has that
 \begin{enumerate}[(a)]
     \item $P_*$ is a chain map w.r.t. the BV Laplacians on half-densities:
     \begin{equation}
     \Delta' P_*= P_* \Delta,
     \end{equation}
     with $\Delta,\Delta'$ the BV Laplacian on the half-densities on $V$ and on $V'$ respectively.
     \item %Assume for simplicity that $V,V',V''$ are graded vector spaces and (\ref{intro V=V'xV''}) is a linear splitting.
     Denote the inclusion of $V'$ into $V$ in the splitting (\ref{intro V=V'xV''}) by $i$ and the projection of $V$ onto $V'$ by $p$. 
     Then, for an infinitesimal deformation 
     of $i,p$ and the Lagrangian $\mc{L}$,
     %$i\ra i+\delta i,\;p\ra p+\delta p$ and of the Lagrangian $\mc{L}\mapsto \mr{graph}(d\Psi)$, with $\Psi\in C^\infty(\mc{L})_{-1}$ 
     %If the Lagrangian $\mc{L}$ is deformed into another Lagrangian $\til{\mc{L}}$ by a Lagrangian homotopy, then 
     the induced variation of the BV pushforward is $\Delta'$-exact,
     \begin{equation} \label{intro delta P_*}
    % P_*^{\til{\mc{L}}} \alpha = P_*^{\mc{L}} \alpha + \Delta'(\cdots)
    \delta_{i,p,\mc{L}} \left(P_* \alpha\right) = \Delta'R,
     \end{equation}
     for any fixed $\alpha\in \mr{Dens}^{\frac12}(V)$ satisfying $\Delta\alpha=0$, with the generator 
     %$R=P_*\left(\alpha\cdot (\omega(\phi,\delta i(\phi'))+\frac12 \Psi)\right)$
     $R$ given explicitly in terms of the variation of $i,p,\mc{L}$, see \cite{Cattaneo2008}, \cite{Cattaneo2017}.
 \end{enumerate}
%\end{rem}

The two properties of $P_*$ above are a theorem in the finite-dimensional case; for infinite-dimensional BV pushforward (defined via perturbative path integral) they become a heuristic statement -- an expectation -- that has to be proven independently at the level of Feynman diagrams.

In the example of Chern-Simons theory, restricted to perturbations of a fixed flat connection $A_0$, we have $V=\Omega^\bt(M,\g)[1]$ with $V'$ being the $A_0$-harmonic forms and $V''$ their orthogonal complement (w.r.t. the Hodge inner product), with 
\begin{equation}\label{intro L=im(d^*_A0)}
\mc{L}=\mr{im}(d^*_{A_0})
\end{equation}
being the coexact forms. Then one has a function on $V$,
\begin{multline}\label{intro S_CS(S_0+B)}
    f(B)\colon=S_\mr{CS}(A_0+B)\\
    =S_\mr{CS}(A_0)+\int_M\mr{tr}\left(\frac12 B\wedge d_{A_0} B+\frac16 B\wedge [B,B]\right).
\end{multline}
As a function of $B$ it satisfies the BV classical master equation $\{f,f\}=0$. Denoting by $\mu_0$ the formal translation-invariant half-density on $V$,  one has 
$\Delta (e^{\frac{i}{\hbar}f}\mu_0 )=0$ where the l.h.s. should be appropriately regularized \cite{Costello2011}. Then the perturbative partition function (\ref{Zpert}) is the perturbative evaluation of the BV pushforward
\begin{equation}\label{intro BV pushforward}
    Z= P_*\left(e^{\frac{i}{\hbar}f}\mu_0\right)
\end{equation}
for the gauge-fixing Lagrangian $\mc{L} = \mr{im}(d^*_{A_0})$ -- this is the origin of the Chern-Simons path integral (\ref{intro: CS path integral}).
In particular, from this viewpoint it is natural that $ Z$ is a \emph{half-density} (rather than a function) on $V'$.

%\marginpar{Remark added April 8, 2026}
\begin{rem}\label{rem: torsion in Gaussian integrals}
The free Gaussian part  of the BV pushforward (\ref{intro BV pushforward}) is 
% the path integral
% \begin{equation}
%     \int_{\mr{im}(d^*_{A_0})}\mc{D}\alpha_\fl \, e^{\frac{i}{\hbar} }
% \end{equation}
the formal Gaussian integral of $e^{\frac{i}{\hbar} \int_M \mr{tr}\frac12 B\wedge d_{A_0} B}$ which is interpreted (as guided by finite-dimensional oscillatory Gaussian integrals) as 
\begin{equation}\label{intro Gaussian integral as reg det times reg phase}
``{({\det}^\mr{reg} d_{A_0})^\frac12 e^{\frac{\pi i}{4}\mr{sign}^\mr{reg}(d_{A_0})},}"
\end{equation}
with ``${\det}^\mr{reg} d_{A_0}$'' the gauge-fixed (gauge-fixing can be thought of replacing $d_{A_0}$ with $(d^*_{A_0}d_{A_0})^{\frac12}$ and then restricting it to coexact forms) and zeta-regularized determinant of $d_{A_0}$, a.k.a. the Ray-Singer torsion. We refer to \cite{Schwarz1978}, \cite{Witten1989}, \cite[Section 1.7]{mnev2014lecture}, \cite{cattaneo2020cellular} for details on Ray-Singer torsion $\tau(A_0)$ (or Reidemeister torsion in the combinatorial setting of \cite{cattaneo2020cellular}) and how it appears as a Gaussian integral with kinetic operator given by de Rham operator twisted by a flat connection. The gauge-fixed and zeta-regularized signature of the kinetic operator in (\ref{intro Gaussian integral as reg det times reg phase}) is understood as the Atiyah-Patodi-Singer eta-invariant $\psi(A_0)$, see \cite{Witten1989}.

Ray-Singer torsion is the regularized version of the general algebraic notion of torsion of a cochain complex $C^\bt$ with inner product. This torsion $\tau(C^\bt)$ is an element of the determinant line $\mr{Det}\,C^\bt \cong \mr{Det}\,H^\bt$, with $H^\bt$ the cohomology of $C^\bt$, see \cite{milnor1966whitehead}, \cite{mnev2014lecture}. (In particular, if $C^\bt$ is acyclic, the torsion is a number.) In the case of $C^\bt$ being the de Rham complex twisted by a flat connection $A_0$, the torsion becomes the Ray-Singer torsion -- an element of the determinant line of twisted de Rham cohomology $H^\bt_{A_0}$, or, equivalently, a constant density on the graded space $H^\bt_{A_0}[1]$. Thus, the square root of torsion participating in the perturbative Chern-Simons path integral (\ref{Zpert}) is a (constant) half-density on the space of zero-modes $H^\bt_{A_0}[1]$. We refer also to \cite{bismut1992extension}, \cite{braverman2003new} for a discussion of Ray-Singer torsion for non-acyclic flat connections.
\end{rem}

Note that the property (\ref{intro delta P_*}) suggests that under a deformation of Riemannian metric on $M$, $Z$ should change by a $\Delta'$-exact term.

\begin{rem}%[\bl framing anomaly]
%\marginpar{\bl Oct 4: I rephrased this slightly, to avoid saying that non-renormalized $Z$ does not change by $\Delta'(\cdots)$ under changes of $g$ (which is false by vanishing of $\Delta'$-cohomology)}
\label{rem: def Z pert framing}
    %The expression \eqref{Zpert} may seem complicated, but it is completely natural from a the point of view of perturbative path integrals. However, two comments are in order: 
%\marginpar{\bl this needs to be said more carefully: it is invariant up to $\Delta'$-exact terms for a stupid reason of vanishing of $\Delta'$-cohomology}
There is a correction to the expected statement above 
% due to a quantum effect 
-- a path integral phenomenon, not visible at the level of finite-dimensional integrals: The partition function \eqref{Zpert} exhibits anomalous dependence on metric. For an acyclic flat connection $A_0$, 
%Contrary to the expectation from the finite-dimensional case, the partition function \eqref{Zpert} is \emph{not} invariant under change of metric, not even up to $\Delta'$-exact terms. 
         at the 1-loop level, as already observed by Witten \cite{Witten1989}, this is due to the fact that the eta invariant depends on the metric. At higher loop orders this phenomenon is due to contributions from hidden boundary strata of %\marginpar{Sep 26: This is not true for theta graph where there are only two points}
         compactified configuration spaces, 
         %(corresponding to collapse of at least three points), 
         as observed by Axelrod and Singer \cite{Axelrod1991}. One can cancel this dependence on the metric at the cost of ``renormalizing'' the partition function by multiplying it with a factor that depends on the metric and a framing (trivialization of the tangent bundle of $M$). The resulting renormalized partition function is independent of metric but depends on the framing
         --- this is the well-known \emph{framing anomaly} of Chern-Simons theory.
         %\marginpar{\bl Oct 4. I'd say ``framing anomaly'' is the dependence of $Z^\mr{ren}$ on framing.} 
         In the case of non-acyclic $A_0$, as it turns out, one needs to include the same renormalization factor, and then this renormalized $Z$ changes under the variation of metric by a BV-exact term.
         %We will comment more on this problem and how to fix it in
          We refer to
          Section \ref{sec:intro metric} below for details (see also Appendix \ref{sec: metric}). 
         %\marginpar{Be more explicit about framing anomaly.}
         \end{rem}
         \begin{rem}\label{rem: def Z pert normalization}
        Ultimately, the goal of this activity  
        %\marginpar{\bl look for a better wording?}
        is to compare the perturbative Chern-Simons partition function and the asymptotics of the Reshetikhin-Turaev invariants \cite{Reshetikhin1991}. Experiments in the literature have shown \cite{Freed1991}, \cite{Jeffrey1992}, \cite{Rozansky1995} that to this end one needs to do two things: 
        \begin{enumerate}[a)]
        \item 
        Be more careful in the normalization of the path integral measure.\footnote{E.g. in the quantum mechanics path integral for a particle in $\R^d$, the ``correct''  measure on paths is $\prod_t \frac{dp(t)dq(t)}{\sqrt{2\pi\hbar}^d}$, rather than the Lebesgue measure that we are implictly using here.} 
        \item Refine the framing correction to 2-framings and use the \emph{canonical} 2-framing.\footnote{2-framings are trivializations of $TM \oplus TM$, introduced by Atiyah \cite{Atiyah1990}. The canonical 2-framing $\alpha$ is the one for which the Hirzebruch defect $\mr{sign}(Y) - \frac16p_1(2TM,\alpha) =0$, where $Y$ is any 4-manifold with boundary $M$ and $p_1(2TM,\alpha)$ is the relative Pontryagin number of the bundle $2TM$ over $M \times I$, trivialized by $\alpha$ over the endpoints of the cylinder.  }
        \end{enumerate} 
        In this paper we will largely ignore these questions and only comment on them briefly in the motivating example in Section \ref{sec:intro example}. 
  
\end{rem}
%subject to the $A_0$-twisted Lorenz gauge condition $d^*_{A_0} \alpha_\mr{fl}=0$.\footnote{}

%{\color{red} [Remark/reminder on BV pushforwards: why is it natural for $Z$ to be a half-density]}

\subsubsection{$Z$ as a family over the moduli space of flat connections.}
Let $\M$ be the moduli space of flat $G$-connections on $M$ and $\M'\subset \M$ the smooth irreducible locus.
%(we assume that the pair $M,G$ is such that this locus is open dense in $\M$).
%\marginpar{\bl Oct4: do we need this assumption? (except maybe for comparison with RT)}
%\marginpar{\red Do we have any examples where this assumption is satisified? Is $\Sigma \times S^1$ an example?}
\footnote{A point $[A_0]\in \M$ is ``smooth'' if 
the minimal model of the 
dg Lie algebra $(\Omega^\bt(M,\g),d_{A_0},[-,-])$ is the cohomology $H^\bt_{A_0}$ with \emph{vanishing} $L_\infty$ operations 
%the dg Lie algebra $\Omega^\bt(M,\g),d_{A_0},[-,-]$ is quasi-isomorphic, as an $L_\infty$ algebra, to the cohomology $H^\bt_{A_0}$ with \emph{vanishing} $L_\infty$ operations 
(which implies that $\M$ is locally a manifold around $[A_0]$). A flat connection is irreducible if $H^0_{A_0}=0$.}\footnote{
For comparisons with nonperturbative answers in Chern-Simons theory one may want to assume that the pair $M,G$ is such that $\M'\subset \M$ is an open dense subset. However, results of this paper don't need this assumption.}

The partition function (\ref{Zpert}) depends only on the gauge equivalence class $[A_0]$ of the flat connection $A_0$, and thus defines a section of the bundle of formal vertical (i.e., fiberwise) half-densities on the graded vector bundle $\mathbb{T}\M'$ over $\M'$, where the fiber of $\mathbb{T}\M'$ over $[A_0]$ is $H^\bt_{A_0}[1]$ (in particular, the degree zero part of $\mathbb{T}\M'$ is the tangent bundle of $\M'$):
\begin{equation}
    Z\in \Gamma(\M',\mr{Dens}^{\frac12,\mr{formal}}(\mathbb{T}\M')).
\end{equation}
%the moduli space of flat connections $\mc{M}$.

We remark that on $\M'$: 
\begin{itemize}
\item The exponential factor in the partition function (\ref{Zpert}) is $1+O(\hbar)$ (tree Feynman graphs vanish%cancel out
due to smoothness of $[A_0]$); $S_{CS}(A_0)$ and $\psi(A_0)$ are locally constant functions on $\M'$.\footnote{In fact, $S_{CS}(A_0)$ and $\psi(A_0)$ are locally constant on the entire moduli space $\M$, including singular/reducible strata.}
\item By irreducibility of $[A_0]$ and by Poincar\'e duality, one has $H^0_{A_0}=H^3_{A_0}=0$ and $H^2_{A_0}\cong (H^1_{A_0})^*$. Therefore, vertical half-densities on $\mathbb{T}\M'$ are naturally identified with vertical 1-densities on the tangent bundle $T\M'$ and in turn, using an orientation\footnote{One has a natural orientation on $\M'$, see \cite[Theorem 4.5]{joyce2020orientations}. %defined by the Ray-Singer torsion. 
%{\color{red} Is that correct? We usually define torsion modulo sign...}
} on $\M'$, with vertical top-degree forms on $T\M'$.
\end{itemize}

\subsubsection{The global partition function}\label{sss intro Zglob}
Restriction of the perturbative partition function (\ref{Zpert}) to the zero-section of $\mathbb{T}\M'$ (i.e., setting $\sfa=0$)  yields the ``naive global partition function''
\begin{equation}\label{Zglob naive intro}
    Z^\mr{glob,naive}_{A_0}=Z_{A_0}(\sfa=0) \quad \in  \mr{Dens}^{\frac12}_{\mr{base}}(T^*[-1]\M')\cong \Omega^\mr{top}(\M')
\end{equation}
where $\mr{Dens}^{\frac12}_{\mr{base}}(T^*[-1]\M')$ denotes half-densities on the shifted cotangent bundle that are independent of the fiber coordinates. 
It is given by the same formula as (\ref{Zpert}), where the sum over graphs ranges over trivalent graphs without leaves.

One of the main results of this work is that one can modify (\ref{Zglob naive intro}), by adding certain explicit corrections,
%\footnote{See (\ref{Zglob formula}), (\ref{Zglob path integral formula}), (\ref{Zglob Feynman diagram expansion}) in Section \ref{ss Zglob}. The corrections correspond to graphs with $\geq 1$ dashed edges in the conventions of Figure \ref{fig:Zglob graph}.} 
to
\begin{equation}
    Z^\mr{glob}_{A_0}= Z^\mr{glob,naive}_{A_0} (1+O(\hbar)),
\end{equation}
in such a way that:
\begin{enumerate}[(a)]
    \item 
    With the renormalization %framing anomaly 
    factor included (as in (\ref{eq: renormalization})), $Z^\mr{glob}$ defines a cohomology class of $\M'$, which is \emph{independent of the metric} (Theorem \ref{thm:metric dependence intro}/Theorem \ref{thm: coho class invariant}). 
    In particular, if $\{\M'_\alpha\}$ are the connected components of $\M'$, then the collection $\{\int_{\M'_\alpha} Z^\mr{glob,ren}\}$ of elements of $e^{\frac{ic_\alpha}{\hbar}}\mathbb{C}[[\hbar]]$ is an invariant of a framed 3-manifold, where $c_\alpha=S_{CS}|_{\M'_\alpha}$.
    \item The pullback of $Z^\mr{glob}$ by the formal exponential map on $\M'$ recovers the perturbative partition function (\ref{Zpert}) up to a BV-exact term (Theorem \ref{intro thm A} (\ref{intro thm A (c)})/Corollary \ref{cor: Tphi^* Zglob=Z+Delta(...)}).
\end{enumerate}

The construction of $Z^\mr{glob}$ is as follows:
\begin{itemize}
    \item First, one extends the perturbative partition function (\ref{Zpert}) to a nonhomogeneous form $\til{Z}$ on the moduli space $\M'$ with values in $\mr{Dens}^{\frac12,\mr{formal}}(H_{A_0}^\bt[1])$. 
    %\marginpar{\bl Oct 4: there is also %$\langle\sfa^2, [\delta A_0]\rangle$ term and the $\Theta$-term}
    This extension is constructed by taking the formula (\ref{Zpert}) and changing the assignment to edges and leaves of a Feynman graph to appropriate objects\footnote{
    The extended propagator $\wh{K}$ and the extended zero-mode inclusion $\wh{i}(\sfa)$, cf. (\ref{ihat, phat, Khat, Phihat}) with $A'=A=A_0$. One also needs to include a special graph consisting of a single edge with the weight $\frac{i}{2\hbar} \langle \sfa, \wh\Theta(\sfa)\rangle$, with $\wh\Theta$ as in (\ref{ihat, phat, Khat, Phihat}).
    } valued in $\Omega^\bt(\M')$.
    \item Then one constructs $Z^\mr{glob}$ as 
    %\marginpar{\bl Dec 23: sign corrected}
    \begin{equation}\label{Zglob formula intro}
        Z^\mr{glob}=\left(\sum_{k\geq 0}\frac{(-i\hbar)^k}{k!}\left\langle \frac{\partial}{\partial \sfa^2},\frac{\partial}{\partial [\delta A_0]} \right\rangle^k \til{Z} \right)\Bigg|_{\sfa^1=[\delta A_0]=0}.
    \end{equation}
    Here $\sfa^{1,2}$ are the components of $\sfa$ in $H^1_{A_0}$, $H^2_{A_0}$.
\end{itemize}
The term $k=0$ in (\ref{Zglob formula intro}) is $Z^\mr{glob,naive}$, and $k\geq 1$ terms are the corrections we referred to above. We refer to Section \ref{ss Zglob} for details on the construction and properties of $Z^\mr{glob}$.

\subsection{Dependence of the perturbative partition function on the flat connection: horizontality of $Z$, recovering $Z$ from $Z^\mr{glob}$}%\leavevmode \\
\subsubsection{The sum-over-trees formal exponential map on $\M'$.}
One can define a map %``formal exponential map''
\begin{equation}
    \phi\colon V \ra \M'
\end{equation}
where $V$ is an open neighborhood of the zero-section of the tangent bundle $T\M'$ such that (a) the restriction of $\phi$ to the zero-section is the identity map $\M'\ra \M'$ and (b) the vertical component of the differential $d\phi$ on the zero-section is identity. Such a map $\phi$ is called, in the language of formal geometry, a ``formal exponential map.''\footnote{More precisely: the formal exponential map is the vertical $\infty$-jet of $\phi$. See e.g. \cite{Bonechi2012}.}

One can define a particular formal exponential map $\phi$ explicitly, as a sum over binary rooted trees (modulo isomorphism) with leaves decorated by $\sfa\in T_{[A_0]}\M'=H^1_{A_0}$, edges (and the root) decorated by Hodge chain homotopy 
%(\ref{K intro}) 
and vertices decorated by the Lie bracket in $\Omega^\bt(M,\g)$.
%\marginpar{\bl refer to a formula in the main text for details}

\subsubsection{Main result 1}
%{\bl For the following result, we make an extra assumption that smooth irreducible flat connections $A_0$ additionally satisfy ``1-extended smoothness,'' see Definition \ref{def: extended smoothness} and Remark \ref{rem: 1-extended smoothness}: not only the $L_\infty$ algebra on $H_{A_0}^\bt$ induced from $\Omega^\bt(M,\g)$ vanishes, but also $L_\infty$ automorphisms of $H_{A_0}^\bt$ induced by variations of the homotopy transfer data (variations of gauge-fixing) vanish.}
\begin{thm}\footnote{
This is Proposition \ref{prop ulZ finite horizontality}, Corollary \ref{cor ulZ nabla^G horizontality} and Corollary \ref{cor: Tphi^* Zglob=Z+Delta(...)} put together.
} \label{intro thm A}
\begin{enumerate}[(a)]
    \item \label{intro thm A (a)}
    The perturbative partition function (\ref{Zpert}) satisfies %\marginpar{Can write $[A_0]$ instead of $A_0$. {\bl maybe better not}}
    \begin{equation}\label{intro: Z finite horizontality eq}
       % Z(\phi(A_0,\alpha), d^\mr{vert}\phi|_{(A_0,\alpha)}(a)) \cdot \det \left(d^\mr{vert}\phi|_{(A_0,\alpha)}\right)
     \det(B^\vee)\circ  Z_{\phi(A_0,\alpha)}( B(\sfa))
        = Z_{A_0}(\alpha+\sfa)-i\hbar\Delta_\sfa  R(A_0,\alpha;\sfa)
    \end{equation}
    for any smooth irreducible flat connection $A_0$.\footnote{To lighten the notations we write $A_0$ instead of $[A_0]$ for a point in $\M'$.} Here: 
    \begin{itemize}
   \item  $\sfa,\alpha \in H^1_{A_0}$  are formal variables; 
   \item $\phi$ is the sum-over-trees formal exponential map; 
   \item $B\colon=d^\mr{vert}\phi|_{(A_0,\alpha)}\colon H^1_{A_0}\ra H^1_{\phi(A_0,\alpha)}$ is the differential of $\phi$ in the second %(tangent fiber) 
    argument; the determinant of the dual of $B$ is a map between determinant lines $\det (B^\vee)=\wedge^\mr{top}B^\vee\colon \mr{Det}(H^1_{\phi(A_0,\alpha)})^*\ra \mr{Det}(H^1_{A_0})^*$;
    %\footnote{
    %\bl Throughout the paper, we use the superscript $*$ or $\vee$ interchangeably to denote duals of spaces and maps.
    %}
    \item $\Delta_\sfa = \left\langle \frac{\partial}{\partial \sfa^{1}},\frac{\partial}{\partial \sfa^{2}}\right\rangle$, with brackets meaning the Poincar\'e duality pairing, is the BV Laplacian on formal half-densities on the fiber of $\mathbb{T}\M'$ over $A_0$; %\marginpar{\red April 8 26: Added formula for $\Delta_a$}
    \item $R(A_0,\alpha;\sfa)$ is some degree $-1$ formal half-density on the fiber of $\mathbb{T}\M'$ over $A_0$ (in a family parametrized by $\alpha$). 
    \end{itemize}
    %\marginpar{We are thinking of $A_0,\alpha$ as fixed and $\sfa$ as variable..}
    \item \label{intro thm A (b)}
    The formal exponential map induces a flat connection $\nabla^G$ (``Grothendieck connection'')
%    \marginpar{\bl is $\nabla^G$ flat? or only flat up to $\Delta(\cdots)$? (follows from existence of a horizontal section)}
    on the bundle of formal fiberwise  half-densities on $\mathbb{T}\M'$, and the perturbative partition function is a horizontal section modulo a $\Delta_\sfa$-exact term:
    \begin{equation}\label{intro: horizontality eq}
        \nabla^G Z= -i\hbar\Delta_\sfa R_1
    \end{equation}
    with some degree $-1$ generator $R_1\in \Omega^1(\M',\mr{Dens}^{\frac12,\mr{formal}}(\mathbb{T}\M'))$.
    \item \label{intro thm A (c)} 
    Under 
    % \emph{1-extended smoothness assumption} (Definition \ref{def: extended smoothness}),\footnote{
    % The assumption is that, for smooth irreducible flat connections $A_0$,
    % not only the $L_\infty$ algebra on $H_{A_0}^\bt$ induced from $\Omega^\bt(M,\g)$ vanishes, but also $L_\infty$ automorphisms of $H_{A_0}^\bt$ induced by variations of the homotopy transfer data (variations of gauge-fixing) vanish. See Definition \ref{def: extended smoothness} and Remark \ref{rem: 1-extended smoothness}.
    % }
    one can recover $Z$ from $Z^\mr{glob}$ modulo a $\Delta_\sfa$-exact term, as
    \begin{equation}\label{intro phi^* Z^glob = Z^pert+Delta(...)}
        \mathbf{T} (\phi^*)^\mr{vert} Z^\mr{glob} = Z+i\hbar\Delta_\sfa R_\mr{glob-pert},
    \end{equation}
    On the left, $(\phi^*)^\mr{vert}$ stands for the fiberwise top form on $V\subset T\M'$ obtained from the pullback of a top form on $\M'$;
    $\mathbf{T}$ stands for taking the Taylor expansion in the fiber coordinates on $T\M'$. The resulting formal fiberwise top form on $T\M'$ is reinterpreted as a formal fiberwise degree zero half-density on $\mathbb{T}\M'$. $R_\mr{glob-pert}$ is some degree $-1$ formal fiberwise half-density on $\mathbb{T}\M'$.
\end{enumerate}  
%\marginpar{\bl Nov 11: 1-ext smoothness assumption removed}
% {\bl Parts (\ref{intro thm A (b)}), (\ref{intro thm A (c)}) above rely on the
%  \emph{1-extended smoothness assumption} (Definition \ref{def: extended smoothness}),}\footnote{
%     The assumption is that, for smooth irreducible flat connections $A_0$,
%     not only the $L_\infty$ algebra on $H_{A_0}^\bt$ induced from $\Omega^\bt(M,\g)$ vanishes, but also $L_\infty$ automorphisms of $H_{A_0}^\bt$ induced by variations of the homotopy transfer data (variations of gauge-fixing) vanish. See Definition \ref{def: extended smoothness} and Remark \ref{rem: 1-extended smoothness}.
%     }
\end{thm}

Equation (\ref{intro: Z finite horizontality eq}) expresses the fact (expected from the heuristic formula (\ref{intro: CS path integral})) that a shift $\alpha$ of the flat connection $A_0$ can be absorbed into a shift of the zero-mode $\sfa$, modulo a 
%BV canonical transformation. 
BV-exact term.
Formula (\ref{intro: horizontality eq}) expresses the same fact infinitesimally (to first order) in the shift $\alpha$.

Equation (\ref{intro phi^* Z^glob = Z^pert+Delta(...)}) 
%is a specialization of (\ref{intro: Z finite horizontality eq}) to $\sfa=0$.
%expresses 
is related to
the fact that $Z$ can be modified by a BV-exact term to a strictly global object (horizontal w.r.t. $\nabla^G$), see Theorem \ref{thm 5.15}.

The generator $R_1$ in (\ref{intro: horizontality eq}) is given explicitly as a sum over graphs with one marked edge or leaf,  cf. Corollary \ref{cor ulZ nabla^G horizontality}. 
%Proposition \ref{prop: variation of Z wrt A'}.\marginpar{\bl Nov 11: check that the ref is correct: we want to exclude trees with a marked element}
Generators $R(A_0,\alpha;\sfa)$ in (\ref{intro thm A (a)}) and $R_\mr{glob-pert}$ in (\ref{intro thm A (c)}) are also given explicitly by (\ref{R for Z underline finite almost-Grothendieck-horizontality}), (\ref{rho from Zmod=ulZ-Delta(rho)}).

\begin{rem}
\label{rem: Poincare lemma}
Cohomology of $\Delta_\sfa$ acting on (formal or smooth) half-densities on the odd-symplectic graded vector space $H_{A_0}[1]=T^*[-1]H^1_{A_0}$ is concentrated in ghost number $-\dim H^1_{A_0}$ and has rank one there. This is a consequence of Poincar\'e lemma, since the odd Fourier transform gives a chain isomorphism $\mr{Dens}^{\frac12}(T^*[-1]V)_{-k},\Delta_\sfa \cong \Omega^{\dim V-k}(V),d$, for $V=H^1_{A_0}$ a vector space. Thus, $H_{\Delta_\sfa}^{-k}$ is the de Rham cohomology of a point in degree $\dim V-k$.
%(where we take $V=H^1_A$).

Thus, if $\dim H^1_{A_0}>0$ (i.e., $A_0$ is not an isolated point of $\M'$), the perturbative partition function $Z$ (which is automatically $\Delta_\sfa$-closed for degree reason) is in fact $\Delta_\sfa$-exact. From this standpoint, statements like (\ref{intro: horizontality eq}), saying that something holds for $Z$ up to \emph{some} BV-exact term $\Delta_\sfa R$ might look trivial, since $Z$ is itself BV-exact. What makes these statements nontrivial is that (i) we give a formula for $R$, (ii) the statement holds in a family over $\M'$, with a coherent choice of $R$.
%the generator $R$ of the BV-exact term. 
More precisely, $Z$ possesses an extension to a nonhomogeneous form $\check{Z}$ on the space of background data, whose 0-from component is $Z$ and 1-form component is $R$, satisfying the ``differential quantum master equation,'' see Section \ref{ss intro:Zcheck} below.
%
%An analogy: saying that two volume forms $v_1$, $v_2$ on a manifold $X$ (of positive dimension) are related by a pullback by some diffeomorphism $f$ doesn't say anything about the relation of jets of $v_1$ at $x$ and of $v_2$ at $f(x)$  for any one point $x\in X$ (unless $f$ is specified). But if this statement is known to hold in a family over $X$, then we know, e.g., that $\int_X v_1=\int_X v_2$.
%Throughout this paper we are making statements type ``something holds for a variation of the partition function up to a BV-exact term.''    
\end{rem}

\subsection{``Desynchronized'' Chern-Simons partition function -- main result 2}
Parts (\ref{intro thm A (a)}), (\ref{intro thm A (b)}) of Theorem \ref{intro thm A} follow from an auxiliary statement on the ``desynchronized'' partition function which we explain below.

In the partition function (\ref{Zpert}), the flat connection $A=A_0$ played two different roles: it was the local system for the kinetic operator $d_{A_0}$ (cf. the quadratic term $B\wedge d_{A_0} B$ in (\ref{intro S_CS(S_0+B)})) and it was a parameter in the Lorenz gauge-fixing (\ref{intro L=im(d^*_A0)}). 

One can 
%To understand Theorem \ref{intro thm A}, it is useful to
allow the parameter in the kinetic operator $d_{A}$ %(we will write $A$ instead of $A_0$) 
and in the gauge-fixing operator $d^*_{A'}$ to be two different (but sufficiently close)
%\footnote{
%More precisely, the construction assumes that $(A_0,A_1)$ is in an open neighborhood $U$ of the diagonal in the space of pairs of smooth irreducible flat connections, with $U$ sufficiently thin that one has an analog of Hodge decomposition based on the operators $d_{A_0}, d^*_{A_1}$.
%}) 
flat connections. This leads to the ``desynchronized'' partition function $Z_{A,A'}(\sfa)$ which is given by the formula (\ref{Zpert}) with the following modification:
% \begin{itemize}
%     \item %\marginpar{EDIT}
    weights of Feynman graphs are based on a ``desynchronized'' analog of Hodge decomposition of forms on $M$, based in turn on the operators $d_{A}$, $d^*_{A'}$.
    %In the propagator (\ref{K intro}), one replaces  $d^*_{A_0}$ with $d^*_{A_1}$, the Laplacian $\Delta_{A_0}$ with the ``desynchronized'' Laplacian $\Delta_{A_0,A_1}=[d_{A_0},d^*_{A_1}]_+$, harmonic forms are replaced with forms in the kernel of $\Delta_{A_0,A_1}$ -- the ``$(A_0,A_1)$-harmonic forms''. The inclusion ${\red i}_{A_0}$ is replaced by the inclusion $\iota_{A_0,A_1}$ of $d_{A_0}$-cohomology classes as their $(A_0,A_1)$-harmonic representatives.
%     \item %Ray-Singer torsion is %and the eta-invariant are 
%     %The expression $e^{\frac{i\pi}{4}\psi(A_0)}\tau(A_0)^\frac12$ in (\ref{Zpert}) is
%     %replaced by its desynchronized variant $I_{A_0,A_1}$ (\ref{I def}).
%     Ray-Singer torsion is replaced by its desynchronized variant.
% \end{itemize}

The desynchronized partition function is still a formal half-density on $H^\bt_{A}[1]$. By construction, it satisfies the ``extension property'': the restriction of $Z_{A,A'}(\sfa)$ to the diagonal $A=A'$ coincides with $Z_{A}(\sfa)$, 
    \begin{equation}
        Z_{A,A}(\sfa) = Z_{A}(\sfa).
    \end{equation}

%One has a desynchronized variant of the sum-over-trees map $\phi(A,A',-)\colon H_A^1\ra \FC$.

We denote $\FC$ the space of flat connections and $\FC'\subset \FC$ the subspace of smooth irreducible connections.

Theorem \ref{intro thm A} above (parts (\ref{intro thm A (a)}) and (\ref{intro thm A (b)})) is a consequence of the following collection of results on the desynchronized partition function. 
\begin{thm}\footnote{
This is Proposition \ref{prop: Z desync equivariance under diagonal gauge transf}, Theorem \ref{thm 4.2}, Corollary \ref{cor: horizontality wrt nabla_G}, Theorem \ref{thm: change gf}, Proposition \ref{prop: variation of Z wrt A'} put together.
}\label{thm: intro desy Z}
Let $A,A'$ be a pair of sufficiently close smooth irreducible flat connections. Then we have:
    \begin{enumerate}[a)]
    \item \emph{Gauge invariance:} We have that $Z_{A,A'}(\sfa)$ is invariant under ``diagonal'' gauge transformations $(A,A',\sfa) \mapsto ({}^\sfg A,{}^\sfg A',{}^\sfg\sfa)$. 
        \item \emph{Variation of kinetic operator:} The desynchronized partition function satisfies 
        \begin{equation}
            \det(B^\vee) \circ Z_{\phi(A,A',\alpha),A'}(B(\sfa)) = Z_{A,A'}(\alpha + \sfa) \label{eq: thm desy Z intro 1}
        \end{equation}
        with notations as in Theorem \ref{intro thm A} above;  $\phi(A,A',-)\colon H_A^1\ra \FC'$ is the desynchronized variant of the sum-over-trees map.
 %       \marginpar{\bl comment on $\phi$ here?}
 %       with notation as in Theorem \ref{intro thm A} above. 
        %\marginpar{Is $\alpha$ a cohomolohy class or a closed form?}
        \item \emph{Infinitesimal variation of kinetic operator:} The %formal exponential map 
        map $\phi$
        induces a partial connection $\widetilde{\nabla}_G$ in the direction of harmonic shifts of $A$ on the %subbundle of $(A_0,A_1)$-harmonic forms, 
         bundle of formal half-densities on $H_{A}[1]$
        such that 
        \begin{equation}\label{intro til nabla^G horizontality}
            \widetilde{\nabla}_GZ_{A,A'} = 0. 
        \end{equation}
        \item \emph{Variation of gauge fixing operator:}  We have that, for $A_1'$ sufficiently close to $A'_0$, %\marginpar{\bl add $I_{A_0,A_1}$ factors}
        \begin{equation}\label{intro desync variation of Z wrt A'}
            Z_{A,A_1'}(\sfa) = %\frac{I_{A_0,A'_1}}{I_{A_0,A_1}} 
            Z_{A,A'_0}(\sfa) - i\hbar\Delta_\sfa R(A,A'_0,A_1',\sfa).
        \end{equation}
        For an infinitesimal variation of $A'\ra A'+\delta A'$, one has
        \begin{equation}\label{intro desync infinitesimal variation of Z wrt A'}
            \delta_{A'}Z_{A,A'}(\sfa)=-i\hbar\Delta_\sfa R_{\delta A'}(A,A',\sfa),
        \end{equation}
        with $R_{\delta A'}$ given by a sum over graphs with one marked edge or one marked leaf (cf. Proposition \ref{prop: variation of Z wrt A'}).
        %with $I$ as in (\ref{I def}).
        %\marginpar{Is $\alpha$ a cohomology class or a closed form? What happens with determinant?}
    \end{enumerate}
\end{thm}
%\marginpar{\bl correct notations in the picture to $A,A'$}
\begin{figure}
    \centering
    \includegraphics[scale=0.7]{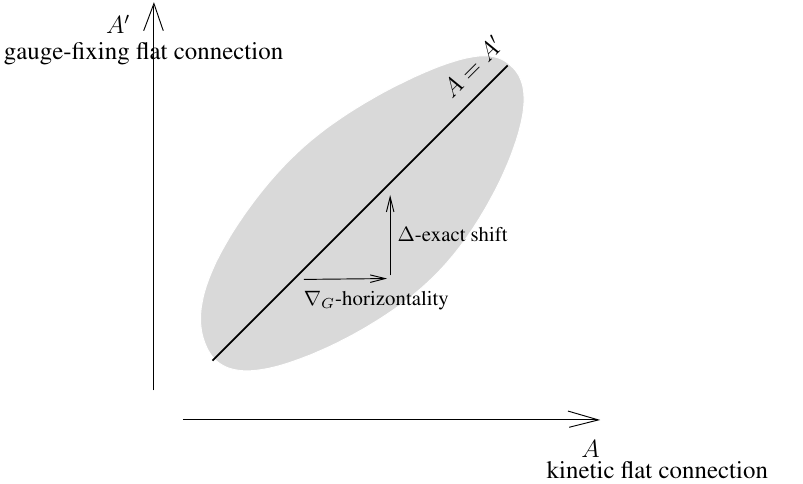}
    \caption{Desynchronized partition function $Z$ is a section of the bundle of formal half-densities on $H_{A}[1]$ over a neighborhood of the diagonal in $\FC'\times \FC'$. Under harmonic shifts of $A$ the section is $\til{\nabla}^G$-horizontal, shifts of $A'$ change $Z$ by a BV-exact term.}
    \label{fig:enter-label}
\end{figure}

Parts (\ref{intro thm A (a)}), (\ref{intro thm A (b)}) of Theorem \ref{intro thm A} follow by setting $A = A'$. The original motivation to consider the ``desynchronized'' partition function was precisely that the shift in the zero mode as on the right hand side of \eqref{eq: thm desy Z intro 1} produces a variation in the ``kinetic operator'' $A$ while keeping the gauge fixing operator fixed.  
%\marginpar{Comment on ``disentangling''of movement along diagonal {\red OK like that?}}
%Note that the first point is obvious by construction. 

% {\bl %For $A'_1$ an infinitesimal variation of $A'_0$, the 
% The generator $R$ in (\ref{intro desync infinitesimal variation of Z wrt A'}) is given explicitly as a sum over graphs with one marked edge or leaf, see Proposition \ref{prop: variation of Z wrt A'}.}

\subsection{Metric (in)dependence of the global partition function -- main result 3}\label{sec:intro metric}
%\marginpar{This sectioning is just a proposal and can be edited}
For the perturbative partition function, one has the following result \cite{Cattaneo2008,mnev2019quantum,Wernli2022}:
%\marginpar{\bl Oct 5: not sure why we are giving these references. First point is trivial for degree reason. Second is Witten/AS for $A_0$ acyclic and the current paper for non-acyclic case.}
\begin{itemize}
    \item The perturbative partition function $Z(A_0,\sfa)$ is closed with respect to the canonical BV Laplacian on formal half-densities on $H^\bullet_{A_0}[1]$. 
    \item There is a universal power series $c(\hbar) = \frac{\dim G}{24}\hbar + c'(\hbar), c'(\hbar)  \in \hbar^2\mathbb{R}[[\hbar^2]]$ such that for every framing $f\colon TM \xrightarrow{\cong} M \times \R^3$ the ``renormalized'' partition function
    %\marginpar{\bl Sep4: renamed $Z'$ to $Z^\mr{rem}$}
    %\marginpar{\bl Sep 13: changed the f-la for $c(\hbar)$ to match the one in sec 4.5}
\begin{equation}
    Z^\mr{ren}_{A_0}(\sfa) := e^{\frac{i}{\hbar}c(\hbar)\frac{S_\mr{grav}(g,f)}{2\pi}}Z_{A_0}(\sfa) \label{eq: renormalization}
\end{equation}
is independent of the metric $g$, up to BV-exact terms:\footnote{ 
See Theorem \ref{thm: metric dependence}.
}
\begin{equation}\label{intro delta_g Z}
    \delta_g Z^\mr{ren}_{A_0}(\sfa)=-i\hbar\Delta_\sfa R_{\delta g}.
\end{equation}
Here $S_\mr{grav}$ is the evaluation of the Chern-Simons action on the Levi-Civita connection of $g$ via the framing.\footnote{For a 2-framing $\alpha$ one defines $S_\mr{grav}({g,\alpha}) =\frac12 S_{CS}(A_g \oplus A_g)$ by evaluating $\frac12$ the Chern-Simons action on the direct sum of the Levi-Civita connection with itself. }
\end{itemize}
%\marginpar{\bl Oct 4: this statement about BV cohomology is trivial for silly (Poincare lemma) reason. remove it?}
For these statements one does not need to assume that $A_0$ is smooth or irreducible.
In particular, the BV cohomology class of $Z^\mr{ren}_{A_0}(\sfa)$ is independent of the metric $g$.\footnote{For $A_0$ an irreducible flat connection, this statement is trivial by Poincar\'e lemma, cf. Remark \ref{rem: Poincare lemma}.   Otherwise, for $[A_0]\in \M\setminus \M'$, this statement is not a triviality.
%In fact, the discussion above in this subsection doesn't need $A_0$ to be smooth or irreducible, in which case 
}

In this paper, we investigate the metric dependence of the ``global partition function'' $Z^\mr{glob}$.  
Since it does not depend on the fiber coordinates, the global partition function $Z^\mr{glob} \in \mr{Dens}^{\frac12}_\mr{base}(T^*[-1]\M')$ is trivially BV-closed on $T^*[-1]\M'$, 
\begin{equation}
    \Delta_{\M'}Z^\mr{glob} = 0.
\end{equation}
Our main result in this direction is that the  BV cohomology class of the renormalized global partition function is independent of the metric used to define the gauge-fixing. 
\begin{thm}\footnote{
This is Theorem \ref{thm: coho class invariant}/Proposition \ref{prop: properties of Zcheckglob} (\ref{delta_g Zglob = Delta(...)}).
}\label{thm:metric dependence intro}
Suppose $g_t, t\in (-\varepsilon,\varepsilon)$, is a smooth family of Riemannian metrics on $M$, and denote by $Z_t^\mr{glob,ren}$ the global partition function defined using the metric $g_t$, renormalized as in \eqref{eq: renormalization}. Then we have 
\begin{equation}
    \frac{d}{dt}\bigg|_{t=0} Z_t^\mr{glob,ren} = -i\hbar\Delta_{\M'}(R^\mr{glob}),
\end{equation}
where $R^\mr{glob}$ is a degree -1 half-density given explicitly by the 1-form component of (\ref{Zcheckglob formula})  along the space of metrics (evaluated on the tangent vector $\dot{g}$), see also (\ref{Zglobren (1) Feynman diagram expansion}).
\end{thm}
Put differently, if we think of the global partition function as a top form on $\M'$, then under a change of metric, it changes by an exact form 
\begin{equation}
    \frac{d}{dt}\bigg|_{t=0} Z_t^\mr{glob,ren} = -i\hbar\, d_{\M'}R^\mr{glob}.
\end{equation}
Here we interpret the degree -1 half-density $R^{\mr{glob}}$ as a differential form of degree $\mr{top} -1$.  

\subsection{Extension of %desynchronized 
the partition function to a nonhomogeneous form in $A,A',g$ -- main result 4
}
\label{ss intro:Zcheck}

In Section \ref{ss: full extension of Z} we consider an open set $\UU=\{(A,A',g)\;|\;A'\mr{\;and\;}A\;\mr{close}\}$ in $\FC'\times\FC'\times \mr{Met}$ and construct: 
\begin{enumerate}
\item A connection\footnote{Connection $\nabla^\mathbb{H}$ arises as the projection -- using the desynchronized Hodge decomposition -- of the trivial connection in the trivial bundle over $\UU$ with fiber $\Omega^\bt(M,\g)$ onto the subbundle of harmonic forms, cf. Remark \ref{rem: nabla^Harm as shift-and-project connection}.} 
$\nabla^{\mathbb{H}}$ on the ``cohomology bundle'' $\mathbb{H}$ over $\UU$ with fiber $H_A$ and the induced connection
$\nabla^\mc{D}$ on the ``half-density bundle'' $\mc{D}$ over $\UU$ with fiber $\mr{Dens}^{\frac12,\mr{formal}}(H_A[1])$. 
%This connection is trivial over $A$-fixed slices in $\UU$.
\item An ``extended'' partition function $\check{Z}\in \Omega^{\bt,\bt,\bt}(\UU,\mc{D})$ (see (\ref{Z check perturbative})) -- a nonhomogeneous form on $\UU$ valued in $\mc{D}$ -- defined similarly to (\ref{Zpert}), where the weights of Feynman graphs are extended appropriately to differential forms on $\UU$; we also include the framing-dependent renormalization factor as in (\ref{eq: renormalization}).
\end{enumerate}

\begin{thm}\footnote{See Theorem \ref{thm dQME on Zcheck}.}
The extended partition function satisfies the following ``differential quantum master equation'' (dQME):
\begin{equation}\label{intro dQME}
    (\nabla^\mc{D}-i\hbar\Delta_\sfa-\frac{i}{\hbar}\frac12 \langle \sfa, F \sfa \rangle) \check{Z}%^\mr{ren}
    =0.
\end{equation}
\end{thm}
Here $F$ is the curvature of $\nabla^\mathbb{H}$. The expression in brackets acting on $\check{Z}$ above is a \emph{flat} superconnection concentrated in de Rham degrees 0,1,2 along $\UU$. By abuse of terminology, we call it the \emph{Gauss-Manin superconnection}.\footnote{In \cite{Cattaneo2019}, \cite{Cattaneo2020} a similar operator was called ``Grothendieck operator''. Here, we use the term Gauss-Manin connection since it is a connection on the (half-density bundle of the) cohomology bundle $\mathbb{H}$.}

Low-degree components of $\check{Z}$ and of the dQME (\ref{intro dQME}) yield various objects and infinitesimal variation statements we have encountered in the earlier subsections:
\begin{itemize}
    \item Degree $(0,0,0)$ component of $\check{Z}$ is the desynchronized partition function $Z_{A,A'}(\sfa)$.\footnote{We remark that the degree $(0,0,0)$ component of the dQME is the ordinary QME $\Delta_\sfa Z_{A,A'}(\sfa)=0$ (which is trivial for degree reasons at irreducible connections $A$).}
    \item Degree $(0,1,0)$ component of $\check{Z}$ is the generator $R_{\delta A'}$ appearing in (\ref{intro desync infinitesimal variation of Z wrt A'}). Degree $(0,1,0)$ component of the dQME is the equation (\ref{intro desync infinitesimal variation of Z wrt A'}).
    \item Degree $(0,0,1)$ component of $\check{Z}$, evaluated at $A'=A$, is the generator $R_{\delta g}$ in (\ref{intro delta_g Z}); the corresponding component of the dQME is the equation (\ref{intro delta_g Z}).
    \item $(1,0,0)$ component of the dQME, contracted with a tangent vector to $\UU$ representing a harmonic shift of $A$, is equivalent to horizontality w.r.t. partial Grothendieck connection $\til{\nabla}^G$ (\ref{intro til nabla^G horizontality}).
    \item Restricting dQME to the diagonal $A'=A$ and fixing $g$ and then taking the degree $1$ component in $A$ yields (\ref{intro: horizontality eq}).
\end{itemize}

Thus, the dQME (\ref{intro dQME}) is an ``omnibus equation'' implying as low-degree specializations all the infinitesimal variation statements from before.

The restriction of $\check{Z}$ to the diagonal $A=A'$ is the key ingredient in the construction of the global partition function, see Section \ref{ss Zglob}.

% \subsubsection{Extension in $A'$}
% \subsubsection{Extension in $g$}
% \subsubsection{Full extension in $A,A',g$}

\subsection{Motivating example: $\Sigma \times S^1$}\label{sec:intro example}
As a motivating example where the (leading) asymptotics of the Reshetikhin-Turaev invariants agree with the integral of $Z^\mr{glob}$ over the moduli space of flat connections, consider the case $M = \Sigma \times S^1$, with $\Sigma$ a Riemann surface of genus $\gamma \geq 2$. Then the non-perturbative Chern-Simons partition function -- the RT invariant -- at level $k$ for $G = SU(2)$ is given by the Riemann-Roch-Hirzebruch formula as %Verlinde formula
%\marginpar{\bl Sep 15: uniformize the notations for the moduli space - they are all over the map here}
\begin{multline} Z^{(k)} = \dim H^0(\M(\Sigma)),\mathcal{L}^{\otimes k}) = 
\int_{\M'(\Sigma)}e^{k\omega_{AB}}\mathrm{Td}  (T\M) \\ 
= k^N \int_{\M'(\Sigma)} \frac{w_{AB}^N}{N!}+O(k^{N-1}) \label{eq: intro ex 1} \end{multline}
where $\omega_{AB}$ denotes the Atiyah-Bott symplectic form on $\M(\Sigma)$, $\M'(\Sigma)$ denotes the subset corresponding to irreducible flat connections,\footnote{Recall that irreducible flat connections satisfy $H^0_{A_0} =0$, in particular in dimension 2 this implies $H^2_{A_0} =0$ and smoothness. For $\gamma \geq 2$, $\M' \subset \M$ is an open dense subset.} $\mr{Td}$ the Todd class and $N = \frac12\dim \M(\Sigma)$. By work of Witten \cite{Witten1991}, the symplectic volume is related to the torsion as \begin{equation}
     k^N \int_{\M'(\Sigma)} \frac{w_{AB}^N}{N!} = \frac{k^N}{(2\pi)^{2N}}\int_{\M'(\Sigma)}\tau_\Sigma.
     \label{eq: intro ex 3}
\end{equation}
We want to compare this with the integral 
$$ Z^\mr{num} =  \int_{\M'(\Sigma \times S^1)} Z^\mr{glob}$$
-- the number-valued partition function. 
Let $A_0$ be an irreducible flat connection on $\Sigma \times S^1$. Then $A_0$ is gauge equivalent to a connection of the form 
\begin{equation}
   A_0 = \pi^*\phi\, dt + \pi^*\alpha\label{eq: special conn}
\end{equation}
where $\alpha$ is a flat connection on $\Sigma$ and $\phi \in \Omega^0(\Sigma,\g)$ is $d_{A_0}$-closed. In particular, all irreducible connections are smooth and there is a bijection (in fact a diffeomorphism)
\begin{equation}
    \M'(\Sigma \times S^1) \to \bigsqcup_{g \in Z(G)} \M'(\Sigma)\label{eq: diffeo}
\end{equation}
which sends the class of $\pi^*\phi\, dt + \pi^*\alpha$ to the pair $([\alpha],g)$ where $g \in G$ is the holonomy of $A_0$ along the circle direction, if $\alpha$ is irreducible, then $g$ is necessarily central. 

% There is a map 
% $ i \colon M_{flat}(\Sigma) \to M_{flat}(\Sigma \times S^1)
% $ given by extending a representation trivially from $\pi_1(\Sigma)$ to $\pi_1(\Sigma \times S^1) \cong \pi_1(\Sigma) \times \Z$. {\red On the level of connection 1-forms, it is given by pullback along the projection $\pi\colon \Sigma \times S^1 \to \Sigma$. In particular, $S_{CS}[A_0]=0$ for $A_0$ in the image of $i$.\footnote{Because such $A_0$ have no component along the circle and the other components are constant along the circle, the Chern-Simons form of such connections vanishes identically.}
For flat connections of the form \eqref{eq: special conn}, we have that $S_{CS}(A_0)=0$.\footnote{For connections of this form we have $S_{CS}(A_0) = \int_\Sigma \langle \phi, F_\alpha\rangle$ which of course vanishes for flat connections $\alpha$. }
This also implies that $\psi(A_0,g) = 0$.\footnote{The Atiyah-Patodi-Singer theorem implies $\frac{\pi i}{4}\psi(A_0,g) = \frac{\pi i}{4}\dim G\, \psi_0(g) - \frac{c_2(G)}{2\pi i}S_{CS}(A_0)$. The second term vanishes by the previous argument, the first term ---because the eta invariant of a product manifold satisfies $\psi_0(g_{M\times N}) = \psi_0(g_M)\tau(g_N) + \psi_0(g_N)\tau(g_M)$ where $\tau$ denotes the signature, however we have $\psi_0(g_M) = 0$ unless $\dim M = 4k-1$ and $\tau(g_N) = 0$ unless $\dim N = 4k$.   } 
%If the genus of $\Sigma$ is at least 2, then the map $i$ has open dense image [REF/argument] and
Under the identification \eqref{eq: diffeo} we have $\tau_{\Sigma \times S^1} = \tau_\Sigma^2$ and therefore we get
\begin{multline} \int_{\M'(\Sigma \times S^1)}Z^\mr{glob} = \int_{\M'(\Sigma\times S^1)}\tau_{\Sigma \times S^1}^\frac12(1 + O(\hbar)) \\
= 
 |Z(G)|\int_{\M'(\Sigma)}\tau_\Sigma (1 + O(\hbar)).
 \label{eq: intro ex 2}
\end{multline}%\marginpar{Factors of 2 = |Z(G)|? Gravitational CS invariant in canonical framing?}
Equations \eqref{eq: intro ex 1} and \eqref{eq: intro ex 2} agree if we identify $k = \frac{2\pi}{\hbar}$ and divide $Z^\mr{glob}$ by  $|Z(G)|(2\pi\hbar)^{N}$. In particular, here the framing correction vanishes in the canonical 2-framing. This is precisely the factor we were alluding to in Remark \ref{rem: def Z pert normalization} and agrees with the proposals in the literature such as \cite{Freed1991}, \cite{Rozansky1995}, \cite{Reshetikhin2010}.
\subsection{Comparison to literature and historical remarks}
 The problem of studying the perturbative (or semiclassical) behavior of the Chern-Simons partition function around non-acyclic flat connections, where the path integral has zero modes, was already observed in Witten's seminal paper on the subject \cite[p.361]{Witten1989}. Axelrod and Singer studied the perturbative theory around acylic flat connections in detail \cite{Axelrod1991}, \cite{Axelrod1994} but already comment that the assumption on acyclicity should be removed, and state (without proof) that the partition function changes by a total divergence when changing the Riemannian metric. They also identify the problem of defining the integral over the moduli space and proving that it is finite, as well as potential anomalies. Axelrod has a later %\marginpar{\bl Apr19 edit}
 text in conference proceedings %preprint 
 on the subject \cite{Axelrod1995}, where he develops the theory of oscillatory integrals of Morse-Bott functions and announces some theorems on their application to Chern-Simons theory, but without proof. Our work %in this paper 
 is independent from %this preprint
 Axelrod's paper\footnote{In fact, we only learned about its existence shortly before completion of this paper.} and draws on a different background - BV pushforwards. 
 %- but some similarities are recognizable and we will comment on them where appropriate. %\marginpar{Oct 4: We never comment on them and I am not sure what they are, so maybe we reformulate here}
The main body of the literature on perturbative Chern-Simons theory turned to the study of (rational) homology 3-spheres, where one can treat the problem of zero modes either by puncturing \cite{Kontsevich1993}, \cite{Kuperberg1999}, \cite{Lescop2002} (resulting in the Kontsevich-Kuperberg-Thurston-Lescop or KKTL invariant) or by introducing %\marginpar{Maybe mention LMO invariant?} 
extra vertices as in the works of Bott and Cattaneo \cite{Bott1998}, \cite{Bott1999} to cancel the effect of zero modes.\footnote{Here one should mention the recent paper \cite{cattaneo2021note} filling a gap in the construction of Bott and Cattaneo.} Cattaneo later showed those constructions agree \cite{Cattaneo1999}. Another line of research focused on extracting perturbative invariants of 3-manifolds from the Kontsevich integral \cite{Kontsevich1993a}, such as the Aarhus integral \cite{Bar-Natan2002}, \cite{Bar-Natan2002a} and the LMO invariant \cite{Le1998}.\footnote{It is known that the Aarhus integral and the LMO invariant are equivalent \cite{Bar-Natan2004}. It is conjectured that the KKTL invariant and the LMO invariant are equal, but this is known only up to degree 2 for integral homology spheres (accredited to C. Lescop in private communication of K.W. with G. Massuyeau). }   A full definition of the perturbative Chern-Simons partition function at non-acyclic flat connections only appeared with the introduction of the BV formalism to the problem and the works of Cattaneo and the first author \cite{Cattaneo2008} (see also \cite{mnev2019quantum}, \cite{Wernli2022}) and simultaneously Iacovino \cite{Iacovino2008}. \\
In the present paper we show how to use the BV partition function to define a volume form on smooth components of the moduli space  whose cohomology class is a topological invariant of the framed 3-manifold. We defer to future work the question of anomalies and convergence of the integral over noncompact smooth components, as well as a more detailed study of the behavior at singular points. 
%\section{Preliminaries}
%\subsection{Notation}
%\marginpar{Trying to fix some notation throughout the text}

\subsection{Acknowledgements}
We thank Alberto S. Cattaneo for inspiring discussions. We also thank the referees for their helpful comments.

\section{Formal geometry on the moduli space of flat connections}\label{sec: formal geometry}
In this section we discuss %in a pedestrian way 
formal geometry on the moduli space of flat connections on a trivialized principal $G$-bundle $P = M \times G$ over a 3-manifold $M$. We will assume that $G$ is a compact, simple and simply connected matrix group, such as $G = SU(n)$, and denote $\g$ its Lie algebra.  %\textcolor{red}
In particular, we discuss two special types of points in the moduli space, \emph{smooth} and \emph{irreducible points}. Roughly speaking, smooth points are the ones where all obstructions to deformations vanish, while irreducible points are those with a minimal stabilizer, so that one can ignore stacky aspects of the moduli space.
\subsection{Smooth points}
In this subsection we specialize the results and definitions of \cite[Appendix C]{Cattaneo2014} to the case of Chern-Simons theory and its Euler-Lagrange moduli space, the moduli space of flat connections. \\
 Since $P$ is trivialized, we identify connections on $P$ with their connection 1-forms $\mr{Conn}(P) \cong \Omega^1(M,\g)$, and denote 
\begin{equation}
    \FC \equiv \FC(P) = \{A \in \Omega^1(M,\g) \;|\; \dd A + \frac{1}{2}[A,A] = 0 \} \subset \Omega^1(M,\g)
\end{equation}
the space of flat connections on $P$.
We also identify $\mr{Aut}\,P \cong C^\infty(M,G)$, its action on $\mr{Conn}(P)$ is given by 
$$ \sfg\cdot A \equiv {}^\sfg A \equiv \sfg A\sfg^{-1} + \sfg d \sfg^{-1}.$$
The moduli space of flat connections is 
\begin{equation}
    \M \equiv \M(M,P) = \FC/\Aut P.
\end{equation}
Next, we turn to the definition of smooth points in $\FC$ and $\M$. 
Let $A_0 \in \FC$ be a flat connection on $P$. Then % we identify $T_{A_0}[-1]\mathcal{F} \cong \Omega^\bullet(M,\g)$, and the linearization of classical Chern-Simons theory around $A_0$ endows 
$\Omega^\bullet(M,\mr{Ad}P) \cong \Omega^\bullet(M,\g)$ carries the structure of a differential graded Lie algebra with differential the twisted de Rham differential $\dd_{A_0} = \dd + [A_0,\cdot]$ and Lie bracket the extension of the Lie bracket on $\g$ to differential forms. 
We denote this dgla by $\Omega^\bullet_{A_0} = (\Omega^\bullet(M,\g), \dd_{A_0})$ and by
$$ H^\bullet_{A_0} := H^\bullet_{\dd_{A_0}}(M,\g)$$
% \begin{align*}
% l_1 &= \dd_{A_0}, \\
% l_2 &= [\cdot,\cdot]. 
% \end{align*}
the cohomology of $\dd_{A_0}$. By homotopy transfer of $L_\infty$-algebras, $H^\bullet_{A_0}$ is turned into a minimal $L_\infty$-algebra endowed with induced operations $\{l'_{n,A_0}\}_{n\geq 2}$. A choice of %contracting triple (also called {\bl strong deformation retraction}
%special deformation retract 
%or SDR for short) 
SDR data\footnote{``Strong Deformation Retraction data'' \cite{gugenheim1989perturbation}, also known in the literature under the names ``contraction'' \cite{eilenberg1953groups}, ``homotopy equivalence data'' \cite{Crainic2004}, ``induction data,'' ``$(i,p,K)$ triple.''} 
$r_{A_0} = (i_{A_0},p_{A_0},K_{A_0})$ of $\Omega^\bullet_{A_0}$ onto $H^\bullet_{A_0}$ %-- we will just call this \emph{induction data at $A_0$} -- 
provides us with explicit representatives of these operations (see Appendix \ref{app: SDR} for our conventions on SDR data). %\marginpar{ \red Should we include such an appendix? {\bl Yes, I think so}} \marginpar{\color{purple} It would be good to stick to one name for $(i,p,K)$ instead of four different names (contraction, induction, SDR, $(i,p,K)$-triple). }
%\marginpar{\bl Trees should be non-planar (i.e., up to isomorphism) and weighed with $1/|Aut(T)|$ in f-la for $l_n'$}
Denote $T_n$ the set of isomorphism classes of binary rooted trees with $n$ leaves --- here we think of leaves and the root as half-edges emanating from internal vertices. To $T\in T_n$ we can associate an $n$-ary operation  $\lambda_T\colon 
%(H^\bullet_{A_0})^{\otimes n}[2-n] 
\wedge^n H^\bullet_{A_0}
\to H^\bullet_{A_0}$ of degree $2-n$ as follows: To the $n$ leaves we assign the map $i_{A_0}$, to internal vertices we assign the map $l_2$, to internal edges we assign the map %\marginpar{\bl Sign: $-K$ for edges} 
$K_{A_0}$, to the root we assign the map $p_{A_0}$ (see Figure \ref{fig: lambda});  then we skew-symmetrize over the permutations of $n$ inputs on the $n$ leaves.
\begin{figure}
    \centering
    \begin{subfigure}{0.3\textwidth}
    \centering
    \begin{tikzpicture}
        \node[vertex] (r) at (0,0) {}; 
        \draw (r) edge node[pos=0.7, left] {$p_{A_0}$} (0,-.5);
        \node[draw,shape=circle] (l2) at (1,1) {$\sfb$};
        \draw (r) edge node[pos=0.4, right] {$i_{A_0}$} (l2);
        \node[draw,shape=circle] (l1) at (-1,1) {$\sfa$};
        \draw (r) edge node[pos=0.4, left] {$i_{A_0}$} (l1);
    \end{tikzpicture}
    \caption{Unique tree $T \in T_2$}
    \end{subfigure}
    \begin{subfigure}{0.6\textwidth}
    \centering
    \begin{tikzpicture}
        \node[vertex] (r) at (0,0) {}; 
        \draw (r) edge node[pos=0.7, left] {$p_{A_0}$} (0,-.5);
        \node[draw,shape=circle] (l1) at (1,1) {$\sfc$};
        \draw (r) edge node[pos=0.4, right] {$i_{A_0}$} (l1);
        \node[vertex] (v1) at ($(r)+(-1,1)$) {};
        \draw (r) edge node[pos=0.4, left] {$K_{A_0}$} (v1);
        \node[draw,shape=circle] (l2) at ($(v1) +(1,1)$) {$\sfb$};
        \draw (v1) edge node[pos=0.4, right] {$i_{A_0}$} (l2);
        \node[draw,shape=circle] (l3) at ($(v1)+(-1,1)$) {$\sfa$};
        \draw (v1) edge node[pos=0.4, left] {$i_{A_0}$} (l3);
        
        % \begin{scope}[xshift=3cm]
        %     \node[vertex] (r) at (0,0) {}; 
        % \draw (r) edge node[pos=0.7, left] {$p_{A_0}$} (0,-.5);
        % \node[draw,shape=circle] (l1) at (-1,1) {$\sfa$};
        % \draw (r) edge node[pos=0.4, left] {$i_{A_0}$} (l1);
        % \node[vertex] (v1) at ($(r)+(1,1)$) {};
        % \draw (r) edge node[pos=0.4, right] {$K_{A_0}$} (v1);
        % \node[draw,shape=circle] (l2) at ($(v1) +(1,1)$) {$\sfc$};
        % \draw (v1) edge node[pos=0.4, right] {$i_{A_0}$} (l2);
        % \node[draw,shape=circle] (l3) at ($(v1)+(-1,1)$) {$\sfb$};
        % \draw (v1) edge node[pos=0.4, left] {$i_{A_0}$} (l3);
        % \end{scope}
    \end{tikzpicture}
    \caption{Unique tree $T \in T_3$}
    \end{subfigure}
    \caption{Trees with labeling defining $\lambda_T$}
    \label{fig: lambda}
\end{figure}
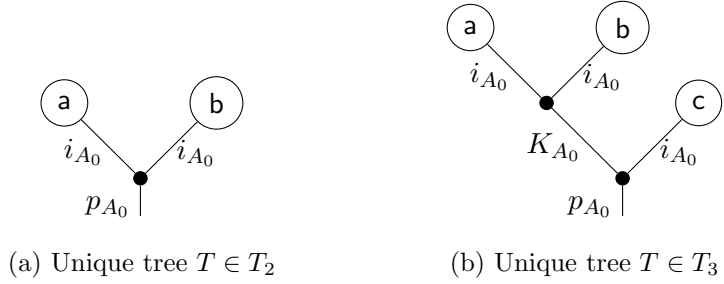
 Then we have 
\begin{equation}
    l_{n,A_0}'=  (-1)^n\sum_{T \in T_n}{ \frac{ n!}{|\operatorname{Aut} T|}}\lambda_T.
\end{equation} Explicitly, the first two non-vanishing operations are given by
\begin{align*}
l'_{2,A_0}(\sfa,\sfb) &= p_{A_0}[i_{A_0}(\sfa),i_{A_0}(\sfb)], \\
l'_{3,A_0}(\sfa,\sfb,\sfc) &= {\mr{Sym}_{\sfa,\sfb,\sfc}}\Big( p_{A_0}[-K_{A_0}[i_{A_0}(\sfa),i_{A_0}(\sfb)],i_{A_0}(\sfc)] \Big).
%+[i_{A_0}(\sfa),K_{A_0}[i_{A_0}(\sfb),i_{A_0}(\sfc)]]) 
\end{align*}
{ Here $\mr{Sym}_{\sfa,\sfb,\sfc}$ stands for skew-symmetrization in $\sfa,\sfb,\sfc$.}
%\marginpar{\red Correct representative of $l_3$? Only 1 tree up to iso? {\bl yes, only one tree; I corrected $l_3$; the right picture should be re-centered }}
\begin{defn}[\cite{Cattaneo2014}]\label{def: smooth}
We say that $A_0$ is smooth if, for all $n \geq 0$, we have $l'_n = 0$.
\end{defn}
We denote by $\FC^\mr{sm} \subset \FC$ the subset of all smooth flat connections. 

%{\bl \cancel{From here on, }
%  In the remainder of this section,}
% we also require our contracting triple to be equivariant with respect to the gauge group actions on $FC$ and the adjoint action on $\Omega^{\bullet}(M,\g)$, and the map $(-)^g\colon H^\bullet_{A_0} \to H^\bullet_{A_0^g}, [\alpha] \mapsto [\alpha^g]$.
% \marginpar{\bl We do not have equivariance of $(i,p,K)$ in the desynchronized setting!}
%\begin{ass}[Equivariance]\label{ass: Eq}
%    For all $\sfa \in H^\bullet_{A_0}$ and $\alpha \in \Omega^\bullet$.  we have 
%    \begin{equation}
%        i_{A_0^g}a^g = (i_{A_0}a)^g, \quad p_{A_0^g}\alpha^g = (p_{A_0}\alpha)^g, \quad K_{A_0^g}\alpha^g = (K_{A_0}\alpha)^g.
%    \end{equation}
%\end{ass} 
%Under assumption \ref{ass: Eq}, smoothness of $A_0$ implies smoothness of $A_0^g$ for all $g \in C^\infty(M,G)$, and in this case we call the corresponding point in the moduli space $[A_0]$ smooth and again denote by $\M^{\mr{sm}} \subset \M$ the subset of smooth points. 

The purpose of this subsection is to prove that both the space of all flat connections and the moduli space of flat connections are smooth manifolds close to a smooth point $A_0$, resp. its class in the moduli space $[A_0]$. It follows that $\FC^\mr{sm}$ and $\M^\mr{sm} = \pi(\FC^\mr{sm})$ are smooth manifolds.\footnote{When going to the moduli space, we will only prove it for the subset of smooth irreducible flat connections $\M' = \pi(\FC')$, i.e. flat connections $A_0$ for which $H^0_{A_0}=0.$} 
%\marginpar{ \red Remark added April 23} 
\begin{rem}[Topologies on $\Omega^\bullet$ and $\FC$]
The vector space $\Omega^\bullet$ carries the natural Fr\'echet topology of uniform convergence, with respect to which it is complete. Moreover, it also carries the coarser Banach topologies induced by the Sobolev norms\footnote{This norm is defined by the $L^p$-norms of $\omega$ and covariant derivatives up to order $l$, with respect to a given Riemannian metric $g$.} $ ||\cdot||_{l,p}$. The spaces $\FC' \subset \FC^\mr{sm} \subset \FC \subset \Omega^1$ carry the corresponding subset topologies. With respect to the $||\cdot||_{l,p}$ norm none of those spaces is complete and we denote the completion by $\Omega^\bullet_{l,p}, \FC_{l,p},\ldots$. We usually work with $p=2$ and denote $\Omega^\bullet_{l,2} = \Omega_{l}$.\footnote{The completions in the Sobolev topologies can contain connections with singularities. However, when we say a connection is smooth, we mean that it is a smooth point in $\FC$.} 
\end{rem}
Given the different topologies, one can make different assumptions on $K_{A_0}$. 
%\marginpar{\bl Apr 18: why ``$k$-Sobolev''? (there is no $k$ later)}
\begin{ass}[$l$-Sobolev boundedness]\label{ass: Banach} 
%\marginpar{\red April 17, rephrased this assumption}
    The chain homotopy $K_{A_0}\colon \Omega^\bullet \to \Omega^{\bullet-1}$ is bounded with respect to the Sobolev norm $||\cdot||_{l,2}$.
\end{ass} 
Actually, the main example considered in this text, the SDR coming from Hodge decomposition (Example \ref{ex: hodge decomposition}) satisfies the following stronger assumption: 
\begin{ass}[Fr\'echet continuity]\label{ass: Frechet}
The chain homotopy $K_{A_0} \colon \Omega^\bullet \to \Omega^{\bullet -1}$ is continuous in the Fr\'echet topology. 
\end{ass}
%\sout{In our examples, this Banach norm will be a Sobolev norm.

%\red We denote by $\Omega^\bullet_l, \Omega^k_l, \FC_l,\ldots $ the completion of those space with respect to $||\cdot||_{l,2}$, which from hereon we denote $||\cdot||_l$.}
The goal of this subsection is to prove the following theorems:
\begin{thm}\label{thm: banach smoothness}
    For every $l \geq 2$, the completion of the set of smooth points $\FC^\mr{sm}_l \subset \FC_l$ has the structure of a Banach manifold. For every smooth point $A_0$, there is a neighborhood $V_{A_0}$ of $A_0 \in \FC_l$ modeled on a neighborhood $U_{A_0}$ of $0 \in (\Omega^1_{A_0-\mr{cl},l},||\cdot||_l)$. Given SDR data $r_{A_0}$ at $A_0$ satisfying Assumption \ref{ass: Banach}, we have a homeomorphism in the Banach topologies  
    \begin{equation} \label{phi tilde - chain level exp map}
        \til{\varphi}_{A_0} \colon U_{A_0} \to V_{A_0}, \qquad \alpha \mapsto \til{\varphi}_{A_0}(\alpha) = A_0 + \sum_{k\geq 0}\alpha^{(k)}
    \end{equation}
     with $\alpha^{(1)} = \alpha$ and $\alpha^{(k)}$ given by 
    \begin{equation}
    \alpha^{(k)} = (-1)^{k-1} \sum_{T \in T_k} \til{\mu}_{T}(\alpha,\ldots,\alpha)
    \end{equation}
    where $\til{\mu}_T$ is defined in Equation (\ref{eq: def mu tilde}) below. 
\end{thm}
The chart $\til{\varphi}$ is constructed using the chain homotopy $K_{A_0}$. If $K_{A_0}$ is Fr\'echet continuous, then $\til{\varphi}_{A_0}$ is actually a chart for the  Fr\'echet manifold $\FC^\mr{sm}$. 
\begin{thm}\label{thm: Frechet smoothness}
    The set $\FC^\mr{sm} \subset \FC$ has the structure of a Fr\'echet manifold. If the chain homotopy $K_{A_0}$ satisfies Assumption \ref{ass: Frechet}, then the map \eqref{phi tilde - chain level exp map}, restricted to $U_{A_0} \cap \Omega^1_{d_{A_0}-\mr{cl}}$, is a Fr\'echet homeomorphism.
    %\marginpar{\bl Apr 25: $U_{A_0} \cap \FC^\mr{sm}$ looks like a typo? some mixture between domain and codomain? {\red fixed}}
\end{thm}
%Then $EL/G = MFC(M,P)$ is the moduli space of flat connections on $P$. %\marginpar{needs to made more precise. in particular describe topology on MFC} 
In the later sections of the text, we consider $\FC' \subset \FC^\mr{sm}$, with this Fr\'echet manifold structure. 
\subsubsection{Outline of the proof}
We briefly outline the strategy of the proof. The Banach case was inspired by \cite{Ho2019}.
\begin{enumerate}
    \item We define the formal power series $\delta_{A_0} = \sum_{k \geq 0} \alpha^{(k)}$ as a sum over trees, and show that $A_0 + \delta_{A_0}$ is flat for $\alpha^{(1)}$ \emph{harmonic} -- Proposition \ref{prop: deformation of A0}. 
    \item We show that $\til\delta_{A_0}$ (defined using the sum over trees for all closed forms) converges in the Banach topology under Assumption \ref{ass: Banach} (Proposition \ref{prop: banach convergence})  or in the Fr\'echet topology under Assumption \ref{ass: Frechet} (Remark \ref{rem: frechet convergence}).
    \item We define the Kuranishi map \eqref{eq: kuranishi map}, show that
    %\marginpar{\bl Apr 24: something missing -- ``it is [?] on $\Omega^1$''}
    it is invertible on $\Omega^1$ and a compositional inverse of $\delta_{A_0}$ on $\Omega^1$, and show that it maps Maurer-Cartan forms to closed forms (Proposition \ref{prop: kuranishi}).
    \item We establish the Banach smoothness of $\FC^\mr{sm}$ by means of the chart $\til{\psi}$, see Proposition \ref{prop: bij}. 
    \item Now, one can use the inverse function theorem for $\til{\kappa}$ on $\FC^\mr{sm}$ to conclude that $\til{\kappa}$, restricted to the Maurer-Cartan set, surjects onto a neighborhood of $0 \in \Omega^1_{cl,l}$, this then finally implies that $\til{\delta}$ lands in the Maurer-Cartan set (Proposition \ref{prop: til delta MC}). This proves Theorem \ref{thm: banach smoothness}.
    \item Since $\til{\kappa}, \til{\delta}$ are both continuous in the Fr\'echet topology if $K_{A_0}$ is, restricting them to $\FC^\mr{sm}$ proves Theorem \ref{thm: Frechet smoothness}.
\end{enumerate}
 \subsubsection{Formal deformations of flat connections}
 %\marginpar{\bl Maybe the title should be ``Formal deformations of flat connections''}
 We recall the following elementary facts about flat connections. If $A_t \colon (-\epsilon,\epsilon) \to \Omega^1$ is a smooth curve of flat connections, then from differentiating $F_{A_t} = 0$ we obtain $d_{A_0}\dot{A}_0 =0$, i.e., tangent vectors at $A_0$ to curves of flat connections are $d_{A_0}$-closed 1-forms. If $\sfg_t\colon (-\epsilon,\epsilon) \to C^\infty(M,G)$ is a curve with $\sfg_0 \equiv 1$ and 
% \marginpar{\bl the annoying sign here is due to switching to gauge group acting from the left}
 $\dot{\sfg}_0 = \gamma \in \Omega^0(M,
 \g)$, then the tangent vector at $0$ to the curve of flat connections $A_t = {}^{\sfg_t}A_0$ is $\dot{A}_0 = -d_{A_0}\gamma$, i.e. the tangent vector  {
 %\bl
 to a curve along the gauge orbit of a flat connection is a $d_{A_0}$-exact 1-form}. This is equivalent to saying that infinitesimal\footnote{ By ``infinitesimal'' we everywhere mean ``first-order,'' as opposed to formal deformations (of infinite order) discussed below.} deformations of flat connections are $d_{A_0}$-closed 1-forms, and two such deformations are equivalent whenever they differ by an exact 1-form, i.e. equivalence classes of infinitesimal deformations are in 1-to-1 correspondence with the first twisted cohomology group $H^1_{{A_0}}$ (sometimes called the \emph{Zariski tangent space} to the moduli space of flat connections at $[A_0]$). % We claim that at smooth points $[A_0]$ (in the sense of the definition above) the moduli space of flat connections has a smooth structure and the Zariski tangent space coincides with the actual tangent space at $[A_0]$. We start with the following proposition. 
% \marginpar{\bl Is this really if and only if? In the def of a smooth point we don't impose conditions on the degree of inputs in $l_n'$ but for lifting gh=0 deformations of flat conn., only deg=1 inputs are relevant, right?}
\begin{prop}\label{prop: deformation of A0}
Let $A_0$ be a flat connection on $P$, and let $(i_{A_0},p_{A_0},K_{A_0})$ be SDR %induction
data at $A_0$.  If $A_0$ is smooth, then all infinitesimal deformations of $A_0$ lift to formal deformations of $[A_0]$, i.e., for every $\sfa \in H^1_{A_0}$ there exists a formal power series 
\begin{equation}\delta = \delta_{A_0}(t\sfa) = \sum_{n\geq 1}t^n \alpha^{(n)} \in \Omega^1[[t]]\label{eq: formal delta}\end{equation}
with $\alpha^{(1)} = i_{A_0}\sfa$, such that $A_t:=A_0 + \delta_{A_0}(t\sfa)$ is flat, i.e. it satisfies $$\dd A_t + \frac12 [A_t,A_t] =0$$ with $\dd, [\cdot,\cdot]$ the induced operations on $\Omega^1[[t]]$. 
Moreover,  we have $K_{A_0}\delta_{A_0}(t\sfa) = 0$. 
\end{prop}
\begin{proof}
We can expand the flatness equation $dA_t + \frac12[A_t,A_t] = 0$ in powers of $t$, obtaining  %\marginpar{\bl removed redundant ${\red i}_{A_0}$ in (10)}
\begin{align}
dA_0 + \frac12[A_0,A_0] &= 0,  \\ 
d_{A_0} %{\red i}_{A_0}
\alpha^{(1)} &= 0,  \\ 
d_{A_0}\alpha^{(2)} + \frac{1}{2}[\alpha^{(1)},\alpha^{(1)}] &= 0, \\ 
&\vdots \notag \\
d_{A_0}\alpha^{(n)} + \frac12\sum_{k=1}^{n-1}[\alpha^{(k)},\alpha^{(n-k)}] &= 0. \label{eq:deformation}
\end{align}
The first two equations are satisfied by our assumptions. It is instructive to look at the third equation in detail. We see that we can solve it for $\alpha^{(2)}$ if and only $\frac12[\alpha^{(1)},\alpha^{(1)}]$ is $d_{A_0}$-exact. Because $\alpha^{(1)} = \alpha$ is closed by assumption, and the bracket is compatible with the differential, $\frac12[\alpha,\alpha]$ is always $d_{A_0}$ closed. It is exact if and only if $l_2'(\sfa,\sfa) = 0$. %\marginpar{\bl Apr 19 $"l_2'([\alpha,\alpha]) = 0" \ra "l_2'(\sfa,\sfa) = 0"$}
%$l_2'([\alpha,\alpha]) = 0$. 
In this case, we can write down an explicit solution: 
$$\alpha^{(2)} = -\frac{1}{2}K_{A_0}[\alpha,\alpha].$$
Indeed, 
$$d_{A_0}\alpha^{(2)} =-\frac12d_{A_0}K_{A_0}[\alpha,\alpha] = -\frac12(\mathrm{id} - P_{A_0} - K_{A_0}d_{A_0})[\alpha,\alpha]$$ 
and $d_{A_0}[\alpha,\alpha] = P_{A_0}[\alpha,\alpha] = 0$ by $d_{A_0}$-closedness and vanishing of $l_2$ respectively. The rest of the proof now follows by induction. Suppose we are given $\alpha^{(1)},\ldots,\alpha^{(n-1)}$ satisfying 
$$d_{A_0}\alpha^{(k)} = -\frac12\sum_{j=1}^{k-1}[\alpha^{(j)},\alpha^{(k-j)}] $$ 
for $k= 1,\ldots,n-1$ and we are looking for $\alpha^{(n)}$ to solve \eqref{eq:deformation}. Then it is a straightforward application of the Jacobi identity that the right hand side is $d_{A_0}$-closed, and is exact if and only if $l_{n}' = 0$, in which case we can define 
$$\alpha^{(n)} = -\frac12K_{A_0}\sum_{k=1}^{n-1}[\alpha^{(k)},\alpha^{(n-k)}].$$ 
Notice that by construction $K_{A_0}\alpha^{(k)} = 0$ for $k\geq 0$. Therefore $K_{A_0}\delta = 0$ if and only if $K_{A_0}\alpha = 0$.  
\end{proof}
We denote the corresponding map by %defined on the vector bundle over $H^\bullet_{(-)} \to FC^{\mr{sm}}$ with fiber $H^\bullet_{A_0}$, by 
\begin{equation}\label{phi - exp map on moduli space}
    \varphi_{A_0}\colon H^1_{A_0} \to \Omega^1[[t]],\quad (A_0,\sfa) \mapsto  \varphi_{A_0}(t\sfa)= A_0 +\delta_{A_0} (t\sfa) . 
\end{equation}
In fact, we can extract from the proof the following slightly more precise fact:
%\marginpar{\bl I edited/rephrased this proposition}
\begin{prop}
Let $A_0$ be a flat connection (not necessarily smooth), $\sfa \in H^1_{{A_0}}$ and $k \geq 2$ an integer. Then $\sfa$ can be lifted to an order $k$ deformation $t\alpha^{(1)} + \ldots + t^k\alpha^{(k)}$ if and only if $l_2'(\sfa,\sfa) = \ldots =l_k'(\sfa,\ldots,\sfa) = 0$, and in this case 
\begin{equation}\alpha^{(n)} ={ (-1)^{n-1}} \sum_{T \in T_n}{ \frac{1}{|\operatorname{Aut} T|}}\mu_T(\sfa,\ldots,\sfa)\label{eq:alphaj}\end{equation}
for $1\leq n \leq k$. %\marginpar{\red Changed $\left(\frac{-1}{2}\right)^{n-1}$ to $\frac{1}{|\operatorname{Aut} T|}$  }
Here the sum is over { isomophism classes of} rooted binary trees $T$ with $n$ leaves. 
{
The 
$n$-ary multilinear operation $\mu_T \colon \mr{Sym}^n H^\bullet_{A_0}[1]  \to\Omega^\bullet [1]$ is 
%given by summing over rooted binary trees $T$ with $j$ leaves which are 
the evaluation of the tree $T$ by putting inputs on the leaves, $l_2$ on internal vertices, $K_{A_0}$ on the internal edges and $K_{A_0}$ on the root, and symmetrizing over inputs, see Figure \ref{fig: mu}. 
}
%\marginpar{\bl changed $K\ra -K$}
\begin{figure}
    \centering
    \begin{subfigure}{0.3\textwidth}
    \begin{tikzpicture}
        \node[vertex] (r) at (0,0) {}; 
        \draw (r) edge node[pos=0.7, left] {$K_{A_0}$} (0,-.5);
        \node[draw,shape=circle] (l2) at (1,1) {$\sfb$};
        \draw (r) edge node[pos=0.4, right] {$i_{A_0}$} (l2);
        \node[draw,shape=circle] (l1) at (-1,1) {$\sfa$};
        \draw (r) edge node[pos=0.4, left] {$i_{A_0}$} (l1);
    \end{tikzpicture}
    \caption{Unique tree $T \in T_2$}
    \end{subfigure}
    \begin{subfigure}{0.6\textwidth}
    \centering
    \begin{tikzpicture}
        \node[vertex] (r) at (0,0) {}; 
        \draw (r) edge node[pos=0.7, left] {$K_{A_0}$} (0,-.5);
        \node[draw,shape=circle] (l1) at (1,1) {$\sfc$};
        \draw (r) edge node[pos=0.4, right] {$i_{A_0}$} (l1);
        \node[vertex] (v1) at ($(r)+(-1,1)$) {};
        \draw (r) edge node[pos=0.4, left] {$K_{A_0}$} (v1);
        \node[draw,shape=circle] (l2) at ($(v1) +(1,1)$) {$\sfb$};
        \draw (v1) edge node[pos=0.4, right] {$i_{A_0}$} (l2);
        \node[draw,shape=circle] (l3) at ($(v1)+(-1,1)$) {$\sfa$};
        \draw (v1) edge node[pos=0.4, left] {$i_{A_0}$} (l3);
        
        % \begin{scope}[xshift=3cm]
        %     \node[vertex] (r) at (0,0) {}; 
        % \draw (r) edge node[pos=0.7, left] {$K_{A_0}$} (0,-.5);
        % \node[draw,shape=circle] (l1) at (-1,1) {$\sfa$};
        % \draw (r) edge node[pos=0.4, left] {${\red i}_{A_0}$} (l1);
        % \node[vertex] (v1) at ($(r)+(1,1)$) {};
        % \draw (r) edge node[pos=0.4, right] {$K_{A_0}$} (v1);
        % \node[draw,shape=circle] (l2) at ($(v1) +(1,1)$) {$\sfc$};
        % \draw (v1) edge node[pos=0.4, right] {${\red i}_{A_0}$} (l2);
        % \node[draw,shape=circle] (l3) at ($(v1)+(-1,1)$) {$\sfb$};
        % \draw (v1) edge node[pos=0.4, left] {${\red i}_{A_0}$} (l3);
        % \end{scope}
    \end{tikzpicture}
    \caption{Unique tree $T \in T_3$}
    \end{subfigure}
    \caption{Trees with the labeling defining $\mu_T$}
    \label{fig: mu}
\end{figure}
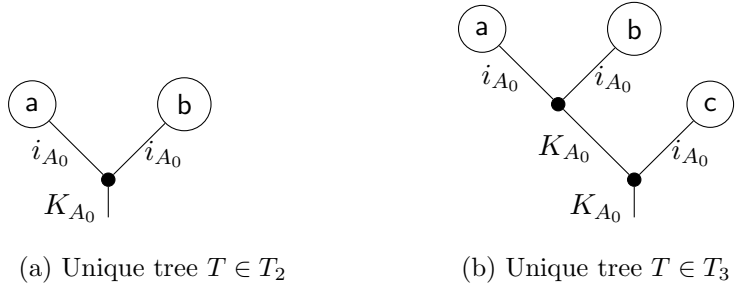
\end{prop}
In words, the induced $L_\infty$-operations $l_{n,A_0}'$ on $H^\bullet_{A_0}$ are precisely the obstructions to the lift of infinitesimal deformations to higher-order deformations. At smooth points all those obstructions vanish, so all infinitesimal deformations lift to  formal ones. 
\subsubsection{Lifting to forms}
For any binary tree $T$  with $n$ leaves, we can lift the operations $\lambda_T\colon \mr{Sym}^n H^\bullet_{A_0}[1] \to H^\bullet_{A_0}[2]$ and $\mu_T\colon \mr{Sym}^n H^\bullet_{A_0}[1]  \to \Omega^\bullet[1]$ to operations \begin{align}
    \til{\lambda}_T&\colon \mr{Sym}^n \Omega^\bullet[1] \to \Omega^\bullet[2], \\
    \til{\mu}_T&\colon \mr{Sym}^n \Omega^\bullet[1] \to \Omega^\bullet[1] \label{eq: def mu tilde}
\end{align}
by replacing $i_{A_0}$ and $p_{A_0}$ with $\mr{id}_{\Omega^\bullet}$. Obviously, we have $\lambda_T = p_{A_0} \circ \til{\lambda}_T \circ i_{A_0}^{\otimes n}$ and $\mu_T = \til{\mu}_T \circ i_{A_0}^{\otimes n}$. 
    We have defined the map
  %  \marginpar{\bl more pedantically, instead of $\mu_T$ in the sum it should be maybe $\mu_T\circ (-)^{\otimes \#\mr{leaves}(T)}$}
    $$\delta_{A_0} = t\cdot i_{A_0} +  \sum_{k \geq 2} t^k \sum_{T \in T_k} (-1)^{k-1} \mu_T
     \circ (-)^{\otimes k}\colon H^1_{A_0} \to t\Omega^1[[t]],$$
    but it is clear from the definition that it factors through a map 
    \begin{equation} 
    \til{\delta}_{A_0} = t\cdot\mr{id}_{\Omega^\bullet} +  \sum_{k \geq 2} t^k \sum_{T \in T_k} (-1)^{k-1} \til{\mu}_T \circ (-)^{\otimes k}  \colon \Omega^1 \to t \Omega^1[[t]]\label{eq: def tilde delta}
    \end{equation} 
    by $\delta_{A_0} = \til{\delta}_{A_0} \circ i_{A_0}.$   
    %\marginpar{Slightly edited, changed to proposition}
    
     We now establish by defining another lift of $\alpha$ to a formal deformation.
    %\footnote{Other authors [https://arxiv.org/pdf/1610.09987v3] use as a condition that the harmonic part of the bracket of any two 1-forms vanishes: $p_{A_0}[\alpha,\alpha] = 0$ for all $\alpha \in \Omega^1(M,\g)$, under this condition, which implies vanishing of all $L_\infty$ operations on $H^1_{d_{A_0}}(M,\g)$, one can apply formula \eqref{eq: formal delta} to any closed 1-form.} 
    Namely, we split a closed 1-form $\alpha$ as $\alpha = i_{A_0}p_{A_0}\alpha + d_{A_0}K_{A_0}\alpha$. Denoting  $\beta =  K_{A_0}\alpha$, we can lift $\alpha$ { to} a (formal) curve of flat connections with tangent vector $\dot{A}_0 = \alpha$ by setting $\sfg_{t} = \exp(-  t\beta)$ and  setting
    %\marginpar{\bl Should the argument of $\varphi, \til{\psi}$ be $\alpha$ or $t\alpha$ everywhere?}
    \begin{equation}
        \til{\psi}_{A_0}(t\alpha) = {}^{\sfg_t}\varphi_{A_0}(t[\alpha]) = \sfg_t {\varphi}_{A_0}(t[\alpha])\sfg_t^{-1} + \sfg_t d \sfg_t^{-1}  \in \Omega^1[[t]].
    \end{equation}
    % where 
    % \marginpar{\bl notation $\varphi(A_0,-)$ was not introduced before, I think}
    % \begin{equation}
    %     \til{\varphi}^{\mr{form}}_{A_0}\colon \Omega^1 \to \Omega^1[[t]],\quad  \alpha \mapsto A_0 + \til{\delta}({\bl t}\alpha).
    % \end{equation} 
    Notice that $\varphi_{A_0}(t[\alpha])$ is flat and therefore $\til{\psi}_{A_0}(t\alpha)$ is also.
    Expanding $\til{\psi}_{A_0}(t\alpha)$ in powers of $t$, we obtain %\marginpar{\bl changed signs due to switching to gauge action from the left}
    %\marginpar{\bl There was a mistake in the formula (it didn't work for $\beta=0$). it is corrected now.}
\begin{multline}
    \til{\psi}_{A_0}(t\alpha) = A_0 + \sum_{k\geq 1}t^k\left(\frac{
    (-1)^{k-1}
    }{k!}\mathrm{ad}_\beta^{k-1} d_{A_0}\beta +\sum_{j+l = k, j\geq 0, l\geq 1} \frac{
    (-1)^{j}
    }{j!}\mathrm{ad}_{\beta}^j \alpha_h^{(l)}\right) 
\end{multline}
with $\alpha_h^{(l)}$ as in Proposition \ref{prop: deformation of A0} applied to $\sfa=p_{A_0}\alpha$.
    In this way, we have constructed the map $\til{\psi}_{A_0}(t\alpha) \colon \Omega^1_{A_0-\mr{cl}} \to \Omega^1[[t]]$ lifting any infinitesimal deformation of $A_0$, i.e. a closed form, to a formal deformation.    
\subsubsection{Convergence in Banach norm}
It is natural to ask whether the formal deformations defined in the previous sections actually converge. One instance where this happens is the case when $K_{A_0}$ is continuous with respect to a Banach norm on $\Omega^\bullet$. 
\begin{prop}\label{prop: banach convergence}
Suppose $K_{A_0}$ is bounded with respect to $||\cdot||_{l}$ for some $l \geq 2$ on $\Omega^\bullet$ (Assumption \ref{ass: Banach}). Then there is an open interval $I$ around zero such that  $\til{\delta}_{A_0}$ defined by \eqref{eq: def tilde delta}, seen as a formal power series in $t$, converges in $||\cdot||_l$ for $t\in I$.     
\end{prop}
Notice that $i_{A_0}$ and $p_{A_0}$ are automatically bounded, since we are assuming that $M$ is compact which implies that $H^1_{A_0}$ is finite-dimensional.
\begin{proof}
    The summand $\alpha^{(n)}$ is a sum over binary trees with $n$ leaves, involving at most $n$ applications of $K_{A_0}$. The number of such binary trees is given the $(n-1)$-th Catalan number, and by assumption $K_{A_0}$ is bounded in $||\cdot||_l$, say with constant $L$. The Lie bracket $[\alpha,\alpha]$ is bounded by \begin{equation}
        ||[\alpha,\alpha]||_l \leq F ||\alpha||^2_l
    \end{equation} for some other constant $F$.\footnote{This is because for $k \cdot p \geq n$, one has, for $f,g \in W_{k,p}(M)$ the estimate $||f \cdot g || \leq \til{C} ||f||_{k,p} ||g||_{k,p}$ which in turns follows from the Sobolev embedding $W_{k,p}(M) \hookrightarrow C^{0}(M)$, which one has for $kp > n$ on a compact Riemannian manifold. See e.g. \cite[Theorem 2.10, Theorem 2.20]{Aubin}.} We therefore have the simple estimate 
    $$||\alpha^{(n)}||_l \leq C_{n-1} (1+L)^n (1+F)^n ||\alpha||^n_l < \til{C}^n ||\alpha||^n_l$$
    with $\til{C} = 4 (1+L)(1+F)$. Therefore, the sum converges in $||\cdot||_l$ for $|t|< 1/ (\til{C}||\alpha||_l)^n$. 
\end{proof}
Equivalently, we can set $t=1$ by letting $||\alpha||_l$ small enough, i.e. there exists some open set  $\til{U}_{A_0} \subset \Omega^1$  on which the formal power series  $\til{\delta}_{A_0}(\alpha)$ converges,  i.e. there is a map 
\begin{equation}
    \til{\delta}_{A_0}\colon \Omega^1_l \supset \til{U}_{A_0} \to \Omega^1_l
\end{equation} given by \eqref{eq: def tilde delta} with $t=1$. 
\begin{rem} \label{rem: frechet convergence}A priori, the limit of this power series lives in the completion $\Omega^{1}_l$ of $\Omega^1$ with respect to $||\cdot||_l$. 
However, under the additional Assumption \ref{ass: Frechet} of Fr\'echet continuity the sum actually converges in the Fr\'echet topology and $\til{\delta}_{A_0}$ actually lands in $\Omega^1$. 
%\marginpar{\bl Apr 25. $\delta$ is a map, so writing $\til\delta_{A_0}(\alpha)\in \Omega^1$ doesn't look good, right?}
\end{rem}
%\marginpar{\red for this we need to assume that $K$ comes from something elliptic}

%Since $\til{\delta}_{A_0}$ satisfies the equations $d_{A_0}\til{\delta}_{A_0} + \frac12[\til\delta_{A_0},\til\delta_{A_0}] = 0$ and $d^*_{A_0}\til{\delta} = 0$, it is smooth by elliptic regularity. 
Convergence of the map $\til{\delta}_{A_0}$ implies that there are well-defined maps
%\marginpar{\bl the domain for $\til{\psi}$ should be an open in $\Omega^1_{cl}$, not in $\Omega^1$}
\begin{align}
&\delta_{A_0} = \til{\delta}_{A_0} \circ i_{A_0} \colon U_{A_0} \to \Omega^1,\\ 
&\varphi_{A_0} = A_0 + \delta_{A_0} \colon U_{A_0} \to \FC^\mr{sm}, \\ 
& \til{\psi}_{A_0}  \colon \til{U}_{A_0} \cap\; \Omega^1_{d_{A_0}\mr{-cl}} \to \FC^\mr{sm}, \\
&\til{\varphi}_{A_0} = A_0 + \til{\delta}_{A_0} \colon \til{U}_{A_0} \to \Omega^1.
\end{align}
 %for every flat connection $A_0$ there is a neighbourhood $U_{A_0}$ of $0$ in $H^1_{{A_0}}$ on which the map $\varphi^{\mr{form}}_{A_0}$ converges to a map $\varphi_{A_0}$ and satisfies 
% \begin{align*}
%     F_{\varphi_{A_0}(\sfa)} =& 0 \\
%     K_{A_0}(\varphi_{A_0}(\sfa) - A_0) = 0. 
% \end{align*}
% We denote by $\til{\varphi}$ its lift to  1-forms, and extend the $\til{\varphi}$ to small closed forms by 
% \begin{equation}
%     \til{\varphi}_{A_0}(\alpha) = \til{\varphi}_{A_0}(\alpha)^{\exp K_{A_0}} = \exp(-K_{A_0}\alpha)\til{\varphi}_{A_0}(p_{A_0}\alpha)\exp(K_{A_0}\alpha) + \exp(-K_{A_0}\alpha)d\exp(K_{A_0})
% \end{equation}
% This map is defined whenever $\til{\varphi}_{A_0}(p_{A_0}\alpha)$ converges, and its image is flat, since gauge transformations preserve flatness. 

%EDIT FROM HERE
 \subsubsection{Kuranishi map}
% \marginpar{\bl Oct 6. $\til\delta_{A_0}$ was missing the subscript in many places. I tried to restore it.}
%\marginpar{\bl Apr 19: switched $\delta$ to $\ddd$ throughout this subsubsection. (it is a macros, easy to change)}
The map $\til\delta_{A_0}$ admits a compositional inverse known as Kuranishi map.\footnote{It appeared in the context of deformations of complex structures in \cite{Kuranishi1965}. }
\begin{defn}
    We define the Kuranishi map $\til{\kappa}_{A_0}\colon \Omega^1 \to \Omega^1$ to be given by 
\begin{equation}
 \til{\kappa}_{A_0}(\ddd) = \ddd + \frac{1}{2}K_{A_0}[\ddd,\ddd] \label{eq: kuranishi map}.
\end{equation}
\end{defn}
This map is as regular as $K_{A_0}$, in particular if $K_{A_0}$ is $l$-Sobolev bounded (for $l \geq 2$) then so is $\til\kappa_{A_0}$ and if $K_{A_0}$ is Fr\'echet continuous then so is $\til{\kappa}_{A_0}$.
Some salient properties of the this map are:
\begin{prop}\label{prop: kuranishi} Assume $r_{A_0}$ satisfies the $l$-Sobolev boundedness assumption for some $l \geq 2$ (Assumption \ref{ass: Banach}). Then $\til{\kappa}_{A_0}$ extends to the completion $\Omega^1_l$ and we have: %\marginpar{\bl there were some missing tildes on $\kappa$; I tried to restore them}
\begin{enumerate}[i)]
    \item The map $\til{\kappa}_{A_0} \colon \Omega^1_l \to \Omega^1_l$ is invertible in a neighborhood $V_{A_0}$ of $ 0\in \Omega^1_l$, and 
     $$V_{A_0} \supset \{\ddd \in \Omega^1|\; ||K_{A_0}\mathrm{ad}_\ddd||_\mr{op} < 1\}. $$
     Here $||\cdot||_\mr{op}$ denotes the operator norm induced by the Sobolev norm $||\cdot||_l$.
    \item 
%    \marginpar{\bl $U$ or $\til{U}$?}
    Recall that there exists an open set $U_{A_0}$ on which $\til{\delta}_{A_0}$ converges. We have $U_{A_0} \subset V_{A_0}$, and on $U_{A_0}$ the inverse of $\til{\kappa}_{A_0}$ is given by the map $\til{\delta}_{A_0}$. 
    \item Define the Maurer-Cartan set $MC_{A_0} \subset \Omega^1_{A_0}$ by 
    \begin{equation}
        MC_{A_0} = \{m \in \Omega^1_{A_0} \;|\; \dd_{A_0}m + \frac12[m,m] = 0\}.
    \end{equation}
    Then $\til{\kappa}_{A_0}(MC_{A_0}) \subset \Omega^1_{A_0-\mr{cl}}$, i.e. $m \in MC_{A_0}$ implies $\dd_{A_0}\til\kappa_{A_0} m = 0$.
    \item We have
    \begin{equation}
        \til{\varphi}_{A_0}^{-1}(A_1) = \til\kappa_{A_0}(A_1-A_0).\label{eq: til phi inverse}
    \end{equation}
\end{enumerate}
\end{prop} 
The proof is in Appendix \ref{app: proofs 1}.
%We can think of $\delta$ as an element of $\widehat{\Sym}(\Omega^1(M,\g)^*)\otimes \Omega^1(M,\g)$. Since its degree zero term vanishes and its degree one term is the identity, $\delta$ has a compositional inverse, which turns out to be simple:
% \begin{prop}
% The formal power series $\kappa \in (\Sym^1 \oplus \Sym^2)\Omega^1(M,g) \subset \widehat{\Sym}(M,\g)$
% \begin{equation}
% \kappa(\delta) = \delta + \frac{1}{2}K_{A_0}[\delta,\delta] \in \Omega^1(M,\g) 
% \end{equation}
% is a compositional inverse to the power series $\delta(\alpha)$, i.e. $\delta(\kappa(\alpha)) = \kappa(\delta(\alpha))= \alpha$. 
% \end{prop} 
\subsubsection{Banach and Fr\'echet smoothness}
We know that $\til{\psi}_{A_0}(\alpha)$ is flat for any closed 1-form $\alpha$ in its domain. We claim that locally, it is actually invertible and thus provides $\FC$ with the structure of a Banach manifold around $A_0$.
\begin{prop}\label{prop: bij}
   \begin{enumerate}[i)]
   \item There exists a neighborhood $U_{A_0}$ of $0$ in $\Omega^1_{A_0-\mr{cl},l}$ such that the map $\til{\psi}_{A_0}\colon U_{A_0} \to \FC_l\subset \Omega^1_l$ is a homeomorphism onto its image.  
   \item If $A_1 \in \til{\psi}_{A_0}(U_{A_0})$, then the composition $\til\psi_{A_1}^{-1} \circ \til\psi_{A_0}$ is smooth\footnote{In the Banach space setting this means smoothness w.r.t. the Fr\'echet derivative.} on its domain. 
   \end{enumerate} 
\end{prop}
%\marginpar{\bl Apr 28: $\FC^\mr{sm}\ra \FC^\mr{sm}_l$ (two places btw Prop 2.13 and Prop 2.14)}
In particular, $\FC^\mr{sm}_l$ is a Banach manifold and $\{\til\psi_{A_0}\}_{A_0 \in \FC^\mr{sm}_l}$ is an atlas for $\FC^\mr{sm}_l$. We give the proof in Appendix \ref{app: proofs 2}. We then have the following corollary. 
    %\marginpar{\bl $\delta\ra \delta_{A_0}$?}
    \begin{prop}\label{prop: til delta MC}
    There is a neighborhood $0 \in V_{A_0} \subset \Omega^1_{A_0-\mr{cl},l}$ such that $\til{\delta}_{A_0}$ 
    %\marginpar{\bl Apr 19 edit}
    applied to a closed 1-form $\alpha \in V_{A_0}$
    defines a Maurer-Cartan element in $\Omega^1_{A_0,l}$: 
    \begin{equation} 
    d_{A_0}\til{\delta}_{A_0} (\alpha) + \frac{1}{2}[\til{\delta}_{A_0}(\alpha),\til{\delta}_{A_0} (\alpha)] = 0.\label{eq: MC delta tilde}
    \end{equation}
    \end{prop}
    \begin{proof}
    We know that $\FC^\mr{sm}_l$ is a Banach manifold and for any 
    %\marginpar{\bl Apr 25. $\cdot$ should be in $\FC^\mr{sm}_l$, while $A_0$ is a fixed point in $\FC^\mr{sm}$, right? {\red right}}
    $A_0 \in \FC^\mr{sm}$, the map $A_1 \mapsto \til\kappa_{A_0}(A_1 - A_0)$ sends $A_1 \in\FC^\mr{sm}_l$ to $\Omega^1_{A_0-\mr{cl},l}$ (point iii of Proposition \ref{prop: kuranishi}). We can now invoke the inverse function theorem to conclude that this map is surjective onto a neighborhood $V_{A_0}$ of zero. We then have for all $\alpha \in V_{A_0}$ that $\alpha = \til\kappa_{A_0}(m)$, with $m \in MC_{A_0}$. Therefore, $\til\delta_{A_0}(\alpha) = \til\delta_{A_0}\til\kappa_{A_0}(m) = m$ satisfies the Maurer-Cartan equation. % Hence, for $t$ small enough, covfefe
\end{proof}
    % Notice that in the proof of Proposition \ref{prop: deformation of A0}, we were using vanishing of $l_{n,A_0}'$, which only guarantees exactness
    % \marginpar{\bl Is it true that $\til\lambda|_{im(i_{A_0})}$ is exact? why cannot there be a coexact component?}
    % of $\til{\lambda}$ of those when we input forms in the image of $i_{A_0}$. Therefore,
This implies the following corollary. 
\begin{cor}\label{cor: banach smoothness}
    %The subset of smooth points $\FC^{\mr{sm}}_l$ is a Banach manifold.
    For each point $A_0 \in \FC^{\mr{sm}}_l$, the map 
    %\marginpar{\bl Changed $\varphi$ to $\til\varphi$ (to be consistent with notations of Thm 2.4)}
    $\til\varphi_{A_0}\colon \Omega^1_{A_0-\mr{cl}} \supset U_{A_0} \to \FC_l^\mr{sm}$
    is a chart with inverse given by Equation \eqref{eq: til phi inverse}. 
\end{cor}

In particular, Theorem \ref{thm: banach smoothness} follows. 
Finally, if $K_{A_0}$ is Fr\'echet continuous, we also have smoothness of $\FC^\mr{sm}$ in the Fr\'echet topology:   
\begin{cor}\label{cor: Frechet smoothness}
    The set of smooth points $\FC^{\mr{sm}}$ is a Fr\'echet manifold. If $K_{A_0}$ satisfies Assumption \ref{ass: Frechet}, then for each point $A_0 \in \FC^{\mr{sm}}$, restriction of $\til{\varphi}_{A_0}$ to $\FC^\mr{sm}$ defines a chart 
    %\marginpar{\bl Changed $\varphi$ to $\til\varphi$ (to be consistent with notations of Thm 2.4)}
    $\til\varphi_{A_0}\colon \Omega^1_{A_0-\mr{cl}} \supset U_{A_0} \to \FC^\mr{sm}$.
    %\marginpar{\bl Apr 28. Sobolev completions should not be there, right? (in the Cor, two places) {\red no, they should not! removed them}}
\end{cor}
\begin{proof}
    The maps $\til\kappa_{A_0}$ and $\til\delta_{A_0}$ are both Fr\'echet continuous and inverse to each other, hence they are inverse homeomorhisms in the Fr\'echet topology when restricted to $\FC^\mr{sm}$ and $\Omega^1_{A_0-\mr{cl}}$ respectively. Also, note that $\til{\kappa}_{A_0}$ and $\til\delta_{A_1}$ are both smooth when considered on open neighborhoods of zero in $\Omega^1$. Therefore, if $A_1 \in \til\varphi_{A_0}(U_{A_0})$, then the composition $\til\varphi^{-1}_{A_1} \circ \til\varphi_{A_0}$ is smooth as a map between neighborhoods of zero in $\Omega^1$.\footnote{In Fr\'echet spaces derivatives are defined using directional (Gateaux) derivatives for which the chain rule holds.} Therefore, its restriction to the subspace $\Omega^1_{A_0-\mr{cl}}$ is also smooth.  
\end{proof}
This finishes the proof of Theorem \ref{thm: Frechet smoothness}. 
\begin{rem}
    Given a (small) closed form $\alpha \in \Omega^1_{A_0-\mr{cl}}$, we now have two different ways to lift it to a flat connection, namely as $\til{\psi}_{A_0}(\alpha)$, or $\til{\varphi}_{A_0}(\alpha)$. By definition, their restrictions to $\operatorname{Im} i_{A_0}$ agree, and they agree up to first order. However, from second order, they disagree. Considering for example $\alpha= \dd_{A_0}\beta$, we have 
    $$\til\psi_{A_0}(\dd_{A_0}\beta) = A_0 + \dd_{A_0}\beta -\frac{1}{2}[\beta,d_{A_0}\beta] + \frac12[\beta,[\beta,\dd_{A_0}\beta]] + O(\beta^4)$$ 
    whereas the sum-over-trees map is 
    \begin{multline*}
        \til{\varphi}_{A_0}(\dd_{A_0}\beta) = A_0 + \dd_{A_0}\beta - \frac{1}{2}K_{A_0}[\dd_{A_0}\beta,\dd_{A_0}\beta] + \frac12K_{A_0}[\dd_{A_0}\beta, K_{A_0}[\dd_{A_0}\beta,\dd_{A_0}\beta]] \\ 
 = A_0 + \dd_{A_0}\beta - \frac12[\beta,d_{A_0}\beta] + \frac12 d_{A_0}K_{A_0}[\beta,d_{A_0}\beta] +\frac12 i_{A_0}p_{A_0}[\beta,d_{A_0}\beta] + O(\beta^3).
    \end{multline*}
\end{rem}
%\marginpar{Banach convergence is ok, but what about uniform convergence/Frechet norms?}
{
%\color{green}
%alternate proof: use defining equation for $\delta$ and Inverse Function Theorem}
%{\color{red}

%{\color{red}

%\marginpar{\bl Apr 18: swapped subsections 2.2 and 2.3 (need irred connections for the new Prop \ref{prop: deformation of harmonic forms with any path of A}}
\subsection{Irreducible points}
% It is well known that every flat connection $A_0$ on a principal $G$-bundle induces a representation of the fundamental group $\rho_{A_0}\colon \pi_1(M) \to G$. 
% \begin{defn}
% Suppose that $G \subset GL(n,\C)$.  We say a flat connection $A_0$ is \emph{reducible} if the corresponding representation $\rho_{A_0}$ is reducible, i.e there is $\rho_{A_0}$-invariant subspace $V \subset \C^n$. Otherwise we say $A_0$ is irreducible. 
% \end{defn}
% \marginpar{\bl Oct 11. missing ref}
% \begin{prop}
%     A flat connection is irreducible if and only if $H^0_{A_0}(M,\g) = 0$.[?]
% \end{prop}
% In other words, a flat connection $A_0$ is irreducible if its stabilizer under the gauge action is minimal, i.e. equal to the center of the gauge group. 
\begin{defn}
    We say that a flat connection is \emph{irreducible} if  $H^0_{A_0}(M,\g) = 0$.
\end{defn}

The following example shows that connections can define smooth points without being irreducible.
\begin{expl}
    Let $G = SU(2)$. For $p > 2$ prime, on a lens space $L(p,q)$ with fundamental group $\Z_p$, there are, up to conjugation $\frac{p+1}{2}$ different representations labeled by $k = 0,1,\ldots, \frac{p-1}{2}$ defined by 
    $$\rho_k(\gamma) = \begin{pmatrix} e^{2\pi i k/p} & 0 \\ 0 & -e^{2\pi i k/p} \end{pmatrix},$$ 
    with $\gamma$ the generator of the fundamental group.
    Clearly those representations are reducible, but for $k\neq 0$ all the induced $L_\infty$ operations vanish as $H^1 = 0$ and $H^0 = \mathfrak{t}$ is the abelian subalgebra of diagonal matrices. 
\end{expl}
We denote the irreducible locus by $\FC^\mr{irr}$. On the irreducible locus, the quotient of the gauge group by its center acts freely and properly. Therefore, the quotient of the smooth irreducible locus 
\begin{equation}
    \FC' = \FC^\mr{sm} \cap \FC^\mr{irr}
\end{equation}
by the gauge group is a smooth manifold that we denote by $\M' \subset \M$.

\subsection{Transporting ``harmonic'' forms}  %\marginpar{\bl In this subsection we are more specific about the gauge-fixing and assume that it is Lorentz? (unlike in the rest of section 2?)}
%{\bl In this subsection we consider $(i,p,K)$ at $A_0$ to be the SDR data associated to the usual Hodge decomposition twisted by $A_0$.}\marginpar{\red I would not say that we need this assumption. Probably we can just replace harmonic with $i_{A_0}(H^\bullet_{A_0}$}

For a given smooth flat connection $A_0$ and SDR data $(i_{A_0},p_{A_0},K_{A_0})$ we will call forms in $\operatorname{Im}(i_{A_0}) = \ker d_{A_0} \cap \ker K_{A_0}$ ``harmonic.''\footnote{For Hodge SDR data, this is the space of harmonic forms in the usual sense of the word.}
\begin{prop}\label{prop: deformation harmonic}
%\marginpar{slightly edited}
Suppose we are given a smooth flat connection $A_0$ and a harmonic 1-form $\sfa$. Let 
\begin{equation}\label{prop: deformation harmonic A_t = A_0+...}
A_t = A_0 + \til\delta_{A_0}(t\sfa) = \til{\varphi}_{A_0}(t\sfa)
\end{equation}
be the path of flat connections given %corollary \ref{cor: convergence}. \marginpar{\bl ... 
by Proposition \ref{prop: deformation of A0}. %?}
Let $\chi$ be another harmonic form. Then, there exists a deformation 
\begin{equation}\label{chi_t}
    \chi_t = \sum_{k\geq 0} t^k \chi^{(k)}
\end{equation}
with $\chi^{(0)}= \chi$, such that $d_{A_t}\chi_t =  K_{A_0}\chi_t = 0$.
\end{prop}
\begin{proof}
    We have 
    $$d_{A_t} \chi_t = \left(d_{A_0} + \sum_{k\geq 1}t^k\mathrm{ad}_{\alpha^{(k)}}\right)\left(\sum_{l\geq 0}t^l\chi^{(l)}\right) = 0$$ 
    and again we can look at the equations in powers of $t$: 
    \begin{equation}
        d_{A_0}\chi^{(n)} = -\sum_{k+k' = n, k\geq 1} \mathrm{ad}_{\alpha^{(k)}}\chi^{(k')}.
    \end{equation}
Similarly to the proof of Proposition \ref{prop: deformation of A0}, the right hand side here is a representative of the induced $L_\infty$-operation $l_n'(\sfa,\ldots,\sfa,\chi)$, its vanishing in cohomology implies that it is exact and that we can set 
\begin{equation}
    \chi^{(n)} = -K_{A_0}\sum_{k+k' = n, k\geq 1} \mathrm{ad}_{\alpha^{(k)}}\chi^{(k')}. 
\end{equation}
\end{proof}
\begin{rem}\label{rem: sum over trees d til varphi}
    Again, one has a similar sum-over-trees formula for $\chi^{(n)}$, namely it is a sum over rooted binary trees  with $n$ leaves where one leaf is labeled with $\chi$. I.e. we have that 
    \begin{equation}
        %\chi_1 = (d\til{\varphi}_{A_0}(a))_a.
        \chi_t = (d\til{\varphi}_{A_0})_{t\sfa}(\chi).
    \end{equation}
\end{rem}

%}
\begin{rem}\label{rem 2.18}
    One can also understand the deformation (\ref{chi_t}) of a harmonic form $\chi$ as $\chi_t=i_t ([\chi])$ with $i_t=\sum_{n\geq 0}(-K_{A_0}\ad_{\til\delta_{A_0}(t\sfa)})^n\circ i_{A_0}$ the deformation of inclusion $i_{A_0}\colon H_{A_0}^\bt\hookrightarrow \Omega^\bt$ of cohomology as harmonic forms, induced via homological perturbation lemma from the deformation of the differential $d_{A_0}\ra d_{A_t}=d_{A_0}+\ad_{\til\delta_{A_0}(t\sfa)}$ on $\Omega^\bt$. Note that the corresponding induced differential on $H_{A_0}^\bt$ is $d'_t\colon=\sum_{n\geq 1}\frac{1}{n!} l'_{n+1}(\underbrace{t\sfa,\ldots,t\sfa}_n,-)$; it vanishes since $A_0$ is assumed to be a smooth point. Hence, $d_{A_t}i_t [\chi]= i_t d'_t [\chi]=0$.
\end{rem}

%\marginpar{\bl added this generalization Apr 18: needed for Remark \ref{rem: change of harm form under change of gf operator}}
Under the additional assumption that $A_0$ is irreducible, Proposition \ref{prop: deformation harmonic} can be generalized to paths $A_t$ not restricted to the form (\ref{prop: deformation harmonic A_t = A_0+...}):

\begin{prop}\label{prop: deformation of harmonic forms with any path of A}
    Let $A_0\in \FC'$ be a smooth irreducible flat connection and let $A_t=A_0+\alpha_t$ with $\alpha_t %=\sum_{n\geq 1}t^n \alpha^{(n)}
    $ be a path of flat connections with $\alpha_0=0$. Fix a harmonic form $\chi$. Then the path of forms
    \begin{equation}\label{prop: deformation of harmonic forms with any path of A, chi_t}
        \chi_t=\sum_{k\geq 0} (-K_{A_0} \ad_{\alpha_t})^k\circ \chi
    \end{equation}
    starting at $\chi_{0}=\chi$ satisfies  $d_{A_t}\chi_t=K_{A_0}\chi_t=0$.
\end{prop}
\begin{proof}
Property $K_{A_0}\chi_t=0$ is immediate by construction (\ref{prop: deformation of harmonic forms with any path of A, chi_t}). The nontrivial part is proving $d_{A_t}\chi_t=0$.

    Due to irreducibility of $A_0$, we only have harmonic 1-forms and 2-forms. We consider these two cases for $\chi$ separately.

    Let $\chi$ be a harmonic 1-form. Then $\chi_t = (d\til\varphi_{A_0})_{\sfa_t}(\chi)$, where $\sfa_t$ is a path of %$d_{A_0}$-
    closed forms such that $\til\varphi_{A_0}(\sfa_t)=\alpha_t$, cf. Theorem \ref{thm: banach smoothness}. Since $\til\varphi_{A_0}(\alpha)$ maps closed 1-forms to flat connections, its derivative maps closed 1-forms to first order deformations of flat connections, i.e., $d_{\til\varphi_{A_0}(\alpha)}$-closed forms.
    Therefore, $d_{A_t}\chi_t=0$.

    Now, let $\chi$ be a harmonic 2-form. We have
    \begin{multline*}
        d_{A_0}\chi_t=\underbrace{d_{A_0}\chi}_0-\sum_{k\geq 1} d_{A_0}K_{A_0}\ad_{\alpha_t}(-K_{A_0}\ad_{\alpha_t})^{k-1}\circ\chi\\
        = - \sum_{k\geq 1}\ad_{\alpha_t}(-K_{A_0}\ad_{\alpha_t})^{k-1}\circ\chi=- \ad_{\alpha_t} \chi_t,
    \end{multline*}
    which implies $d_{A_t}\chi_t=0$. 
    Here in transition to the second line, we note that the operator $d_{A_0}K_{A_0}=\mr{id}-i_{A_0}p_{A_0}-K_{A_0}d_{A_0}$ acts on a 3-form $\beta$, and thus this operator acts as identity ($p_{A_0}\beta=0$ by irreducibility of $A_0$ and $d_{A_0}\beta=0$ by degree reason).
\end{proof}

\subsection{Exponential maps} \label{ssec: exp maps}
% \marginpar{REWRITE}
% % \begin{cor}\label{cor: convergence}
% %     If $A_0$ is a smooth point, then there is an open set $U \subset 
% %     H^1_{d_{A_0}}(M,\g)$  containing 0 such that $\delta(\alpha)$ converges for $\alpha \in U$.
% % \end{cor}
% % \marginpar{converges in what sense? do we need a Banach norm on $Omega^1$ maybe?}
% % \begin{proof}
% %     This follows from the fact that the formal power series $\delta$ has a compositional inverse $\kappa$ which is clearly convergent and whose derivative is invertible in a neighbourhood of $\delta = 0$. 
% % \end{proof}
% \marginpar{\bl I didn't understand whether the domain of $\phi_{A_0}$ is supposed to be  closed 1-forms or 1-cohomology}
% Put differently, we have a map $\varphi_{A_0}\colon U\subset H^1_{d_{A_0}-cl} \to \mathcal{A}_0$ with the properties that $\varphi_{A_0}(0) = A_0$ and $(d_\alpha\varphi_{A_0})_{\alpha=0} = {\red i}_{A_0}$.
% %In particular, $\varphi$ defines a generalized exponential map on $\mathcal{A}_0$ and its jet at zero a formal exponential map.
% Our goal is to show that this map descends to the moduli space of flat connections, which follows from the following lemma:
% % \begin{lem}
% %     Let 
% % \end{lem}
% % \begin{lem}
% %     We have 
% %     \begin{equation}
% %         \varphi_{A_0}[\alpha + d_{A_0}\beta] = (\varphi_{A_0}\alpha)^g
% %     \end{equation}
% %     with $g = \exp \beta$. 
% % \end{lem}
The upshot of the previous discussion is the following. Suppose that we are given a smooth family $(i,p,K)$ of SDR data  on the smooth irreducible locus $\FC' \subset \FC$. Then, we have two exponential maps 
    $\til{\varphi}$ and $\til{\psi}$ on $\FC'$, defined on an open neighborhood $\til{U} \subset T\FC'$ of the zero section, which agree on $\til{U} \cap \im i$. We denote by $\mathbb{H}$ the cohomology bundle over $\FC'$ - the graded vector bundle with fiber over $A_0$ given by $H^\bullet_{A_0}$, and by $U \subset \mathbb{H}$ the preimage of $\til{U}$ under $i$. Then, by restriction of $\til{\varphi}$, we have the map $\varphi = \til{\varphi} \circ i \colon U \to \FC'$. 
\begin{lem} Suppose the family $(i,p,K)$ is equivariant with respect to the action of the gauge group, i.e. 
    for all $\sfa \in H^\bullet_{A_0}$ and $\alpha \in \Omega^\bullet$ we have 
    \begin{equation}
        i_{{}^\sfg A_0}({}^\sfg a) = {}^\sfg (i_{A_0}a), \quad 
        p_{{}^\sfg A_0}({}^\sfg \alpha) = {}^\sfg (p_{A_0}\alpha), 
        \quad K_{{}^\sfg A_0}({}^\sfg \alpha) = {}^\sfg (K_{A_0}\alpha).
    \end{equation}
Then the map $\varphi$ is equivariant with respect to gauge transformations, 
    \begin{equation}
        \varphi_{{}^\sfg A_0}({}^\sfg \alpha) = {}^\sfg (\varphi_{A_0}\alpha).
    \end{equation}
\end{lem}
\begin{proof}
    The only ingredients of the map $\varphi$ are the chain homotopy $K_{A_0}$ and the Lie bracket, which are both equivariant with respect to gauge transformations. 
\end{proof} 
In particular, the map $\varphi$ descends to the moduli space and defines a generalized exponential map that we denote by $\underline{\varphi}$: %the same letter 
\begin{equation}\label{phi underline}
    \underline{\varphi} \colon U \subset T\M \to \M, \qquad ([A_0],[\alpha] \mapsto [\varphi_{A_0}\alpha]).
\end{equation}

\subsubsection{Grothendieck connection} 
%\marginpar{New subsubsection - are these claims correct?}
The exponential maps $\til{\varphi}$ and $\underline{\varphi}$ induce   connections on the tangent bundles (viewed as fiber bundle) of $\FC'$ and $\M'$. These connections are sometimes called the \emph{Grothendieck connections} (see \cite{Cattaneo2001a},\cite{Cattaneo2002}, \cite{Cattaneo2019}, \cite{Cattaneo2020}). Here we present a slightly %{\bl \cancel{alternative} 
different approach. Namely, let $[A],[\til{A}] \in \M'$ and $\alpha \in T_A\M'$. If $[A]$ and $[\til{A}]$ are 
close enough, there exists $\til{\alpha} \in T_{\til{A}}\M'$ such that 
\begin{equation}
    \underline{\varphi}_{\til{A}}\til{\alpha} = \underline{\varphi}_A\alpha \qquad\text{or} \qquad \til{\alpha} = \underline{\varphi}^{-1}_{\til{A}}(\ul{\varphi}_A\alpha). 
\end{equation}
\begin{defn}\label{def: Grothendieck connection}
    The \emph{Grothendieck connection} $\nabla^G$ is the fiber bundle connection on $U \subset T\M'$ whose parallel transport of $\alpha \in U_A$ from $A$ to $\til{A}$ is given by 
    \begin{equation}\label{Grothendieck conn from phi underline}
        \til{\alpha} = \underline{\varphi}^{-1}_{\til{A}}(\underline{\varphi}_A\alpha). 
    \end{equation}
\end{defn}
In other words, if $A_t$ is a path of flat connections starting at $A$, then the horizontal lift of this path starting at $\alpha$ is given by $\alpha_t = \underline{\varphi}_{A_t}^{-1}\underline{\varphi}_A\alpha$. 
From the definition, is it obvious that this connection is flat, since its parallel transport between any two (close enough) points $[A],[\til{A}]$ is independent of the choice of a path between them. 

%{\bl \marginpar{\bl added Aug 20}
\begin{rem}\footnote{See e.g. \cite{Bonechi2012}.}
 The role of $\nabla^G$ is that it ``recognizes'' global objects. %among formal objects. 
 More explicitly, $\nabla^G$ induces a connection in the bundle $\wh{\mr{Sym}^\bt}T^*\M'$ of formal functions on $\M'$ -- let us also denote it $\nabla^G$ by abuse of notations. Then a section $\sigma$ of $\wh{\mr{Sym}^\bt} T^*\M'$  (a formal function) is of the form $\mr{T}\underline{\varphi}^* f$ for some $f\in C^\infty(\M')$ (an actual, ``global,'' function) if and only if $\sigma$ is horizontal w.r.t. $\nabla^G$:  
 \begin{equation}
 \nabla^G\sigma=0.
 \end{equation} 
 Here $\mr{T}$ stands for the Taylor expansion of a function on $U$ in vertical (tangent) coordinates on $T\M'$.
 This discussion applies to any manifold with a formal exponential map, not just $(\M',\underline{\varphi})$. Also,
 one can replace functions with half-densities (especially relevant for BV formalism), differential forms, spinors, etc.
\end{rem}

% \begin{cor}
% Denote $\M(M,P)^{smooth} $ the union of all smooth points in $\M(M,P)$. Then $\M(M,P)^{smooth} = \bigsqcup_i X_i$ where $X_i$ is a smooth manifold whose dimension is equal to the dimension of $H^1_{d_{A_0}}(M,\mathrm{Ad} P)$ for all $A_0$ with $[A_0] \in X_i$. 
% \end{cor}
% \begin{cor}
% Denote $\M(M,P)^{smooth} $ the union of all smooth points in $\M(M,P)$. Then $\M(M,P)^{smooth} = \bigsqcup_i X_i$ where $X_i$ is a smooth manifold whose dimension is equal to the dimension of $H^1_{d_{A_0}}(M,\mathrm{Ad} P)$ for all $A_0$ with $[A_0] \in X_i$. 
% \end{cor}

\subsection{Gauge fixing operators}
We now specialize to SDR data defined by gauge-fixing operators. 
\begin{defn}[Gauge-fixing operator]
    We say that $\sh\colon \Omega^{\bullet}(M,\g) \to \Omega^{\bullet -1}(M,\g)$ is a gauge fixing operator for $d_{A_0}$ if the operator\footnote{Here we are using the graded commutator. Since $d_{A_0}$ and $\sh$ have degree $+1$ and $-1$ respectively, this means $[d_{A_0},\sh] = d_{A_0}\sh + \sh 
    d_{A_0}$.} 
    \begin{equation} \HHH=\HHH_{d_{A_0},\sh}:=[d_{A_0},\sh] \colon \Omega^{\bullet}(M,\g) \to \Omega^{\bullet}(M,\g) 
    \end{equation}
    is a generalized Laplacian, i.e. has symbol $\sigma_2(H)(x,\xi) = |\xi|^2$.
\end{defn}
\begin{expl}
    If $g$ is a Riemannian metric on $M$, then the formal adjoint $d^*_{A_0}$ of $d_{A_0}$ is a gauge fixing operator, with $\HHH_{d_{A_0},d^*_{A_0}} = \Delta_{A_0}$ the (twisted) Hodge-de Rham Laplacian. In fact, if $A'$ is a different flat connection, then $d^*_{A'}$ is still a gauge-fixing operator for $d_{A_0}$, because the difference 
    $$\HHH_{d_{A_0},d^*_{A_0}} - \HHH_{d_{A_0},d^*_{A'}} = [d_{A_0}, \ad_{A_0-A'}^*] $$ 
    is a first-order differential operator. 
\end{expl}

%\marginpar{Can we construct an $(i,p,K)$ triple for any $Q$? what about the 1-loop part of the partition function - the regularized superdeterminant and signature? Does it need an additional metric?}

\subsubsection{Good gauge fixing operators}
Let $\sh$ be a gauge fixing operator for $d_{A_0}$ and $\HHH= [d_{A_0},\sh]$  the corresponding generalized Laplacian.  \begin{defn} 
We say that $\sh$ is a good gauge fixing operator if
\begin{enumerate}
\item $\sh$ is skew-selfadjoint with respect to Poincar\'e pairing,
    \item $\sh^2 = 0$, 
    \item the eigenvalues of $\HHH$ have nonnegative real part,
    \item there is a Hodge decomposition
    %\marginpar{\bl added the underbrace}
     \begin{equation}%\label{Hodge decomp induced by gauge-fixing operator}
     \Omega = \ker \HHH \oplus \underbrace{\operatorname{im} d_{A_0} \oplus \operatorname{im} h}_{\operatorname{im} \HHH} , \label{eq: decomposition disj}
     \end{equation}
     \item we have $\ker \HHH \cong H^\bullet_{A_0}$.
\end{enumerate}
\end{defn}
%{\color{gray}
%\marginpar{\bl We should change the definition of $G$ to $\int_0^\infty dt e^{-t P}\kappa_t$, so that it acts correctly on ``harmonic'' part.{\red if we do not need the heat kernel later, we might just skip this}}
%{\red Denote  $\kappa_t$ the heat kernel of $\HHH$ and $G = \int_0^\infty \kappa_t dt $ its Green's function.} 
%\marginpar{\bl maybe use a different font for $K_t$ as a heat kernel, not to confuse it with chain homotopy? {\red used $\kappa_t$}}
%Then we have for all $v$ that
%\marginpar{\bl We should probably also assume that $Re(\HHH)\geq 0$}
%\marginpar{\bl I think $G$ (Green's fun) was not introduced until this point}
%\begin{align*}  \HHH G v = \HHH \int_0^\infty dt\, \kappa_t v &= \int_0^\infty dt\, \HHH \kappa_tv = -\int_0^\infty dt\, \partial_t\kappa_tv \\ 
%&= \lim_{t\to 0}\kappa_tv-\lim_{t\to \infty} \kappa_tv  =  v-Pv
%\end{align*}
%\marginpar{\bl double check this part}
%}
Denote $P$  the %orthogonal \marginpar{\color{purple} I think the projection $P$ is non-orthogonal for the desync Hodge decomposition}
projection onto the kernel of $\HHH$ along the image of $\HHH$. 
%lor{gray} This shows we have $\Omega = \ker \HHH + \operatorname{im} \HHH$ 
%\marginpar{\bl why is it obvious that $\Omega = \ker \HHH + \mr{im}\HHH$? - I believe it's true if $\HHH$ is diagonalizable, but then it is automatically a direct sum.} 
%and obviously we have $\operatorname{im}\HHH \subset \operatorname{im} d_{A_0} + \operatorname{im} \sh$ and therefore
 %\begin{equation}
%     \Omega = \ker \HHH + \operatorname{im} d_{A_0} + \operatorname{im} \sh \label{eq: decomposition}
% \end{equation} 
% }
 The operator $\HHH + P$ is invertible and we denote $G:=(\HHH+P)^{-1}$ its inverse. It satisfies 
 \begin{equation}
     \HHH G = G \HHH = \mathrm{id} - P. 
 \end{equation}
Also, defining $K = \sh \circ G$ we have $[d_{A_0},K] = \operatorname{id} - P$. 
%{[\red $P - \operatorname{id}$?]}. 

For good gauge-fixing operators, we thus have a strong deformation retraction (SDR)
\begin{equation}\label{eq: SDR good gf}
\begin{gathered}
    i\colon H^\bullet_{A_0}\cong \ker \HHH \hookrightarrow \Omega^\bt, \\
    p\colon \Omega^\bt \twoheadrightarrow \ker \HHH \cong H^\bullet_{A_0}, \\ 
    K= \sh \circ G \colon \Omega^\bullet \to \Omega^{\bullet - 1}.
\end{gathered}
\end{equation}
\begin{expl}[Hodge decomposition] \label{ex: hodge decomposition}
The main example of a good gauge fixing operator is, given the choice of a Riemannian metric on $M$, the codifferential $d^*_{A_0}$. The fact that $d^*_{A_0}$ is a good gauge-fixing operator follows from the well-known Hodge decomposition.\footnote{ 
See e.g. \cite[Section 3]{Ray1971}, \cite{Mathai2011} for details on Hodge-de Rham decomposition theorem for forms on a Riemannian manifold with coefficients in a flat bundle.
}
%\marginpar{\bl Apr 19: added footnote and citations.}
In this case \begin{itemize}
\item The operator $\HHH = \Delta_{A_0}$ is the Hodge-de Rham Laplacian (twisted by the flat connection $A_0$),
\item the map $i_{A_0} \colon H^\bullet_{A_0} \to \ker \Delta_{A_0}$ is the isomorphism between de Rham cohomology and harmonic forms, composed with the inclusion into $\Omega^\bullet$,
    \item  the decomposition (\ref{eq: decomposition disj}) is orthogonal,
     \item and $p_{A_0} = i_{A_0}^{-1} P_{A_0}$ is the orthogonal projection to harmonic forms, composed with the isomorphism with de Rham cohomology. 
\end{itemize}
Moreover, the family of SDR data defined by $A_0 \mapsto (i_{A_0},P_{A_0},K_{A_0})$ defines a global, equivariant family of SDR data and in particular induces a formal exponential map on $\M'$ as explained in Section \ref{ssec: exp maps}.
%\marginpar{\bl Apr 18 $p\geq$ ?}
\begin{prop}\label{prop: bounded}
    On a compact manifold $M$, for any $A_0$ and $k \in \N, p \geq 2$, the SDR data $(i_{A_0},p_{A_0},K_{A_0})$ is bounded with respect to the Sobolev norm $||\cdot||_{k,p}$. In particular, the SDR data is Fr\'echet continuous (satisfies Assumption \ref{ass: Frechet}%\footnote{The norm $||\omega||_{k,p}$ is defined as the sum of the $p$-norms of $\omega$ and its covariant derivatives up to order $k$.}
\end{prop}
\begin{proof}
The maps $i_{A_0}$ and $p_{A_0}$ are bounded since they have finite-dimensional domain (resp. codomain). %\marginpar{\bl Apr 28: Either "denoting" should not be there, or "by $\Omega^\bt_{k,p}$" is missing {\red added it}}
Denoting the completion of $\Omega^\bullet$ with respect to $||\cdot||_{k,p}$ by $\Omega^\bullet_{k,p}$, the map $K$ is bounded since it is a composition of the bounded maps $G_{A_0} \colon \Omega^\bullet_{k,p} \to \Omega^\bullet_{k+2,p}$ and $d^*_{A_0}\colon \Omega^\bullet_{k+2,p} \to \Omega^{\bullet -1}_{k+1,p}$. Therefore it defines a continuous map $K_{A_0}\colon\Omega^\bullet_{k,p} \to \Omega^{\bullet -1}_{k+1,p}$, and the injection $\Omega^{\bullet -1}_{k+1,p} \to \Omega^{\bullet -1}_{k,p}$ is continuous, which proves the first claim. \footnote{Since the inclusion $\Omega^{\bullet -1}_{k+1,p} \to \Omega^{\bullet -1}_{k,p}$ is compact on a compact manifold by the Rellick-Kondranov theorem, we have actually proven that $K_{A_0}$ is compact. } Since it is bounded with respect to all Sobolev norms, it is also continuous in the Fr\'echet topology. 
\end{proof}
\end{expl}
%{\color{gray} 
%\marginpar{\bl added Aug 5 (can remove if we don't need it..)}
%\subsubsection{Infinitesimal variation of a good gauge-fixing operator}
\begin{lem}[Variation of $\sh$]
    An infinitesimal variation of a good gauge-fixing operator $\sh\ra \sh+\delta\sh$ induces the  following first-order deformation of the SDR (\ref{eq: SDR good gf}): $i\ra i+\delta i$, $p\ra p+\delta p$, $K\ra K+\delta K$ with
\begin{equation}
    \delta i=-d_{A_0} \mathbb{I}_{\delta\sh} i,\quad \delta p =-p \mathbb{P}_{\delta\sh} d_{A_0},\quad \delta K= [d_{A_0},\Lambda_{\delta \sh}]+P \mathbb{P}_{\delta\sh}+\mathbb{I}_{\delta \sh}i.
\end{equation}
Here we denoted
\begin{equation}
    \mathbb{I}_{\delta\sh}=G\delta\sh,\quad \mathbb{P}_{\delta\sh}=\delta\sh G,\quad 
    \Lambda_{\delta\sh}=K\delta\sh G.
\end{equation}
\end{lem}
The proof is similar to the proof of Proposition \ref{prop:varAtprime}.
%}

\subsubsection{Desynchronized Hodge decomposition}
%\marginpar{\bl Aug5: promoted this to separate subsubsection}

%\marginpar{\bl Apr 18: added irreducibility requirement}
\begin{prop} \label{lem: Hodge decomposition close connections}
%\marginpar{\bl 
%Maybe this whould be phrased on FC rather than $\M$? What if $[A_1]$ is close to $[A_0]$ but $A_1$ is far from $A_0$? (I.e. $A_1$ is close to a large gauge tranf of $A_0$.)
%Edited
%}
    Let $A_0\in \FC'$ be a smooth irreducible flat connection. 
    %\sout{ such that $[A_0]\in \M$ is a smooth point}. 
    %and irreducible, 
    Then there is a neighborhood $U$ of $A_0$ in $\FC'$ such that,  for any $A' \in U$, $d^*_{A'}$ is a good gauge fixing operator for $d_{A_0}$. 
\end{prop}
%\marginpar{Oct 13, Added superscript $\sm$}
Before giving the proof we need to make the following remark.
%\marginpar{\bl Apr 18: added irreducibility requirement, referred to the new Prop \ref{prop: deformation of harmonic forms with any path of A}, so as to allow any path of $A'$, edited the remark. Notations: $\til\chi\ra \chi$. Renamed the path $A_t\ra A'_t$.}
\begin{rem}\label{rem: change of harm form under change of gf operator}
%\marginpar{\red This remark needs to be moved later after we have introduced GF operators. Is it ok here?}
Let $A_0$ be a smooth irreducible flat connection and $A'_t = A_0 + \alpha_t$, with $\alpha_t=\sum_{k \geq 1}t^k \alpha^{(k)}$, a path of flat connections starting at $A_0$.
    Dually to Proposition \ref{prop: deformation of harmonic forms with any path of A}, one can deform a harmonic form $\chi$ satisfying $d_{A_0}\chi = d^*_{A_0}\chi = 0$ to a path $\chi_t=\sum_{k\geq 0}t^k\chi^{(k)}$ with $\chi^{(0)}=\chi$, satisfying $d_{A_0}\chi_t = d_{A'_t}^*\chi_t = 0$, with 
    $$
    \chi^{(n)} = -d_{A_0}G_{A_0}\sum_{k+k' = n, k\geq 1}\mathrm{ad}^*_{\alpha^{(k)}}\chi^{(k')}.
    $$
    Equivalently,
    $$
    \chi_t=\sum_{k\geq 0}(-d_{A_0}G_{A_0}\ad^*_{\alpha_t})^k\circ \chi.
    $$
    This statement follows by applying Hodge star to (\ref{prop: deformation of harmonic forms with any path of A, chi_t}).
\end{rem}

\begin{proof}[Proof of Proposition \ref{lem: Hodge decomposition close connections}]
Property (1) follows from integration by parts and property (2) from $(d_{A'}^*)^2 = 0$. 

Remark \ref{rem: change of harm form under change of gf operator} implies that, for $A'$ an open neighborhood $U$ of $A_0$, the graded vector space 
\begin{equation}\label{W= space of (A_0,A_1)-harmonic forms}
W\colon=\ker d_{A_0}\cap \ker d^*_{A'}
\end{equation}
satisfies the following:
%\marginpar{\bl Apr 19: edited the proof}
\begin{enumerate}[(a)]
    \item 
    $W$ coincides with $\ker \HHH$ and has constant (graded) rank as $A'\in U$ moves away from $A_0$. Indeed, one has:
    \begin{enumerate}[(i)]
        \item $W\subset \ker\HHH$.
        \item At $A'=A_0$, $W=\ker\HHH$ (by usual Hodge decomposition).
        \item Rank of $W$ is non-decreasing when moving $A'$ away from $A_0$. %(in an open neighborhood). 
        Indeed, by Remark \ref{rem: change of harm form under change of gf operator}, one can deform a basis of $A_0$-harmonic forms to a linear independent set in $W$. (Since linear independence is an open condition.)
        \item Rank of $\ker\HHH$ is non-increasing as $A'$ is moved away from $A_0$. -- The rank of the kernel in a family of elliptic operator jumps (increases) on a closed subset in the space of parameters.
    \end{enumerate}
    These properties imply the claim above: $W=\ker\HHH$ and the rank is constant.
    %$W$  has constant (graded) rank and, since $W$ is contained in $\ker\HHH$  whose rank is non-increasing when moving $A'$ away from $A_0$ (in an open neighborhood) and since $W=\ker\HHH$ at $A'=A_0$, the rank of $\ker\HHH$ must stay constant. Hence, $\ker \HHH=W$ for $A'\in U$.
    \item $W$ contains a single representative of each $d_{A_0}$-cohomology class. Indeed, the map $q\colon W \ra H_{A_0}$ sending $\alpha\mapsto [\alpha]$ is surjective, since Remark \ref{rem: change of harm form under change of gf operator} 
    constructs a preimage under $q$ of a cohomology class $[\alpha]$ as a deformation of its harmonic representative $\chi$ along a path $A'_t$ from $A_0$ to $A'$, to a form $\til\chi$ satisfying $d_{A_0}\til\chi=d^*_{A'}\til\chi=0$.
    %defines a right inverse for $q$ -- a map $\rho\colon H_{A_0}\ra W$ satisfying $q\circ\rho=\mr{id}_{H_{A_0}}$. 
    On the other hand, by (a) $W$ has constant rank when $A'$ is changing, equal to the rank of $H_{A_0}$ at $A'=A_0$. Hence, the fact that $q$ is a surjection implies that it is in fact an isomorphism.
\end{enumerate}
Then, (a) together with (b) proves (5).

For (4), note that $\HHH$ is diagonalizable at $A'=A_0$ and hence is diagonalizable for $A'$ in a neighborhood of $A_0$ (since diagonalizability is an open condition). Thus, $\Omega= \ker\HHH \oplus \mr{im}\HHH$ for $A'\in U$, -- splitting into zero-modes of $\HHH$ and $\mr{im}\HHH=\colon V$ -- the span of eigenforms of $\HHH$ with nonzero eigenvalues. Operators $d_{A_0}$ and $d^*_{A'}$ act on $V$ (since they commute with $\HHH$). Moreover, one has
\begin{equation}\label{im(H) splitting}
\mr{im}\HHH=\mr{im}(d_{A_0})\oplus \mr{im}(d^*_{A'}) 
\end{equation}
Indeed, the intersection of the summands on the right is zero: if $d_{A_0}\alpha=d^*_{A'}\alpha=0$, then $\alpha\in \ker d_{A_0}\cap \ker d^*_{A'}=\ker\HHH$, but since $\alpha\in V$ it must be zero.  Also, if $\alpha\in V$ then $\alpha=d_{A_0} \beta + d^*_{A'}\gamma$ with $\beta=d^*_{A'}\HHH^{-1}\alpha$  and $\gamma=d_{A_0}\HHH^{-1}\alpha$ (here we are using that $\HHH$ is invertible on $V$). This proves that (\ref{im(H) splitting}) is a direct sum.

Property (3) is obvious by a continuity argument: in the deformation of $A'$ away from $A_0$ (in a small enough neighborhood), zero modes of $\HHH$ are deformed to zero modes while eigenvectors with eigenvalues $\lambda>0$ are deformed to eigenvectors with $\mr{Re}(\lambda)>0$.
\end{proof}

\begin{defn}\label{def: close connections}
\begin{enumerate} %\marginpar{Upgraded to defn since we need it often}
\item If $A_0$ and $U$ are as in Proposition \ref{lem: Hodge decomposition close connections} then for any $A' \in U$ we say that $(A_0, A')$ is a pair of \emph{close} flat connections. 
\item If $(A_0,A')$ is a pair of close flat connections, we call the space (\ref{W= space of (A_0,A_1)-harmonic forms}) the \emph{space of $(A_0,A')$-harmonic forms} and denote it $\mr{Harm}_{A_0,A'}$. 
We also call the associated decomposition (\ref{eq: decomposition disj}) the \emph{desynchronized Hodge decomposition}:
\begin{equation}
    \Omega(M,\g)=\mr{Harm}_{A_0,A'}\oplus \mr{im}\, d_{A_0} \oplus \mr{im}\, d^*_{A'}.
\end{equation}
\end{enumerate}
\end{defn}

Further in this section we will suppress the subscript in $A_0$ and just denote it $A$. 
\begin{rem}
    Notice that Proposition \ref{prop: bounded} is true also for the desynchronized Hodge gauge fixing, since we are only perturbing $\Delta_{A_0}$ by differential operators of degree one or less, so its inverse is still a degree -2 operator. 
\end{rem}

\subsubsection{Connection on the bundle of $(A,A')$-harmonic forms}\label{sss nabla^Harm}
%\marginpar{The remainder of this subsection could be moved into a separate Appendix {\bl doesn't make much sense now (probably an old marginpar)}}
Let $\mc{U}\subset \FC'\times \FC'$ be an open neighborhood of the diagonal in $\FC'\times \FC'$ obtained as the union of open sets $U$ from Proposition \ref{lem: Hodge decomposition close connections}.
Consider the vector bundle $\mr{Harm}$ over $\mc{U}$ whose fiber over $(A,A')$ is the space of $(A,A')$-harmonic forms $\mr{Harm}_{A,A'}$.

Consider the connection $\nabla^\mr{Harm}$ on the bundle $\mr{Harm}$ defined by infinitesimal parallel transport as follows. If $\chi\in \mr{Harm}_{A,A'}$ is a harmonic form, then:
\begin{enumerate}[(i)]
    \item For any $\alpha\in \Omega^1_{d_{A}\mr{-closed}}$, when moving from $(A,A')$ to $(A+t \alpha,A')$ on $\mc{U}$, $\chi$ transforms to 
    \begin{equation}\label{nabla^harm infinitesimal horizontal transport}
    \chi-t d^*_{A'}G_{A,A'}\mr{ad}_\alpha \chi\quad \in \mr{Harm}_{A+t\alpha,A'}.
    \end{equation}
    \item For any $\beta\in \Omega^1_{d_{A'}\mr{-closed}}$, when moving from $(A,A')$ to $(A,A'+s \beta)$ on $\mc{U}$, $\chi$ transforms to 
    \begin{equation}\label{nabla^harm infinitesimal vertical transport}
    \chi-s d_{A}G_{A,A'}\mr{ad}^*_\beta\chi\quad \in \mr{Harm}_{A,A'+s\beta}.
    \end{equation}
\end{enumerate}
The formulae above are written in the first order in deformation parameters $s,t$. 
One can consider $\nabla^\mr{Harm}$ as a connection in the trivial bundle over $\mc{U}$ with fiber $\Omega^\bt(M,\g)$ preserving the subbundle $\mr{Harm}$.
The covariant derivative operator associated with the connection $\nabla^\mr{Harm}$  is
\begin{equation}\label{nabla^Harm}
 \nabla^\mr{Harm}=\delta-G_{A,A'}\Big(d^*_{A'}\mr{ad}_{\delta A}+d_{A}\mr{ad}^*_{\delta A'}\Big)
\end{equation}
with $\delta$ the de Rham operator on $\FC'\times \FC'$.

\begin{rem}\label{rem: nabla^Harm as shift-and-project connection}
One can think of $\nabla^\mr{Harm}$ as a ``shift-and-project'' connection: its infinitesimal parallel transport takes 
an $(A,A')$-harmonic form $\chi$ over $(A,A')$ and moves it to the  $(A+t\alpha,A'+s\beta)$-harmonic form $P_{A+t\alpha,A'+s\beta}(\chi)$ over $(A+t\alpha,A'+s\beta)$.  We note that this construction is reminiscent of the construction of Hitchin's (projectively flat) connection in the Verlinde bundle\footnote{The vector bundle with fiber being the space of states of Chern-Simons theory on $\Sigma$, a.k.a. the space of WZW conformal blocks on $\Sigma$, a.k.a. the Verlinde space. See \cite{Axelrod1991a}.} over the moduli space of complex structures on a surface $\Sigma$.
\end{rem}

One has the following:
\begin{prop}\label{prop 3.5}
\begin{enumerate}[(a)]
    \item\label{prop 3.5 (a)}    The curvature of the connection $\nabla^\mr{Harm}$ (restricted to harmonic forms) is
    %\marginpar{\bl remove right $P$ factors? (actually, maybe it's good to leave them: the formula looks manifestly self-adjoint)}
    \begin{multline}\label{nabla^Harm curvature}
        F_{\nabla^\mr{Harm}}=-P\ad_{\delta A} G \ad^*_{\delta A'} P- P\ad^*_{\delta A'} G \ad_{\delta A} P\\ \in \Omega^{1,1}(\mc{U},\mr{End}(\mr{Harm}_{A,A'})),
    \end{multline}
    where we suppress subscripts in $P_{A,A'},G_{A,A'}$.
    In particular, $\nabla^\mr{Harm}$ is flat on $A'$-fixed and on $A$-fixed slices of $\mc{U}$.
    \item \label{prop 3.5 (b)} The restriction of $\nabla^\mr{Harm}$ to the diagonal in $\mc{U}\subset \FC'\times \FC'$ is a Euclidean connection %\marginpar{Euclidean or orthogonal? (terminology)} 
    -- it preserves the Hodge inner product on harmonic forms.
    \item \label{prop 3.5 (c)}
    Given a path $A_t$ of flat connections $0\leq t\leq 1$, from $A$ to $A'$, the parallel transport of an $(A,A')$-harmonic form $\chi$ along the path $(A_t,A')$ is $q(\chi)\in \mr{Harm}_{A',A'}$ with $q\colon \mr{Harm}_{A,A'}\xra{\sim} \mr{Harm}_{A',A'}$ the orthogonal projection onto $A'$-harmonic forms. Likewise, the parallel transport of $\chi$ along the path $(A,A_{1-t})$ is $p(\chi)\in \mr{Harm}_{A,A}$ with $p\colon \mr{Harm}_{A,A'}\xra{\sim} \mr{Harm}_{A,A}$ the orthogonal projection onto $A$-harmonic forms. 
    %Also, $pq^{-1}$ is an isometry.
\end{enumerate}
\end{prop}
\begin{figure}
    \centering
    \includegraphics[width=0.5\linewidth]{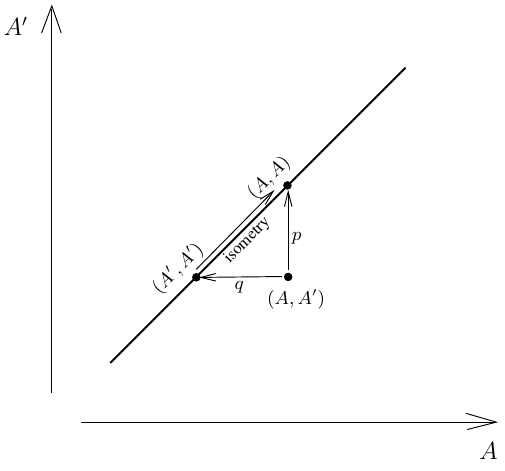}
    \caption{Vertical and horizontal parallel transport to the diagonal on $\FC'\times \FC'$ (Proposition \ref{prop 3.5} (\ref{prop 3.5 (c)})).}
    \label{fig:enter-label}
\end{figure}

%\begin{comment}
\begin{proof}
    (\ref{prop 3.5 (a)}): 
    % Consider the following connection $\nabla$ on the trivial bundle $\mc{U}\times \Omega^\bt(M,\g)\ra \mc{U}$:
    % \begin{equation}
    %     \nabla=\delta-\delta G^2_{A_0,A_1}= \delta - \delta G_{A_0,A_1} G_{A_0,A_1}-  G_{A_0,A_1} \delta G_{A_0,A_1}=
    %     \delta + G [\mr{ad}_{\delta A_0},d^*_{A_1}] 
    % \end{equation}
Note that the connection (\ref{nabla^Harm}) can be equivalently written as 
$$ 
\nabla^\mr{Harm}=\delta-G([d^*,\mr{ad}_{\delta A}]+[d,\mr{ad}^*_{\delta A'}]) = \delta+ G \delta \Delta.
$$
Therefore, the curvature (on harmonic forms) is 
\begin{multline}
(\nabla^\mr{Harm})^2 P = (\delta G \delta \Delta+ G \,\delta \Delta\, G \,\delta \Delta )P
\\
=\Big(G\Big(\underbrace{[d^*,\ad_{\delta A}]+[d,\ad^*_{\delta A'}]}_{-\delta \Delta}-K\ad_{\delta A}P+P\ad_{\delta A}K\\
-dG \ad^*_{\delta A'} P+P\ad^*_{\delta A'} Gd\Big) G\delta \Delta + G\delta\Delta G \delta\Delta \Big)P\\
=-P\ad_{\delta A}\underbrace{KGd}_{G P_\mr{coex}} \ad^*_{\delta A'} P- P\ad^*_{\delta A'} G\underbrace{d G d^*}_{P_\mr{ex}} \ad_{\delta A} P,
\end{multline}
which simplifies to (\ref{nabla^Harm curvature}). 
% Consider a connection on the trivial bundle $\ul\Omega$ over $\mc{U}$ with fiber $\Omega^\bt(M,\g)$ defined as
%     \begin{equation}\label{nabla s 3.1.1}
%         \nabla= \delta-\delta P_{A_0,A_1}
%     \end{equation}
%     where $P_{A_0,A_1}$ is the projector onto $(A_0,A_1)$-harmonic forms in the desynchronized Hodge decomposition $\Omega=\mr{Harm}_{A_0,A_1}\oplus \mr{im}(d_{A_0})\oplus \mr{im}(d^*_{A_1})$. Note that:
%     \begin{enumerate}
%         \item $\nabla$ is a gauge transformation of the trivial connection $\delta$ with generator $e^{-P_{A_0,A_1}}$. In particular, it is automatically flat.
%         \item Restricted to harmonic forms, $\nabla$ yields the connection (\ref{nabla^Harm}) (cf. (\ref{eq:varPharm})). Hence, $\nabla^\mr{harm}$ is flat.
%     \end{enumerate}

(\ref{prop 3.5 (b)}): Infinitesimal parallel transport along the diagonal in $\FC'\times \FC'$, from $(A,A)$ to $(A+t\alpha,A+t\alpha)$ transform a harmonic form $\chi \in \mr{Harm}_{A,A}$ to $\chi'=\chi-t d^*_{A}G_{A,A} \mr{ad}_\alpha \chi -t d_{A} G_{A,A} \mr{ad}^*_{\alpha}\chi$. Note that the three summands in $\chi'$ are mutually orthogonal and two of them are of order $t$, hence $||\chi'||=||\chi||+O(t^2)$. Therefore, if $A_t$ is a path of flat connections and $\chi_t\in \mr{Harm}_{A_t,A_t}$ is the parallel transport of $\chi$ along the corresponding path in the diagonal in $\FC'\times \FC'$, then $\frac{d}{dt}||\chi_t ||=0$.

(\ref{prop 3.5 (c)}): The form of the connection (\ref{nabla^harm infinitesimal horizontal transport}) implies that the parallel transport from $(A',A')$ to $(A,A')$ transforms an $A'$-harmonic form $\psi$ to $\psi+d^*_{A'}(\cdots)=\chi$. Hence, the reverse parallel transport transforms an $d^*_{A'}$-closed form $\chi$ to its projection $q(\chi)$ onto $A'$-harmonic forms. The case of moving from $(A,A')$ to $(A,A)$ is analogous. 
%Since the connection is flat and Euclidean along the diagonal, $p q^{-1}$ is an isometry.
\end{proof}

\begin{rem}\label{rem 3.9}
The connection $\nabla^\mr{Harm}$ is a rephrasing of the result of Proposition \ref{prop: deformation harmonic} and Remark \ref{rem: change of harm form under change of gf operator} (and for paths considered in that Proposition and Remark, $\nabla^\mr{Harm}$ gives the same parallel transport).
\end{rem}
\subsubsection{Cohomology comparison map}\label{sss cohomology comparison map}
%\marginpar{\bl new subsection, May 17}
Let $\mathbb{H}$ be the ``cohomology bundle'' over $\FC'$ -- the graded vector bundle with the fiber over $A$ being $H^\bt_{A}$. For a fixed $A'\in \FC'$, let 
$\mc{U}_{A'}\colon=\mc{U}\cap (\FC'\times\{A'\})$ -- the $A'$-fixed slice of $\mc{U}$.
The connection $\nabla^\mr{Harm}$ of Section \ref{sss nabla^Harm} restricted to $\mc{U}_{A'}$ induces (via the isomorphism $\mr{Harm}_{A,A'}\cong H_{A}$, $\chi\mapsto [\chi]$) a flat connection $\nabla^{\mathbb{H},A'}$ in $\mathbb{H}|_{\mc{U}_{A'}}$.

For $A,\til{A}$ a pair of close flat connections, close to $A'$, we will call the parallel transport of $\nabla^{\mathbb{H},A'}$ -- the linear map 
\begin{equation}\label{cohomology comparison map}
\BB_{\til{A}\leftarrow A;A'}\colon H^\bt_{A}\ra H^\bt_{\til{A}}
\end{equation}
-- the ``cohomology comparison map.'' Notice that due to the curvature of $\nabla^\mr{Harm}$, this map depends nontrivially on $A'$. 
%\marginpar{Just added a line here, not sure if neccessary. Also here notation needs to be changed to $A,A'$.}
% \footnote{
% We may assume (by possibly having to choose a smaller $\mc{U}$) that fibers of the projections of $\mc{U}\subset \FC'\times \FC'$ to the first and second copies of $\FC'$ are simply-connected. Then the map (\ref{cohomology comparison map}) does not depend on the choice of the path from $A_0$ to $A_1$ as long as it stays close to $A_0$.
% }

% An equivalent definition of $\BB_{A'_0\la A_0}$ is that it is the map induced in cohomology by the map $q$ of Proposition \ref{prop 3.5} (\ref{prop 3.5 (c)}). Also, by Proposition \ref{prop 3.5} (\ref{prop 3.5 (b)}), the cohomology comparison map is an isometry.

Sometimes we will need the cohomology comparison map restricted to cohomology in degree $1$; we will denote it $\BB^1_{\til{A}\la A;A'}$.

\begin{rem}
In the special case $\til{A}=A'$ the cohomology comparison map is
 $$\BB_{A'\la A;A'}=[q]\colon H_{A}\ra H_{A'}$$
 -- the map induced in cohomology by the map $q$ of Proposition \ref{prop 3.5} (\ref{prop 3.5 (c)}).
\end{rem}
\subsubsection{Local exponential map for fixed $A'$.}
%\marginpar{new subsubsection August 20}
For a given smooth flat connection $A'$, for $A \in \mc{U}_{A'}$, we have the desynchronized Hodge SDR data $(i_{A,A'},p_{A,A'},K_{A,A'})$. These induce, locally around $A'$, a sum-over-trees exponential map  that we denote $\til\varphi_{\bullet,A'}\colon T\FC' \supset V_A \to \FC'$. Contrary to the ``global'' exponential map $\til\varphi_{\bullet,\bullet}$, the local exponential map $\til\varphi_{\bullet,A'}$ is not equivariant with respect to the gauge group action on its argument. However, it satisfies the following ``convolution'' property: 
\begin{prop}\label{prop: desy exp map conv}
    Let $\beta, \gamma \in T_A\FC'$ such that $\beta, \gamma$ and $\beta + \gamma$ are in the domain of $\varphi_{A,A'}$. Let $\til{A} = \til\varphi_{A,A'}(\beta)$. Then 
    \begin{equation}
        \til\varphi_{A,A'}(\beta + \gamma) = \til\varphi_{\til{A},A'}( (d\til\varphi_{A,A'})_\beta(\gamma)). 
    \end{equation}
\end{prop}
\begin{proof}
    The proof follows from the combinatorics of tress and the homological perturbation lemma. Namely, one can expand the left hand side as a sum-over-trees map where edges are decorated by $K_{A,A'}$ and leaves are decorated either by $\beta$ or by $\gamma$. On the other hand, one can expand the right hand side as a sum-over-trees map where edges are decorated by $K_{\til{A},A'}$ and leaves are decorated by $\BB^1_{\til{A}\la A;A'}(\gamma)$. By the homological perturbation lemma, we have $K_{\til{A},A'} = K_{A,A'} - K_{A,A'}\delta_{A,A'}K_{A,A'} + \ldots$, which we can %\marginpar{\bl Aug 22. rephrased here.}
    represent as a sum over ways to plug in a forest of trees into an edge, with leaves decorated by $\beta$ and edges decorated by $K_{A,A'}$.
    %expand as a sum-over-trees with leaves decorated by $\beta$ and edges decorated by $K_{A,A'}$. 
    Finally, $(d\til\varphi_{A,A'})_\beta(\gamma)$ is itself given as a sum over trees with edges labeled by $K_{A,A'}$, one leaf labeled $\gamma$ and all other leaves labeled $\beta$ (see Remark \ref{rem: sum over trees d til varphi}). In this way, one also on the right hand side obtains a sum over trees with edges labeled by $K_{A,A'}$ and leaves either labeled $\beta$ or $\gamma$. The numerical prefactor of each such tree is the same on both sides and given by $(-1)
    ^{n-1}/|\operatorname{Aut} T|$ where $n$ is the number of leaves and $\operatorname{Aut} T$ are the automorphisms of $T$ respecting the decorations of leaves.  
\end{proof}
This then leads to the following explicit description of the Grothendieck connection on $\FC'$: 
%\marginpar{\bl what is defined here is a \emph{partial} connection (not every flat $\til{A}$ close to $A$ is reachable by $\phi_{A,A'}(-)$). It is the same as $\til{\nabla}^G$ in sec. 4 {\red I changed it but have to include the partial connection again}
%}
%defined, in this setting, on the cohomology bundle over $U_{A'}$). 
Let $\alpha, \beta \in T_A\FC'$, and $\til{A} = \til\varphi_{A,A'}(\beta)$. Denote $\til\alpha$ the parallel transport of $\alpha$ from $A \to \til{A}$.
\begin{prop}
    We have 
    \begin{equation}
        \til\alpha = (d\til\varphi_{A,A'})_\beta(\alpha - \beta). 
    \end{equation}
\end{prop}
\begin{proof}
    This follows from the fact that $\til\alpha = \til\varphi^{-1}_{\til{A},A'}\til\varphi_{A,A'}$ and  % the property of $\varphi$ that for any $\beta,\gamma \in U_A$ we have \marginpar{\bl I think, for this to be true, we need to use the same gauge-fixing operator $h$ at $A$ and $\til{A}$. This is true in the context of desync Hodge, but this was not the setup in this section..}
    %\begin{equation}
       % \varphi_{\til{A}}(B(\gamma)) = \varphi_A(\beta + \gamma).
    %\end{equation}
    Proposition \ref{prop: desy exp map conv} by choosing $\gamma$ such that $\alpha = \gamma + \beta$. 
\end{proof}
By restricting to harmonic forms and passing to cohomology, we obtain a local exponential map $\varphi_{A,A'}\colon V_{A'} \to \mc{U}_{A'}$, defined on an open subset $V_{A'}$ of the cohomology bundle $\mathbb{H}\big|_\mc{{U}_A}\to \mc{U}_{A'}$. Associated to this map is a partial fiber bundle connection, whose parallel transport can be defined as follows: Let $\alpha, \beta \in H^1_A$ and $\til{A} = \varphi_{A,A'}(\alpha)$. Then the parallel transport of $\alpha$ from $A$ to $\til{A}$ is given by the cohomology comparison map
\begin{equation}
    \til{\alpha} = \BB_{\til{A} \la A,A'}(\alpha - \beta)
\end{equation}
since by Remark \ref{rem 3.9} the parallel transport of $\nabla^\mr{Harm}$ from $A$ to $\til{A} = \varphi_{A,A'}(\beta)$ is given by $(d\varphi_{A,A'})_\beta$.
Cf. also Definition \ref{def: partial Grothendieck connection} and Remark \ref{rem: partial Grothendieck connection}.
\section{Perturbative Chern-Simons partition function in the BV formalism}
In this section we recall the definition of the perturbative Chern-Simons partition function at an arbitrary reference flat connection $A_0$ given in \cite{Cattaneo2008} (a detailed review can be found also in \cite{mnev2019quantum},\cite{Wernli2022}), and extend this to the definition of the desynchronized partition function which uses as gauge fixing operator the codifferential $d^*_{A'}$ instead of $d^*_{A_0}$. Let $G$ be a simple, compact and simply connected Lie group and $\langle\cdot,\cdot\rangle$ an $\mathrm{ad}$-invariant pairing on $\g$. Let $P$ be a principal $G$-bundle on a 3-manifold $M$, we will assume that a trivialization\footnote{Our assumptions are such that trivializations are guaranteed to exist. See for instance \cite{Freed1995}.} of $P$ has been fixed: $P \cong M \times G$. We can therefore identify connections with $\g$-valued 1-forms. Our convention for the Chern-Simons action $S_{CS} \colon \Omega^1(M,\g) \to \R$ is 
\begin{equation}\label{S_CS in sec3}
    S_{CS}(A) = \int_M \frac12 \langle A, \dd A\rangle + \frac16\langle A,[A,A]\rangle.
\end{equation}
Its critical points are the flat connections, i.e. those 1-forms $A_0 \in \Omega^1(M,\g)$ satisfying 
\begin{equation}
    dA_0 + \frac{1}{2}[A_0,A_0] = 0.
\end{equation}
For a flat connection $A_0 \in \Omega^1(M,\g)$, we denote the twisted de Rham differential by 
\begin{equation}
    d_{A_0} = d + [A_0,\cdot] \colon \Omega^\bullet \to \Omega^{\bullet + 1}
\end{equation}
and the $A_0$-twisted de Rham cohomology by 
%\marginpar{\bl edited Apr 19}
$H_{A_0}^\bt$.
%$H^\bullet_{\bl A_0}(M, {\bl \g})$.
%\operatorname{Ad} P
% We also fix a Riemannian metric $m$ on $M$. 

\subsection{Perturbative partition function}\label{ss perturbative partition function}
We can now proceed with the definition of  the perturbative Chern-Simons partition function at $A_0$.  % with gauge fixing operator $Q$. {\color{red} 
Formally, we want to define it as the perturbative evaluation of the BV pushforward 
$$ Z_\sh(A_0,\sfa) = \int_{\alpha_\mr{fl} \in \operatorname{im} \sh}e^{\frac{i}{\hbar}S_{CS}(A_0 + \sfa+\alpha_\mr{fl})}\mu^{\frac{1}{2}}. $$
%}
%\marginpar{What about the 1-loop part???
%}
We first define the partition function with gauge fixing operator $\sh=d^*_{A_0}$, and then comment on changing the gauge fixing operator. 
\begin{defn} \label{def: Z pert} Let $A_0$ be a flat connection on $M$, and $g$ a Riemannian metric on $M$. The Chern-Simons partition function at $A_0$ with gauge fixing operator $d^*_{A_0}$ is %denoted $Z_{A_0}$ and  is
defined by

\begin{multline}\label{eq: Z def}
    Z(A_0,\sfa;g)\colon=e^{\frac{i}{\hbar}S_{CS}(A_0)} \tau(A_0)^\frac{1}{2} e^{\frac{\pi i}{4}\psi(A_0;g)}\cdot\\
    \cdot \exp \left(\sum_\Gamma \frac{(-i\hbar)^{-\chi(\Gamma)}}{|\mr{Aut}(\Gamma)|}\Phi_{\Gamma,A_0;g}(\sfa)\right) \\ 
    \in 
    %\mr{Dens}^{\frac12,\mr{formal}}(H^\bt_{A_0}[1])=
    e^{\frac{i}{\hbar}\big(S_{CS}(A_0)+\sum_{n\geq 2}\frac{1}{(n+1)!}\langle \sfa,l_n(\sfa,\ldots,\sfa)\rangle\big)}\cdot\mr{Det}^\frac12(H^\bt_{A_0})\otimes \wh{\mr{Sym}}(H^\bt_{A_0}[1])^*[[\hbar]]
\end{multline}
%\marginpar{To get the motivating example correct, we should include the correct prefactors of $\hbar$, and probably also the framing correction.}
-- a formal half-density on de Rham cohomology twisted by $A_0$.\footnote{Note that there is no sign ambiguity in the square root line bundle $\mr{Det}^{\frac12}(H^\bt_{A_0})$, since by Poincar\'e duality it can be expressed as $\mr{Det}(H^0)\otimes (\mr{Det}(H^1))^*$.} 
Here:
\begin{itemize}
%\item $H_{A_0}^\bt$ is the cohomology of the complex of $\g$-valued differential forms on $M$ with differential $d_{A_0}=d+\mr{ad}_{A_0}$. One calls the variable $a\in H_{A_0}^\bt[1]$ the \emph{zero-mode}.
   % \item $S_{CS}(A)=\int_M \mr{tr} \left( \frac12 A\wedge dA +\frac16 A\wedge [A,A]  \right)$ is the Chern-Simons action functional.
   \item $S_{CS}(A_0)$ is the value of Chern-Simons action (\ref{S_CS in sec3}) on $A_0$.
    \item $\tau(A_0)\in \mr{Det}(H^\bt_{A_0})$ is the Ray-Singer torsion of $M$ with local system $A_0$. $\tau(A_0)^{\frac12}\in \mr{Det}^{\frac12}(H^\bt_{A_0})$ is its square root.
    %\marginpar{Spin structure for choice of sign? or work modulo sign... but maybe trivial over irreducible locus. \textcolor{red}{Joyce}}
    %For $A_0$ non-acyclic, rather than being a number it is an element of the determinant line of the cohomology $H^\bt_{A_0}$.
    \item $\psi(A_0;g)$ is the Atiyah-Patodi-Singer eta-invariant of the %Dirac 
    operator ${L_-\colon=*d_{A_0}+d_{A_0}*}$ acting on forms of odd degree.
    \item The sum ranges over connected 3-valent graphs (``Feynman graphs'') $\Gamma$ with leaves (loose half-edges) allowed. $\chi(\Gamma)$ is the Euler characteristic of the graph and $\mr{Aut}(\Gamma)$ is the automorphism group. The weight of a graph $\Gamma$ is a polynomial in $\sfa$ defined as 
    \begin{multline}\label{eq: def Phi_Gamma}
        \Phi_{\Gamma,A_0;g}(\sfa)=\\
        \int_{\overline{\mr{Conf}}_V(M)} 
        \left\langle  \prod_{\mr{leaves}\; l} \pi^*_{v(l)}i_{A_0}(\sfa) \prod_{\mr{edges}\, e=(uv)} \pi_{uv}^*\eta_{A_0}
        \prod_{\mr{short\; loops}\; e=(vv)} \pi_v^* \eta_{A_0}^\Delta
        , \bigotimes_{\mr{vertices}} f \right\rangle,
    \end{multline}
    where: 
    \begin{itemize}
   \item  $\overline{\mr{Conf}}_V(M)$ is the Fulton-MacPherson-Axelrod-Singer compactification of the configuration space of $V=\#\{\mr{vertices}\}$ points on $M$. 
   \item $\pi_{uv}\colon \overline{\mr{Conf}}_V(M) \ra \overline{\mr{Conf}}_2(M)$   is the map forgetting the positions of all points except points $u$ and $v$; similarly,
   $\pi_v\colon \overline{\mr{Conf}}_V(M) \ra M$ is the map forgetting all points except $v$.
   \item The propagator $\eta_{A_0}\in \Omega^2(\overline{\mr{Conf}}_2(M),\g\otimes\g)$ is minus the integral kernel of the operator 
   \begin{equation}\label{eq: def K}
   K_{A_0}=d^*_{A_0}(\Delta_{A_0}+P_\mr{Harm})^{-1}
   \end{equation} -- the Hodge chain homotopy between $d_{A_0}$ and projection to harmonic forms.
   \item $\eta_{A_0}^\Delta\in \Omega^2(M,\g\otimes\g)$ is the appropriately renormalized evaluation of $\eta_{A_0}$ on the diagonal.\footnote{\label{footnote: tadpoles}
   It is the term $L^{cont}$ in \cite{Axelrod1991}, formula (PL5). It is the limit $\lim_{y\ra x} (\eta_{A_0}(x,y)-(\cdots))$ with $(\cdots)$ the singular part of $\eta$ at the diagonal. 
   %arising from the heat kernel expansion.
   }
   \item $i_{A_0}$ maps a cohomology class to its harmonic representative. 
   \item $f\in \g^{\otimes 3}$ is the structure tensor of the Lie algebra $\g$. 
   \item $\langle,\rangle$ is the inner product on $\g$ extended to $\g^{\otimes \#\{\mr{half-edges}\}}$.
   \item In (\ref{eq: def Phi_Gamma}), the first product is over leaves of $\Gamma$, with $v(l)$ the vertex incident to the leaf; the second product is over edges connecting \emph{distinct} vertices $u,v$; the third product is over ``short loops'' -- edges connecting a vertex $v$ to itself.
    \end{itemize}
\end{itemize}
\end{defn}
\begin{rem}
One can split the partition function according to ``loop number'' as $Z(A_0,\sfa)=Z^{(0)}(A_0,\sfa)Z^{(1)}(A_0,\sfa)Z^{(\geq 2)}(A_0,\sfa)$, 
%\marginpar{In which space does this product happen exactly? unbounded negative powers of hbar are a problem... should think of it as formal expression.}
where 
\begin{align*}
Z^{(0)}(A_0,\mathsf{a}) &:=  
\exp\left(\frac{i}{\hbar}\left(S_{CS}(A_0)+\sum_{\Gamma \in \mr{Gr_{conn}},\, l(\Gamma) =  0} \frac{1}{|\operatorname{Aut}(\Gamma)|}\Phi_{\Gamma,A_0;m}(\mathsf{a})
\right)\right) 
\\
&=\exp\left(\frac{i}{\hbar}\left(S_{CS}(A_0)+\sum_{n \geq 1}\frac{1}{(n+1)!}\langle \sfa, l_n(\sfa,\ldots,\sfa)\rangle\right)\right), \\
Z^{(1)}(A_0,\mathsf{a}) &:= \tau(A_0)^{\frac12}e^{\frac{\pi i }{4}\psi(A_0;g)} \cdot \exp\left(\sum_{\Gamma \in \mr{Gr_{conn}},\, l(\Gamma) =  1} \frac{1}{|\operatorname{Aut}(\Gamma)|}\Phi_{\Gamma,A_0;g}(\mathsf{a})\right)
\\ &\in \mr{Det}^{\frac12}( H^\bullet_{A_0})\otimes\widehat{\Sym}(H^\bullet_{A_0}[1])^*, \notag\\
Z^{(\geq 2)}(A_0,\mathsf{a}) &:= \exp\left(\frac{i}{\hbar}\sum_{\Gamma \in \mr{Gr_{conn}},\,  l(\Gamma) \geq 2} \frac{(-i\hbar)^{l(\Gamma)}}{|\operatorname{Aut}(\Gamma)|}\Phi_{\Gamma,A_0;g}(\mathsf{a})\right) \in \widehat{\Sym}(H^\bullet_{A_0}[1])^*[[\hbar]].
\end{align*}
%\marginpar{Not sure about powers of $i$ here..}
Here $l(\Gamma)$ is the number of loops in a connected graph.

The reason for excluding tree (0-loop) diagrams from the sum in \eqref{eq: Z def} is that they come with a factor of $1/\hbar$ so after taking exponential we would obtain unbounded negative powers oh $\hbar$. Instead, they are included in the prefactor in the form of the induced $L_\infty$ operations $l_n$. 
\end{rem}
%\marginpar{\bl The proof from [CM08] relies on de Rham/Lie algebra factorization and does not extend to arbitrary $A_0$. We should sketch an argument using config space integrals, \`a la [CMRpert]. 
%Also, on $\M'$, $Z$ simplifies and this becomes a moot point.
%}
\begin{thm}%[Cattaneo-Mnev]
    The perturbative partition function is closed with respect to the BV Laplacian on zero modes, 
    \begin{equation} \label{QME in ss 3.2}
        \Delta_\sfa Z_{A_0}(\sfa) = 0.
    \end{equation}
\end{thm}
We refer to \cite[Section 3.4.2]{Wernli2022} for the proof. It is in turn an adaptation of the proof of Lemma 4.11 from \cite{Cattaneo2017}, using Stokes' theorem for configuration space integrals representing Feynman weights.
%This is proven by adapting the proof of Lemma 4.11 from \cite{Cattaneo2017} - see \cite{Wernli2019b}. 
Also, the case $A_0=0$ is a part of Theorem 1 in \cite{Cattaneo2008}.
%\begin{proof}[Sketch of proof]
%\end{proof}

%{\bl
\begin{rem} 
%\marginpar{I added this remark here, but maybe we say this elsewhere besides the intro?}
    \begin{enumerate}[(i)]
        \item If the flat connection $A_0$ is \emph{irreducible}, then $H^0_{A_0}=H^3_{A_0}=0$. 
        %\sout{Hence, by degree reason,} 
        An elementary degree count then shows that $Z(A_0,\sfa)$ depends only on the 1-form component of $\sfa$. In particular, in this case (\ref{QME in ss 3.2}) holds trivially.
        \item If $[A_0]$ is a smooth point in the moduli space of flat connections, then operations $l_n$ on $H^\bt_{A_0}$ vanish (i.e., the tree graphs in (\ref{eq: Z def}) cancel out).
    \end{enumerate}
\end{rem}
%}

%\color{red}
% {\color{gray} \marginpar{\bl Oct 11: remove this part on $Z^\mr{glob}$}
% \begin{defn}
%     The \emph{global perturbative partition function} is the assignment 
%     \begin{equation}
%         Z^\mr{glob}\colon A_0 \mapsto Z_{A_0}(\sfa=0) \;\; \in e^{\frac{i}{\hbar}S_{CS}(A_0)}\cdot \mr{Det}^{\frac12}( H^\bullet_{A_0})[[\hbar]]
%     \end{equation}
% \end{defn}
% \begin{rem} %\marginpar{\bl Is this remark needed?}
%     The global partition function $Z^\mr{glob}$ is trivially closed with respect to the canonical BV operator on half-densities on 
%     %\marginpar{\bl $T^*[-1]M_{irr}\ra \M^{\irr,\sm}$}
%     $T^*[-1]\M$, since it does not depend on the cotangent coordinates.
% \end{rem}
% }
%\marginpar{maybe global is not exactly the right name, since later we will only define it on smooth irreducible components }

\subsection{Desynchronized partition function}
Let $\sh$ be a good gauge fixing operator for $A_0$, and $r_\sh = (i_\sh,p_\sh,K_\sh)$ the corresponding SDR data. The goal of this subsection is to define the ``desynchronized'' perturbative partition function, heuristically given by the BV pushforward 
\begin{equation}\label{Z desync PI}
    Z_{A_0,\sh}(\sfa) = \sqrt{\mu'}\int_{\alpha \in \operatorname{im} \sh}e^{\frac{i}{\hbar}S_{CS}(A_0+i_\sh(\sfa)+\alpha)}\sqrt{\mu''}\bigg|_\calL
\end{equation}
where $\sqrt{\mu} = \sqrt{\mu'}\sqrt{\mu''}$ is the formal Lebesgue half-density on $\Omega^\bullet(M,\g)[1]$ and $\sqrt{\mu'},\sqrt{\mu''}$ the Lebesgue
%induced \marginpar{\bl what "induced" (half density) means here?}
half-densities on $H^\bullet_{A_0}[1]
 \cong \ker \HHH$ and $\im d_{A_0} \oplus \im \sh$ respectively. For the remainder of this section we fix $\sh = d^*_{A'}$, for some flat connection $A'$ close to $A_0$ in the sense of Definition \ref{def: close connections}.

We then define the desynchronized partition function analogously to the synchronized case: 
\begin{defn}
Let $(A,A')$ be a pair of close, smooth flat connections. Then we define the \emph{desynchronized partition function} $$Z_{A,A'} \in e^{\frac{i}{\hbar}S_{CS}(A)}\cdot \Det^\frac12 (H^\bullet_{A})\otimes\widehat{\Sym}(H^\bullet_{A}[1])^*[[\hbar]]$$ as the product 
\begin{equation}
    Z_{A,A'}(\sfa) = Z^{(0)}_{A,A'}\,Z^{(1)}_{A,A'}(\sfa)\, Z^{(\geq 2)}_{A,A'}(\sfa)   
\end{equation}
where  $Z_{A,A'}^{(0)} :=  e^{\frac{i}{\hbar}S_{CS}(A)} $ and 
%\marginpar{Check correct powers of $i$. {\bl done} Is it correct to call this exponential ``1-loop part''? it contains many loops. {\bl I'd say, yes}}
\begin{multline} 
        Z_{A,A'}^{(1)}(\mathsf{a}) :=  %I_{A_0,A_1} 
        e^{\frac{\pi i}{4}\psi_{A}} \tau_{A}^{1/2}
        \exp\left(\sum_{\Gamma \in Gr_{conn}, l(\Gamma) =1 } \frac{1}{|\operatorname{Aut}(\Gamma)|}\Phi_{\Gamma,A,A'}(\mathsf{a})\right) \\
    \in  \Det^\frac12 (H^\bullet_{A})\otimes\widehat{\Sym}(H^\bullet_{A}[1])^*, \label{eq: Z desy 1-loop}
\end{multline}
\begin{multline}
    Z_{A,A'}^{(\geq 2)}(\mathsf{a}) := \exp\left(\sum_{\Gamma \in Gr_{conn}, l(\Gamma) \geq 2} \frac{(-i\hbar)^{l(\Gamma)-1}}{|\operatorname{Aut}(\Gamma)|}\Phi_{\Gamma,A,A'}(\mathsf{a})\right)\\
    \in  %H^\bullet_{A}[1])\otimes
    \widehat{\Sym}(H^\bullet_{A}[1])^*[[\hbar]].
\end{multline}
The Feynman weights $\Phi_{A,A'}(\Gamma)$ are defined as in \eqref{eq: def Phi_Gamma}, where we replace the integral kernel \eqref{eq: def K} of $K_{A}$ by the integral kernel of 
\begin{equation}
    K_{A,A'} = d_{A'}^*\circ (\Delta_{A,A'} + P_{A,A'})^{-1} 
\end{equation} and the map $i_{A}$ with $i_{A,A'}$. 
\end{defn}
Notice that since $A$ is smooth, there are no trees in the zero-loop part -- their weights vanish by the smoothness assumption. 

\subsubsection{Digression: Path integral computation of desynchronized 1-loop part}
The abelian part of the path integral (\ref{Z desync PI}) is
\begin{equation}
    I_{A,A'} \colon = \sqrt{\mu'} \int_{\alpha \in \calL = \im d_{A'^*} }e^{\frac{i}{\hbar}\int_M\frac12\langle \alpha, d_{A}\alpha\rangle}\sqrt{\mu''}\bigg|_\calL .
    %:= \tau_Ae^{\frac{i\pi}{4}\psi_A} 
    \label{eq: int I}
\end{equation}
Perturbative formula (\ref{eq: Z desy 1-loop}) corresponds to evaluating the path integral (\ref{eq: int I}) to
\begin{equation}
    I_{A,A'}\colon= \tau_Ae^{\frac{i\pi}{4}\psi_A}.
\end{equation}
In this digression we want to explain why this is a good definition of the r.h.s. of (\ref{eq: int I}).
% In the definition of the 1-loop part of the desynchronized partition function \eqref{eq: Z desy 1-loop}, we have set 
% \begin{equation}
%     I_{A,A'} \colon = \sqrt{\mu'} \int_{\alpha \in \calL = \im d_{A'^*} }e^{\frac{i}{\hbar}\int_M\frac12\langle \alpha, d_{A}\alpha\rangle}\sqrt{\mu''}\bigg|_\calL := \tau_Ae^{\frac{i\pi}{4}\psi_A} \label{eq: int I}
% \end{equation}
% {\red In this digression we want to explain why this is a good definition of the path integral in the middle. 
% \sout{However, a}} 
Namely, naive evaluation of this path integral would go along the following lines.
For a subspace $V \subset \Omega^\bullet(M,\g)$ and an isomorphism $F \colon V \to V$ we set 
    $$ \int_{\alpha \in V} e^{\frac{i}{\hbar}\frac12(\alpha, F\alpha)_H}\mu_H = e^{\frac{i\pi}{4}\mr{sign}F}\sdet_V^\frac12F$$  
    where $(\cdot,\cdot)_H$ denotes the Hodge inner product 
    $$(\alpha_1, \alpha_2)_H = \int_M \langle \alpha_1 \stackrel{\wedge}{,} * \alpha_2 \rangle $$
    and the signature and superdeterminant have to be understood in a regularized sense. Looking at \eqref{eq: int I}, the map $*d_{A}$ maps $d_{A'}$-coexact forms to $d_{A}$-coexact forms, so it is not an endomorphism of $\calL = \im d_{A'}^*$. The orthogonal (with respect to the Hodge inner product) projector to $\im d_{A'}^*$ is $K_{A'}d_{A'}$, so we obtain 
    \begin{equation} \label{eq: I def PI} I_{A,A'} = \sqrt{\mu'}\, e^{\frac{i\pi}{4}\sign K_{A'}d_{A'}*d_{A}}\sdet_{\im d_{A'}^*}^\frac12(K_{A'}d_{A'}*d_{A}).
    \end{equation}
    We claim that this coincides with the following definition: 
    \begin{lem} For a pair of close flat connections $(A,A')$, $I_{A,A'}$ can be expressed as 
\begin{equation}\label{I def}
    I_{A,A'} = e^{\frac{i\pi}{4} \psi_{A'}}%\det(q)^{-1/2}
    \det(\BB_{A\la A';A'})^{1/2}
    \tau_{A'}^{1/2}\sdet^{\frac{1}{2}}_{\im d_{A'}^*}\left(1 + K_{A'}\operatorname{ad}_\beta\right) \in \Det^{\frac12} H^\bullet_{A},
\end{equation}
where 
\begin{itemize}
    \item $\psi_{A'}$ is the eta-invariant of $*d_{A'} + d_{A'}*$, 
    \item $\BB_{A\la A';A'}\colon H^\bt_{A'}\ra H^\bt_{A}$ is the cohomology comparison map of Section \ref{sss cohomology comparison map}.
    %$q \colon \mr{Harm}_{A_0,A_1} \to \mr{Harm}_{A_1,A_1}$ is the orthogonal projection, cf. Prop \ref{prop 3.5},
    \item $\sdet$ denotes a zeta-regularized superdeterminant, %\marginpar{\red Maybe we should use another concept of superdeterminant here (``generalized Fredholm'').
  %  Oct 13: I no longer remember why we needed that}
    \item $\beta = A - A'$ is the difference between the two flat connections. 
\end{itemize}
\end{lem} 
%     \begin{defn}[Free desynchronized partition function] For a pair of close flat connections $(A,A')$, we define the \emph{free desynchronized partition function} $I_{A,A'}$ as 
% \begin{equation}\label{I def}
%     I_{A,A'} = e^{\frac{i\pi}{4} \psi_{A'}}%\det(q)^{-1/2}
%     \det(\BB_{A\la A';A'})^{1/2}
%     \tau_{A'}^{1/2}\sdet^{\frac{1}{2}}_{\im d_{A'}^*}\left(1 + K_{A'}\operatorname{ad}_\beta\right) \in \Det^{\frac12} H^\bullet_{A},
% \end{equation}
% where 
% \begin{itemize}
%     \item $\psi_{A'}$ is the eta-invariant of $*d_{A'} + d_{A'}*$, 
%     \item $\BB_{A\la A';A'}\colon H^\bt_{A'}\ra H^\bt_{A}$ is the cohomology comparison map of Section \ref{sss cohomology comparison map}.
%     %$q \colon \mr{Harm}_{A_0,A_1} \to \mr{Harm}_{A_1,A_1}$ is the orthogonal projection, cf. Prop \ref{prop 3.5},
%     \item $\sdet$ denotes a zeta-regularized superdeterminant, \marginpar{\red Maybe we should use another concept of superdeterminant here (``generalized Fredholm'').}
%     \item $\beta = A - A'$ is the difference between the two flat connections. 
% \end{itemize}
% \end{defn} 
% \begin{lem}
%     The definitions \eqref{I def} and \eqref{eq: I def PI} of $I_{A,A'}$ agree.
% \end{lem}
%\marginpar{\bl Oct 11: finish changing $(A_0,A_1)$ to $(A,A')$ in this subsection}
\begin{proof}
    We have on $\im d_{A'}^*$ that $K_{A'}d_{A'} = \mr{id}$ and therefore
    $$ \mr{id} + K_{A'}\mathrm{ad_\beta} = \mr{id} + K_{A'}(d_{A} - d_{A'}) = K_{A'}d_{A}.$$  
    Also, 
    $$ \det(\BB_{A\la A';A'})^{\frac12}
    \tau_{A'}^\frac12 = \sqrt{\mu'}\sdet^\frac12_{\im d_{A'}^*} *d_{A'}.$$
    This implies that 
    %\marginpar{\bl I removed the second G (non-needed) factor in the next-to-last line}
    \begin{align*} & %\det(q)^{-\frac12}
    \det(\BB_{A\la A';A'})^{\frac12}
    \tau_{A'}^\frac12\sdet^{\frac{1}{2}}_{\im d_{A'}^*}\left(1 + K_{A'}\operatorname{ad}_\beta\right) \\
    &= \sqrt{\mu'}\sdet_{\im d_{A'}^*}^\frac12(*d_{A'}K_{A'}d_{A}) \\
    &= \sqrt{\mu'}\sdet_{\im d_{A'}^*}^\frac12(*d_{A'}* d_{A'}* G_{A'}d_{A}) \\
    &= \sqrt{\mu'}\sdet_{\im d_{A'}^*}^\frac12(*d_{A'}* G_{A'} d_{A'}* %G_{A_1}
    d_{A})\\
    &= \sqrt{\mu'}\sdet_{\im d_{A'}^*}^\frac12(K_{A'}d_{A'}*d_{A})
    \end{align*}
    where we have used that the Green's function $G_{A'}$ commutes with both the the Hodge star and $d_{A'}$.\footnote{In principle there could be a multiplicative anomaly when combining the regularized superdeterminants, but here it is absent because the equality is trivially true for $\beta = 0$ and we are restricting to small $\beta$.}
    For the phase, we note that the spectrum of $K_{A'}d_{A'} * d_{A}$ is obtained from the spectrum of $*d_{A'}$ through continuous deformation where none of the 
  real parts of the eigenvalues  crosses zero, therefore any regularization of the signature will yield the same result. 
  \end{proof}
However, it turns out that we have the following: 
\begin{lem}
    The expression \eqref{I def} for $I_{A,A'}$ is independent of $A'$ and depends on $g$ only through $\psi_{A'}$. In particular, $I_{A,A'} = I_{A,A} = \tau^\frac12_Ae^{\frac{i\pi}{4}\psi_A}$. 
\end{lem}
The proof is a long computation. Crucially, the non-flatness of $\nabla^\mr{Harm}$ means that the cohomology comparison map $\BB_{A\la A';A'}$ depends both on $A'$ and $g$, this dependence precisely cancels the dependence of $\sdet^{\frac{1}{2}}_{\im d_{A'}^*}\left(1 + K_{A'}\operatorname{ad}_\beta\right)$ on $A'$ and $g$.

\section{Properties of the desynschronized partition function}
This section is devoted to the proof of Theorem \ref{thm: intro desy Z}  %\ref{eq: thm desy Z intro 1} 
which we split up in several subsections. Throughout this section $A$ and $A$ is a pair of close smooth irreducible flat connections. 
\subsection{Gauge invariance} 
%\marginpar{\bl I switched to the left action here}
We first discuss the impact of gauge transformations on $Z_{A,A'}(\sfa)$. Note that the gauge transformation $A \mapsto {}^\sfg A$ induces an isomorphism $H_{A}^\bullet \cong H_{{}^\sfg A}^\bullet$ by the adjoint action on cohomology classes.
\begin{prop}\label{prop: Z desync equivariance under diagonal gauge transf}
     We have that $Z_{A,A'}(\sfa)$ is invariant under ``diagonal'' gauge transformations $(A,A',\sfa) \mapsto ({}^\sfg A,{}^\sfg A',{}^\sfg\sfa)$. 
\end{prop}
\begin{proof}
    This follows from the fact that all the ingredients of $Z_{A,A'}$ are gauge equivariant. I.e., we have $K_{{}^\sfg A}({}^g\omega) = {}^\sfg (K_{A}\omega)$ and $\iota_{{}^\sfg A}[{}^g\omega] = {}^\sfg (\iota_{A}\omega).$ Finally, we contract tensors using the $G$-invariant pairing on $\g$. %\marginpar{Expand, comment on $I_{A_0,A_1}$. {\red Now I think it is fine} }
\end{proof}
\subsection{Horizontality w.r.t. Grothendieck connection (changing the kinetic operator)}
In the next theorem we prove that a shift in the kinetic operator can be expressed as a shift of the zero mode (or vice versa). 

Let $\varphi_{A,A'}(\alpha)=A+\delta_{A,A'}(\alpha)\colon\; U \ra \FC'$ be the sum-over-trees exponential map (\ref{phi - exp map on moduli space}), determined by the 
SDR data
associated to the SDR data $r_\sh = (i_\sh,p_\sh,K_\sh)$ corresponding to the gauge-fixing operator $\sh=d^*_{A'}$. The map $\varphi_{A,A'}(\alpha)$, as a function of $\alpha$, is defined on some open neighborhood $U$ of zero in $H^1_{A}$.

In this section we will denote for brevity  
\begin{equation}\label{Atilde}
\til{A}\colon= \varphi_{A,A'}(\alpha).
\end{equation}

Denote
%\marginpar{\bl Oct 13. To me formulas with two ":=" look hard to parse. Do we even need the notation in the middle? I don't think we use it.}
\begin{equation}\label{B for thm 4.2}
    B\colon = %{\red B_{A,A'}(\alpha)\colon =}
    \BB^1_{\til{A}\la A;A'}\colon H^1_{A}\ra H^1_{\til{A}}
\end{equation}
the cohomology comparison map in degree $1$.

\begin{rem}\label{rem: d underline phi = B}
The map (\ref{B for thm 4.2}) coincides with the differential of $\pi\circ \varphi_{A,A'}(\alpha)$ in the last argument, with $\pi\colon \FC' \ra \M'$ the quotient by gauge transformations. This follows from the fact that
\begin{equation}\label{i B = d phi}
\begin{aligned}
    i_{\til{A},A'}\circ B %\BB^1_{A_0'\la A_0;A_1} 
    &=&
    i_{A,A'}-K_{A,A'}\mr{ad}_{\delta_{A,A'}(\alpha)}i_{A,A'}+\cdots 
    \\
    &=& d_\alpha\varphi_{A,A'}(\alpha)\;\colon H_{A}^1\ra \mr{Harm}^1_{\til{A},A'},
\end{aligned}
\end{equation}
cf. Remarks \ref{rem 2.18}, \ref{rem 3.9}.
\end{rem}

\begin{thm}\label{thm 4.2}
We have that
the desynchronized partition function satisfies 
        \begin{equation}\label{eq thm 4.2}
            \det(B^\vee) \circ Z_{\varphi_{A,A'}(\alpha),A'}(B(\sfa)) = Z_{A,A'}(\alpha + \sfa) %\label{eq: thm desy Z intro 1}
        \end{equation}
       where 
         $\sfa$ and $\alpha$ denote variables  in an open neighborhood of zero in $H^1_{A}$. 
\end{thm}
The proof of this theorem relies on the following fact about the dependence of the Ray-Singer torsion on the local system that we were unable to locate in the literature:  
\begin{prop}\label{prop: tau Grothendieck horizontality}
Ray-Singer torsion satisfies
\begin{equation}\label{tau G horizontality}
  \det(B^\vee)\circ \tau_{\til{A}}^{1/2}=\tau_{A}^{1/2} \exp \sum_{\gamma} \frac{1}{|\mr{Aut}(\gamma)|}\Phi_{\gamma,A,A'}(\alpha).
\end{equation}
    % We have 
    % \begin{equation}\label{I horizontality}
    %     \det(B^\vee)\circ I_{A'_0,A_1} = I_{A_0,A_1}\exp \sum_{\gamma} \frac{1}{|\mr{Aut}(\gamma)|}\Phi_{A_0,A_1,\gamma}(\alpha)
    % \end{equation}
%    \marginpar{check the sign of 1-loop graphs in the rhs}
    Here $\gamma$ runs over 1-loop connected trivalent graphs.
\end{prop}
We give the proof in Appendix \ref{Appendix: proof of Prop 4.4}.

\begin{proof}[Sketch of proof of Theorem \ref{thm 4.2}]
The r.h.s. of (\ref{eq thm 4.2}) is $e^{\frac{i}{\hbar}S_{CS}(A)}  e^{\frac{i}{\hbar}\psi_{A}}\tau_{A}^\frac12$ times 
 the exponential of the sum of connected Feynman graphs $\Gamma$ with $l\geq 1$ loops, with leaves decorated by either $i_{A,A'}(\sfa)$ or $i_{A,A'}(\alpha)$ and edges decorated by $K_{A,A'}$. Given such a graph $\Gamma$, one can represent it -- in a unique way -- as a smaller graph $\Gamma'$ with leaves decorated by subtrees $X_1,\ldots,X_m$ and $Y_1,\ldots,Y_n$ of the original graph $\Gamma$, where
 \begin{itemize}
     \item Subtrees $X_i$ have a single leaf decorated by $\sfa$; all the rest are decorated by $\alpha$.
     \item Subtrees $Y_j$ have all leaves decorated by $\alpha$.
     \item If a vertex of $\Gamma'$ has more than one incident leaves, they must all be decorated by $X$-subtrees. 
 \end{itemize}
 One can think of $\Gamma'$ as $\Gamma$ with subtrees $\{X_i\}$, $\{Y_j\}$ collapsed, each tree to its root. We will also denote $\Gamma''$ the graph obtained from $\Gamma'$ by removing all $Y$-leaves and merging the internal edges incident to them.

 Explicit construction of $\Gamma'$: $\Gamma$ can be thought of a trivalent graph $\til{\Gamma}$ with no leaves, with several rooted trees $T_1,\ldots,T_N$ plugged into the edges of $\til{\Gamma}$. Starting from each leaf of $\Gamma$ decorated by $\sfa$ (which belongs to some tree $T_k$), draw the shortest path along edges connecting it to the root of $T_k$, call it an ``$\sfa$-path.'' The graph $\Gamma'$ is obtained by taking all the edges of $\Gamma$ which are either (a) non-separating (cutting the edge does not make the graph disconnected) or (b) have at least $2$ $\sfa$-paths passing through them; together with each vertex involved we take its neighborhood in $\Gamma$, producing leaves. The graph $\Gamma\setminus \Gamma'$ is disjoint and consists of $X$-subtrees (those containing an $\sfa$-path) and $Y$-subtrees (those not containing an $\sfa$-path).
\begin{figure}
    \centering
    \includegraphics[width=1\linewidth]{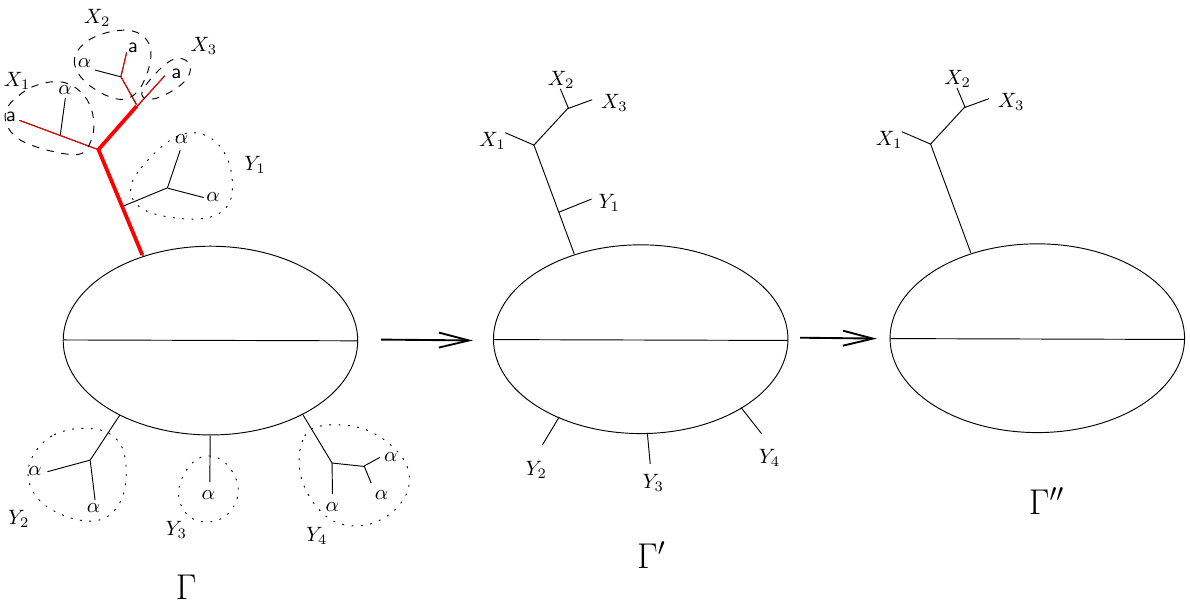}
    \caption{A Feynman graph $\Gamma$ with $X$- and $Y$-subtrees and the resulting $\Gamma'$ and $\Gamma''$ graphs. $\sfa$-paths are shown in red; edges belonging to several $\sfa$ paths are thick red edges.}
    \label{fig: horizontality proof}
\end{figure}

%\textbf{Graph combinatorics argument:} 
The sum over $\Gamma$ can be represented as a sum over graphs $\Gamma''$. %\marginpar{\red Maybe add here that $\Gamma''$ is constructed from $\Gamma'$ by contracting all $Y$ subtrees?}
Summation over possible subtrees $X$ 
%\marginpar{Maybe $X \to Y$ here? Am I missing something?}
on a leaf of $\Gamma''$ yields 
\begin{equation}\label{i deformation thm 4.2}
i_{\til{A},A'}(B(\sfa))=
i_{A,A'}(\sfa)-K_{A,A'}\mr{ad}_{\delta_{A,A'}(\alpha)} i_{A,A'}(\sfa)+\cdots,
\end{equation} 
cf. (\ref{i B = d phi}).
Summation over inserting $k\geq 0$ $Y$-subtrees into an edge $e$ of $\Gamma''$ results in decorating that edge with the chain homotopy 
\begin{equation}\label{K deformation thm 4.2}
K_{\til{A},A'}
 =K_{A,A'}-K_{A,A'}\mr{ad}_{\delta_{A,A'}(\alpha)} K_{A,A'}+\cdots
 \end{equation}
 
Formulae (\ref{i deformation thm 4.2}), (\ref{K deformation thm 4.2}) are the homological perturbation theory expressions for the deformation of an  SDR data $(i,p,K)$ for the deformation retraction $\Omega(M,\g)\ra H_{A}$  induced by a deformation of the differential from $d_{A}$ to $d_{\til{A}}$, 
%\sout{See e.g. %\cite{Crainic2004}, \cite[Section 4]{cattaneo2020cellular}
%for details.} 
see Appendix \ref{app: SDR} for details. 
% \marginpar{Oct 13, changed ref to appendix}

Thus, the sum over Feynman graphs $\Gamma$ in the r.h.s. of (\ref{eq thm 4.2}) equals the sum over Feynman graphs $\Gamma''$ in the l.h.s. of (\ref{eq thm 4.2}). There is one correction: one-loop graphs $\Gamma$ with leaves decorated only by $\alpha$ (no $\sfa$) were omitted in this correspondence, since they result in $\Gamma''$ being a loop with no vertices, which is not a legitimate graph. Thus we have
\begin{multline}
\exp\sum_{\Gamma''}\frac{(-i\hbar)^{l(\Gamma'')-1}}{|\mr{Aut}(\Gamma)|}\Phi_{\Gamma'',\til{A},A'}(B(\sfa))\cdot 
\exp\sum_{\gamma\;\mr{1-loop}} \frac{1}{|\mr{Aut}(\gamma)|}\Phi_{\gamma,A,A'}(\alpha)
\\
=\exp\sum_{\Gamma}\frac{(-i\hbar)^{l(\Gamma)-1}}{|\mr{Aut}(\Gamma)|}\Phi_{\Gamma,A,A'}(\alpha+\sfa) .
\end{multline}
Together with (\ref{tau G horizontality}) this implies (\ref{eq thm 4.2}).
\end{proof}

%Let $\mathbb{H}$ be the ``cohomology bundle'' over $FC'$ -- the graded vector bundle with the fiber over $A_0$ being $H^\bt_{A_0}$. The flat connection $\nabla^\mr{Harm}$ of Section \ref{sss nabla^Harm} restricted to the diagonal in $FC'\times FC'$ induces a flat connection $\nabla^\mathbb{H}$ in $\mathbb{H}$.

\begin{cor}[Infinitesimal variation of kinetic operator]
\label{cor: variation of kin operator}
    Let $A_t$ be a curve of flat connections such that $\dot{A}_0 = i_{A_0,A'}(\alpha)$ and let $B_t =\BB^1_{A_t \la A_0;A'}\colon H^1_{A_0}\ra H^1_{A_t}$ be the cohomology comparison map in degree $1$. Then 
    \begin{equation}\label{variation of kin operator}
        \left.\frac{d}{dt}\right|_{t=0} \Big(\det(B_t^\vee)\circ Z_{A_t,A'}( B_t (\sfa))\Big) = \left\langle  \alpha, \frac{\partial}{\partial \sfa} \right\rangle Z_{A_0,A'}(\sfa).
    \end{equation}
\end{cor}
\begin{proof}[Proof 1]
Follows immediately from Theorem \ref{thm 4.2} by setting $A_t = \varphi_{A_0,A'}(t \alpha)$ and taking the derivative of both sides in $t$ at $t=0$.
%    Infinitesimal deformations of $\iota$ and $K$ with graph combinatorics. 
\end{proof}

One can also prove (\ref{variation of kin operator}) by a standalone combinatorial argument.
\begin{proof}[Proof 2]
One has the following formulae for the infinitesimal variation of $i,K$:
\begin{equation}\label{var of i,K with A}
\begin{aligned}  
\left.\frac{d}{dt}\right|_{t=0}K_{A_t,A'}=&-K_{A_0,A'}\ad_{i_{A_0,A'}(\alpha)} K_{A_0,A'},\\ 
\left.\frac{d}{dt}\right|_{t=0} i_{A_t,A'}(B_t(\sfa))=& - K_{A_0,A'}\ad_{i_{A_0,A'}(\alpha)} i_{A_0,A'}(\sfa).
\end{aligned}
\end{equation}
These imply that the l.h.s. of (\ref{variation of kin operator}) is the sum over graphs $\Gamma$, where either (i) one edge is split into two by an insertion of leaf decorated by $\alpha$, or (ii) one $\sfa$-leaf is replaced by a subtree consisting of an $\sfa$-leaf meeting an $\alpha$-leaf and continuing with an edge (Figure \ref{fig:var of A}). 
\begin{figure}
    \centering
    \includegraphics[width=0.7\linewidth]{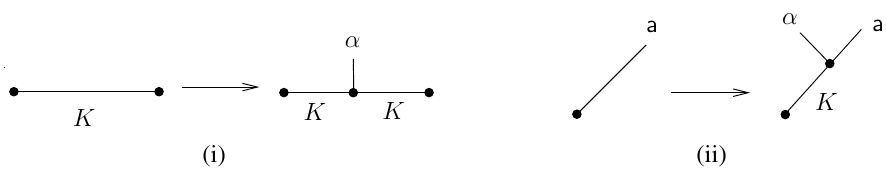}
    \caption{Variation of a Feynman graph $\Gamma$ under a harmonic shift of $A$: local picture on the graph. %Internal vertices are shown as solid dots.
    }
    \label{fig:var of A}
\end{figure}

Additionally, one has a special graph -- the one-loop graph with a single $\alpha$-leaf (the ``tadpole''), arising from the variation of $\tau_{A_t}^{1/2}$ (Proposition \ref{prop: tau Grothendieck horizontality})
%\marginpar{Quote the corresponding lemma here? {\red I did it}}. 
It is easy to see that the r.h.s. of (\ref{variation of kin operator}) yields exactly the same graphs.
\end{proof}

\begin{rem}
    The r.h.s. of (\ref{variation of kin operator}) can also be written as a BV-exact term
    \begin{equation}
        \Delta_\sfa\Big(\langle \alpha, \sfa\rangle Z_{A_0,A'}(\sfa)\Big).
    \end{equation}
The expression in brackets is a formal half-density on $H^\bt_{A_0}$ of ghost degree $-1$. Expanding $\sfa=\sfa^{1}+\sfa^{2}$, with $\sfa^{k}\in H^k_{A_0}[1-k]$, we can write the expression in brackets as 
$\langle \alpha, \sfa^{2}\rangle Z_{A_0,A'}(\sfa^{1})$.

Indeed, one has
\begin{multline}
    \Delta_\sfa\Big(\langle \alpha, \sfa\rangle Z_{A_0,A'}(\sfa)\Big)=\\
    =
   -\underbrace{\Delta_\sfa \langle \alpha, \sfa\rangle}_0 \cdot Z_{A_0,A'}(\sfa)-
    \langle \alpha, \sfa\rangle \cdot \underbrace{\Delta_\sfa Z_{A_0,A'}(\sfa)}_0 -
    \Big\{\left\langle \alpha, \sfa\right\rangle, Z_{A_0,A'}(\sfa)\Big\}
    \\=
   \left \langle  \alpha, \frac{\partial}{\partial \sfa} \right\rangle Z_{A_0,A'}(\sfa)
\end{multline}
-- the r.h.s. of (\ref{variation of kin operator}), as claimed.\footnote{ Here $$\{f,g\} = f\left\langle \frac{\overleftarrow\partial }{\partial \sfa^{1}}, \frac{\overrightarrow\partial }{\partial \sfa^{2}}\right\rangle g - f\left\langle \frac{\overleftarrow\partial }{\partial \sfa^{2}}, \frac{\overrightarrow\partial }{\partial \sfa^{1}}\right\rangle g $$ is the BV bracket of $f$ and $g$, see e.g. \cite[Section 4.4.3]{mnev2019quantum}. } 
%\marginpar{April 8 26. Added footnote 4 }
\end{rem}

% \begin{thm}[Parallel transport]
%     For $\alpha \in H^1_{d_{A_0}}(M,\g)$ we have
%     \begin{equation}
%         Z_{A_0 + \delta_{A_0,A'}(\alpha),A'}(d\varphi_{A_0}(\alpha)a)= Z_{A_0,A'}(a + \alpha)
%     \end{equation}
% \end{thm}

%{\color{red} [DEFINE $\til{\nabla}^G$]}

Consider the graded vector bundle $\mc{D}$ over $\FC'$ with fiber over $A$ being the space of formal half-densities on cohomology
\begin{equation}\label{D factorization}
    \mc{D}_{A}=\mr{Dens}^{\frac12,\mr{formal}}(H^\bt_{A}[1])\cong
    \mr{Det} (H^1_{A})^* \otimes \widehat{\mr{Sym}}(H^\bt_{A}[1])^*.
\end{equation}
% Thus, the bundle $\mc{D}$ can be written as a tensor product  of the line bundle $\mr{Det}(\mathbb{H}^1)$ and the graded vector bundle $\wh{\Sym}(\mathbb{H}[1])^*$.
% %$\mr{Fun}^\mr{formal}$, 
% %with fibers being the two factors in the r.h.s. of (\ref{D factorization}).
% Here $\mathbb{H}$ is the cohomology bundle (Section \ref{sss cohomology comparison map}).
% The flat connection $\nabla^\mathbb{H}$ in the cohomology bundle induces a flat connection $\nabla^\mr{Det}$ in the determinant line bundle $\mr{Det}(\mathbb{H}^1)$. 

The flat connection $\nabla^{\mathbb{H},A'}$ in the cohomology bundle (Section \ref{sss cohomology comparison map}) induces a flat connection in $\mc{D}$ which by abuse of notations we will also denote $\nabla^{\mathbb{H},A'}$.

Let $\mr{pr}_1$ be the projection onto the first factor in $\FC'\times \FC'$.

\begin{defn}[Partial Grothendieck connection]\label{def: partial Grothendieck connection}
One has a partial connection $\til{\nabla}^G$ on the bundle $\mr{pr}_1^* \mc{D}$ over $\mc{U}\subset \FC'\times \FC'$ defined by 
%\marginpar{\bl sign??}
% \marginpar{\bl EDIT: should take into account $B_t$ acting on $\sfa$ (or rephrase as a shift of $\nabla^\mr{triv}$ induced by $\nabla^\mathbb{H}$ on 1/2-densities)}
% \begin{equation}\label{partial Grothendieck connection}
% (\til{\nabla}^G)_{(\chi,0)} (v\otimes \xi(\sfa))=\nabla^\mr{Det}_\chi(v)\otimes \xi(\sfa) + v\otimes \left(\left\langle \chi,\frac{\delta}{\delta A_0} \right\rangle -\left\langle [\chi],\frac{\partial}{\partial\sfa} \right\rangle\right) \xi(\sfa)
% \end{equation}
\begin{equation}\label{partial Grothendieck connection}
    (\til{\nabla}^G)_{(\chi,0)}\xi(\sfa)=\nabla^{\mathbb{H},A'}_\chi\xi(\sfa) -\left\langle [\chi],\frac{\partial}{\partial\sfa} \right\rangle\xi(\sfa)
\end{equation}
for any $\chi\in \mr{Harm}^1_{A,A'}$ and $\xi(\sfa)$ a section  of $\mr{pr}_1^* \mc{D}$.
%$v$ a section of $\mr{Det}(\mathbb{H}^1)$ and $\xi(\sfa)$ a section of $\wh{\Sym}(\mathbb{H}[1])^*$.
Here $(\chi,0)$ is a tangent vector to $\FC'\times \FC'$ at a point $(A,A')$. 

% \marginpar{\bl remove: there is no $\nabla^\mathbb{H}$-trivialization along both $A,A'$}
% {\color{gray}
% Put another way, in a local trivialization induced by $\nabla^\mathbb{H}$, one has
% $$  (\til{\nabla}^G)_{(\chi,0)}\xi(\sfa)=
%     \left(\left\langle \chi,\frac{\delta}{\delta A_0} \right\rangle -\left\langle [\chi],\frac{\partial}{\partial\sfa} \right\rangle\right)\xi(\sfa)
%     $$
% }

Thus (\ref{partial Grothendieck connection}) allows to differentiate sections of $\mr{pr}_1^* \mc{D}$ in the direction of infinitesimal harmonic shifts of $A$.
\end{defn}

%\marginpar{\bl added Aug 17}
\begin{rem}\label{rem: partial Grothendieck connection}
The partial connection (\ref{partial Grothendieck connection}) is induced (via pushforward of half-densities) from the fiber bundle partial connection on the cohomology bundle $\mr{pr}_1^* \mathbb{H}$ over $\mc{U}$ 
defined by the parallel transport
%\marginpar{\bl $\sfa,\alpha$ in $H^\bt$ or $H^1$???}
\begin{equation}\label{nabla^G tilde parallel transport on H}
\begin{array}{ccc}
    H_A^\bt=\mr{pr}_1^* \mathbb{H}|_{A,A'} & \ra & H^\bt_{\til{A}}=\mr{pr}_1^* \mathbb{H}|_{\til{A},A'} \\
    \sfa& \mapsto & \til{\sfa}=\mathfrak{B}_{\til{A}\la A,A'} (\sfa-\alpha)
\end{array}
\end{equation}
with $\alpha\in H^1_A$ sufficiently small and $\til{A}$ as in (\ref{Atilde}). We remark that the parallel transport (\ref{nabla^G tilde parallel transport on H}) satisfies
\begin{equation}\label{rem: partial Grothendieck connection eq1}
    \varphi_{A,A'}(\sfa)=\varphi_{\til{A},A'}(\til{\sfa}).
\end{equation}
for $\sfa\in H^1_A$ sufficiently small.
\end{rem}

\begin{cor}[Horizontality with respect to the partial Grothendieck connection] 
\label{cor: horizontality wrt nabla_G}
One has
        \begin{equation}
            \widetilde{\nabla}^G Z_{A,A'} = 0. 
        \end{equation}
\end{cor}
This is just a rephrasing of Corollary \ref{cor: variation of kin operator}.

\subsection{Changing the gauge-fixing operator}
\begin{thm}[Changing the gauge-fixing operator]\label{thm: change gf}
We have that, for $A_0'$ and $A'_1$ flat connections close to a flat connection $A$ 
        \begin{equation}\label{thm change of A' eq}
            Z_{A,A_1'}(\sfa) = %\frac{I_{A,A'_1}}{I_{A,A'_0}}
            Z_{A,A'_0}(\sfa) -i\hbar \Delta_\sfa R_{A,A'_0,A'_1}(\sfa). 
        \end{equation}
        with $R_{A,A'_0,A'_1}(\sfa)$ a formal half-density on $H^\bt_A[1]$ given by (\ref{R in change of Z with A'}) below. 
        %\marginpar{\bl write a formula for $R$?}
\end{thm}
This is an immediate consequence of Proposition \ref{prop: variation of Z wrt A'} below, by integrating over a path $A'_t$ from $A'_0$ to $A'_1$.

Let us consider the effect of an infinitesimal change of $A'\ra A'+\delta A'$ on $Z_{A,A'}$. 
%Let $A'_t$ be a path of flat connections close to $A$. 
Consider the following endomorphisms of $\Omega^\bt(M,\g)$:
\begin{equation}\label{Lambda, II, PP}
    \Lambda = K_{A,A'} \mr{ad}^*_{\delta A'} G_{A,A'}, \quad
    \mathbb{I} = G_{A,A'} \mr{ad}^*_{\delta A'}, \quad
    \mathbb{P} = \mr{ad}^*_{\delta A'} G_{A,A'}.
\end{equation}
They arise in the first-order deformation of the SDR data $(i,p,K)_{A,A'}$ resulting from the deformation of $A'$:
\begin{equation}\label{variation of ipK in A'}
\begin{aligned}
    \delta_{A'} K_{A,A'}=&[d_A,\Lambda]+P_{A,A'}\mathbb{P} +\mathbb{I} P_{A,A'}, \\
   \delta_{A'} i_{A,A'}  =& -d_A\mathbb{I} i_{A,A'},\\
   \delta_{A'} p_{A,A'}  =& -p_{A,A'}\mathbb{P} d_A,
\end{aligned}
\end{equation}
cf. Lemma \ref{lem: SDR deformations}. 
%\marginpar{Changed ref to appendix} 
Here $P_{A,A'}=i_{A,A'}p_{A,A'}$ is the projection onto $(A,A')$-harmonic forms. We note that $\mathbb{I}$ and $\mathbb{P}$ are mutually transpose w.r.t. Poincar\'e pairing on forms and cohomology, while $\Lambda$ is symmetric w.r.t. Poincar\'e pairing.

\begin{comment}  % OLD Xi
Let $\Xi$ be the following endomorphism of $H^\bt_A$:\marginpar{\bl write the simpler f-la for $\Xi$ instead, using $P_{A'}$?}
\begin{multline}\label{Xi}
    \Xi=p \Big( \mathbb{P}\,\ad_\beta (1-K\ad_\beta)^{-1}-(1-\ad_\beta K)^{-1} \ad_\beta\, \mathbb{I} +\\
    +  (1-\ad_\beta K)^{-1}\ad_\beta\, \Lambda\, \ad_\beta (1-K\ad_\beta)^{-1} \Big) i
%   \Xi=p_{A',A'} \Big( \mathbb{P}\ad_\beta -\ad_\beta \mathbb{I}+ \ad_\beta \Lambda \ad_\beta \Big) i_{A',A'}
\end{multline}
\end{comment}

% \marginpar{\bl remove/ integrate into Appendix}

% {\color{gray}
% Denote 
% \begin{equation}\label{Xi}
% \Xi= \frac12\mr{Str}_{\Omega^\bt} P_{A'}(\mathbb{P} \ad_\beta-\ad_\beta \mathbb{I} - \ad_\beta \Lambda \ad_\beta),
% \end{equation}
% where $\beta\colon = A-A'$ and $P_{A'}$ the projection onto $(A',A')$-harmonic forms.\footnote{Note that $\Xi$ can be seen as a finite-dimensional supertrace over $H_{A'}$.}
% %and we are suppressing the subscript $A,A'$ in $i,p,K$.
% }

\begin{prop}\label{prop: variation of Z wrt A'}
%Given a path $A'_t$ of flat connections close to $A$, 
For $A,A'\in\FC'$ a pair of close flat connections, the variation of 
$Z_{A,A'}(\sfa)$ with respect to variation of $A'$ is given by
\begin{equation}\label{variation of Z_A,A' wrt A'}
%    \frac{d}{dt} Z_{A,A'_t}(\sfa) = \Delta_{\sfa} \left(\left(\frac12 \langle \sfa, \Xi(\sfa) \rangle +r(\sfa)\right) Z_{A,A'}(\sfa)\right)
\delta_{A'} Z_{A,A'}(\sfa)=
-i\hbar
\Delta_\sfa \left(r_{A,A';\delta A'}(\sfa) Z_{A,A'}(\sfa)\right)
\end{equation}
%\marginpar{\bl Apr21 changed sign in rhs}
where $r_{A,A';\delta A'}(\sfa)$ is the sum of connected Feynman graphs with one marked edge decorated by $\Lambda$ or one marked leaf decorated by $\mathbb{I}$.
\end{prop}

Note that if $A'_t$ is a path from $A'_0$ to $A'_1$, integrating (\ref{variation of Z_A,A' wrt A'}) we obtain (\ref{thm change of A' eq}) with 
\begin{equation}\label{R in change of Z with A'}
    R_{A,A'_0,A'_1}(\sfa)=\int_0^1 dt\, r_{A,A'_t;\dot{A}'_t}(\sfa) Z_{A,A'_t}(\sfa).
\end{equation}

\begin{proof}[%Sketch of 
Proof of Proposition \ref{prop: variation of Z wrt A'}]
Denote $Z^0_A=e^{\frac{i}{\hbar}S_{CS}(A)}e^{\frac{\pi i}{4}\psi_A}\tau_A^{\frac12}$.
We have
\begin{equation}
  %Z_{A,A'}(\sfa)^{-1}\, 
  (Z^0_A)^{-1}
  \delta_{A'} Z_{A,A'}(\sfa) =
%\underbrace{I_{A,A'}^{-1}\delta_{A'}I_{A,A'}}_\Xi + 
\sum_{\Gamma} \delta_{A'}\overline{\Phi}_{\Gamma,A,A'}(\sfa)
\end{equation}
-- a sum over possibly disconnected trivalent graphs with leaves allowed,
where we denoted 
$$\overline{\Phi}_{\Gamma,A,A'}(\sfa)=\underbrace{\frac{(-i\hbar)^{-\chi(\Gamma)}}{|\mr{Aut}(\Gamma)|}}_{c_\Gamma}\Phi_{\Gamma,A,A'}(\sfa)=\int_{\overline{\mr{Conf}}_V(M)} \omega_\Gamma.  $$ 
Here $\omega_\Gamma$ is the integrand of (\ref{eq: def Phi_Gamma}) in $(A,A')$-gauge, normalized with an appropriate power of $\hbar$ and a combinatorial factor:
\begin{equation}
    \omega_\Gamma=c_\Gamma \Big\langle \bigwedge_{e=(uv)\in E} \pi^*_{uv} \eta\wedge \bigwedge_{l\in L} \pi^*_{v(l)}i(\sfa),\bigotimes_{v\in V}f \Big\rangle.
\end{equation}
Here:
\begin{itemize}
    \item $V,E,L$ are the sets of vertices, edges and leaves of $\Gamma$. 
    %$v(l)$ is the vertex adjacent to the leaf $l$.
    \item For a short loop $e=(vv)$, expression $\pi_{vv}^*\eta$ is interpreted as $\pi_v^* \eta^\Delta$ with $\eta^\Delta$ the regularized diagonal evaluation of $\eta$.
    \item $A,A'$ subscripts are suppressed.
\end{itemize}
We will also use shorthand notation $\eta_e= \pi^*_{uv}\eta$, $i(\sfa)_l=\pi_{v(l)}^* i(\sfa)$.

The rest of the proof is parallel to %follows 
the arguments of \cite[Section 4.5, Appendix A]{Cattaneo2008}, \cite[Lemma 4.11]{Cattaneo2017}.\footnote{The technology of proving the properties of perturbative partition functions via Stokes' theorem on the configuration space was developed in particular in \cite[\S 6]{Axelrod1994}, \cite{Kontsevich1994}, \cite{Kontsevich2003}.
}

Variation of the value of a Feynman graph in $A'$ is the sum over edges and leaves of the graph of replacing that edge or leaf with its variation:
\begin{multline}\label{prop 4.11 delta_A' Phi}
    \delta_{A'}\overline\Phi_{\Gamma,A,A'}=\\
    \sum_{e_\mm\in E} \int_{\Conf} c_\Gamma \Big\langle \delta_{A'}\eta_{e_\mm}\wedge\bigwedge_{e\in E\setminus e_\mm}\eta_e\wedge \bigwedge_{l\in L}i(\sfa)_l,\bigotimes_{v\in V} f \Big\rangle
    \\+ \sum_{l_\mm\in L} \int_{\Conf} c_\Gamma \Big\langle 
    \bigwedge_{e\in E}\eta_e\wedge \delta_{A'}i(\sfa)_{l_\mm} \wedge\bigwedge_{l\in L\setminus l_\mm}i(\sfa)_l,\bigotimes_{v\in V} f
    \Big\rangle .
\end{multline}
(Subscript $\mm$ in $e_\mm,l_\mm$ is for ``marked'' edge/leaf.) We introduce the following notation: for $O\in \{\Lambda, P,\II P, P\PP\}$ an operator on $\Omega^\bt(M,\g)$, we will denote $\eta_O \in \Omega^p(\Conf[2],\g\otimes \g)$ its integral kernel, seen as a form on the compactified configuration space of two points of degree $p=1$ for $\Lambda$, $p=2$ for $\II P$ and $P\PP$, $p=3$ for $P$.
From (\ref{variation of ipK in A'}) we have
\begin{equation}\label{prop 4.11 delta_A' eta}
    \delta_{A'} \eta= -d\eta_\Lambda + 
    %\eta_\mathbb{I}+ \eta_\mathbb{I}^T,
    \eta_{\II P}+ \eta_{P \PP},
    \quad \delta_{A'}i(\sfa)=-d\mathbb{I}i(\sfa).
\end{equation}
Note that operators $\II P$ and $P\PP$ are mutually transpose.
%Here $\eta_\Lambda \in \Omega^1(\Conf[2],\g\otimes\g)$, $\eta_\II,\eta_\II^T\in \Omega^2(\Conf[2],\g\otimes\g) $ are the integral kernels of operators $\Lambda$, $\II P$, $P\PP$. (The latter two operators are mutually transpose w.r.t. Poincar\'e pairing, hence the notation.) Note that $\eta_\II,\eta_\II^T$ extends smoothly to forms on $M\times M$. \marginpar{is this remark needed? (smoothness of $\eta_\II$ across the diagonal)}

Next, we substitute (\ref{prop 4.11 delta_A' eta}) in (\ref{prop 4.11 delta_A' Phi}) and use Stokes' theorem on the configuration space to move $d$ from the marked edge $e_\mm$ or leaf $l_\mm$ to another edge $e'_\mm$ (note that\footnote{\label{footnote: 1 as principal boundary stratum}
As a distributional form on $M\times M$, the integral kernel of the operator $[d,K]=1-P$ contains the integral kernel of the identity -- the delta-distribution on the diagonal in $M\times M$. On the  configuration space of pairs of \emph{distinct} points on $M$, ${M\times M\setminus \mr{Diag}}$, this delta-distribution disappears. However, in the formalism of configuration space integrals, it effectively reappears as the contribution of the principal boundary stratum of the configuration space, see (\ref{prop 4.11 delta_A' Phi conf space integral line 4}), (\ref{prop 4.11 delta_A' Phi conf space integral line 5}) and the discussion below it.
} $d\eta=\eta_P$, since $[d,K]=1-P$). In the process, Stokes' theorem produces additional contributions of the boundary of the configuration space.%\marginpar{\bl reformat}
%\begin{subequations}
\begin{align}%\label{prop 4.11 delta_A' Phi conf space integral line 1}
    &\delta_{A'}  \overline\Phi_{\Gamma,A,A'}=  \notag \\ 
    \label{prop 4.11 delta_A' Phi conf space integral line 1}
    &
    \sum_{e_\mm\neq e'_\mm\in E}\int_{\Conf} c_\Gamma \cdot \\ 
    \notag
    &\cdot \Big\langle (\eta_P)_{e'_\mm}\wedge (\eta_\Lambda)_{e_\mm}\wedge  \bigwedge_{e\in E\setminus\{e_\mm,e'_\mm\}}\eta_e \wedge \bigwedge_{l\in L} i(\sfa)_l,\bigotimes_{v\in V} f  \Big\rangle\\
    \label{prop 4.11 delta_A' Phi conf space integral line 2}
    &
    +\sum_{l_\mm\in L,\, e'_\mm\in E} \int_{\Conf} c_\Gamma \cdot \\ 
    \notag &\cdot
    \Big\langle  (\eta_P)_{e'_\mm} \wedge \bigwedge_{e\in E\setminus e'_\mm}\eta_e \wedge \II i(\sfa)_{l_\mm}\wedge \bigwedge_{l\in L\setminus l_\mm} i(\sfa)_l, \bigotimes_{v\in V} f\Big\rangle\\
    \label{prop 4.11 delta_A' Phi conf space integral line 3}
    &
    +\sum_{e_\mm\in E} \int_{\Conf} c_\Gamma \cdot \\ 
    \notag
    &\cdot \Big\langle (\eta_{\II P}+\eta_{P\PP})_{e_\mm}\wedge\bigwedge_{e\in E\setminus e_\mm}\eta_e\wedge \bigwedge_{l\in L}i(\sfa)_l,\bigotimes_{v\in V} f \Big\rangle\\
    \label{prop 4.11 delta_A' Phi conf space integral line 4}
    &+\sum_{e_\mm\in E} \int_{\partial \Conf} c_\Gamma \cdot \\
    \notag
    &\cdot \Big\langle 
    (\eta_\Lambda)_{e_\mm}\wedge\bigwedge_{e\in E\setminus e_\mm}\eta_e\wedge \bigwedge_{l\in L}i(\sfa)_l,\bigotimes_{v\in V} f
    \Big\rangle \\
    \label{prop 4.11 delta_A' Phi conf space integral line 5}
    &+ \sum_{l_\mm\in L} \int_{\partial \Conf} c_\Gamma \cdot \\
    \notag&\cdot \Big\langle 
    \bigwedge_{e\in E}\eta_e\wedge \II i(\sfa)_{l_\mm} \wedge\bigwedge_{l\in L\setminus l_\mm}i(\sfa)_l,\bigotimes_{v\in V} f
    \Big\rangle .
\end{align}
%\end{subequations}

The boundary contributions here (\ref{prop 4.11 delta_A' Phi conf space integral line 4}), (\ref{prop 4.11 delta_A' Phi conf space integral line 5}) cancel out in the sum over graphs $\Gamma$ as follows. The boundary of the configuration space is stratified (see \cite{Axelrod1994}, \cite{Cattaneo2008}) into ``principal strata'' (differential geometric blow-ups of collapses of two points on $M$) and ``hidden strata'' (collapses of $\geq 3$ points). Contributions of principal strata cancel out in the sum over graphs $\Gamma$ as a consequence of Jacobi identity  (or, equivalently, as a consequence of the classical master equation on BV Chern-Simons action).

Hidden boundary strata of the configuration space do not contribute to the integrals in (\ref{prop 4.11 delta_A' Phi conf space integral line 4}), (\ref{prop 4.11 delta_A' Phi conf space integral line 5}); we defer the proof of this to Appendix \ref{Appendix: vanishing of hidden boundary strata}.

% {\color{gray} We denote:
% \begin{enumerate}[(i)]
%     \item The summand in (\ref{prop 4.11 delta_A' Phi conf space integral line 1}) by $\overline\Phi_{\Gamma^{(1)}_{e_\mm,e'_\mm}}$ -- the weight of a graph $\Gamma^{(1)}_{e_\mm,e'_\mm}$ with one edge $e_\mm$ decorated by $\Lambda$ and one other edge $\eta'_\mm$ decorated by $P$. 
%     \item Likewise, we denote the summand in (\ref{prop 4.11 delta_A' Phi conf space integral line 2}) by $\overline\Phi_{\Gamma^{(2)}_{l_\mm,e'_\mm}}$ -- the weight of a graph $\Gamma^{(2)}_{l_\mm,e'_\mm}$  with one leaf $l_\mm$ decorated by $\II i(\sfa)$ and one edge $e'_\mm$ decorated by $P$. 
%     \item
% \end{enumerate}
% }

\begin{figure}
    \centering
    \includegraphics[width=0.8\linewidth]{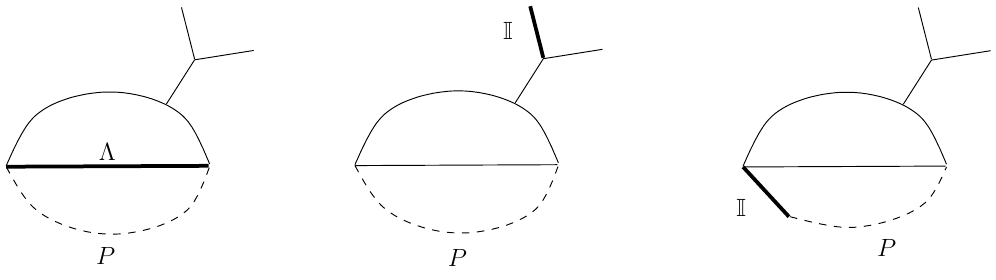}
    \caption{Examples of graphs contributing to (\ref{prop 4.11 delta_A' Phi conf space integral line 1}), (\ref{prop 4.11 delta_A' Phi conf space integral line 2}), (\ref{prop 4.11 delta_A' Phi conf space integral line 3}). 
    %$\Gamma^{(1)}$, $\Gamma^{(2)}$, $\Gamma^{(3)}$.
    }
    \label{fig:variation of A'}
\end{figure}

Denote 
\begin{multline}\label{prop 4.11 zeta}
    \zeta=\\
    \sum_{\Gamma^\#} \sum_{e_\mm\in E^\#} \int_{\overline{\mr{Conf}}_{V^\#}(M)} c_{\Gamma^\#} 
    \Big\langle (\eta_\Lambda)_{e_\mm}\wedge \bigwedge_{e\in E^\#\setminus e_\mm}\eta_e\wedge \bigwedge_{l\in L^\#} i(\sfa)_l,\bigotimes_{v\in V^\#} f \Big\rangle\\
    +\sum_{\Gamma^\#} \sum_{l_\mm\in L^\#} \int_{\overline{\mr{Conf}}_{V^\#}(M)} c_{\Gamma^\#} 
    \Big\langle \bigwedge_{e\in E^\#}\eta_e\wedge \II i(\sfa)_{l_\mm} \wedge \bigwedge_{l\in L^\#\setminus l_\mm} i(\sfa)_l,\bigotimes_{v\in V^\#} f \Big\rangle.
\end{multline}
Here $\Gamma^\#$ runs over possibly disconnected trivalent graphs with leaves allowed; $V^\#,E^\#,L^\#$ stand for the sets of vertices/edges/leaves of $\Gamma^\#$.

Note that 
\begin{equation}
\zeta=(Z^0_A)^{-1} r_{A,A';\delta A'}(\sfa) Z_{A,A'}(\sfa)    
\end{equation}
with $r$ defined as in Proposition \ref{prop: variation of Z wrt A'}.

On the other hand, applying $\Delta_\sfa$ to each term in (\ref{prop 4.11 zeta}) yields a sum over  pairs of $i(\sfa)$ leaves, of replacement of that pair by a new edge adjoined to the graph $\Gamma^\#$ and decorated by $P$:
%\begin{subequations}
\begin{align}
    &\Delta_\sfa \zeta = \notag\\
    &\sum_{\Gamma^\#} \sum_{e_\mm\in E^\#,\, l_\mm\neq l'_\mm\in L^\#} \int_{\overline{\mr{Conf}}_{V^\#}(M)} c_{\Gamma^\#}\cdot  \label{prop 4.11 Delta zeta line 1} \\
    &\notag
     \cdot\Big\langle (\eta_P)_{(v(l_\mm) v(l'_\mm))} \wedge (\eta_\Lambda)_{e_\mm}\wedge 
    \bigwedge_{e\in E^\#\setminus e_\mm}\eta_e\wedge \bigwedge_{l\in L^\#\setminus\{l_\mm,l'_\mm\}} i(\sfa)_l,\bigotimes_{v\in V^\#} f \Big\rangle\\
    &+
    \sum_{\Gamma^\#} \sum_{l_\mm,l'_\mm,l''_\mm\in L^\#\,\mr{distinct}} \int_{\overline{\mr{Conf}}_{V^\#}(M)} c_{\Gamma^\#}\cdot \label{prop 4.11 Delta zeta line 2}\\
    & \notag  \cdot
    \Big\langle \bigwedge_{e\in E^\#}\eta_e\wedge \II i(\sfa)_{l_\mm} \wedge 
    (\eta_P)_{(v(l'_\mm) v(l''_\mm))} \wedge
    \bigwedge_{l\in L^\#\setminus \{l_\mm,l'_\mm,l''_\mm\}} i(\sfa)_l,\bigotimes_{v\in V^\#} f \Big\rangle\\
    &+\sum_{\Gamma^\#} \sum_{l_\mm\neq l'_\mm\in L^\#} \int_{\overline{\mr{Conf}}_{V^\#}(M)} c_{\Gamma^\#} \cdot \label{prop 4.11 Delta zeta line 3}
    \\ & \notag \cdot
    \Big\langle \bigwedge_{e\in E^\#}\eta_e\wedge (\eta_{\II P})_{(v(l_\mm) v(l'_\mm))}  \wedge \bigwedge_{l\in L^\#\setminus \{l_\mm, l'_\mm\} } i(\sfa)_l,\bigotimes_{v\in V^\#} f \Big\rangle.
\end{align}    
%\end{subequations}
Comparing with the computation of $\delta_{A'}\overline\Phi_{\Gamma,A,A'}$, we find that:
\begin{itemize}
    \item $\sum_\Gamma (\ref{prop 4.11 delta_A' Phi conf space integral line 1})=-i\hbar(\ref{prop 4.11 Delta zeta line 1})$. Here $\Gamma$ is $\Gamma^\#$ with a new $P$-edge $e'_\mm=(v(l_\mm) v(l'_\mm))$ attached.  The extra $-i\hbar$ factor appears because attaching a new $P$-edge to $\Gamma^\#$ changes its Euler characteristic by $-1$ and thus changes the power of $(-i\hbar)$ in $c_\Gamma$ by $+1$. 
    %The minus on the left is to compensate the sign coming from Stokes' theorem in (\ref{prop 4.11 delta_A' Phi conf space integral line 1}).
    \item $\sum_\Gamma (\ref{prop 4.11 delta_A' Phi conf space integral line 2})=-i\hbar (\ref{prop 4.11 Delta zeta line 2})$. Here $\Gamma$ is $\Gamma^\#$ with a new $P$-edge $e'_\mm=(v(l'_\mm)v(l''_\mm))$ attached. 
    \item $\sum_\Gamma (\ref{prop 4.11 delta_A' Phi conf space integral line 3})=-i\hbar (\ref{prop 4.11 Delta zeta line 3})$. Here $\Gamma$ is $\Gamma^\#$ with a new $\II P$-edge attached, $e_\mm=(v(l_\mm)v(l'_\mm))$, or new $P\PP$-edge attached, $e_\mm=(v(l'_\mm)v(l_\mm))$.
\end{itemize}
Therefore we have
\begin{equation}
    \underbrace{\sum_\Gamma \delta_{A'} \overline\Phi_{\Gamma,A,A'}}_{(Z^0_A)^{-1}\delta_{A'} Z_{A,A'}(\sfa)}=\underbrace{-i\hbar\Delta_\sfa\zeta}_{i\hbar (Z^0_A)^{-1} \Delta_\sfa(r_{A,A';\delta A'}(\sfa)Z_{A,A'}(\sfa))}
\end{equation}
Multiplying both sides by $Z^0_A$, we have (\ref{variation of Z_A,A' wrt A'}).

\end{proof}

\begin{comment}
\begin{lem}
If $A_t'$ is a curve of flat connections with $\alpha = \dot{A_0}'$, then 
\begin{equation}
   \left.\frac{d}{dt}\right|_{t=0} Z_{A_0,A'_t} = \Delta_\sfa (R_{A_0,A_0',\alpha}(\sfa))
\end{equation}
where $(R_{A_0,A_0',\alpha}(\sfa))$ is given by summing over all diagrams with either one marked edge or one marked leaf, evaluated by sending the marked edge to 
\begin{equation}
    \Lambda = K_{A_0,A_0'}\mr{ad}_\alpha^*G_{A_0,A_0'}
\end{equation}
and sending marked leaves to
\begin{equation}
    \mathbb{I} = G_{A_0,A_0'}\mr{ad}_\alpha^*\iota_{A_0,A_0'}
\end{equation}
\end{lem}
\end{comment}

\subsubsection{Total horizontality (modulo $\Delta$-exact terms) on $\FC\times \FC$}
\label{sss: T^I+T^II+T^III}
%\begin{rem}
    One can summarize the results of Propositions \ref{prop: variation of Z wrt A'}, \ref{prop: Z desync equivariance under diagonal gauge transf}  and Corollary 
    \ref{cor: horizontality wrt nabla_G}
    %\ref{cor: variation of kin operator} 
    as follows. One has splitting of the tangent bundle of $\mc{U}\subset \FC'\times \FC'$ into a direct sum of three integrable distributions
    \begin{equation}\label{T^I+T^II+T^III}
        T\mc{U}=T^I\oplus T^{II} \oplus T^{III}.
    \end{equation}
Here:
\begin{itemize}
    \item $T^I_{A,A'}=\{(\chi,0)\;|\; \chi\in \mr{Harm^1_{A,A'}}\}$ -- harmonic shifts of $A$.
    \item $T^{II}_{A,A'}=\{(0,\gamma)\;|\; \gamma\in \Omega^1_{A'\mr{-closed}}\}$ -- shifts of $A'$.
    \item $T^{III}_{A,A'}=\{(d_A\beta,d_{A'}\beta)\;|\; \beta\in \Omega^0\}$ -- diagonal gauge transformations of $(A,A')$.
\end{itemize}
For the bundle $\mc{D}$ of formal half-densities on $H_A$ over $\mc{U}$, one has three flat partial connections
$\til\nabla^G$ (\ref{partial Grothendieck connection}), $\nabla_{A'}=\delta_{A'}$, $\nabla^\mr{gauge}$ along these three distributions. Here $\nabla^\mr{gauge}$ is such that its parallel transport takes a formal half-density $\psi(\sfa)$ over $(A,A')$ to the half-density $\psi(\sfg\sfa \sfg^{-1})$ over $({}^\sfg A,{}^\sfg A')$ for any $\sfg\colon M\ra G$.
These three partial connections can be assembled into a total flat connection
\begin{equation} \label{nabla^tot}
\nabla^\mr{tot}= \til\nabla^G+\nabla_{A'}+\nabla^\mr{gauge}.
\end{equation}
The extended partition function then satisfies the total horizontality equation (modulo $\Delta$-exact terms):
%\marginpar{\bl Apr 21 changed sign (for $\delta A'$ even)}
\begin{equation}\label{nabla^tot Zhat = 0}
    \nabla^\mr{tot} Z=
    -i\hbar\Delta_\sfa\Big( r_{{\til{\delta A'}}}(\sfa) Z_{A,A'}(\sfa)\Big).
\end{equation}
%\marginpar{\bl Edit or remove: as stated it doesn't account correctly for diag gauge transformations}
% Explicitly, in the trivialization of $\mc{D}$ induced by the flat connection $\nabla^\mathbb{H}$ in $H_A$-cohomology bundle, one can write the equation (\ref{nabla^tot Zhat = 0})  as
% \begin{equation}
%     \left(\delta_A+\delta_{A'}-\left\langle p(\delta A),\frac{\partial}{\partial \sfa} \right\rangle -\Xi_{\til{\delta A'}}\right)  %Z 
%     Z_{A,A'}(\sfa)
%     =-i\hbar\Delta_\sfa(
%     r_{{\til{\delta A'}}}(\sfa) Z_{A,A'}(\sfa)
%      %r_{{\til{\delta A'}}} Z
%     ).
% \end{equation}
Here $\delta A'$ is replaced in %$\Xi$ and 
$r$ (as in Proposition \ref{prop: variation of Z wrt A'})  with the expression
\begin{equation}\label{tilde delta A'}
\til{\delta A'}=\delta A'-d_{A'}K_{A,A'}\delta A     
\end{equation} 
-- %the projection of $\delta A'$ onto $T^{II}$ in the decomposition (\ref{T^I+T^II+T^III}).
a 1-form on $\mc{U}$ valued in $\Omega^1(M,\g)$ which vanishes along $T^I$ and $T^{III}$ and coincides with $\delta A'$ on $T^{II}$. Put another way, $\til{\delta A'}$ is the projector onto $T^{II}$ in (\ref{T^I+T^II+T^III}).

% $\nabla^\mr{tot}$ explicitly as
% $$
% \nabla^\mr{tot}=\delta_A+\delta_{A'}-\left\langle p(\delta A),\frac{\partial}{\partial \sfa} \right\rangle+i\hbar \Delta_\sfa -\Xi-\left\{\frac12\langle\sfa,\wh\Theta(\sfa)\rangle,- \right\}.
% $$
%\end{rem}

\subsection{Extension of $Z_{A,A'}$ to a horizontal nonhomogeneous form in $A'$}\label{sec: ext Z A prime}
In this section we describe a refinement of Proposition \ref{prop: variation of Z wrt A'}: one combines the partition function $Z_{A,A'}(\sfa)$ with the BV generator in the r.h.s. of (\ref{variation of Z_A,A' wrt A'}) into a nonhomogeneous form in $A'$
$$ \wh{Z}_{A,A'}(\sfa)= Z_{A,A'}(\sfa)+ r(\sfa) Z_{A,A'}(\sfa)+\underbrace{\cdots}_{\mr{degree}\;\geq 2\;\mr{in}\; A'}$$
so that the full object satisfies horizontality condition
\begin{equation}\label{horizontality of Zhat}
\wh\nabla_{A'} \wh{Z}_{A,A'}(\sfa)=0
\end{equation}
with respect to the flat partial superconnection\footnote{
%When we say ``superconnection,'' we always mean $\mathbb{Z}$-graded
For the definition of a superconnection, see e.g. \cite[Definition 1.2]{igusa2009iterated}.
The superconnection (\ref{nabla_A' hat}) can be thought of as a correction of the trivial superconnection $\delta_{A'}$ by  $\Delta_\sfa$ -- a $0$-form in $A'$.
%and $\Xi$ -- a $1$-form. 
} 
\begin{equation} \label{nabla_A' hat}
\wh\nabla_{A'}=\delta_{A'}-i\hbar \Delta_{\sfa}
\end{equation}
in the direction of $A'$ in the bundle of formal half-densities in $H_A$ over $\mc{U}\subset \FC'\times \FC'$. 
In particular, in degree zero in $A'$, (\ref{horizontality of Zhat}) is the equation $\Delta_\sfa Z_{A,A'}(\sfa)=0$ (the BV quantum master equation) and in degree one in $A'$, it yields the equation (\ref{variation of Z_A,A' wrt A'}).

We proceed to the detailed construction.

Consider the following quadruple of nonhomogeneous forms in $A'$ valued in linear maps:
\begin{equation}\label{ihat, phat, Khat, Phihat}
\begin{aligned}
    \wh{i}&=\sum_{k\geq 0} (-G_{A,A'}\ad^*_{\delta A'})^k i_{A,A'}=i+\mathbb{I}i+\cdots \quad \in \Omega^{0,\bt}(\mc{U},\mr{Hom}(H_A^\bt,\Omega^\bt(M,\g))), \\
    \wh{p}&= \sum_{k\geq 0} p_{A,A'}(-\ad^*_{\delta A'}G_{A,A'})^k=p+p\mathbb{P}+\cdots\quad \in \Omega^{0,\bt}(\mc{U},\mr{Hom}(\Omega^\bt(M,\g),H_A^\bt)), \\
     \wh{K}&= \sum_{k\geq 0} K_{A,A'}(-\ad^*_{\delta A'}G_{A,A'})^k=K+\Lambda+\cdots\quad \in \Omega^{0,\bt}(\mc{U},\mr{End}(\Omega^\bt(M,\g))),\\
     \wh\Theta&=p_{A,A'}\ad^*_{\delta A'} d_A G_{A,A'}(-G_{A,A'}\ad^*_{\delta A'}) i_{A,A'}\quad  \in \Omega^{0,2}(\mc{U},\mr{End}(H^\bt_A)).
\end{aligned}
\end{equation}
%\marginpar{\bl Sep 6: only $k=1$ term survives for $\wh\Theta$ for degree reason; bounds for $k$ for other series can be improved also}
Here for the first three maps, the 0-form component in $A'$ is  given by the usual maps $i,p,K$ in $(A,A')$-gauge and the 1-form component is given by the maps $\mathbb{I}i$, $p\mathbb{P}$, $\Lambda$, with $\mathbb{I},\mathbb{P},\Lambda$ as in (\ref{Lambda, II, PP}). In $\Omega^{i,j}(\mc{U})$, $(i,j)$ is the form bi-degree along the two factors in $\FC'\times \FC'$.
%, we mean an $i$-form in $A$ and $j$-form in $A'$.

We remark that sums in (\ref{ihat, phat, Khat, Phihat}) are finite (stop at $k=2$)
%for $\wh{i},\wh{p},\wh{K}$ and at $k=1$ for $\wh\Theta$), 
since each factor $\ad^*_{\delta A'}$ drops the form degree on $M$ by one. %\marginpar{Maybe at $k=2$ for $\widehat{K}$, since $K$ also decreases the form degree.}

The triple $(\wh{i},\wh{p},\wh{K})$ above can be thought of as a promotion of the $(i,p,K)$ triple (SDR data) associated to the $(A,A')$-gauge-fixing to a \emph{differential family} over $A'\in\FC'$, cf. Lemma \ref{lemma: ipK hat} below.

\begin{defn}
We define the extended desynchronized Chern-Simons perturbative partition function as
\begin{multline}\label{Zhat}
    \wh{Z}_{A,A'}(\sfa) 
    =e^{\frac{i}{\hbar}S_{CS}(A)} 
    %I_{A,A'}\;
    e^{\frac{\pi i}{4}\psi_A}\tau_{A}^{1/2}
e^{\frac{i}{\hbar}\,\frac12\langle \sfa, \wh\Theta(\sfa)\rangle}\;\exp \sum_\Gamma \frac{(-i\hbar)^{l(\Gamma)-1}}{|\mr{Aut}(\Gamma)|}\wh{\Phi}_{\Gamma,A,A'}(\sfa)\\
    \in \Omega^{0,\bt}(\mc{U},\mr{Dens}^{\frac12,\mr{formal}}(H_A[1]))
\end{multline}
where $\Gamma$ runs over connected graphs with 
%\marginpar{\bl I think tree graphs with $\wh{i},\wh{K}$ may be nonzero}
$l\geq 0$ loops and Feynman weights of graphs are defined as in (\ref{eq: def Phi_Gamma}), where we replace $K$ with $\wh{K}$ and $i$ with $\wh{i}$.
\end{defn}

% Also, consider the following partial superconnection in the direction of $A'$ on the bundle of formal half-densities on $H_A$ over $\mc{U}$:
% \begin{equation}\label{nabla_A' hat}
%     \wh\nabla_{A'}=\delta_{A'}-i\hbar\Delta_\sfa-\Xi_{\delta A'}.
%     %-\Big\{\frac12 \langle \sfa, \wh\Theta(\sfa) \rangle,-\Big\}
% \end{equation}
% Here $\Xi$ is as in (\ref{Xi}). 

%The superconnection (\ref{nabla_A' hat}) can be thought of as a correction of the trivial superconnection $\delta_{A'}$ by  $\Delta_\sfa$ -- a $0$-form in $A'$ and $\Xi$ -- a $1$-form. 
%and the last term concentrated in de Rham degrees $2$ and $3$ in $A'$.

\begin{thm}\label{thm: horizontality of Zhat}
    We have the horizontality relation
    \begin{equation}\label{dQME}
        \wh\nabla_{A'}\wh{Z}_{A,A'}(\sfa)=0,
    \end{equation}
    with $\wh\nabla_{A'}$ as in (\ref{nabla_A' hat}).
\end{thm}
An alternative name for equation (\ref{dQME}) is the \emph{differential quantum master equation}, cf. \cite{Bonechi2012}.
%\marginpar{add references}

The proof is based on the fact that the triple $(\wh{i},\wh{p},\wh{K})$ satisfies the the relations of %{\bl SDR data} 
an $(i,p,K)$ triple, 
with the differential $d_A$ replaced by the total differential $\delta_{A'}+d_A$. %see Lemma \ref{lemma: ipK hat}. 
Interestingly, in these relations cohomology $H_A$ as a family over $A'$ attains a nontrivial total differential\footnote{Or: flat (partial) superconnection.} $\delta_{A'}+\wh\Theta$, with $\wh\Theta$ of mixed degree along $\mc{U}$ and along $H_A$ but of total degree one. 
%\marginpar{move the lemma and its proof to an appendix?}

\begin{lem}\label{lemma: ipK hat}
The triple $(\wh{i},\wh{p},\wh{K})$ satisfies the following relations:
%\marginpar{notations: 1 or $\mr{id}$?}
\begin{subequations}
\begin{eqnarray}
    (\delta_{A'}+[d_A,-])\wh{K} &=&1-\wh{i}\,\wh{p},  \label{ipK hat a} \\
    (\delta_{A'}+d_A)\wh{i} &=& \wh{i}\, \wh\Theta, \label{ipK hat b}\\
    \delta_{A'} \wh{p} - \wh{p}\, d_A &=&  -\wh\Theta\, \wh{p}, \label{ipK hat c}\\
    \wh{K}\, \wh{i} &=& 0,\label{ipK hat d}\\
    \wh{p}\, \wh{K} &=& 0, \label{ipK hat e}\\
    \wh{K}^2 &=& 0, \label{ipK hat f}\\
    \wh{p} \,\wh{i} &=& 1. \label{ipk hat pi=1}
\end{eqnarray}
\end{subequations}
\begin{comment}
Additionally, for the differential in $A$ one has
\begin{subequations}
    \begin{eqnarray}
        \delta_A\, \wh{i} &=& \wh{K} \ad_{\delta A} \wh{i}, \label{ipK hat g} \\
        \delta_A\, \wh{p} &=& - \wh{p}\, \ad_{\delta A} \wh{K}, \label{ipK hat h}\\
        \delta_A\, \wh{K} &=& \wh{K} \ad_{\delta A} \wh{K}. \label{ipK hat i}
    \end{eqnarray}
\end{subequations}
\end{comment}
\end{lem}
%We give the proof in Appendix \ref{appendix: proof oh Lemma on ipKhat}. 
Formulae (\ref{ihat, phat, Khat, Phihat}) and Lemma \ref{lemma: ipK hat} are an application of homological perturbation lemma, see Appendix \ref{sss: extended ipK triples from families}.

\begin{rem}\label{rem: Thetahat squared}
    As an immediate consequence of Lemma \ref{lemma: ipK hat} we have that $\wh\Theta$ satisfies 
    \begin{equation}
        (\delta_{A'}+\wh\Theta)^2=0,
    \end{equation}
or, equivalently, $\wh\Theta$ satisfies the Maurer-Cartan equation 
\begin{equation}\label{Theta MC}
\delta_{A'}\wh\Theta+\wh\Theta^2=0.
\end{equation}
Indeed, one has
\begin{multline*} 
\wh\Theta^2\underset{(\ref{ipk hat pi=1})}{=}\wh{p}\wh{i}\wh\Theta\wh\Theta\underset{(\ref{ipK hat b})}{=}\wh{p} \big((\delta_{A'}+d_A)\wh{i}\big)\wh\Theta=
\wh{p} (\delta_{A'}+d_A)(\wh{i}\,\wh\Theta)-\underbrace{\wh{p}\,\wh{i}}_1\delta_{A'}\wh\Theta\\
\underset{(\ref{ipK hat b})}{=} 
\wh{p}\underbrace{(\delta_{A'}+d_A)^2}_0\wh{i}-\delta_{A'}\wh\Theta=-\delta_{A'}\wh\Theta.
\end{multline*}
\end{rem}

\begin{proof}[Proof of Theorem \ref{thm: horizontality of Zhat}]
The proof is similar to the proof of Proposition \ref{prop: variation of Z wrt A'}.

Consider the Feynman graph part of $\wh{Z}$, $\sum_\Gamma \wh{\Phi}_\Gamma$, with $\Gamma$ running over possibly disconnected graphs, with edges decorated by $\wh{K}$ and leaves decorated by $\wh{i}(\sfa)$. (We include powers of $\hbar$ and the symmetry factor in $\wh\Phi$.)  The action of $\delta_{A'}$ on $\wh\Phi_\Gamma$ can be computed as a sum of (a) decorations of one edge of $\Gamma$ with $(\delta_{A'}+[d_A,-])\wh{K}=1-\wh{i}\,\wh{p}$  (\ref{ipK hat a}), plus  (b) decorations of one leaf of $\Gamma$ with $(\delta_{A'}+d_A)\wh{i}=\wh{i}\wh\Theta$ (\ref{ipK hat b}).
%\marginpar{\bl footnote added Apr 23}
Upon summing over graphs, contributions of $1$ cancel out;\footnote{Cf. cancellation of principal boundary strata and footnote \ref{footnote: 1 as principal boundary stratum} in the proof of Proposition \ref{prop: variation of Z wrt A'}. } contributions of 
%\marginpar{\bl Oct13. sign in $\delta_{A'}\pm i\hbar\Delta_\sfa$?}
$\wh{i}\wh{p}$ yield $-i\hbar\Delta_\sfa \sum_\Gamma \wh\Phi_\Gamma$; contributions of $\wh{i}\wh\Theta$ add up to 
$$-\Big\{\frac12 \langle \sfa,\wh\Theta(\sfa) \rangle,\sum_\Gamma \wh\Phi_\Gamma\Big\}.$$
Thus, we have
$$ (\delta_{A'}-i\hbar \Delta_\sfa
 ) \sum_\Gamma \wh\Phi_\Gamma =-\Big\{\frac12 \langle \sfa,\wh\Theta(\sfa) \rangle,\sum_\Gamma \wh\Phi_\Gamma\Big\}. $$
Next, we have
\begin{multline*}
(\delta_{A'}-i\hbar\Delta_\sfa)\Big(e^{\frac{i}{\hbar}\,\frac12\langle \sfa,\wh\Theta(\sfa) \rangle} \sum_\Gamma \wh\Phi_\Gamma\Big)=\\
=e^{\frac{i}{\hbar}\,\frac12\langle \sfa,\wh\Theta(\sfa) \rangle}\frac{i}{\hbar}\Big(\frac12\langle \sfa,(\delta_{A'}\wh\Theta+\wh\Theta^2)\sfa\rangle+
\Big\{\frac12\langle \sfa,\wh\Theta(\sfa) \rangle,-\Big\}-\Big\{\frac12\langle \sfa,\wh\Theta(\sfa) \rangle,-\Big\}
\Big) \sum_\Gamma \wh\Phi_\Gamma =0,
\end{multline*}
where we used (\ref{Theta MC}).
%Together with result on the variation of $I_{A,A'}$ w.r.t. $A'$ (Lemma \ref{lemma: dependence of I on A'}), this yields the 
This proves the horizontality equation (\ref{dQME}).
\end{proof}

% \marginpar{\bl remove this remark}
% \begin{rem}\label{rem: absorbing Theta into Zhat}
%     In the horizontality (or differential QME) result (\ref{dQME}), the $\wh\Theta$ term can be eliminated from the superconnection $\wh\nabla_{A'}$ by a gauge transformation, leading to the equivalent equation
%     \begin{equation}\label{dQME for Zhat*Theta}
%         (\delta_{A'}+i\hbar\Delta_\sfa-\Xi)\left(e^{\frac{i}{\hbar}\,\frac12\langle \sfa, \wh\Theta(\sfa)\rangle}\; \wh{Z}_{A,A'}(\sfa)\right)  =0.
%     \end{equation}
% \end{rem}

\begin{comment}
Corollaries \ref{cor: variation of kin operator}, \ref{cor: horizontality wrt nabla_G} extend to the form-valued partition function $\wh{Z}$:
\begin{prop}\label{prop: horizontality of Zhat wrt nabla_G}
    The extended partition function is horizontal with respect to the partial Grothendieck connection (\ref{partial Grothendieck connection}) along harmonic shifts of $A$:
    \begin{equation}\label{Zhat nabla_G-horizontality}
        \til\nabla^G \wh{Z}_{A,A'}(\sfa)=0.
    \end{equation}
\end{prop}

\begin{proof}
    Using (\ref{ipK hat g}), (\ref{ipK hat i}), the proof 2 of Corollary \ref{cor: variation of kin operator} translates verbatim to the setting of $\wh{Z}$.
\end{proof}
\end{comment}

\subsubsection{A path integral formula for $\wh{Z}$}
%\begin{rem}
    The extended partition function (\ref{Zhat}) 
    %(or rather its modification by the $\Theta$-term arising in Remark \ref{rem: absorbing Theta into Zhat}) 
    can be seen as a perturbative expansion of the following path integral:
\begin{multline}\label{Zhat as path integral}
    %e^{\frac{i}{\hbar}\,\frac12\langle \sfa, \wh\Theta(\sfa)\rangle}\;
    \wh{Z}_{A,A'}(\sfa)=\int_{\mc{L}=\Omega_{d^*_{A'}\mr{-ex}}[1]}\mc{D}\alpha_\mr{fl}
    \; \exp \frac{i}{\hbar}\Big( 
    S_{CS}(A+i(\sfa)+\alpha_\mr{fl})+\\
    +\int_M\frac12 \underbrace{\big\langle 
    \alpha_\mr{fl}, d_A G\,\ad^*_{\delta A'}\alpha_\mr{fl}
    \big\rangle}_\tau+
    \underbrace{\big\langle \alpha_\mr{fl},d_A G\,\ad^*_{\delta A'} i(\sfa) \big\rangle}_\sigma
    \Big)
\end{multline}
where we suppress the indices in $i_{A,A'},G_{A,A'}$. The addition of the second and third terms in the exponential generates Feynman diagrams with edges and leaves decorated by $\wh{K}$ and $\wh{i}$ instead of $K$ and $i$. Additionally, one has a diagram consisting of two source terms $\sigma$ connected by %a chain of $\tau$-vertices 
an edge
-- this accounts for the exponential prefactor containing $\wh\Theta$ in 
%the l.h.s. of (\ref{Zhat as path integral}). 
(\ref{Zhat}).
See Figure \ref{fig:sigma-tau}.
%\marginpar{\bl actually we cannot have so many tau's in Figure, for degree reasons}
\begin{figure}
    \centering
    \includegraphics[width=0.7\linewidth]{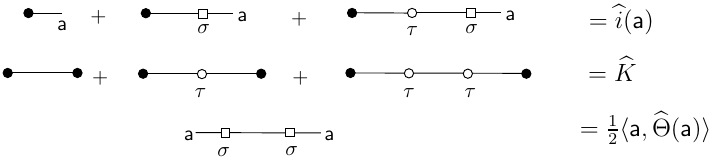}
    \caption{Feynman diagrams containing vertices $\tau$ and $\sigma$. Black dots are the usual internal vertices of Chern-Simons graphs. 
    %Sums are over the number of insertions of $\tau$.
    }
    \label{fig:sigma-tau}
\end{figure}

\subsection{Metric dependence of the desynchronized partition function}
Our definition of the desynchronized partition function $Z_{A,A'}$ (and its extended version) rely on a Riemannian metric $g$ on $M$. In this subsection, we analyze the dependence of $Z$ on this metric. Changing the metric induces another deformation of the gauge-fixing operator $d_{A'}^*$. The goal of this subsection is to prove Theorem \ref{thm: metric dependence} below. Most of the discussion is analogous to the previous subsection and we only sketch the proofs.
\begin{rem}[Framing anomaly and renormalization] It is well known that the (synchronized) perturbative Chern-Simons partition function exhibits metric dependence known as \emph{framing anomaly}. Namely, the phase $e^{\frac{i\pi}{4}\psi_A}$ of the synchronized 1-loop part 
$$ %I_{A,A} = 
I_A=e^{\frac{i\pi}{4}\psi_A}\tau^\frac12_A$$ 
depends on the metric through the eta invariant $\psi$, as already discussed in Witten's treatment \cite{Witten1989}.\footnote{Ray-Singer torsion, as an element of the determinant line $\mr{Det}\,H_A^\bt$, is invariant under changes of metric (this result does not require the flat connection to be acyclic).}
The dependence on the metric can be canceled by choosing a framing or 2-framing $\phi$ of $M$,i.e. a trivialization of $TM$ or $TM \oplus TM$\footnote{Working with 2-framings has the advantage that 3-manifolds have a canonical 2-framing \cite{Atiyah1990}. Choosing this canonical 2-framing both simplifies the 1-loop part and conjecturally agrees with asymptotics of WRT invariants, see \cite{Freed1991}.}, and multiplying $I_{A,A}$ by $e^{\frac{i\dim G}{24}\cdot \frac{S_\mr{grav}(g,\phi)}{2\pi}}$, where 
$S_\mr{grav}(g,\phi) = S_{CS}(\phi^*A_{g})$ denotes the evaluation of the Chern-Simons action on Levi-Civita connection $A_g$. For a 2-framing $\phi$, one defines $S_\mr{grav}(g,\phi) = \frac12S_{CS}(A_g \oplus A_g)$. Then, one has that 
\begin{equation}\label{eq: ren I}
    I_{A}^\mr{ren}:= e^{\frac{i\dim G}{24}\cdot \frac{S_\mr{grav}(g,\phi)}{2\pi}}I_{A}
\end{equation}
is invariant under variations of the metric. Axelrod and Singer \cite{Axelrod1991},\cite{Axelrod1994} showed that the anomaly persists at 
%\marginpar{\bl Sep4: only even loop orders?}
even loop orders: boundary strata of the compactified configuration spaces corresponding to the collapse of all vertices of a connected component of a Feynman graph yield potentially non-zero contributions. However, one can show that there exists a power series 
%\marginpar{\bl Sep 13: I corrected the two-loop term in accordance with AS91 (5.113) and Witten's nonperturbative prediction}
\begin{equation}
c(\hbar)  =  %\frac{\hbar\dim G}{24} + \frac{C_2(\mathfrak{g})\hbar^2}{48} + O(\hbar^4) 
\frac{\dim G}{24}\hbar  - \frac{h^\vee\dim G}{24\cdot (2\pi)}\hbar^2  + O(\hbar^4)
\in  \hbar\mathbb{R}[[\hbar]]\label{eq: c(hbar)}
\end{equation}
such that the renormalized perturbative partition function 
\begin{equation}\label{eq: ren Z}
    Z^\mr{ren}_{A}:=Z_A e^{\frac{i}{\hbar}c(\hbar)\frac{S_\mr{grav}(g,\phi)}{2\pi}}
\end{equation}
is independent of the metric $g$, up to an explicit BV exact term (also for $A$ non-smooth), see \cite{Cattaneo2008}, \cite{mnev2019quantum}, \cite{Wernli2019}.  Here $h^\vee$ is the dual Coxeter number of $\g$.  Comparison with the non-perturbative answer suggests that all higher order terms in \eqref{eq: c(hbar)} vanish, see e.g. the discussion below Eq. (6.124) in \cite{Axelrod1991}.
In our desynchronized setting, the same anomalies appear, and they can be canceled in the same way. 
%{\red \sout{To avoid cluttering the notation, from now on we understand $I_{A}$ and $Z_{A,A'}$ as renormalized according to \eqref{eq: ren I} and \eqref{eq: ren Z} respectively.}} 
%\marginpar{Is this OK? Can expand and change to subsubsection if necessary, or change notation. }
\end{rem}
Let us now consider the effect of an infinitesimal deformation $g \mapsto g + \delta g$. We have 
\begin{equation}
    \delta_g d_{A'}^* = [d_{A'}^*,\lambda_{\delta g}] 
\end{equation}
%\marginpar{New notation: $\lambda \to \lambda_{\delta g}, \mathbb{I}_g \to \mathbb{I}_{\delta g},\ldots $}
where $\lambda_{\delta g} \in \Omega^1(\mr{Met}, \operatorname{End}(\Omega^\bullet(M))$ is given by (see Lemma \ref{lem: lambda formula}) 
\begin{equation}
    \lambda_{\delta g} = \star^{-1} \delta_g\star = \frac12 \tr g^{-1}\delta g - \iota_{g^{-1}\delta_g}. 
\end{equation}
We note that we have 
\begin{equation}
    \delta_g\lambda_{\delta g} = - \star^{-1}(\delta_g\star)\star^{-1}\delta_g = -\lambda_{\delta_g}^2. 
\end{equation}
The $(i,p,K)$ triple transforms as 
\begin{align}
    \delta_g i_{A,A'} &= -d_A\mathbb{I}_{\delta g}i_{A,A'}, \\ 
    \delta_g p_{A,A'} &= -p_{A,A'}\mathbb{P}_{\delta g} d_A, \\
    \delta_g K_{A,A'} &= [d_A,\Lambda_{\delta g}] + P_{A,A'} \mathbb{P}_{\delta g} + \mathbb{I}_{\delta g}P_{A,A'}
\end{align}
where 
\begin{equation}
    \Lambda_{\delta g} = K_{A,A'}\lambda_{\delta g} K_{A,A'},\qquad \mathbb{I}_{\delta g} = K_{A,A'}\lambda_{\delta g}, \qquad \mathbb{P}_{\delta g} = \lambda_{\delta g}  K_{A,A'}
\end{equation} are the endomorphisms of $\Omega^\bullet(M,\g)$ analogous to $\Lambda,\mathbb{I},\mathbb{P}$ defined in \eqref{Lambda, II, PP}.

\begin{comment}
\marginpar{\bl remove}
{\color{gray}
We have the following result analogous to Lemma \ref{lemma: dependence of I on A'}: 
\begin{lem}\label{lem: metric dependence of I}
    The (renormalized) 1-loop part $I_{A,A'}$ satisfies 
    \begin{equation}
        \delta_gI_{A,A'} = \Xi_g I_{A,A'} 
    \end{equation}
    where 
    \begin{equation}\label{eq: metric desy anomaly}
    \Xi_g = \frac12\mr{Str}_{\Omega^\bt} P_{A'}(\mathbb{P}_g \ad_\beta-\ad_\beta \mathbb{I}_g - \ad_\beta \Lambda_g \ad_\beta). 
\end{equation}
\end{lem}
\begin{proof}[Sketch of the proof]
    The proof is analogous to the one of \ref{lemma: dependence of I on A'} given in \ref{appendix: proof of Lemma: dependence of I on A'}, with the difference that $\delta_g \tau = \delta_g\beta = 0$. The nonzero $\delta_g e^{\frac{i\pi}{4}\psi_{A'}}$ is canceled by the framing renormalization. For the remaining term 
    $$\mr{Str}_{\Omega^\bt}(1+K_{A'}\mr{ad}_\beta)^{-1}\delta_gK_{A'}\mr{ad}_\beta $$ 
    the computation in \ref{appendix: proof of Lemma: dependence of I on A'} goes through \emph{mutatis mutandis}. 
\end{proof}
}
\end{comment}

We then have the following theorem:
\begin{thm}\label{thm: metric dependence}
For $A,A' \in \FC$ a pair of close flat connections, we have 
%\marginpar{\bl notation: $r_g$ or $r_{\delta g}$?}
\begin{equation}\label{eq: metric dependence}
    \delta_g Z_{A,A'}^\mr{ren}(\sfa) = -i\hbar\Delta_\sfa \left(r_{\delta g}(\sfa)Z^\mr{ren}_{A,A'}(\sfa)\right)
\end{equation}
where $r_{\delta g}(\sfa)$ is given by the sum of connected Feynman diagrams with one edge marked by $\Lambda_{\delta g}$ or one leaf marked by $\mathbb{I}_{\delta g}$. 
%and $\Xi_g$ given by \eqref{eq: metric desy anomaly}.
\end{thm}
\begin{proof}[Sketch of the proof]
    This proof is analogous to the proof of Proposition \ref{prop: variation of Z wrt A'}, using %Lemma \ref{lem: metric dependence of I} and 
    the fact that the renormalization cancels potential contributions from hidden boundary strata.\footnote{
    By Corollary \ref{cor: eta_A,A' boundary behavior}, the restriction of the propagator to the boundary of configuration space in the desynchronized theory agrees modulo regular terms to the propagator in the case $A=A'$, so the analysis of the hidden boundary strata of \cite{Axelrod1994} carries over without change.
    } %\marginpar{\bl Apr 21 added footnote}
\end{proof}

\subsubsection{Extension of $Z$ to a horizontal nonhomogeneous form in $g$.}
\label{sss Zhat g}
One can perform a construction analogous to the one in Section \ref{sec: ext Z A prime} and extend the desynchronized partition function to a horizontal non-homogeneous differential form in the metric direction. This generalizes the construction of Axelrod and Singer \cite{Axelrod1994} to the desynchronized case and arbitrary kinetic operators.\footnote{Axelrod and Singer always work under the assumption that $d_A$ is acyclic.} Namely, one has the following analog of \eqref{ihat, phat, Khat, Phihat}: 
%\marginpar{\bl signs? $\wh\Theta$?}
\begin{equation}\label{ihatg, phagt, Khatg, Phihatg}
\begin{aligned}
    \wh{i}_{\delta g}&=\sum_{k\geq 0} (K_{A,A'}\lambda_{\delta g})^k i_{A,A'}=i+\mathbb{I}_{\delta g}i+\cdots \quad \in \Omega^{0,0,\bt}(\mc{U}\times \mr{Met},\mr{Hom}(H_A^\bt,\Omega^\bt(M,\g))), \\
    \wh{p}_{\delta g}&= \sum_{k\geq 0} p_{A,A'}(\lambda_{\delta g} K_{A,A'})^k=p+p\mathbb{P}_{\delta g}+\cdots\quad \in \Omega^{0,0,\bt}(\mc{U}\times \mr{Met},\mr{Hom}(\Omega^\bt(M,\g),H_A^\bt)), \\
     \wh{K}_{\delta g}&= \sum_{k\geq 0} K_{A,A'}(\lambda_{\delta g} K_{A,A'})^k=K+\Lambda_{\delta g}+\cdots\quad \in \Omega^{0,0,\bt}(\mc{U}\times \mr{Met},\mr{End}(\Omega^\bt(M,\g))),\\
     \wh\Theta_{\delta g}&=-p_{A,A'}\lambda_{\delta g} K_{A,A'}\lambda_{\delta g} i_{A,A'}\quad  \in \Omega^{0,0,2}(\mc{U},\mr{End}(H^\bt_A)).
\end{aligned}
\end{equation}
Here $\mr{Met}$ denotes the Fr\'echet manifold of metrics and we are using its exterior calculus. Alternatively, one can restrict to an arbitrary finite-dimensional submanifold thereof.
%\marginpar{\red Added April 21} 
Again, we note that since $K$ reduces the form degree along $M$, the sums in \eqref{ihatg, phagt, Khatg, Phihatg} are finite: The sum stops at $k=2$ for $\wh{i}_{\delta g}$, $\wh{p}_{\delta g}$ and  $\wh{K}_{\lambda_{\delta g}}$. %and $\Theta_{\delta g}$. 
\begin{rem} 
Another way to construct the triple $(\wh{i}_{\delta g},\wh{p}_{\delta g},\wh{K}_{\delta g})$ is to consider the operator\footnote{This is the construction used by Axelrod and Singer \cite{Axelrod1994}.} 
\begin{equation}
    \widehat{H}_g = [d_A + \delta_g,d^*_{A',g}] = \wh{H}_g + (\delta_g d_{A',g}^*) = \wh{H}_g + [\lambda_{\delta g}, d^*_{A',g}].
\end{equation}
The operator $H_g + P_{A,A',g} + [\lambda_{\delta g}, d^*_{A',g}]$ is invertible and its Green's function can be computed as
$\wh{G} = G - G[\lambda_{\delta g}, d^*_{A',g}]G + \ldots$. 
Upon applying $d^*_{A',g}$, one recovers \eqref{eq: deformed induction}. In particular, $\wh{K}_{\delta g}$ coincides, for $A = A'$ an acyclic flat connection, with the extended propagator of Axelrod and Singer \cite[Section 4]{Axelrod1994}. 
\end{rem}
Similarly to Lemma \ref{lemma: ipK hat} we have: % which we spell out for completeness: 
\begin{lem}\label{lem: ipk hat met}
    The triple $(\wh{i}_{\delta g},\wh{p}_{\delta g},\wh{K}_{\delta g})$ satisfies the following relations:
\begin{subequations}
\begin{eqnarray}
    (\delta_{g}+[d_A,-])\wh{K}_{\delta g} &=&1-\wh{i}_{\delta g}\,\wh{p}_{\delta g},  \label{ipK hat a met} \\
    (\delta_{g}+d_A)\wh{i}_{\delta g} &=& \wh{i}_{\delta g}\, \wh\Theta_{\delta g}, \label{ipK hat b met}\\
    \delta_{g} \wh{p}_{\delta g} - \wh{p}_{\delta g} \, d_A &=&  -\wh\Theta_{\delta g}\, \wh{p}_{\delta g}, \label{ipK hat c met}\\
    \wh{K}_{\delta g}\, \wh{i}_{\delta g} &=& 0,\label{ipK hat d met}\\
    \wh{p}_{\delta g}\, \wh{K}_{\delta g} &=& 0, \label{ipK hat e met}\\
    \wh{K}_{\delta g}^2 &=& 0, \label{ipK hat f met}\\
    \wh{p}_{\delta g} \,\wh{i}_{\delta g} &=& 1. \label{ipk hat pi=1 met}
\end{eqnarray}
\end{subequations}
\end{lem}
%We sketch the proof in Appendix \ref{app: proof lem ipKhat met}.  
This is a special case of Proposition \ref{prop: extended ipK triple, App C} for $\mathbb{GF}=\mr{Met}$.
The extended partition function 
\begin{multline}\label{Zhat g}
    \wh{Z}_{A,A';\delta g}(\sfa) 
    =e^{\frac{i}{\hbar}S_{CS}(A)} 
    %I_{A,A'}\;
    e^{\frac{\pi i}{4}\psi_A}\tau_{A}^{1/2}
e^{\frac{i}{\hbar}\,\frac12\langle \sfa, \wh\Theta(\sfa)\rangle}\;\exp \sum_\Gamma \frac{(-i\hbar)^{l(\Gamma)-1}}{|\mr{Aut}(\Gamma)|}\wh{\Phi}_{\Gamma,A,A';\delta g}(\sfa)\\ 
    \in \Omega^{0,0,\bt}(\mc{U}\times \mr{Met},\mr{Dens}^{\frac12,\mr{formal}}(H_A [1]))
\end{multline}
then satisfies the differential Master Equation 
\begin{equation}\label{dQME on Zhat g}
    (\delta_g - i\hbar\Delta_\sfa)\wh{Z}_{A,A';\delta g}^\mr{ren} = 0,
\end{equation}
which one can prove analogously to Theorem \ref{thm: horizontality of Zhat}; the superscript ``ren'' means that we include the renormalization factor as in (\ref{eq: ren Z}). Again, one can view $\wh{Z}_{A,A';\delta g}$ as the perturbative expansion of the path integral
\begin{multline}\label{Zhat g as path integral}
    %e^{\frac{i}{\hbar}\,\frac12\langle \sfa, \wh\Theta(\sfa)\rangle}\;
    \wh{Z}_{A,A';\delta g}(\sfa)=\int_{\mc{L}=\Omega_{d^*_{A'}\mr{-ex}}[1]}\mc{D}\alpha_\mr{fl}
    \; \exp \frac{i}{\hbar}\Big( 
    S_{CS}(A+i(\sfa)+\alpha_\mr{fl})
    - \\ -\int_M\frac12 \big\langle 
    \alpha_\mr{fl}, \lambda_{\delta g}\alpha_\mr{fl}
    \big\rangle+\big\langle \alpha_\mr{fl},\lambda_{\delta g} i(\sfa) \big\rangle
    \Big)
\end{multline}
with the last two terms generating additional vertices that sum up to $K_{\delta g}, i_{\delta g}$ and $\Theta_{\delta g}$ respectively.

\subsection{Partition function extended to a nonhomogeneous form %in $A,A',g$
on $\FC\times \FC\times \mr{Met}$. %A BV integral construction.
}
\label{ss: full extension of Z}
%\marginpar{Maybe we can drop the last part of the title? {\bl done!}}
\subsubsection{Connection $\nabla^\mr{Hodge}$}
In this section, $\mr{Met}$ will denote the space of Riemannian metrics on $M$ and $\UU=\{(A,A',g)\}$ will stand for a sufficiently 
% small\footnote{
% ``Sufficiently small'' means here that (a) for each fixed metric $g$ all pairs $(A,A')$ are close w.r.t. $g$ and (b) subbundles of $\underline\Omega$ (see below) given by harmonic, exact and coexact forms are trivial over $\UU$.
% We will also assume (for the construction of Section \ref{sss: construction of nabla^Hodge by flattening H^0}) that $\UU$ is connected.
% } %open neighborhood of a point $(A_0,A_0,g_0)$ in $\FC'\times\FC'\times \mr{Met}$.
thin open neighborhood of $\mr{Diag}\times\mr{Met}$ in $\FC'\times\FC'\times \mr{Met}$, where each fixed-$g$ slice consists of pairs $(A,A')$ that are close w.r.t. $g$.

Let
\begin{equation}\label{H}
H=H_{\delta A}+H_{\delta A'}+H_{\delta g}\quad \in \Omega^1(\UU,\mr{End}(\Omega^\bt(M,\g))),
% \begin{aligned}
%     H =&
%     -\Big(K\ad_{\delta A} dK+K\ad_{\delta A} P + P\ad_{\delta A} K\Big)\\
%     &-
%     \Big(dG\ad^*_{\delta A'}Kd + dG\ad^*_{\delta A'}P+P\ad^*_{\delta A'}Gd \Big)\\
%     &+
%     \Big(dK\lambda_{\delta g} Kd+  dK \lambda_{\delta g} P+ P\lambda_{\delta g} Kd\Big) \\
%     &
%     \hspace{2cm} \in \Omega^1(\UU,\mr{End}(\Omega^\bt(M,\g))).
% \end{aligned}
\end{equation}
where
\begin{subequations}
\begin{eqnarray}
    H_{\delta A} &=&  -\Big(K\ad_{\delta A} dK+K\ad_{\delta A} P + P\ad_{\delta A} K\Big), 
    \label{H delta A}
    \\
    H_{\delta A'} &=& -\Big(dG\ad^*_{\delta A'}Kd + dG\ad^*_{\delta A'}P+P\ad^*_{\delta A'}Gd \Big), \label{H delta A'} \\
    H_{\delta g} &=& dK\lambda_{\delta g} Kd+  dK \lambda_{\delta g} P+ P\lambda_{\delta g} Kd. \label{H delta g}
\end{eqnarray}
\end{subequations}
We are suppressing the subscripts in $d_A, d^*_{A'}, P_{A,A'}, K_{A,A'}$. Furthermore, let
\begin{equation}
    \Psi=\int_M \frac12\langle B, H(B) \rangle\quad \in \Omega^1(\UU)\otimes \mr{Sym}^2(\Omega^\bt(M,\g)[1])^*.
\end{equation}
Here $B\in \Omega^\bt(M,\g)[1]$ is the field. Also, we consider the connection
\begin{equation}
    \nabla^\mr{Hodge}=
    \delta^\mr{tot}+ H = \delta^\mr{tot} + \{\Psi,-\}_B
\end{equation}
on the trivial bundle $\ul\Omega$ over $\UU$ with fiber $\Omega^\bt(M,\g)[1]$. Here $\{,\}_B$ stands for the Poisson bracket in the fiber and $\delta^\mr{tot}=\delta_A+\delta_{A'}+\delta_g$ is the de Rham differential on $\UU$.

We have the following.
\begin{lem}\label{lemma: nabla^Hodge}
    \begin{enumerate}[(a)]
        \item \label{lemma: nabla^Hodge (a)} 
        The connection $\nabla^\mr{Hodge}$ preserves the harmonic, exact and coexact subbundles in $\ul\Omega$.\footnote{
        Put another way:
        the parallel transport of the connection $\nabla^\mr{Hodge}$ along a path $(A_t,A'_t,g_t)$, $t\in [0,1]$ maps the desynchronized $(A_0,A'_0)$-Hodge decomposition 
        $\Omega=\mr{Harm}_{A_0,A'_0}\oplus \mr{im}(d_{A_0})\oplus \mr{im}(d^*_{A'_0})$ (with metric $g_0$) to the desynchronized $(A_1,A'_1)$-Hodge decomposition  $\Omega=\mr{Harm}_{A_1,A'_1}\oplus \mr{im}(d_{A_1})\oplus \mr{im}(d^*_{A'_1})$ (with metric $g_1$) term-to-term.}
        \item\label{lemma: nabla^Hodge (b)}  $\nabla$ is symplectic, i.e., its parallel transport along any path in the base is a symplectomorphism between fibers w.r.t. the BV symplectic form $\omega=\int_M \langle - \stackrel{\wedge}{,} -\rangle$.
        \item \label{lemma: nabla^Hodge (c)} The curvature of $\nabla^\mr{Hodge}$ is
        \begin{multline}\label{F_nabla^Hodge}
            F_{\nabla^\mr{Hodge}}=P\ad_{\delta A}(K\lambda_{\delta g}-G\ad^*_{\delta A'})P+
        P(\lambda_{\delta g} K-\ad^*_{\delta A'}G)\ad_{\delta A} P+\\
        +d(K\lambda_{\delta g}-G\ad^*_{\delta A'})(Kd+P)\ad_{\delta A}K+K\ad_{\delta A}(dK+P)(\lambda_{\delta g} K-\ad^*_{\delta A'}G)d.
        \end{multline}
        In particular, the curvature has vanishing $(\delta A)^2$, $(\delta A')^2$, $(\delta g)^2$ and $\delta A' \delta g$ terms; only $\delta A \delta A'$ and $\delta A \delta g$ terms are nonzero.
        \item \label{lemma: nabla^Hodge (d)} For a fixed metric $g$ and restricted to harmonic forms, $\nabla^\mr{Hodge}$ coincides with $\nabla^\mr{Harm}$ of Section \ref{sss nabla^Harm}.
    \end{enumerate}
\end{lem}
(Proven by an explicit computation.)
% \begin{proof}
%     Explicit computation. 
%     ((\ref{lemma: nabla^Hodge (c)}) is trivial, (\ref{lemma: nabla^Hodge (a)}) is a straightforward check, (\ref{lemma: nabla^Hodge (b)}) is quite a long computation.) \marginpar{For (b) I only have partial checks}
% \end{proof}
%[TO HERE]
%}
%\end{comment}

%\marginpar{\bl added Aug 24}
\begin{rem}\footnote{We thank S. Stolz for 
%suggesting this mechanism as a possible origin of 
%$\nabla^\mr{Hodge}$.
this remark. 
}
    $\nabla^\mr{Hodge}$ can be seen as a sum of three ``shift-and-project connections'' (cf. Remark \ref{rem: nabla^Harm as shift-and-project connection}) induced on the harmonic, exact and coexact subbundles in $\ul\Omega$ from the trivial connection on $\ul\Omega$. Put another way, one can write
    \begin{equation}
        \nabla^\mr{Hodge}=\delta^\mr{tot}-\delta^\mr{tot} (P)\, P-
        \delta^\mr{tot}(dK)\, dK - \delta^\mr{tot}(Kd)\, Kd.
    \end{equation}
    Here $P,dK,Kd$ are the fiberwise projections onto the three terms in the Hodge decomposition.
\end{rem}

By restricting $\nabla^\mr{Hodge}$ to harmonic forms and projecting to cohomology $H_A$, $\nabla^\mr{Hodge}$ induces the connection $\nabla^\mathbb{H}$ on the cohomology bundle bundle  $\HH$ over $\UU$ with fiber $H_A$.\footnote{
%We are abusing notations slightly, using the same notations as in Section 
These objects are a natural extension of the corresponding objects of Section
\ref{sss cohomology comparison map} by allowing variation of metric. By an abuse of notations, we use $\mathbb{H}$, $\nabla^\mathbb{H}$ for the extension.} 
Notice that the cohomology bundle is trivial along $\mr{GF} = \FC_{A'} \times \mr{Met}$  directions and the connection is also trivial in these directions, i.e. 
\begin{equation}
 \nabla^\HH = \nabla^\HH_{A} + \delta_{\mr{GF}}.
\end{equation} %\marginpar{Added Oct 2}
The curvature of $\nabla^\mathbb{H}$ corresponds to the harmonic-harmonic block of (\ref{F_nabla^Hodge}):
\begin{equation}\label{F_nabla^H}
    F_{\nabla^\mathbb{H}}=p\Big((\ad_{\delta A}(K\lambda_{\delta g}-G\ad^*_{\delta A'}))+
    (\lambda_{\delta g} K-\ad^*_{\delta A'}G)\ad_{\delta A}\Big)i.
\end{equation}
%We will denote $\nabla^\mc{D}$ the flat connection induced by $\nabla^\HH$ on the bundle $\mc{D}$ of formal half-densities on $H_A$ over $\UU$.
Furthermore, we will denote by $\nabla^\mc{D}$ the connection induced by $\nabla^\HH$ on the bundle $\mc{D}$ of formal half-densities on $H_A{ [1]}$ over $\UU$.

\subsubsection{Extended partition function}
Denote\footnote{This term is completely analogous to the term in the extended action denoted $S_R$ in \cite{Bonechi2012} and is the Hamiltonian (in an appropriate sense) for  
the Grothendieck connection, hence the superscript $G$.} 
$$ \Psi^G= \int_M\langle \delta A,B \rangle. $$
Let 
\begin{equation}
    \check{S}(B)=S_{CS}(A+B)-\Psi^G - \Psi
     \quad \in \Omega^{\bt}(\underline{\mc{U}})\otimes \mr{Sym}(\Omega^\bt(M,\g)[1])^*.
\end{equation}
-- a form on $\underline{\mc{U}}$ valued in polynomials in $B$. We split the field as $B=i_{A,A'}(\sfa)+\alpha_\mr{fl}$ and consider the perturbative path integral
\begin{equation}\label{Z check}
    \check{Z}(\sfa)=\int_{\mr{im}(d^*_{A'})} \mc{D}\alpha_\mr{fl} \; e^{\frac{i}{\hbar} \check{S}(i_{A,A'}(\sfa)+\alpha_\mr{fl})}\quad \in \Omega^\bt(\underline{\mc{U}}, \mr{Dens}^{\frac12,\mr{formal}}(H_A[1]) ).
\end{equation}

Perturbative evaluation of (\ref{Z check}) yields the following:
\begin{multline}\label{Z check perturbative}
    \check{Z}(\sfa)=e^{\frac{i}{\hbar}S_{CS}(A)} e^{\frac{\pi i}{4}\psi_A}\tau_A^{1/2}
    e^{\frac{i}{\hbar} (-\langle [\delta A],\sfa \rangle
    %+\frac12 \langle \sfa, pHi(\sfa) \rangle 
    +\frac12\langle \sfa, \check\Theta(\sfa)\rangle)}\cdot \\
    \cdot \exp \sum_\Gamma \frac{(-i\hbar)^{l(\Gamma)-1}}{|\mr{Aut}(\Gamma)|}\check\Phi_\Gamma.
\end{multline}
Here:
\begin{itemize}
    \item $\Gamma$ runs over connected trivalent graphs with leaves, as usual.
    %\item $H$ is the 1-form (\ref{H}) of the connection $\nabla^\mr{Hodge}$.
    \item $\check\Phi_\Gamma$ is the Feynman weight of the graph $\Gamma$, where an edge is assigned the extended propagator
    \begin{equation}\label{eq: extended K Z check}
        \check{K}=\sum_{k=0}^2 K(HK)^k %{\red = \sum_{k=0}^2 K(HK)^k}
    \end{equation}
    -- a nonhomogeneous form on $\UU$ valued in $\mr{End}(\Omega^\bt(M,\g))$;
    here $H$ is the 1-form (\ref{H}) of the connection $\nabla^\mr{Hodge}$.
    A leaf is assigned the expression
    %\marginpar{\bl not sure about the relative sign in $i(\sfa)\pm K\delta A$}
    \begin{equation}\label{eq: extended leaf Z check}
        \check{i}(\sfa)=\sum_{k= 0}^2 (KH)^k i(\sfa)+K\delta A. 
        %{\red = \sum_{k= 0}^2(KH)^k(i(\sfa)-K\delta A)}
    \end{equation}
    Note that $\check{i}(\sfa)$ is affine-linear in $\sfa$, rather than just linear. Both \eqref{eq: extended K Z check} and \eqref{eq: extended leaf Z check} stop at $k=2$ because $K$ decreases form degree by 1. 
    \item $\check\Theta(\sfa)$ stands for
    \begin{equation}
        \check\Theta(\sfa)=
       -p H K H i(\sfa).
        % -\int_M \frac12 \langle Hi(\sfa)-\delta A, 
        % %\sum_{k\geq 0} (KH)^k 
        % K (Hi(\sfa)-\delta A) \rangle.
    \end{equation}
    It is a 2-form on $\UU$ with values in endomorphisms of $H_A[1]$. %\marginpar{\red updated Oct 7}
\end{itemize}
%{\red \marginpar{added remark Sep 22}
\begin{rem}\label{rem: feynman rules check Z }
    To elucidate the relationship between $\check{Z}$ and $Z$, one can express \eqref{Z check perturbative} in terms of the regular Feynman rules, where an edge is assigned $K$ and a leaf is assigned $i(\sfa)$, by adding extra vertices carrying the form degree along $\UU$: bivalent black and white vertices and gray univalent vertices. White vertices are assigned $H_{\delta A'}$, black vertices are assigned $H_{\delta g}$, gray vertices are assigned $\delta A$. In addition, edges decorated by more than two bivalent vertices vanish automatically. See Figure \ref{fig: check Z diagram A}.  
    When sandwiched between $K$ and $K$ (or $K$ and $i$) only a single term in $H_{\delta A'}$ and $H_{\delta g}$ survives. Therefore, one can evaluate black vertices with $\lambda_{\delta g}$ and white vertices to $dG\mr{ad}^*_{\delta A'}$. The latter effectively acts by applying ``$(d^*)^{-1}$'' on one of the internal edges incident to the vertex. See for instance the diagram in Figure \ref{fig: check Z diagram B}.
    % \begin{equation}
    %     \int_{M^4} i(\sfa)^a_1(\mr{ad}^*_{\delta A'_1})^b_a(G_{12})^c_bf^e_{cd}(K_{23})^d_g(\lambda_{\delta g_3})^g_h(K_{34})^h_i(\mr{ad}^*_{\delta A'_4})^i_l(G_{42})^l_e 
    % \end{equation}
    % Here subscripts $i,j \in {1,2,3,4}$ denote evaluation at $x_i,x_j$, indices $a,b,c,\ldots$ denote expansion in a basis of $\g$. 
\end{rem}
%\marginpar{(a),(b) is impossible for degree reasons, find a better one (add more leaves)}
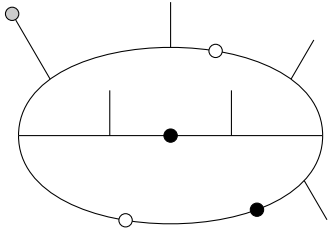
\begin{figure}
\centering

   \begin{tikzpicture}[scale = 2]
       \node[coordinate] (v1) at (0,0) {}; 
       \node[coordinate] (v2) at (2,0) {}
        edge[bend left=90] 
        node[coordinate, pos=.15] (d) {}
        node[black, pos=.3] {} 
       node[white, pos=.6] {}
       (v1) 
       edge[bend right=90] node[coordinate,pos=.2] (a) {}
       node[white,pos=.4] {}
       node[coordinate, pos=.5] (b) {} 
       node[coordinate, pos=.8] (c) {}
       (v1) 
       edge
       node[coordinate, pos=.3] (g) {}
       node[black] {} 
       node[coordinate,pos=.7] (h) {}
       (v1);
     \draw (a) -- ($(a) + (60:.3)$);
     \node[coordinate] (b1) at ($(b) + (90:.3)$) {} 
     edge (b);
     \node[coordinate]  at ($(g) + (90:.3)$) {} 
     edge (g);
     \node[coordinate]  at ($(h) + (90:.3)$) {} 
     edge (h);
     \node[coordinate]  at ($(d) + (-60:.3)$) {} 
     edge (d);
     \node[grey] (c1) at ($(c) + (120:.5)$) {} 
     edge (c);
   \end{tikzpicture}
   \caption{This graph evaluates to an element of $\Omega^{1,2,2}(\UU,\mr{Dens}^{\frac12,\mr{formal}}(H_A[1])$ (notice gray/white/black vertices carry form degree 1 along $A/A'/g$). The ghost number of this graph is $-5$. }
    \label{fig: check Z diagram A}
\end{figure}
   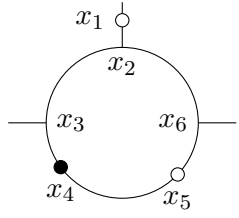
\begin{figure}
      \centering
      \begin{tikzpicture}
          \draw (1,0) arc (0:360:1) 
          node[pos=0,coordinate, label=left:{$x_6$}] (b) {}
          node[pos=.25,coordinate,label=below:{$x_2$}] (a) {}
            node[pos=.5,coordinate, label=right:{$x_3$}] (c) {}
          node[pos=.6,black,label=below:{$x_4$}] {} node[pos=.88,white,label=below:{$x_5$}] {}; 
          \draw (a) to node[pos=.6, white,label=left:{$x_1$}] {} ($(a) + (0,.6)$);
          \draw (b) to ($(b) + (0:.5)$);
          \draw (c) to ($(c) + (180:.5)$);
      \end{tikzpicture}
      \captionsetup{singlelinecheck=off}
          \caption[.]{This diagram evaluates to 
          $$
          \operatorname{Str} \mr{ad}_{G\mr{ad}^*_{\delta A'}i(\sfa)}K \mr{ad}_{i(\sfa)}K\lambda_{\delta_g}G\mr{ad}^*_{\delta A}K \mr{ad}_{i(\sfa)}K \in \Omega^{0,2,1}(\UU,\mr{Dens}^{\frac12,\mr{formal}}(H_A[1]))
          $$
          (it can be written as a supertrace because it is a 1-loop graph).}         \label{fig: check Z diagram B}
    \end{figure}
\subsubsection{Differential quantum master equation}
\begin{thm} \label{thm dQME on Zcheck}
%In the trivialization of the bundle of formal half-densities induced by $\nabla^\mathbb{H}$,
%    $\check{Z}$ satisfies
The following differential quantum master equation holds:
    \begin{equation}\label{dQME on Zcheck}
        \Big(\nabla^\mc{D}
        -i\hbar \Delta_\sfa -\frac{i}{\hbar} \frac12 \langle\sfa, F_{\nabla^\HH}\sfa \rangle
        \Big) (e^{\frac{i}{\hbar} c(\hbar)\frac{S_\mr{grav}(g,\phi)}{2\pi}} \check{Z}) =0,
    \end{equation}
    with $F_{\nabla^\HH}$ as in (\ref{F_nabla^H}).
    % Here $\delta^\mr{tot}=\delta_A+\delta_{A'}+\delta_g$ is the de Rham differential on $\underline{\mc{U}}$.
\end{thm}

%Let us work in the trivialization of the bundle of formal half-densities on $H_A$ induced by $\nabla^\mc{D}$; in this trivialization one has $\nabla^\mc{D}=\delta^\mr{tot}$.
\textbf{Heuristic Path Integral Argument:}
Denote $\mc{L}=\mr{im}(d^*_{A'})$ the gauge fixing Lagrangian. We have
\begin{multline}\label{dQME on Zcheck proof eq1}
    \nabla^\mc{D} \int_{\mc{L}} e^{\frac{i}{\hbar} \check{S}(B)}=\\=
    \int_{\mc{L}} e^{\frac{i}{\hbar} \check{S}(B)} \frac{i}{\hbar}(\delta_A S_{CS}(A+B)-\delta^\mr{tot}\Psi+\{\check{S},\Psi\}_B-i\hbar \Delta_B\Psi).
\end{multline}
Here in the brackets in the r.h.s., the first two terms account for $\delta^\mr{tot}$ acting on the integrand and the last two terms account for the change of gauge-fixing induced by an infinitesimal change of $A,A',g$, cf. \cite[Proposition 2]{Cattaneo2008}.
%\footnote{We think of $\nabla^\HH$ in the l.h.s. of (\ref{dQME on Zcheck proof eq1}) as as an infinitesimal shift along $\UU$ (accommodated by the first two terms in the r.h.s.) accompanied by an infinitesimal symplectic transformation of zero-modes, }  
Also, note that we can write the first term in the r.h.s. of (\ref{dQME on Zcheck proof eq1}) as 
$$\delta_A S_{CS}(A+B)=\{S_{CS}(A+B), \Psi^G \}_B.$$ 
Next, applying the BV Laplacian in zero-modes to $\check{Z}$, we have, using the BV-Stokes' theorem,
\begin{multline}\label{dQME on Zcheck proof eq2}
    -i\hbar\Delta_\sfa \int_{\mc{L}} e^{\frac{i}{\hbar} \check{S}(B)}=
    \int_{\mc{L}} -i\hbar\Delta_B e^{\frac{i}{\hbar} \check{S}(B)}
    =\\=
     \int_{\mc{L}}e^{\frac{i}{\hbar} \check{S}(B)} \frac{i}{\hbar}(\frac12\{\check{S},\check{S}\}_B-i\hbar \Delta_B \check{S})
     \\
     =\int_{\mc{L}}e^{\frac{i}{\hbar} \check{S}(B)} \frac{i}{\hbar}(
     -\{S_{CS},\Psi^G\}_B-\{\check{S},\Psi\}_B-\frac12 \{\Psi,\Psi\}_B+i\hbar \Delta_B \Psi
     )\\
     =\int_{\mc{L}}e^{\frac{i}{\hbar} \check{S}(B)} \frac{i}{\hbar}(
     -\delta_AS_{CS}-\{\check{S},\Psi\}_B+\delta^\mr{tot}\Psi-\Psi^F+i\hbar \Delta_B \Psi
     )
     .
\end{multline}
In the last transition we used that
%flatness of $\nabla$, which is tantamount to 
$\delta^\mr{tot}\Psi+\frac12\{\Psi,\Psi\}_B=\Psi^F$, with $\Psi^F=\frac12 \langle B,F_{\nabla^\mr{Hodge}} B\rangle$ -- the quadratic form associated with the curvature (\ref{F_nabla^Hodge}) of $\nabla^\mr{Hodge}$. Note that $\Psi^F|_{B=i(\sfa)+\alpha_\mr{fl}}=\frac12 \langle \sfa,F_{\nabla^\HH} \sfa\rangle$ for $\alpha_\mr{fl}\in\mc{L}$.

Thus, (\ref{dQME on Zcheck proof eq1}) and (\ref{dQME on Zcheck proof eq2}) differ by 
$\check{Z}\frac{i}{\hbar}\frac12 \langle \sfa,F_{\nabla^\HH} \sfa\rangle$.

\begin{proof}[Sketch of diagrammatic proof of Theorem \ref{thm dQME on Zcheck}] 
For the purpose of the proof %it is convenient to introduce a ``dashed'' vertex which is the sum of the white and black vertices (i.e. it evaluates to $H$):  \begin{tikzpicture}
%     \node[striped] at (0,0) {}; 
%     \node[] at (.5,0) {$=$};
%         \node[black] at (1,0) {}; 
%     \node[] at (1.5,0) {$+$};
%         \node[white] at (2,0) {}; 
% \end{tikzpicture}
% We also introduce a square bivalent vertex to denote $d$: 
% \begin{tikzpicture}
%     \node[rectangle,draw] at (0,0) {$d$}; 
%     \node[] at (.5,0) {$=d$};
% \end{tikzpicture}
we expand the partition function $\check{Z}$ as a sum over graphs with trivalent and univalent vertices and leaves. Edges are assigned the extended propagator $\check{K}= \sum_{k=0}^2 K(HK)^k$, leaves the extended inclusion $\check{i}_{\mr{GF}}(\sfa) = \sum_{k=0}^2 (KH)^ki(\sfa)$, and univalent vertices are assigned $\delta A$. 
Notice that we have 
%\marginpar{\bl Oct 8: maybe change $\nabla_A^\mathbb{H}\ra \nabla_A$ here -- it is a connection in some different vector bundle ($End(\Omega)$ or $Hom(H_A,\Omega)$)}
\begin{align}
    %\nabla^\HH_A 
    \delta_A
    \check{K} &= \check{K}\mr{ad}_{\delta A}\check{K}, \\
    \nabla^\HH_A \check{i}_{GF} &= \check{K}\mr{ad}_{\delta A}\check{i}_\mr{GF}, \\
    \delta_{\mr{GF}}\check{K} &= - [d,\check{K}] + \mr{id} - \check{i}_{\mr{GF}}\check{p}_\mr{GF}, \\
    \delta_\mr{GF}\check{i}_\mr{GF} &= - d\check{i}_\mr{GF} + \check{i}_{\mr{GF}}\check{\Theta}.
\end{align}
%\marginpar{define $\Theta$ correctly}
Now the proof follows closely the proof of Theorem \ref{thm: horizontality of Zhat}. When computing $\delta_{\mr{GF}}\check{Z}$, there are now additional terms when $\delta_\mr{GF}$ hits the edge incident to a univalent vertex. Here, the terms $\mr{id} - \check{i}_{\mr{GF}}\check{p}_\mr{GF}$ survive. The first term is canceled by a graph with a $\delta A$ leaf that is produced when applying $\nabla^\HH_A$. A special case occurs for the graph consisting of a graph with one univalent vertex, one trivalent vertex and two leaves, here one such term survives and is canceled by the curvature \eqref{F_nabla^H} and $\nabla^\HH_A\check{\Theta}$, see Figure \ref{fig: curvature}. The second term involving $\check{i}_\mr{GF}\check{p}_\mr{GF}$ (which applied to $\delta A$ is simply $P$) is canceled by a term in  $\Delta_\sfa(e^{-\frac{i}{\hbar}\langle [\delta A], \sfa\rangle}\Phi_\Gamma (\sfa))$. 
\begin{figure}
\begin{tikzpicture}
    \node[] at (-1,0) {$\delta_\mr{GF}$} ;
    \node[grey] at (0,1) {} edge node[left] {$K$} (0,0); 
    \node[circle, draw, inner sep=1pt] at (-120:1) {$\sfa$} edge[double=black] (0,0); 
     \node[circle, draw, inner sep=1pt] at (-60:1) {$\sfa$} edge[double=black] (0,0); 
     \node at (1,0) {$=$};
     \begin{scope}[xshift=-1cm]
       \node[circle, draw, inner sep=1pt] at ($(3,0)+(-60:1)$) {$\sfa$} edge[double=black] (3,0); 
     \node[circle, draw, inner sep=1pt] at ($(3,0)+(-120:1)$) {$\sfa$} edge[double=black] (3,0);
     \node[above] at (3,0.3) {$\delta A$} edge (3,0);
     \node at (4,0) {$+\ldots$};
     \end{scope}
\end{tikzpicture}
\caption{The graph on the right evaluates to $\frac12\langle \sfa, F_\nabla(\sfa)\rangle + \nabla^\HH_A\check{\Theta}(\sfa)$, with $F_\nabla$ given in \eqref{F_nabla^H}. Thick edge stands for $\check{i}$, thin edges stand for $K$ (between vertices) or $i$ (at leaves). }\label{fig: curvature}
\end{figure}
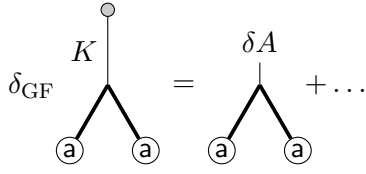
 As always, there are extra terms in the metric dependence of the partition function due to the total collapse of a connected component of a Feynman graph (canceled by $S_\mr{grav}$ counterterm). Incorporating that, we get (\ref{dQME on Zcheck}).
\end{proof}

\begin{rem}
    Note that while $\nabla^\mc{D}$ is not flat, the superconnection 
    \begin{equation}
    \check\nabla=\nabla^\mc{D}-i\hbar \Delta_\sfa -\frac{i}{\hbar} \frac12 \langle\sfa,F_{\nabla^\HH}\sfa \rangle
    \end{equation}
        on the bundle of formal half-densities on $H_A$, appearing in (\ref{dQME on Zcheck}), is flat. In more detail: one has
        \begin{equation}\label{nabla^D squared}
            (\nabla^\mc{D})^2=\langle F_{\nabla^\HH}(\sfa),\frac{\partial}{\partial\sfa} \rangle+\frac12 \mr{Str}_{H_A[1]} F_{\nabla^\HH}
        \end{equation}
        and
        \begin{multline}
            \check\nabla^2= (\nabla^\mc{D})^2 + \underbrace{[\nabla^\mc{D},-i\hbar\Delta_\sfa]}_0 - \frac{i}{\hbar} \frac12\underbrace{[\nabla^\mc{D}, \langle\sfa,F_{\nabla^\HH}\sfa \rangle]}_{=\langle \sfa, (\nabla^\HH F_{\nabla^\HH})\sfa\rangle=0 \,\mr{by\,Bianchi\, identity}}+\\
            +
            \underbrace{(-i\hbar \Delta_\sfa)^2}_0-[\Delta_\sfa,\frac12 \langle\sfa,F_{\nabla^\HH}\sfa \rangle]+\underbrace{(-\frac{i}{\hbar} \frac12 \langle\sfa,F_{\nabla^\HH}\sfa \rangle)^2}_0\\
            =
            (\nabla^\mc{D})^2+\Big\{\frac12 \langle\sfa,F_{\nabla^\HH}\sfa \rangle,-\Big\}_\sfa-\Delta_\sfa\Big(\frac12 \langle\sfa,F_{\nabla^\HH}\sfa \rangle\Big)\underset{(\ref{nabla^D squared})}{=}0.
        \end{multline}
\end{rem}

\subsubsection{Partial extensions of $Z$ along $A$, $A'$ and $g$ -- comparison with previous results}\leavevmode
%\begin{rem}[Partial extensions of $Z$]
%\begin{enumerate}
%    \item 

\textbf{Varying $A'$.}
    If we consider the slice of $\UU$ with fixed $A,g$ and varying $A'$, 
    % and take $\nabla$ with $\delta A'$ component given by (\ref{H delta A'}),\footnote{
    % Such $\nabla$ can be produced by the flattening construction above, taking the path $\gamma_p$ to be a piecewise smooth path $\gamma^3*\gamma^2*\gamma^1$ where $\gamma^1$ moves from $g_0$ to $g$, $\gamma^2$ moves from $A_0$ to $A$, $\gamma^3$ moves from $A'_0$ to $A'$.
    % } 
    the path integral (\ref{Z check}) reduces to (\ref{Zhat as path integral}). The corresponding restriction of (\ref{dQME on Zcheck}) to $\delta A=\delta g=0$ is equivalent to (\ref{dQME}). Also, in this case $\check{K},\check{i},\check\Theta$ reduce to $\wh{K},\wh{i}$ and $\wh\Theta$ of Section \ref{sec: ext Z A prime}.
    
%\item 
\textbf{Varying $g$.}
Similarly, if we consider the slice of $\UU$ with fixed $A,A'$ and varying $g$,
%and take $\nabla$ with $\delta g$ component given by (\ref{H delta g}), 
$\check{Z}$ reduces to (\ref{Zhat g as path integral}), (\ref{Zhat g}); dQME (\ref{dQME on Zcheck}) becomes (\ref{dQME on Zhat g}); $\check{K},\check{i},\check\Theta$ become $\wh{K}_{\delta g},\wh{i}_{\delta g}$ and $\wh\Theta_{\delta g}$ of Section \ref{sss Zhat g}, respectively.

%\item 
\textbf{Varying $A$.}
One can also consider varying $A$ while keeping $A',g$ fixed. 
%If we take $\nabla$ to have $\delta A$ component given by (\ref{H delta A}), 
In this case (\ref{Z check perturbative}) simplifies to
\begin{equation}\label{Z check delta A}
    \check{Z}_{\delta A}=e^{\frac{i}{\hbar} S_{CS}(A)} e^{\frac{\pi i}{4}\psi_A}\tau_A^{1/2} e^{-\frac{i}{\hbar}
    %(
    \langle [\delta A],\sfa \rangle
    %-\int_M \frac12\langle \delta A,K \delta A \rangle )
    }
    \exp \sum_\Gamma \frac{(-i\hbar)^{l(\Gamma)-1}}{|\mr{Aut}(\Gamma)|} \check\Phi_\Gamma
\end{equation}
where in $\check\Phi_\Gamma$, edges are decorated with the usual non-extended propagator $K$ and leaves are decorated with $i(\sfa)-K\delta A$. By (\ref{dQME on Zcheck}), $\check{Z}_{\delta A}$ satisfies the dQME in the direction of shifts of $A$:
\begin{equation}
    (\nabla^{\HH,A'}-i\hbar \Delta_\sfa)\check{Z}_{\delta A}=0
\end{equation}
%or, equivalently,
which implies
\begin{equation}\label{tilZ dQME}
     (\nabla^{\HH,A'}-\langle [\delta A],\frac{\partial}{\partial \sfa} \rangle 
     +\langle p\, \ad_{K\delta A} i(\sfa),\frac{\partial}{\partial \sfa}\rangle 
     + O((\delta A)^2) 
     -i\hbar \Delta_\sfa)\til{Z}_{\delta A}=0
\end{equation}
where $\til{Z}_{\delta A}$ is (\ref{Z check delta A}) with two modifications: (i) factor  $e^{-\frac{i}{\hbar}\langle [\delta A],\sfa \rangle}$ is removed, (ii) the  graph consisting of the single cubic vertex is removed from the sum over $\Gamma$.

We note that the result (\ref{tilZ dQME}) restricted to degree $1$ in $\delta A$ is equivalent to (\ref{nabla^tot Zhat = 0}) restricted to $\delta A'=0$ as one has, from inspection of Feynman diagrams,
\begin{equation}
    \til{Z}_{\delta A}=Z+r_{\til{\delta A'}=-d_{A'}K\delta A}Z+\underbrace{\cdots}_{\mr{degree}\;\geq 2\;\mr{in}\;\delta A},
\end{equation}
with $r$ and $\til{\delta A'}$ as in (\ref{nabla^tot Zhat = 0}), (\ref{tilde delta A'}).\footnote{
This is a consequence of the following observations: (i) an edge with a $K\delta A$ leaf plugged in contributes to $\check\Phi$ as $K\ad_{K\delta A} K=K[d^*_{A'},\ad^*_{K\delta A}]G= K \ad^*_{d_{A'}K\delta A}G = \Lambda_{\til{\delta A'}}$; (ii) a $K\delta A$ leaf joining an $i(\sfa)$ leaf and continuing with an edge contributes $K\ad_{K\delta A}i(\sfa)=G[d^*_{A'},\ad^*_{K\delta A}]i(\sfa)=G\ad^*_{d_{A'}K\delta A}i(\sfa)=\mathbb{I}_{\til{\delta A'}}i(\sfa)$. Thus, allowing one $K\delta A$ leaf in a graph results in graphs of $rZ$. Here $\Lambda$ and $\mathbb{I}$ are as in (\ref{Lambda, II, PP}).
}
To elucidate this equivalence, we note for harmonic shifts $\delta A$, the first two terms in (\ref{tilZ dQME}) yield the connection $\til\nabla^G$ (\ref{partial Grothendieck connection}); for $\delta A$ exact, there is a discrepancy between $\nabla^{\HH,A'}$ and $\nabla^\mr{tot}$ (\ref{nabla^tot}) compensated by the third term in (\ref{tilZ dQME}).\footnote{This discrepancy arises from the fact that the connection $\nabla^\mr{Harm}$, for an infinitesimal exact shift $\delta A= d\beta$ (and $\delta A'=0$), shifts a harmonic form $\chi$ to $\chi-K\ad_{d\beta}\chi= \chi-\ad_\beta\chi+dK\ad_\beta \chi + P\ad_\beta\chi$. Here $\ad_\beta\chi$ corresponds to $\nabla^\mr{gauge}+\nabla_{A'}$ (Section \ref{sss: T^I+T^II+T^III}), the term $d(\cdots)$ is irrelevant in cohomology and $P\ad_\beta\chi=P\ad_{K\delta A}\chi$ is the discrepancy.
}

\section{The global partition function} \label{sec: global partition function}
%\marginpar{\bl Sep 9: I rewrote this paragraph}
In this section we will discuss the properties of the ``synchronized'' partition function $Z(A_0,\sfa) = Z_{A_0,A_0}(\sfa)$ defined in Definition \ref{def: Z pert} seen as a family $\ul{Z}$ over the moduli space $\M'$. We introduce the \emph{global partition function} $Z^\mr{glob}$  %= Z(A_0,a=0)$.
-- a volume form on the moduli space $\M'$ arising from modifying $\ul{Z}$ to a global ($\nabla^G$-horizontal) object by a BV-exact term and then restricting it to $\sfa=0$. Finally, we study the dependence of $Z^\mr{glob}$ on metric.

%\subsection{Equivariance and descent to the moduli space}
\subsection{Perturbative partition function $\ul{Z}$ on the moduli space}
 %\marginpar{\bl Maybe "descent" (or "reduction") instead of "descending"? {\red descent} }
 \subsubsection{Bundles over $\FC'$ and $\M'$.}
We recall that $\FC$ denotes the space of all flat connections, $\mathcal{G} = C^\infty(M,G)$ the gauge group and 
\begin{equation}
    \M = \FC/\mathcal{G}
\end{equation}
the moduli space of flat connections, with $\pi \colon \FC \to \mathcal{M}$ the projection. Consider the subsets
\begin{equation}
   \FC'\subset \FC, \qquad \M' \subset \mathcal{M}
\end{equation}
of smooth irreducible points, by the results of Section \ref{sec: formal geometry}, they are smooth manifolds. 
Over $\FC'$, we have the cohomology bundle $\mathbb{H} \to \FC'$, with fiber over $A_0$ given by $H^\bullet_{A_0}$. The gauge group acts on $\mathbb{H}$ by conjugation, we have ${}^\sfg (H^\bullet_{A_0}) = H_{{}^\sfg A_0}$. The quotient
    $\mathbb{H}[1]/\mathcal{G} \to \M'$ is isomorphic to the bundle $T\M' \oplus T^*[-1]\M'$. 
    \subsubsection{Perturbative partition function}
    Restricted to  $\FC'$, the perturbative partition function  $Z$ defined in \eqref{eq: Z def} defines a section 
    %\marginpar{\bl $Det^{\frac12}$ should be not of $\mathbb{H}[1]$ (as was written) but of $\mathbb{H}[1]^*$ or equivalently of $\mathbb{H}$.}
\begin{equation}\label{eq: Z pert section}
    Z \in e^{\frac{i}{\hbar}S_{CS}(-)} \Gamma(\FC', \wh{\Sym}\,\mathbb{H}[1]^* \otimes \Det^\frac12 \mathbb{H} 
    %[1]^{\bl *}
    )[[\hbar]].
\end{equation}
% $$H_{A_0}:= H^\bullet_{d_{A_0}}(M,\g).$$
% Then for $A_0 \in \mathcal{A}_0^{\mr{sm},\mr{irr}}:=\pi^{-1}( \M^{\mr{sm},\mr{irr}}) $ we have 
% $$Z_{A_0} \in e^{\frac{i}{\hbar}{\bl S_{CS}(A_0)}}\wh{\Sym}(H^\bullet_{A_0}[1])^*\otimes\mr{Det}^{\frac{1}{2}}H_{A_0}.$$
% Thus, we can think of $Z$ as giving a map \marginpar{\bl are you saying that we should be thinking of $Z$ as a section of a bundle over $FC=\mathcal{A}_0$?}
% \begin{equation}
%     Z\colon \mathcal{A}_0^{\mr{sm},\mr{irr}} \to \bigsqcup_{A_0}\wh{\Sym}(H^\bullet_{A_0}[1])^*\otimes\mr{Det}^{\frac{1}{2}} H_{A_0} \label{eq:Z}
% \end{equation}
From \cite{Cattaneo2008} we have the following result. 
\begin{lem}\label{lem: 1 form part}
Suppose $A_0$ is an irreducible flat connection. Then 
$$Z_{A_0}(\sfa) \in \wh{\Sym}\,H^1_{A_0} \otimes \Det^\frac12(H^\bullet_{A_0})[[\hbar]]$$ %restricted to $m$-harmonic forms 
depends only on the 1-form part $\sfa^{1}$ of $\sfa$.  
\end{lem} 
Finally, we have the following:
\begin{prop}\label{prop: equivariance}
The Chern-Simons partition function $Z$ is equivariant with respect to the  action of the gauge group on $\FC'$ and $\mathbb{H}$.
\end{prop}
\begin{proof}
Follows immediately from Proposition \ref{prop: Z desync equivariance under diagonal gauge transf} by restricting to the diagonal $A=A'$.
\end{proof}
% \begin{proof}[Sketch of the proof, TO BE IMPROVED]
% { \red
% \marginpar{Moved equivariance here {\bl Oct 12. This is just Prop 4.1 restricted to $A=A'$, right?} {\red yes, that's right}}
% The SDR data associated to the diagonal Hodge decomposition satisfies the following equivariance: 

%     For all $\sfa \in H^\bullet_{A_0}$ and $\alpha \in \Omega^\bullet$.  we have 
%     \begin{equation}
%         i_{{}^\sfg A_0} {}^\sfg \sfa = {}^\sfg (i_{A_0}\sfa), \quad p_{{}^\sfg A_0} {}^\sfg \alpha = {}^\sfg(p_{A_0}\alpha), \quad K_{{}^\sfg A_0}{}^\sfg \alpha = {}^\sfg(K_{A_0}\alpha).
%     \end{equation}
% }
% %We have $K_{A^g_0} = (K_{A_0})^g$.
% Hence $Z_{{}^\sfg A_0}({}^\sfg \alpha) = Z_{A_0}(\alpha)$ since it is a contraction of invariant tensors. But this means that $Z_{{}^\sfg A_0}(\alpha) = Z_{A_0}({}^{\sfg^{-1}}\alpha)= {}^\sfg Z_{A_0}(\alpha)$ (natural action on the RHS of \eqref{eq:Z} is the dual action). 
% \end{proof}
\begin{cor}
    The perturbative partition function defines a section 
    \begin{equation}
        \underline{Z} \in  e^{\frac{i}{\hbar}S_{CS}(-)} \Gamma(\M', \wh{\Sym}\,T^*\M' \otimes \Det^\frac12(T\M'\oplus T^*[-1]\M')^*)[[\hbar]].
    \end{equation}
\end{cor}
%\marginpar{new notation: $\underline{Z}$ for the one defined on the moduli space}
\begin{proof}
    This follows from \eqref{eq: Z pert section} together with Lemma \ref{lem: 1 form part} and Proposition \ref{prop: equivariance} (notice that the quotient bundle $\mathbb{H}^1[1]/\mathcal{G} \cong T\M'$.) 
\end{proof}
\subsubsection{Naive global partition function} 
%\marginpar{\bl \sout{Sep9: move/merge to a later section on Zglob} Sep14: added ``naive'' everywhere}
Restricting $Z, \underline{Z}$ to $\sfa=0$, we obtain the naive global partition functions 
\begin{equation}
    \begin{aligned}
        Z^\mr{glob,naive}_A &= Z_A(0) \in e^{\frac{i}{\hbar}S_{CS}(A)}\Gamma(\FC', \Det^\frac12\mathbb{H} )[[\hbar]], \\
        \underline{Z}^\mr{glob,naive}_{A} &= \underline{Z}_{A}(0) \in e^{\frac{i}{\hbar}S_{CS}(A)}\Gamma(\M', \Det^\frac12(T\M'\oplus T^*[-1]\M')^*)[[\hbar]] .
    \end{aligned}
\end{equation}
\begin{rem}
    The bundle $\Det^\frac{1}{2}(T\M'\oplus T^*[-1]\M')^*$ is canonically isomorphic to the bundle of top forms $\Det(T^*\M')$ over $\M'$. From the BV viewpoint, it is also natural to identify this bundle with the subbundle of half-densities on $T^*[-1]\M'$ which do not depend on the fiber coordinates. %(this viewpoint will likely be helpful once one considers also singular points). 
\end{rem}

In Section \ref{ss Zglob} below we will construct the non-naive global partition function -- a modification of $\ul{Z}^\mr{glob,naive}$ yielding an invariant of $M$ as a framed 3-manifold (Theorem \ref{thm: coho class invariant}).

%\subsection{Comparing $Z^\mr{glob}$ and $Z$}
% \subsubsection{Formal and global objects}\marginpar{Does this subsection make sense?}
% If we are given an arbitrary manifold $M$ and a formal exponential map $\varphi \colon U \subset TM \to M$, then one gets a map $T\varphi^*\colon C^\infty(M) \to \Gamma(M, \operatorname{\widehat{\Sym}}T^*M).$ This map has a left inverse $ev_0$ given by restriction to the zero section. Given a section $\sigma \in \Gamma(M,\wh{\Sym}T^*M)$ (a ``formal object''), one can then ask whether $\sigma = T\varphi^*f$, with $f = ev_0\sigma$. If this is the case, we say $\sigma$ is global and the function $f$ is then considered the global object associated with $\sigma$. This discussion can be extended to section of tensorial bundles over $M$, see e.g. \cite[Section 2]{Bonechi2012}, MORE REFS for further discussions.
% In particular, the perturbative partition function defines a section 
% \begin{equation}
%     Z \in \Gamma(\M^{\mr{sm},\mr{irr}}, \wh{\Sym}\,T^*\M^{\mr{sm},\mr{irr}}\otimes\mr{Dens}^{\frac{1}{2}}\M^{\mr{sm},\mr{irr}})
% \end{equation}

%\marginpar{\bl Nov 11: new subsection}
\subsection{A formula for the Grothendieck connection}
First we introduce some notations that we will be using in the remainder of this section.

Let $\Xi$ be the one-form component (as a form on $\UU$) of the tree part of the extended partition function $\check{Z}$ (\ref{Z check perturbative}):
\begin{equation}
    \Xi= -\langle [\delta A],\sfa \rangle
    %+\frac12 \langle \sfa, \check{\Theta}(\sfa) \rangle
    +\sum_T \frac{1}{|\mr{Aut}(T)|}\check{\Phi}_{T}(\sfa) \Big|_{\Omega^1(\UU)}\quad \in 
    \Omega^1(\UU,\wh{\mr{Sym}} (H_A[1])^*)
\end{equation}%\marginpar{$\check{\Theta}$ is a 2-form and does not belong here {\bl removed it}}
with $T$ running over trivalent trees, with notations as in (\ref{Z check perturbative}). $\Xi$ splits as a 1-form along $A$ plus a 1-form along $A'$ plus a 1-form along $g$: 
%\marginpar{dec 11 fixed typo 1-from $\to$ 1-form}
\begin{equation}
    \Xi=\Xi_{\delta A}+\Xi_{\delta A'}+\Xi_{\delta g}.
\end{equation}
Furthermore, let
\begin{equation}
    I^*\colon \Omega^\bt(\mc{U})\ra \Gamma (\FC',\wedge^\bt (H^1_A)^*)
\end{equation}
be the map which restricts a form on $\mc{U}\subset \FC'\times \FC'$ to the diagonal $A'=A$ and evaluates it on \emph{$(A,A)$-harmonic} tangent vectors $\delta A'=\delta A$.

Then, given a form $\omega$ on $\mc{U}$ valued in the bundle $\mc{D}=\mr{Dens}^{\frac12,\mr{formal}}(H_A[1])$ and invariant under gauge transformations, we can construct a $\mc{D}$-valued form on $\M'$ by evaluating $\omega$ on the diagonal $A'=A$ and harmonic tangent vectors $\delta A'=\delta A$, and passing to gauge-equivalence classes. I.e., we have a map
\begin{equation}
    \pi_*I^*\colon \Omega^\bt(\mc{U},\mc{D})^\mr{Gauge}\ra \Omega^\bt(\M',\mc{D}).
\end{equation}
Here $\pi\colon \FC\ra \M$ is the projection to gauge-equivalence classes. The superscipt ``Gauge'' means ``invariants under the action of the gauge group.'' 

Denote %\marginpar{\bl Nov 11: draw diagrams? Nov 15: done!}
\begin{multline}
\xi=\pi_*I^*(\Xi_{\delta A}+\Xi_{\delta A'})\\
= -\langle [\delta A],\sfa \rangle+
\vcenter{\hbox{\includegraphics[scale=0.6]{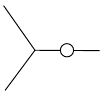}}}+
\vcenter{\hbox{\includegraphics[scale=0.6]{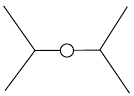}}}+
\vcenter{\hbox{\includegraphics[scale=0.6]{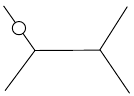}}}+\cdots
\\
=-\langle [\delta A],\sfa \rangle+\frac12\langle [i(\sfa),i(\sfa)],G\mr{ad}^*_{i[\delta A]}i(\sfa)  \rangle +O(\sfa^4) \quad
 \in \Omega^1(\M',\wh{\mr{Sym}} (H_A[1])^*)
\end{multline}
-- it is  the term $-\langle [\delta A],\sfa \rangle$ plus the sum of trees with one white binary vertex $H_{\delta A'=i[\delta A]}$, in terms of Remark \ref{rem: feynman rules check Z }.

\begin{prop}\label{prop formula for nabla^G}
    Grothendieck connection associated to the formal exponential map $\underline{\varphi}$ on $\M'$, seen as an operator on $\Omega^\bt(\M',\mc{D})$, can be written as
    \begin{equation}\label{nabla^G formula}
        \nabla^G=\nabla^{\mathbb{H}}_\M+\{\xi,-\}+\Delta_\sfa \xi.
    \end{equation}
\end{prop}
Here $\nabla^{\mathbb{H}}_\M=\pi_*I^* \nabla^{\mathbb{H}}$ is the connection on the bundle $\mc{D}\ra \M'$ obtained by restricting the connection in the cohomology bundle $\nabla^{\mathbb{H}}$ over $\mc{U}$ (induced by $\nabla^\mr{Harm}$ of Section \ref{sss nabla^Harm}) to the diagonal $A'=A$ and harmonic tangent vectors, and passing to gauge-equivalence classes.

For the proof of Proposition \ref{prop formula for nabla^G} we will need the following lemma.
\begin{lem}\label{lemma 5.6}
    Consider a path $(A,A'_t)$ in $\mc{U}$ passing through $(A,A')$ at $t=0$. Then the sum-over-trees map $\varphi_{A,A'_t}\colon H^1_A\supset U\ra \FC'$ satisfies
    \begin{equation}\label{variation of phi wrt A'}
        \frac{d}{dt}\Big|_{t=0}\varphi_{A,A'_t}(\sfa)=d_{\varphi_{A,A'}(\sfa)}\gamma + (d_\sfa \varphi_{A,A'})(\beta)
    \end{equation}
    for sufficiently small $\sfa$. Here:
    \begin{itemize}
        \item $\gamma$ is given by the sum over trees with $K$ on the root and either one edge (or the root) marked by $\Lambda_{\dot{A}'}$ or one leaf marked by $\mathbb{I}_{\dot{A}'}i(\sfa)$. One also allows for  $\gamma$ a special graph consisting of a single marked leaf with value $\mathbb{I}_{\dot{A}'}i(\sfa)$.
        \item $\beta=\{\Xi_{\delta A'=\dot{A}'},-\}$ is given by the sum over trees with $p$ on the root, with one edge marked by $\Lambda_{\dot{A}'}$ or one leaf marked by $\mathbb{I}_{\dot{A}'}i(\sfa)$ or the root marked by $p\mathbb{P}_{\dot{A}'}$.
    \end{itemize}
\end{lem}
Note that the first term in the r.h.s. of (\ref{variation of phi wrt A'}) is an infinitesimal gauge transformation.

\begin{proof}[Sketch of proof]
We have $\varphi_{A,A'_t}(\sfa)=A+\delta_{A,A'_t}(\sfa)$, with $\delta$ being the sum over trees with $K$ on the root. Therefore, the l.h.s. of (\ref{variation of phi wrt A'}) is the sum over trees with one edge (or the root) marked by $[d,\Lambda] +\mathbb{I}P+P\mathbb{P}$  
or one leaf marked by $d\mathbb{I}i$ (we suppress the subscripts, in particular, $d=d_A$). As in the proof of Proposition \ref{prop: variation of Z wrt A'}, using the Stokes' theorem on the configuration space to move $d$ from the marked edge or leaf to other edges, we obtain a sum of:
\begin{enumerate}
    \item Trees with one edge (or root) marked by $\Lambda$ or one leaf marked by $\mathbb{I}i$ and one other edge marked by $[d,K]=1-P$.
    \item Trees with one edge marked by $\Lambda$ or one leaf marked by $\mathbb{I}i$ and the root marked by $Kd=1-P-dK$.
    \item Trees with the root marked by $d\Lambda$.
    \item Trees with one edge (or root) decorated by $\mathbb{I}P$.
    \item Trees with one edge (or root) decorated by $P\mathbb{P}$.
\end{enumerate}
Contributions of 1 on the edge from (1) cancel out in the sum over graphs by the classical master equation (IHX relation). Contributions of $P$ on the edge from (1) and (2), together with (4) and (5), add up to the second term in the r.h.s. of (\ref{variation of phi wrt A'}): (a) if the subtree between $P$ and the leaves does not contain $\mathbb{I},\Lambda$, such subtrees add up to zero by smoothness, and the same applies to (4); (b) if the subtree does contain $\mathbb{I}$ or $\Lambda$, then such subtrees, as well as the ones coming from with (5), add up to $\beta$. The remaining contributions are: $dK$ from (2) and $d\Lambda$ from (3) on the root add up to $d\gamma$; $1$ from (2) yields $[\delta_{A,A'}(\sfa),\gamma]$  (from pairs of trees joined at the root, one tree containing $\Lambda$ or $\mathbb{I}$ and one not). Thus, we obtain
\begin{equation}
    \frac{d}{dt}\Big|_{t=0}\varphi_{A,A'_t}(\sfa)=d\gamma+[\delta_{A,A'}(\sfa),\gamma]+(d_\sfa \varphi_{A,A'})(\beta)= d_{\varphi_{A,A'}(\sfa)}\gamma+(d_\sfa \varphi_{A,A'})(\beta).
\end{equation}
\end{proof}

\begin{proof}[Proof of Proposition \ref{prop formula for nabla^G}]
    Fix $A\in \FC'$ and $\sfa\in H^1_A$ sufficiently small. Consider the first-order deformation $A\ra \til{A}=A+\epsilon i(\alpha)$ with $\alpha\in H^1_A$. Then the infinitesimal parallel transport of $\nabla^G$, seen as a fiber bundle connection on a neighborhood of the zero-section in $T\M'$ (Definition \ref{def: Grothendieck connection}), maps $\sfa\ra \til{\sfa}$ where $\til\sfa$ is determined by
    \begin{equation}\label{Prop 5.5 proof eq1}
        \ul\varphi_A(\sfa)=\ul\varphi_{\til{A}}(\til{\sfa}).
    \end{equation}
    We will write $\til{\sfa}=B(\sfa+\epsilon \sfb)$ (all computations are $\bmod \,\epsilon^2$), with $B=\mathfrak{B}^1_{\til{A}\la A;\til{A}}\colon H^1_A\ra H^1_{\til{A}}$ the cohomology comparison map.\footnote{Note that in the first order in $\epsilon$ it does not matter whether we choose $A'=A$ or $A'=\til{A}$ in the cohomology comparison map $\mathfrak{B}_{\til{A}\la A;A'}$.} The r.h.s. of equation (\ref{Prop 5.5 proof eq1}) reads
    %\marginpar{\bl Apr 19: edited/ equation reformatted. (can also remove eq. numbers here)}
    \begin{align}
        %[\varphi_{A,A}(\sfa)]&\stackrel{\mr{want}}{=} 
        &[\varphi_{\til{A},\til{A}}(\til\sfa)] 
        \underset{\mr{Prop.\;}\ref{prop: desy exp map conv}}{=} 
        [\varphi_{A,\til{A}}(B^{-1}(\til\sfa)+\epsilon\alpha)]
        \\&=
        [\varphi_{A,\til{A}}(\sfa)+\epsilon (d_\sfa\varphi_{A,A})(\alpha+\sfb)]\\
        &\underset{\mr{Lemma\;}\ref{lemma 5.6}}{=}[\varphi_{A,A}(\sfa)+\epsilon (d_\sfa\varphi_{A,A})(\alpha+\sfb+\beta)+\underbrace{\epsilon d_{\varphi_{A,A}(\sfa)}(\gamma)}_{\mr{gauge\;transf.}}]\\
        &=[\varphi_{A,A}(\sfa)+\epsilon (d_\sfa\varphi_{A,A})(\alpha+\sfb+\beta)].
    \end{align}
    Here $\beta,\gamma$ are as in Lemma \ref{lemma 5.6} with $A'=A$ and $\dot{A}'=i(\alpha)$. Thus,  (\ref{Prop 5.5 proof eq1}) holds if the $O(\epsilon)$ term in the last line vanishes, i.e., if $\sfb=-\alpha-\beta$ and thus
    \begin{equation}
        \til\sfa=B(\sfa-\epsilon (\alpha+\beta)).
    \end{equation}
    Therefore, as a fiber bundle connection on (a neighborhood of the zero-section in) $T\M'$, 
    \begin{equation}\label{Prop 5.5 proof eq2}
    \nabla^G= \nabla^\mathbb{H}_{\M}+\langle [\delta A],\frac{\partial}{\partial \sfa} \rangle+\{\pi_* I^* \Xi_{\delta A'=i[\delta A]},-\}.
    \end{equation}
    Its action on half-densities is given by (\ref{nabla^G formula}).
\end{proof}
\begin{rem}
    When we consider $\nabla^G$ as a fiber bundle connection on $T\M'$, $\sfa=\sfa^1$ is always an element of $H^1_A$; second and third terms in (\ref{Prop 5.5 proof eq2}) give a 1-form on $\M'$ valued in formal vector fields on $H^1_A$. On the other hand, when we allow $\nabla^G$ to act on half-densities on $H^\bt_A[1]$, we have the cotangent lift of that vector field, 
    $\{\xi,-\}$ (which involves $\sfa^2$ -- the component $H^2_A$ of $\sfa$) appearing in (\ref{nabla^G formula}). The term $\Delta_\sfa\xi$ arises from the natural action of a hamiltonian vector field on $H^\bt_{A}[1]$ on a half-density.
\end{rem}
}

\subsection{Almost horizontality of $\ul{Z}$ w.r.t. Grothendieck connection}

%\marginpar{\bl Notations: we can use pullbacks $\mathfrak{B}^*$ instead of determinant factors}
\begin{prop}\label{prop ulZ finite horizontality}
    Let $\underline\varphi \colon U \subset T\M'\to \M'$ be the sum-over-trees exponential map (\ref{phi underline}) induced by the SDR $(i_{A},p_{A},K_{A})$. Fix $A$ and small $\alpha\in H^1_{A}$, %and denote $\til{A}=\varphi_{A,A}(\alpha)$ 
    and let $B=d_\alpha\ul\varphi_{A}(\alpha)\colon H^1_{A}\ra H^1_{\ul\varphi_A(\alpha)}$. Then, for $\sfa\in H^1_{A}$ small, we have
    \begin{equation}\label{Z underline finite almost-Grothendieck-horizontality}
        \det(B^\vee)\circ \underline{Z}_{\underline{\varphi}_{A}(\alpha)} (B(\sfa))=\underline{Z}_{A}(\alpha + \sfa)  -i\hbar \Delta_\sfa R(A,\alpha,\sfa) ,
    \end{equation}
    where 
    \begin{multline}\label{R for Z underline finite almost-Grothendieck-horizontality}
        R(A,\alpha,\sfa)=\det (B^\vee)\circ R_{\til{A},A,\til{A}}(\mathfrak{B}(\sfa))\\
        = \int_0^1 dt\, r_{\til{A},A_t;\dot{A}_t}(\mathfrak{B}(\sfa))\cdot \det(B^\vee)\circ Z_{\til{A},A_t}(B(\sfa)).
    \end{multline}
    Here: $R_{\til{A},A,\til{A}}$ %in the r.h.s. in the first line 
    is as in (\ref{R in change of Z with A'}),
    $\til{A}=\varphi_{A,A}(\alpha)$ (hence $[\til{A}]=\ul\varphi_{A}(\alpha)$), $A_t=\varphi_{A,A}(t\alpha)$ is a path from $A$ to $\til{A}$,  $\mathfrak{B}=\mathfrak{B}_{\til{A}\la A,A}\colon H^\bt_A\ra H^\bt_{\til{A}}$ is the promotion of $B$ to a map between full cohomology, $r$ is as in Proposition \ref{prop: variation of Z wrt A'}.
    % where
    % \begin{equation}\label{R for Z underline finite almost-Grothendieck-horizontality}
    % R(A_0,\alpha,\sfa)=%R_{A,A,\varphi_A(\sfa)}(\sfa+\sfb)
    % \int_0^1 dt\, r_{A_0,A'_t;\dot{A}'_t}(\alpha+\sfa) Z_{A_0,A'_t}(\alpha+\sfa)
    % \end{equation}
    % with $A'_t=\varphi_{A_0}(t\alpha)$ %and $r$ as in Proposition \ref{prop: variation of Z wrt A'}.
    % and $r_{A_0,A'_t;\dot{A}'_t}(\cdots)$ denotes the sum of connected Feynman diagrams with $(A_0,A'_t)$ Feynman rules, one edge decorated by $\Lambda$, or one leaf decorated by $\mathbb{I}$, with $\delta A' = \dot{A}'_t  $ (cf. Proposition \ref{prop: variation of Z wrt A'} and equation \eqref{R in change of Z with A'}).
\end{prop}
\begin{proof}
This is an immediate consequence of Theorems \ref{thm 4.2} and \ref{thm: change gf}.
Indeed, we have
\begin{multline}
 \det(B^\vee)\circ \underline{Z}_{\underline{\varphi}_{A}(\alpha)} (B(\sfa))=\det(B^\vee)\circ Z_{\til{A},\til{A}}(B(\sfa))
\\ 
\underset{\mr{Theorem\,}\ref{thm: change gf}}{=}
\det(B^\vee)\circ \Big( Z_{\til{A},A}(B(\sfa))-i\hbar (\Delta_\sfb R_{\til{A},A,\til{A}}(\sfb))\Big|_{\sfb=B(\sfa)} \Big)\\
\underset{\mr{Theorem\,}\ref{thm 4.2}}{=}
Z_{A,A}(\alpha+\sfa)-i\hbar \Delta_\sfa \Big(\det(B^\vee)\circ R_{\til{A},A,\til{A}}(\mathfrak{B}(\sfa))\Big).
     % \det(B^\vee)\circ \underline{Z}_{\underline{\varphi}_{A_0}(\alpha)} (B(\sfa))=
     %  \det(B^\vee)\circ Z_{\varphi_{A,A}(\alpha),\varphi_{A,A}(\alpha)}(B(\sfa))\\
     %  \underset{\mr{Theorem\,}\ref{thm 4.2}}{=}
     %  Z_{A,\varphi_{A,A}(\alpha)}(\alpha+\sfa)
\end{multline}
\end{proof}

\begin{cor}\label{cor ulZ nabla^G horizontality}
%     {\bl Assuming 1-extended smoothness (see Definition \ref{def: extended smoothness} below),}\footnote{
% \bl For the reason why 1-extended smoothness assumption is important, see Corollary \ref{cor: partial G = nabla G} below. 
%     }
    Partition function $\ul{Z}$ is horizontal w.r.t. Grothendieck connection modulo a BV-exact term:
    \begin{equation}\label{Z underline almost-horizontality}
        \nabla^G \ul{Z}=-i\hbar\Delta_\sfa 
        %R_{A,\delta A}(\sfa)
        %\left(r_{A,A;\delta A}(\sfa)Z_{A,A}(\sfa)\right),
        (\ul{r}\, \ul{Z}),
    \end{equation}
    where $\ul{r}$ is $r_{A,A;i[\delta A]}$ of Proposition \ref{prop: variation of Z wrt A'} with trees removed. I.e., $\ul{r}$ is the sum of connected Feynman graphs with $\geq 1$ loops, with one edge marked by $\Lambda_{i[\delta A]}$ or one leaf marked by $\mathbb{I}_{i[\delta A]}i(\sfa)$.
%    as in Proposition \ref{prop: variation of Z wrt A'}.
%    \marginpar{\bl Nov 11: need to exclude trees from $r$}
    % with 
    % \begin{equation}
    % R_{A,\delta A}(\sfa)=r_{A,A;\delta A}(\sfa)Z_{A,A}(\sfa).
    % \end{equation}
    %the derivative in $\alpha$ at $\alpha=0$ of the generator (\ref{R for Z underline finite almost-Grothendieck-horizontality}).
\end{cor}

\begin{proof}
    Taking the derivative of (\ref{Z underline finite almost-Grothendieck-horizontality}) in $\alpha$ at $\alpha=0$ one obtains the following:
    \begin{equation}
        \nabla^{G,\mr{triv}}\ul{Z}= -i\hbar\Delta_\sfa( %r_{A,A;i[\delta A]}
        \pi_*I^*(r)\,
        \ul{Z}),
    \end{equation}
    where $r$ is as in Proposition \ref{prop: variation of Z wrt A'} and $\nabla^{G,\mr{triv}}\colon=\nabla^{\mathbb{H}}_\M+\langle [\delta A],\frac{\partial}{\partial \sfa}\rangle$. Adding 
    $%-i\hbar
    \Delta_\sfa (
    %\frac{i}{\hbar}
\pi_* I^* \Xi_{\delta A'}\, \ul{Z})$ to both sides and using Proposition \ref{prop formula for nabla^G}, we obtain (\ref{Z underline almost-horizontality}).
\end{proof}

\subsection{Correcting $\ul{Z}$ to a global object. Definition of 
%the global partition function 
$Z^\mr{glob}$.}\footnote{Main statements of this section are a Chern-Simons counterpart of Theorem 6.1 in \cite{Bonechi2012}.}
\label{ss Zglob}
%Throughout this subsection we are assuming 1-extended smoothness.
%{\bl It is a natural question whether $\underline{Z}$ is the Taylor expansion of $\underline{Z}^\mr{glob}$ with respect to the sum-over-trees formal exponential map. The following theorem states that this is only true up to a fiberwise canonical transformation: }
\begin{thm}\label{thm 5.15}
%Under 1-extended smoothness assumption,   we have the following.
$\ul{Z}$ can be modified by a BV-exact term (pointwise on the moduli space) to a global object. I.e., there exists a degree $-1$ element $\rho\in \Gamma(\M',\mr{Dens}^{\frac12,\mr{formal}}(H^\bt_{A}[1]))$ such that 
\begin{equation}\label{Zmod}
Z^\mr{mod}\colon= \ul{Z}+i\hbar\Delta_\sfa \rho
\end{equation}
satisfies
    \begin{equation}
        \nabla^G Z^\mr{mod}=0.
    \end{equation}
%     \begin{equation}
%         (T\ul\varphi^* \ul{Z}^\mr{glob})(A,\sfa)=\ul{Z}_{A}(\sfa)-i\hbar \Delta_\sfa \rho(A,\sfa)
%     \end{equation}
%     with a certain degree $-1$ generator $\rho\in \Gamma(\M',\mr{Dens}^{\frac12,\mr{formal}}(H^\bt_{A}[1]))$.
\end{thm}

For the proof (and for the proof of Proposition \ref{prop: properties of Zcheckglob} below) we will need the following lemma.

\begin{lem}\label{lemma 5.11}
\begin{enumerate}[(a)]
    \item 
    Let 
    \begin{equation}\label{olZ def}
        \overline{Z}\colon=\pi_*I^* (e^{-\frac{i}{\hbar}\Xi}\check{Z}) \quad \in \Omega^\bt(\mr{Met}\times \M',\mc{D})
    \end{equation}
    and 
    \begin{equation}\label{olZren def}
    \overline{Z}^\mr{ren}\colon= e^{\frac{i}{\hbar}c(\hbar) \frac{S_\mr{grav}(g,\phi)}{2\pi}} \overline{Z}.
    \end{equation}
    Then  $\overline{Z}^\mr{ren}$ satisfies
    \begin{equation}\label{olZ equation}
        (\delta_g+\nabla^G+\{\mu,-\}+\Delta_\sfa \mu-i\hbar \Delta_\sfa) \overline{Z}^\mr{ren}=0.
    \end{equation}
    Here $\mu=\pi_*I^* \Xi_{\delta g}\in \Omega^{1,0}(\mr{Met}\times \M',\wh{\mr{Sym}}(H_A[1])^*)$.
   \item  Let
    \begin{equation}\label{tilZ def}
        \til{Z}\colon= \overline{Z}|_{g\mr{\;fixed,\;}\delta g=0}\; \in \Omega^\bt(\M',\mc{D})
    \end{equation}
    -- the restriction of $\overline{Z}$ to a $g$-fixed slice of $\mr{Met}\times \M'$.\footnote{%\bl
    Note that $\til{Z}=\ul{Z}+\ul{r}\,\ul{Z}+\cdots$ with $\ul{r}\,\ul{Z}$ the generator in the r.h.s. of (\ref{Z underline almost-horizontality}) and $\cdots$ being a sum of forms of degree $\geq 2$ on $\M'$. 
    } Then $\til{Z}$ satisfies
    \begin{equation}\label{tilZ equation}
        (\nabla^G-i\hbar\Delta_\sfa)\til{Z}=0.
    \end{equation}
\end{enumerate}
\end{lem}

%\marginpar{\bl can move the proof to the appendix (with technical proofs)}
\begin{proof}%\marginpar{\bl check that signs are locally reasonable}
    Denote $\xi^\mr{tot}\colon = \pi_*I^* \Xi=\xi+\mu$. The dQME (\ref{dQME on Zcheck}) implies 
    \begin{multline}\label{lemma 5.11 proof eq1}
    0=\pi_*I^* e^{-\frac{i}{\hbar}\Xi}\underbrace{(\nabla^\mc{D}-i\hbar\Delta_\sfa-\frac{i}{\hbar}\frac12 \langle \sfa, F_{\nabla^\mathbb{H}}\sfa\rangle)}_{\nabla^{G-M}} (e^{\frac{i}{\hbar}\Xi}(e^{-\frac{i}{\hbar}\Xi} \check{Z}^\mr{ren}) )\\
    =
    \Big(\nabla^\mathbb{H}_{\mr{Met}\times \M'}+\{\xi^\mr{tot},-\}+\Delta_\sfa \xi^\mr{tot}-i\hbar\Delta_\sfa+\\
    +\frac{i}{\hbar}(\nabla^\mathbb{H}_{\mr{Met}\times \M'}\xi^\mr{tot}+\frac12 \{\xi^\mr{tot},\xi^\mr{tot}\}-\frac12 \langle \sfa, F_{\nabla^\mathbb{H}}\sfa\rangle +\underbrace{\frac12\langle [\sfa,\sfa],KHi[\delta A] \rangle}_Y) 
    \Big) \overline{Z}^\mr{ren}.
    \end{multline}
    Here  the differential operator in brackets in lines 2 and 3 -- except the term $Y$ -- arises as a conjugation of the Gauss-Manin superconnection by $e^{\frac{i}{\hbar}\Xi}$, restricted to $A'=A$ and $\delta A'=\delta A$ harmonic and reduced modulo gauge transformations. 
    
    Next we comment on the $Y$ term appearing in (\ref{lemma 5.11 proof eq1}). First consider for simplicity the case of $g$ fixed. In fact, $I^* (\nabla^{G-M} \check{Z}^\mr{ren})$ cannot be computed a priori as a differential operator (denote it $I^*\nabla^{G-M}$) acting on $I^*\check{Z}^\mr{ren}$. The reason is that when we move infinitesimally on $\FC' \stackrel{\mr{Diag}}{\hookrightarrow}\FC'\times \FC'$, %(along the diagonal), 
    the notion of harmonic forms changes, so, in order to evaluate $I^* (\nabla^{G-M} \check{Z}^\mr{ren})$ on a tuple of harmonic tangent vectors, it is insufficient to know $I^*(\check{Z}^\mr{ren})$ -- one needs also the value of $\check{Z}^\mr{ren}$ on a collection of harmonic tangent vectors and one non-harmonic (arising as a variation of a harmonic tangent vector). 
    %$$I^*(\nabla^{G-M} \check{Z})(\alpha_0,\ldots, \alpha_p)=(I^*\nabla^{G-M})(I^*\check{Z})$$
    More precisely: given harmonic tangent vectors $\alpha_0,\ldots,\alpha_p \in T_{(A,A)}\FC'_\mr{Diag}$, we extend them to harmonic vector fields $\til\alpha_0,\ldots,\til\alpha_p$ on $\FC'_\mr{Diag}$. Then we have
    \begin{multline}\label{lemma 5.11 proof eq3}
        (\nabla^\mc{D}\check{Z}^\mr{ren})(\til\alpha_0,\ldots,\til\alpha_p)=
        \sum_{i=0}^p %(-1)^i
        \nabla^\mc{D}_{\til\alpha_i}\check{Z}^\mr{ren}(\til\alpha_0,\ldots\wh{\til\alpha_i}\ldots, \til\alpha_p)\\
        -
        \sum_{i,j=0}^p \frac12 %(-1)^{i+j} 
        \check{Z}^\mr{ren}([\til\alpha_i,\til\alpha_j],\til\alpha_0,\ldots\wh{\til\alpha_i}\ldots\wh{\til\alpha_j}\ldots,\til\alpha_p).
    \end{multline}
    Then, choosing the extensions $\til\alpha_i|_{A+\epsilon \beta}=\alpha_i-\epsilon (K \mr{ad}_\beta \alpha_i+dG \mr{ad}^*_{\beta}\alpha_i)+O(\epsilon^2)$ (cf. (\ref{nabla^harm infinitesimal horizontal transport}), (\ref{nabla^harm infinitesimal vertical transport})), we have 
    $[\til\alpha_i,\til\alpha_j]|_{(A,A)}=2dG*[\alpha_i,*\alpha_j]$. Hence, we have
    \begin{equation}
        I^* \nabla^{G-M} \check{Z}^\mr{ren}=(I^* \nabla^{G-M}) (I^* \check{Z}^\mr{ren})+I^*(\langle -d G \ad^*_{\delta A}\delta A,\frac{\delta}{\delta(\delta A)} \rangle \check{Z}^\mr{ren} )
    \end{equation}
    Here the two terms on the right corresponds to the two terms on the right in (\ref{lemma 5.11 proof eq3}).
    Allowing $g$ to vary (and so allowing the tangent vectors to have component along $T_g\mr{Met}$), the obvious extension of this computation yields 
    \begin{equation}\label{lemma 5.11 proof eq2}
        I^* \nabla^{G-M} \check{Z}^\mr{ren}=(I^* \nabla^{G-M}) (I^* \check{Z}^\mr{ren})+I^*(\langle d K\underbrace{H}_{H_{\delta A'}+H_{\delta g}}\delta A,\frac{\delta}{\delta(\delta A)} \rangle \check{Z}^\mr{ren} )
    \end{equation}
    In terms of the Feynman diagram expansion of $\check{Z}$ of Remark \ref{rem: feynman rules check Z }, the rightmost term in (\ref{lemma 5.11 proof eq2}) has the following contributions:
%    \marginpar{\bl Nov 14: draw diagrams?}
    \begin{enumerate}[(i)]
        \item $Hi[\delta A]$ replacing a gray vertex in a graph (from the derivation acting on $\delta A$ in a gray vertex),
        \item replacing a white vertex in a graph by $\ad_{KHi[\delta A]}$ (from the derivation acting on $\delta A'=\delta A$ in a white vertex).
    \end{enumerate}
    Contributions (i) and (ii) above mutually cancel, except for one special graph in (i) -- a cubic corolla with one leaf attached to a gray vertex, replacing the latter with $Hi[\delta A]$ as in (i) above, see Figure \ref{fig: lemma 5.11}. This yields $Y \check{Z}^\mr{ren}$, resulting eventually in the $Y$ term correction in (\ref{lemma 5.11 proof eq1}).
    \begin{figure}
        \centering
        \includegraphics[width=1\linewidth]{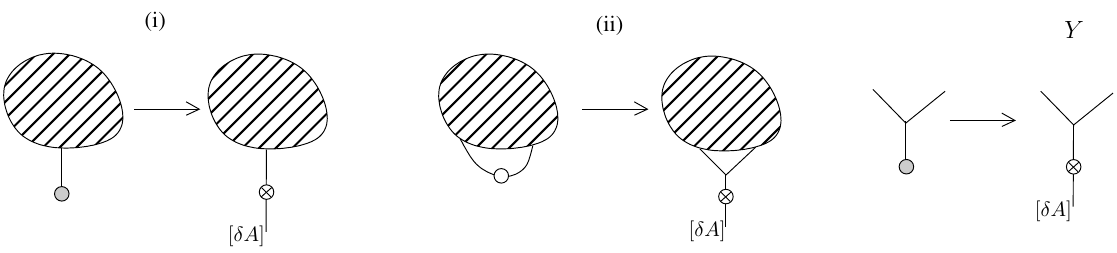}
        \caption{Cancellation mechanism in the rightmost term in (\ref{lemma 5.11 proof eq2}). Circle vertex with a cross stands for either white or black vertex. Dashed region represents some graph.}
        \label{fig: lemma 5.11}
    \end{figure}

    Now that (\ref{lemma 5.11 proof eq1}) is proven, we show that it can be simplified. Consider the dQME (\ref{dQME on Zcheck}) on $\check{Z}^\mr{ren}$, restrict it to lowest order in $\hbar$ (tree diagram contributions) and to 2-forms on $\UU$. Then, reducing to $\mr{Met}\times \M'$ by applying $\pi_*I^*$, we obtain the equation
    \begin{equation}
        \nabla^\mathbb{H}_{\mr{Met}\times \M'}\xi^\mr{tot}+\frac12 \{\xi^\mr{tot},\xi^\mr{tot}\}-\frac12 \langle \sfa, F_{\nabla^\mathbb{H}}\sfa\rangle +Y =0
    \end{equation}
    where the $Y$ term appears by the same mechanism as above.
    Hence, (\ref{lemma 5.11 proof eq1}) simplifies to 
    \begin{equation}
        (\delta_g+\underbrace{\nabla^\mathbb{H}_\M+\{\xi,-\}+\Delta_\sfa \xi}_{\nabla^G} +\{\mu,-\}+\Delta_\sfa\mu -i\hbar\Delta_\sfa)\overline{Z}^\mr{ren}=0,
    \end{equation}
    which proves (\ref{olZ equation}).

    Finally, (\ref{tilZ equation}) is an immediate consequence of (\ref{olZ equation}) by restricting to a $g$-fixed slice of $\mr{Met}\times \M'$.
\end{proof}

\begin{proof}[Proof of Theorem \ref{thm 5.15}]

    \textbf{Step 1.} (Building a chain contraction for $\nabla^G-i\hbar\Delta_\sfa$.)
    The cohomology of the complex $(\Omega^\bt(\M', 
    \mc{D}
    %\mr{Dens}^{\frac12,\mr{formal}}(H^\bt_A[1])
    ),\nabla^G)$ is concentrated in form degree $0$ and is isomorphic to global half-densities $\mr{Dens}^{\frac12}(T^*[-1]\M')$.\footnote{
    See \cite[Section 2]{Bonechi2012}.
    }
    More precisely, one has SDR data $(\mathsf{i},\mathsf{p},\mathsf{K})$ with inclusion $\mathsf{i}=T\ul\varphi^*$ and projection $\mathsf{p}$ given by evaluating a formal half-density at $\sfa^1=0$,
    \begin{equation}
        \mathsf{p}: \psi([A],\sfa^1,\sfa^2)D^{\frac12} \sfa^1 D^{\frac12}\sfa^2\mapsto \psi([A],0,\sfa^2) D^{\frac12} [A] D^{\frac12}\sfa^2
    \end{equation}
    where $\sfa^1,\sfa^2$ are the components of $\sfa$ in $H^1_A=T_{[A]}\M'$ and $H^2_A=T^*_{[A]}[-1] \M'$. It is also understood that $\mathsf{p}$ sends forms of positive degree on $\M'$ to zero.

    Next, deform the differential on $\Omega^\bt(\M', %\mr{Dens}^{\frac12,\mr{formal}}(H^\bt_A[1])
    \mc{D}
    )$ from $\nabla^G$ to $\nabla^G-i\hbar\Delta_\sfa$. By homological perturbation lemma (Lemma \ref{lem: HPL}), one has deformed SDR data
    \begin{equation}\label{(i',p',K')}
    (\mathsf{i}',\mathsf{p}',\mathsf{K}')\colon (\Omega^\bt(\M', %\mr{Dens}^{\frac12,\mr{formal}}(H^\bt_A[1])
    \mc{D}
    ), \nabla^G -i\hbar\Delta_\sfa) \leadsto (\mr{Dens}^{\frac12}(T^*[-1]\M'),\delta).
    \end{equation} 
    %\marginpar{\red Displayed this equation for potentially better readability(?) {\bl looks good!}}
    Moreover, the fact that $\mathsf{K}$ lowers form degree along $\M'$ by one, implies that 
    \begin{itemize}
    \item the induced differential $\delta=-i\hbar\Delta$ is the BV Laplacian on (global) half-densities $T^*[-1]\M'$,
    \item $\mathsf{i}'=\mathsf{i}$.
    \end{itemize}
    
\textbf{Step 2.} Using the fact that $K'$ defined above satisfies the chain homotopy property $\mr{id}=\mathsf{i}\mathsf{p}'+[\nabla^G-i\hbar\Delta_\sfa,\mathsf{K}']$, we have
\begin{equation}\label{thm 5.17 proof eq1}
\til{Z}=\mathsf{i}\mathsf{p}'\til{Z}+(\nabla^G-i\hbar\Delta_\sfa) \mathsf{K}' \til{Z}+
   \mathsf{K}'\underbrace{(\nabla^G-i\hbar\Delta_\sfa)\til{Z}}_{=0\mr{\; by\; (\ref{tilZ equation})}} ,
\end{equation}
where $\til{Z}$ is as in (\ref{tilZ def}).
Denote the first term on the r.h.s. by $Z^\mr{mod}\colon= \mathsf{i}\mathsf{p}'\til{Z}$. Since it is in the image of $\mathsf{i}$, it is $\nabla^G$-closed, and hence a global object. Restricting (\ref{thm 5.17 proof eq1}) to form degree zero along $\M'$ and denoting 
\begin{equation}\label{rho from Zmod=ulZ-Delta(rho)}
\rho=(\mathsf{K}' \til{Z})|_{\Omega^0(\M',\mc{D})},
\end{equation}
we obtain (\ref{Zmod}).
% Restricting this expression to form degree zero along $\M'$, we obtain
% \begin{equation}
%     \ul{Z}=T\ul\varphi^* \ul{Z}^\mr{glob}+i\hbar \Delta_\sfa \rho
% \end{equation}
% with $\rho=(K' \ul{\til{Z}})|_{\Omega^0(\M',\cdots)}$.
\end{proof}

%\marginpar{\bl Nov15: remark edited and reformatted as a lemma}
\begin{lem}\label{lem Zmod = Z (1+O(hbar))}
    We have the following ansatz for $\hbar$-dependence of $\rho$ and $Z^\mr{mod}$: %1-extended smoothness implies 
    \begin{eqnarray}
        \rho&=&\ul{Z}\cdot f(\hbar), \label{rho hbar dependence}\\
        Z^\mr{mod}&=&\ul{Z}\cdot (1+\hbar\, g(\hbar)), \label{Zmod hbar dependence}
    \end{eqnarray}
    where $f,g\in C^\infty(\M',\underbrace{\wh{\mr{Sym}}(H_A[1])^*}_{\mathbb{F}})[[\hbar]]$ are some formal power series in $\hbar$ and $\sfa$.
%    $\rho=\ul{Z}\cdot O(\hbar^0)$ and $Z^\mr{mod}=\ul{Z}\cdot (1+O(\hbar))$. 
\end{lem}
\begin{proof}
    Both properties follow from the fact that 
    \begin{equation}
    \til{Z}=\underbrace{e^{\frac{i}{\hbar}S_{CS}(A)}\tau_A^{\frac12}e^{\frac{\pi i}{4}\psi_A}}_{\ul{Z}_0}\,e^{\frac{i}{\hbar}\omega+h}
    \end{equation} 
    with $\omega\in \Omega^{\geq 2}(\M',\mathbb{F})$ independent of $\hbar$ and with $h\in \Omega^{\geq 0}(\M',\mathbb{F})[[\hbar]]$. Hence, the $p$-form component of $\til{Z}$ along $\M'$ is
    \begin{equation}
        \til{Z}^{(p)}=\ul{Z}_0\cdot \hbar^{-\left[\frac{p}{2}\right]}F^{(p)}=\ul{Z}\cdot \hbar^{-\left[\frac{p}{2}\right]}F^{'(p)}
    \end{equation}
    with $F^{(p)},F^{'(p)}\in \Omega^p(\M',\mathbb{F})[[\hbar]]$. Therefore,
    \begin{equation}
        \rho=\mathsf{K}'\til{Z}\big|_{\Omega^0(\M',\mc{D})}=\sum_{k\geq 0}(\mathsf{K} i\hbar \Delta_\sfa)^k \mathsf{K}\left(\til{Z}|_{\Omega^{k+1}(\M',\mc{D})}\right)=
        \ul{Z}_0 \cdot \sum_{k\geq 0}\hbar^{k-\left[\frac{k+1}{2}\right]}G_k
    \end{equation}
    with $G_k\in \Omega^0(\M',\mathbb{F})[[\hbar]]$. This implies (\ref{rho hbar dependence}) and -- using (\ref{Zmod}) -- implies also (\ref{Zmod hbar dependence}).
\end{proof}

\begin{defn}
    We define the global partition function as the degree zero half-density on $T^*[-1]\M'$ (or, equivalently, a volume form on $\M'$) given by restriction $Z^\mr{mod}(A,\sfa)$ to $\sfa=0$:
    \begin{equation}
    Z^\mr{glob}\colon= Z^\mr{mod}\Big|_{\sfa=0} \in \mr{Dens}^{\frac12}(T^*[-1]\M').
    \end{equation}
\end{defn}
In the notations of the proof of Theorem \ref{thm 5.15}, we have 
\begin{equation}
Z^\mr{glob}=\mathsf{p}'\til{Z}. 
\end{equation}
Lemma \ref{lem Zmod = Z (1+O(hbar))} above implies
    \begin{equation}
        Z^\mr{glob}=\ul{Z}\big|_{\sfa=0}\,(1+O(\hbar)).
    \end{equation}

%\marginpar{\bl not sure about signs}
\begin{prop}\label{prop: Zglob formula}
    One has
    %\marginpar{\bl Dec 23: signs corrected}
    \begin{equation}\label{Zglob formula}
    \begin{aligned}
        Z^\mr{glob}&=\sum_{k=0}^{\dim \M'} \frac{(-i\hbar)^k}{k!} \left\langle \frac{\partial}{\partial \sfa^2},\frac{\partial}{\partial [\delta A]}\right\rangle^k \til{Z}^{(k)}\Big|_{\sfa^1=0}\\
        &=\left(e^{-i\hbar \left\langle \frac{\partial}{\partial \sfa^2},\frac{\partial}{\partial [\delta A]}\right\rangle} \til{Z} \right) \Big|_{\sfa^1=[\delta A]=0}
        .
        \end{aligned}
    \end{equation}
    Here $\dim \M'=\dim H^1_A$ is the dimension of the connected component of $\M'$ containing $[A]$; $\til{Z}^{(k)}$ is the $k$-form component of $\til{Z}$, as a form on $\M'$.
\end{prop}
\begin{proof}
In the notations of the proof of Theorem \ref{thm 5.15}, we have
\begin{equation}\label{prop 5.17 proof eq1}
    Z^\mr{glob}=\mathsf{p}'\til{Z}=\sum_{k\geq 0} \mathsf{p}(i\hbar\Delta_\sfa \mathsf{K})^k \til{Z}.
\end{equation}
The chain homotopy $\mathsf{K}$ increases the polynomial degree in $\sfa^1$, and in the lowest degree in $\sfa^1$ is given by 
%\marginpar{\bl Dec 23: sign corrected}
\begin{equation}\label{prop 5.17 proof eq2}
    \mathsf{K}\omega|_{\sfa^1\ra 0}  \sim -\frac{1}{\deg\omega} \left\langle \sfa^1,\frac{\partial}{\partial [\delta A]}\right\rangle \omega|_{\sfa^1=0},
\end{equation}
cf. the homotopy $\delta^*$ in \cite[Section 2]{Bonechi2012}. On the other hand $\Delta_\sfa$ lowers the degree in $\sfa^1$ by one and $p$ sets $\sfa^1$ to zero. So, in the r.h.s. of (\ref{prop 5.17 proof eq1}), only the constant term in $\sfa^1$ contributes, and for the purpose of evaluating the r.h.s., $\mathsf{K}$ can be replaced by its asymptotics (\ref{prop 5.17 proof eq2}). Formula (\ref{Zglob formula}) follows.
\end{proof}

% \begin{cor}
%     One has
%     \begin{equation}
%         Z^\mr{glob}=\ul{Z}\big|_{\sfa=0}\,(1+O(\hbar)).
%     \end{equation}
% \end{cor}

\begin{cor}\label{cor: Tphi^* Zglob=Z+Delta(...)}
    Global partition function $Z^\mr{glob}$ is related to the perturbative partition function $\ul{Z}$ by
    \begin{equation}
        (T\ul\varphi^* Z^\mr{glob})(A,\sfa)=\ul{Z}_A(\sfa)+i\hbar\Delta_\sfa\rho(A,\sfa).
    \end{equation}
\end{cor}
\begin{proof}
    This is an immediate consequence of (\ref{thm 5.17 proof eq1}) restricted to form degree zero along $\M'$.
\end{proof}

\begin{rem}[A path integral formula for $Z^\mr{glob}$]\label{rem: PI formula for Zglob}
    Formula (\ref{Zglob formula}) can be seen as the perturbative evaluation of the following path integral:\footnote{
    The sign convention for the Berezinian in the outer integral is: $\mc{D}\sfa^2\,\mc{D}\zeta=\prod_a \mc{D}\sfa^2_a\,\mc{D}\zeta^a$, for $(\sfa^2_a,\zeta^a)$ coordinates on $H^2_A[-1]\oplus H^1_A[1]$ associated to some basis $\chi_a$ in $H^1_A$ and the dual basis $\chi^a$ in $H^2_A$.
    }
%    \marginpar{\bl Nov 14: this needs to be corrected: we are dropping some diagrams}
%\marginpar{\bl Dec 23: swapped the order of factors in the Berezinian and the integration domain and added a footnote. Dec 25: swapped back but changed the sign of the pairing in the exp.}
%\marginpar{\bl Feb 22: removed RT from f-la and explanation of RT below}
\begin{multline}\label{Zglob path integral formula}
    Z^\mr{glob}(A)=%\mr{RT}
    \int_{H^2_A[-1]\oplus H^1_A[1]
    %\bl H^1_A[1]\oplus H^2_A[-1]
    }
    \mc{D} \sfa^2\, \mc{D} \zeta 
    %{\bl \mc{D} \zeta\, \mc{D} \sfa^2 }
    \int_{\mc{L}=\Omega_{d_A^*\mr{-ex}}[1]} \mc{D} \alpha_\mr{fl}\\
    \exp \frac{i}{\hbar}\Bigg(S_{CS}(A+i(\sfa^2)+\alpha_\mr{fl})+
    \langle \zeta,\sfa^2\rangle 
    \\
    + \int_M \frac12 \langle \alpha_\mr{fl},d_A G \mr{ad}^*_{i(\zeta)} \alpha_\mr{fl} \rangle +\langle \alpha_\mr{fl},d_A G \mr{ad}^*_{i(\zeta)} i(\sfa^2) \rangle\Bigg).
\end{multline}
The last two terms can also be written as 
\begin{equation}\label{Zglob interaction vertex}
\int_M \frac12 \langle i(\sfa^2)+\alpha_\mr{fl}, H_{\delta A'=i(\zeta)}(i(\sfa^2)+\alpha_\mr{fl}) \rangle,
\end{equation}
with $H_{\delta A'}$ as in (\ref{H delta A'}). Note that in the integral formula (\ref{Zglob path integral formula}), $\sfa^2$ and $\zeta=[\delta A]$ become dynamical variables (integrated over).
\begin{figure}
    \centering
    \includegraphics[scale=0.75]{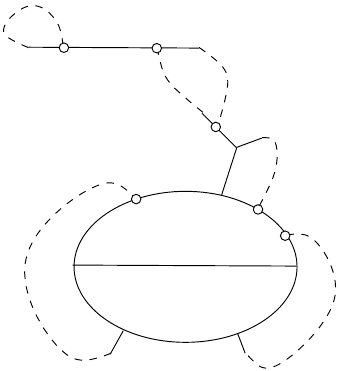}
    \caption{
    Example of a Feynman graph for $Z^\mr{glob}$ (\ref{Zglob path integral formula}). Dashed edges correspond to $\sfa^2$--$\zeta$ propagators; white vertices correspond to (\ref{Zglob interaction vertex}). Selection rules: $\leq 2$ white vertices on a solid edge, $\leq \dim \M'$ dashed edges in total. A solid edge not incident to Chern-Simons cubic vertices should have exactly two white vertices (as in the top part of the picture).
    }
    \label{fig:Zglob graph}
\end{figure}

%The symbol $\mr{RT}$ ``remove trees'' in (\ref{Zglob path integral formula}) stands for removing from the perturbative expansion Feynman graphs with the property that if one cuts all dashed edges ($\sfa^2-\zeta$ propagators), one of the remaining connected components is a tree with exactly one $H$-vertex.

The Feynman graph expansion of $Z^\mr{glob}$ has the form
\begin{multline}\label{Zglob Feynman diagram expansion}
    Z^\mr{glob}(A)=
    e^{\frac{i}{\hbar}S_{CS}(A)}e^{\frac{\pi i}{4}\psi_A}\tau_A^{1/2}\Big(1+
    \vcenter{\hbox{\includegraphics[scale=0.4]{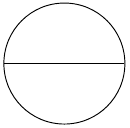}}}+
    \vcenter{\hbox{\includegraphics[scale=0.55]{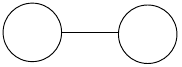}}}\\
    +
     \vcenter{\hbox{\includegraphics[scale=0.4]{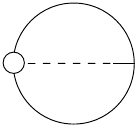}}}+
    \vcenter{\hbox{\includegraphics[scale=0.55]{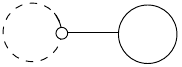}}}+
     \vcenter{\hbox{\includegraphics[scale=0.4]{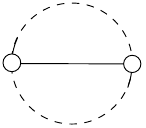}}}+
    \vcenter{\hbox{\includegraphics[scale=0.55]{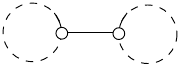}}}+
    \cdots
    \Big).
\end{multline}
The graphs shown contribute in the order $O(\hbar)$ and $\cdots$ is of order $\geq 2$ in $\hbar$, with graphical conventions as in Figure \ref{fig:Zglob graph}. 
% {\bl In particular, we have
%     \begin{equation}
%         Z^\mr{glob}=\ul{Z}\big|_{\sfa=0}\,(1+O(\hbar)).
%     \end{equation}
% }
\end{rem}

\subsection{Metric dependence of the global partition function}\label{ss Zglob metric independence}

To define $Z^\mr{glob}$ we needed to choose a metric $g$ on $M$. In this section we analyze the dependence of $Z^\mr{glob}$ on this metric. We have shown previously that $Z^\mr{glob}$ can be interpreted either as a top form on $\M'$ or, equivalently, a half-density on $T^*[-1]\M'$ that does not depend on the fiber coordinates (has degree zero). As such it is trivially closed (w.r.t. de Rham differential or BV Laplacian).  

Let 
\begin{equation}
    Z^\mr{glob,ren}_{g,\phi}(A)\colon= e^{\frac{i}{\hbar}c(\hbar)\frac{S_\mr{grav}(g,\phi)}{2\pi}} Z^\mr{glob}_g(A)
\end{equation}
be the renormalized global partition function, with $c(\hbar)$ as in (\ref{eq: ren Z}) and $\phi$ a framing of $M$.

%\marginpar{\bl Nov 14: metric-extended smoothness assumption removed}
%In this section we are assuming metric-extended smoothness, cf. Definition \ref{def: metric-extended smoothness}.

The main result of this section is the following theorem. 
\begin{thm}\label{thm: coho class invariant}
    The cohomology class of $Z_{g,\phi}^\mr{glob,ren}$ is independent of the choice of metric $g$.
    %used to define $Z^\mr{glob}$. 
\end{thm}
\begin{defn}
    We call the cohomology class $[Z^\mr{glob,ren}_{g,\phi}] \in H^\mr{top}(\M')$ the \emph{Chern-Simons volume class} on $\M'$. 
    %\marginpar{\bl Sep 10: maybe ``volume class''?}
\end{defn}

% Consider the following object:\marginpar{\bl Nov 6: corrected sign, $e^{-\frac{i}{\hbar}\langle [\delta A],\sfa \rangle}\ra e^{\frac{i}{\hbar}\langle [\delta A],\sfa \rangle}$}
% \begin{multline}
%     \overline{Z}^\mr{ren}\colon= e^{\frac{i}{\hbar}c(\hbar)\frac{S_\mr{grav}(g,\phi)}{2\pi}}\cdot \pi_*\Big( 
%     \check{Z}\cdot e^{\frac{i}{\hbar}\langle [\delta A],\sfa \rangle}\Big|_{A=A',\delta A=\delta A'\mr{\; harmonic}}
%     \Big) 
%     \\ 
%     \in \Omega^\bt(\mr{Met}\times \M', \mr{Dens}^{\frac12,\mr{formal}}(H_A[1]))
% \end{multline}
% where $\check{Z}$ is as in (\ref{Z check}) and $\pi_*$ stands for the passage to the quotient by gauge transformations, cf. \ref{pi projection}. 
% %the r.h.s. is considered modulo gauge transformations

Consider $\overline{Z}^\mr{ren}$ defined by (\ref{olZren def}) 
and let 
\begin{equation}
    \check{Z}^\mr{glob,ren}\colon=\mathsf{p}' \overline{Z}^\mr{ren} 
    \quad \in \Omega^\bt(\mr{Met})\otimes\mr{Dens}^{\frac12}(T^*[-1]\M')
    %\cong \Omega^\bt(\mr{Met}\times \M')
\end{equation}
with $\mathsf{p}'$ as in (\ref{(i',p',K')}). Note that $\check{Z}^\mr{glob,ren}$ is an extension of $Z^\mr{glob,ren}$ to a nonhomogeneous form on $\mr{Met}$; we denote its $k$-form component by $\check{Z}^{\mr{glob,ren}(k)}$. 
%be the component of $\check{Z}^\mr{glob}$ that is a $k$-form on $\mr{Met}$.

\begin{prop}\label{prop: properties of Zcheckglob}
$\check{Z}^\mr{glob,ren}$ satisfies the following. 
\begin{equation}\label{Zcheckglob killed by delta_g+Delta}
(\delta_g-i\hbar \Delta)\check{Z}^\mr{glob,ren}=0,
\end{equation}
\begin{equation}\label{delta_g Zglob = Delta(...)}
\delta_g Z^\mr{glob,ren} = i\hbar\Delta \check{Z}^{\mr{glob,ren}(1)},
\end{equation}
%\marginpar{\bl Dec 23: changed sign}
\begin{equation}\label{Zcheckglob formula}
\begin{aligned}
    \check{Z}^\mr{glob,ren} &= \sum_{k=0}^{\dim \M'}\frac{(-i\hbar)^k}{k!} \left\langle \frac{\partial}{\partial \sfa^2},\frac{\partial}{\partial [\delta A]} \right\rangle^k \overline{Z}^{\mr{ren}(\bt,k)}\Big|_{\sfa^1=0} \\
    &= \left(e^{-i\hbar \left\langle \frac{\partial}{\partial \sfa^2},\frac{\partial}{\partial [\delta A]} \right\rangle}\overline{Z}^\mr{ren}\right)\Big|_{\sfa^1=[\delta A]=0}.
\end{aligned}
\end{equation}
Here $\Delta$ is the BV Laplacian on half-densities on $T^*[-1]\M'$. Superscript $(\bt,k)$ means the component of de Rham degree $k$ along $\M'$ (and arbitrary degree along $\mr{Met}$).
\end{prop}
Note that (\ref{delta_g Zglob = Delta(...)}) immediately implies Theorem \ref{thm: coho class invariant}.

%\marginpar{\bl Nov 14: the proof is rewritten}
\begin{proof}
First, recall from Lemma \ref{lemma 5.11} that $\overline{Z}^\mr{ren}$ satisfies
\begin{equation}
    (\delta_g+\nabla^G+\{\mu,-\}+\Delta_\sfa\mu-i\hbar \Delta_\sfa) \overline{Z}^\mr{ren}=0.
\end{equation}
% note that, by restricting (\ref{dQME on Zcheck}) to the diagonal $A=A'$ and setting $\delta A=\delta A'$ to be harmonic, one obtains, under the metric-extended smoothness assumption, the equation
% \begin{equation}
%     (\delta_g+\nabla^G-i\hbar \Delta_\sfa) \overline{Z}^\mr{ren}=0.
% \end{equation}
Next, consider the differential perturbation of the contraction
\begin{multline}
    (\mathsf{i},\mathsf{p},\mathsf{K})\colon (\Omega^\bt(\mr{Met}\times \M',\mc{D},\nabla^G)  \\ \leadsto 
    (\Omega^\bt(\mr{Met})\otimes \mr{Dens}^{\frac12}(T^*[-1]\M'),\mr{zero\; differential}),
\end{multline}
deforming $\nabla^G\ra \nabla^G+\underbrace{\delta_g+\{\mu,-\}+\Delta_\sfa\mu-i\hbar \Delta_\sfa}_\varkappa$. By homological perturbation lemma, we obtain the deformed contraction
% Next, consider the contraction (\ref{(i',p',K')}). Note that, by metric-extended smoothness assumption, all maps involved do not depend on the metric.  
% Tensoring (\ref{(i',p',K')}) with the de Rham complex of $\mr{Met}$, 
% %and deforming the differential from $\nabla^G+i\hbar\Delta_\sfa$ to $\delta_g+\nabla^G+i\hbar\Delta_\sfa$, 
% we obtain the contraction
\begin{multline}\label{(i'',p'',K'')}
    (\mathsf{i}'',\mathsf{p}''=\mathsf{p}',\mathsf{K}'')\colon (\Omega^\bt(\mr{Met}\times \M',\mc{D}),
    \nabla^G+\varkappa
    )  \\ \leadsto 
    (\Omega^\bt(\mr{Met})\otimes \mr{Dens}^{\frac12}(T^*[-1]\M'),\delta_g-i\hbar\Delta).
\end{multline}
%The fact that the deformation of the differential does not induce a deformation of the SDR data $(i',p',K')$ is a consequence of metric-extended smoothness (which implies that $\delta_g$ commutes with the relevant maps). 
We remark that:
\begin{enumerate}[(i)]
    \item The induced differential in (\ref{(i'',p'',K'')}) is 
    $\sum_{k\geq 0}\mathsf{p}\varkappa (-\mathsf{K}\varkappa)^k \mathsf{i}$. Since the image of $\mathsf{i}$ is in zero-forms on $\M'$, $\mathsf{K}$ reduces the form degree on $\M'$ by one and $\varkappa$ does not change the form degree on $\M'$, all terms with $k>0$ vanish and the induced differential is $\mathsf{p}\varkappa \mathsf{i}= \delta_g-i\hbar\Delta$.
    \item One has $\mathsf{p}''=\mathsf{p'}$ -- the chain projection in (\ref{(i',p',K')}) extended trivially along $\mr{Met}$. Indeed, one has
    \begin{equation}\label{p''}
    \mathsf{p}''=\sum_{k\geq 0}\mathsf{p} (-\varkappa \mathsf{K})^k.
    \end{equation}
    Since $\mathsf{p}$ evaluates the constant term in $\sfa^1$, $\mathsf{K}$ increases the degree in $\sfa^1$ and $\varkappa=-i\hbar\Delta_\sfa$ plus terms that do not decrease the degree in $\sfa^1$, in (\ref{p''}) one can replace $\varkappa$ with $-i\hbar\Delta_\sfa$, leading to $\mathsf{p}''=\mathsf{p}'$.
\end{enumerate}

Since $\mathsf{p}'$ is a chain map, it sends the cocycle $\overline{Z}^\mr{ren}$ of the complex upstairs in (\ref{(i'',p'',K'')}) to a cocycle $\check{Z}^\mr{glob,ren}$ of the complex downstairs. This proves (\ref{Zcheckglob killed by delta_g+Delta}).

Equation (\ref{delta_g Zglob = Delta(...)}) is the restriction of (\ref{Zcheckglob killed by delta_g+Delta}) to form degree $1$ on $\mr{Met}$.

Formula (\ref{Zcheckglob formula}) for $\check{Z}^\mr{glob,ren}$ is proven similarly to Proposition \ref{prop: Zglob formula}.
\end{proof}

\begin{rem}
    The path integral formula for $Z^\mr{glob}$ from Remark \ref{rem: PI formula for Zglob} extends -- via (\ref{Zcheckglob formula}) -- to the extended global partition function $\check{Z}^\mr{glob,ren}$ as follows:
    %\marginpar{\bl Nov 23: PI f-la corrected -- $\sfa^2=\sfa^2_0+\sfa^2_\mr{fl}$.}
    %\marginpar{\bl Dec 23: swapped the order of factors in the Berezinian/ integration domain. Dec 25: swapped back but changed $\langle \sfa^2,\zeta\rangle\ra \langle \zeta,\sfa^2 \rangle$ in the exp}
    %\marginpar{\bl Feb 22: removed RT and explanation of RT}
\begin{multline}\label{Zcheckglob path integral formula}
    \check{Z}^\mr{glob,ren}(A,\sfa^2_0)=
    e^{\frac{i}{\hbar}c(\hbar)\frac{S_\mr{grav}(g,\phi)}{2\pi}}
    %\cdot\mathrm{RT} 
    \int_{H^2_A[-1]\oplus H^1_A[1]
    %\bl H^1_A[1]\oplus H^2_A[-1]
    }
    \mc{D} \sfa^2_{\mr{fl}}\, \mc{D} \zeta 
    %{\bl \mc{D} \zeta\, \mc{D} \sfa^2_\mr{fl}}
    \int_{\mc{L}=\Omega_{d_A^*\mr{-ex}}[1]} \mc{D} \alpha_\mr{fl}\\
    \exp \frac{i}{\hbar}\Big(S_{CS}(A+i(\sfa^2)+\alpha_\mr{fl})+\langle \zeta,\sfa^2_{\mr{fl}}\rangle \\
    + \int_M 
    \underbrace{\frac12 \langle \alpha_\mr{fl},d_A G \mr{ad}^*_{i(\zeta)} \alpha_\mr{fl} \rangle +\langle \alpha_\mr{fl},d_A G \mr{ad}^*_{i(\zeta)} i(\sfa^2) \rangle}_{-\frac12\langle i(\sfa^2)+\alpha_\mr{fl},H_{\delta A'=i(\zeta)}(i(\sfa^2)+\alpha_\mr{fl}) \rangle}
    \underbrace{-\frac12 \langle\alpha_\mr{fl},\lambda_{\delta g} \alpha_\mr{fl}\rangle - 
    \langle\alpha_\mr{fl},\lambda_{\delta g} i(\sfa^2)\rangle}_{-\frac12\langle i(\sfa^2)+\alpha_\mr{fl},H_{\delta g}(i(\sfa^2)+\alpha_\mr{fl}) \rangle}
\Big)\Bigg|_{\sfa^2=\sfa^2_0+\sfa^2_\mr{fl}}
\end{multline}
with $H_{\delta A'}$ and $H_{\delta g}$ as in (\ref{H delta A'}), (\ref{H delta g}); $\sfa^2_0$ is interpreted as a vector in $T^*_A[-1]\M'$.

In particular, the Feynman diagram expansion of the generator $\check{Z}^{\mr{glob,ren}(1)}$ in the r.h.s. of (\ref{delta_g Zglob = Delta(...)}) is
\begin{multline}\label{Zglobren (1) Feynman diagram expansion}
    \check{Z}^{\mr{glob,ren}(1)}=\\
    e^{\frac{i}{\hbar}c(\hbar)\frac{S_\mr{grav}(g,\phi)}{2\pi}}\cdot
    e^{\frac{i}{\hbar}S_{CS}(A)}e^{\frac{\pi i}{4}\psi_{A,g}}\tau_A^{1/2}\left(
    \vcenter{\hbox{\includegraphics[scale=0.45]{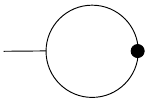}}}+
    \vcenter{\hbox{\includegraphics[scale=0.45]{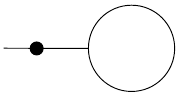}}}+
     \vcenter{\hbox{\includegraphics[scale=0.45]{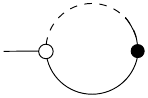}}}+
    \vcenter{\hbox{\includegraphics[scale=0.45]{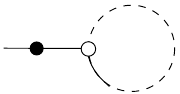}}}+
    \cdots
    \right).
\end{multline}
Here the graphical conventions are as in Figure \ref{fig:Zglob graph}; loose half-edges are decorated by $i(\sfa^2_0)$; black circle vertex is decorated by $H_{\delta g}$. The graphs shown contribute in zeroth order in $\hbar$ and $\cdots$ is of order $\geq 1$ in $\hbar$.

%The symbol $\mr{RT}$ in (\ref{Zcheckglob path integral formula}) stands for removing Feynman graphs such that if all dashed edges are deleted, there is a connected component which is a tree with one white vertex or one black vertex.
\end{rem}

% Theorem \ref{thm: coho class invariant} follows from the following result. 
% \begin{prop}
%     We have 
%     \begin{equation}
%         \delta_g\underline{Z}^\mr{glob} = -i\hbar\Delta_{T^*[-1]\M'}R
%     \end{equation}
% %    \marginpar{\bl Sep4: I am not sure $R$ is so simple here, since we also use the comparison of $Z^{glob}$ and $Z^{pert}$ which has a complicated $R$}
%     where $R$ is given by the sum of connected Feynman diagrams with a single leaf,  which is labeled by $\mathbb{I}_{\delta g}$. 
% \end{prop}
% \begin{proof}
%     This follows by restricting to the diagonal in \eqref{eq: metric dependence} and noticing both sides are equivariant in the gauge group action. 
% \end{proof}

\subsubsection{Relation to the asymptotic expansion conjecture}
Now fix $\g = \mathfrak{su}(N)$ and denote $\tau_{k,N}$ the $(SU(N)$-Reshetikhin-Turaev invariants \cite{Reshetikhin1991}.
We recall the statement of the asymptotic expansion conjecture, which we cite from \cite[Conjecture 7.7]{Andersen2002}
\begin{conj}
 Let $\{c_0,\ldots,c_m\}$ be the Chern-Simons invariants of $M$. Then there exist  $d_j \in \mathbb{Q}, \tilde{I}_j 
\in \mathbb{Q}/\mathbb{Z}, v_j \in \mathbb{R}_+$ and $a_j^e$ for $j = 0,\ldots,m$ and $e \in \mathbb{N}$ such that for $r = k + h^\vee$:

    \begin{equation} 
        \tau_{k,N} \underset{k\ra\infty}{\sim}\sum_{j=0}^m e^{\frac{i r c_j}{ 2\pi}}r^{d_j}e^{\frac{i \pi}{4} \tilde{I}_j}v_j 
        %\left(1 +
        \exp\sum_{e=1}^\infty a_j^e \left( \frac{r}{2\pi}\right)^{-e}. %\right)
    \end{equation}
\end{conj}
Other forms of this conjecture have appeared in the literature, for instance in  \cite[Section 6]{Reshetikhin2010}. 
We conjecture that if $c_j$ comes from a union of smooth, irreducible components of the moduli space, then the $j$-th summand above coincides with  the integral of Chern-Simons volume class over that preimage. More precisely, fix $\phi$ to be the canonical 2-framing of $M$ and denote $\M_j = S_{CS}^{-1}(c_j)$. If $\M_j \subset \M'$,  then we conjecture that $a_j^e$ is given by the  coefficient of $\hbar^e$ (contribution of connected $(e+1)$-loop graphs) in  $\log\int_{\M_j}Z^\mr{glob,ren}_{g,\phi}$. 
This assumes that at higher loop orders one has to identify $\hbar =
%(k + h^\vee)^{-1}
\frac{2\pi}{k + h^\vee}
$ (as discussed by Axelrod-Singer \cite[Section 6]{Axelrod1991}).

% under the identification $\hbar = (k + h^\vee)^{-1}$, we have
% \begin{equation}
%     \int_{\M'_j}Z^\mr{glob,ren}_{g,\phi} = e^{2\pi i r c_j}r^{d_j}e^{\pi i \tilde{I}_j/4}v_j \left(1 + \sum_{e=1}^\infty a_j^e r^{-e} \right)
% \end{equation}
% This gives explicit proposals for all the constants in the conjecture. 
An interesting class of examples where this conjecture could potentially be checked are Seifert fibered homology spheres. For these, the only reducible connection is the trivial one, all other components of the moduli space are closed manifolds (\cite{Fintushel1990}). For this class, the asymptotic expansion conjecture has recently been proven in \cite{Andersen2025}, where the authors show that the asymptotic expansion is given in terms of integrals over the smooth components of the moduli space. Comparison with this and other results will be addressed in future work.

% \subsubsection{Extended global partition function}
% One can extend also extend $\underline{Z}^\mr{glob}$ to a inhomogeneous form $\check{\underline{Z}} \in \Omega^{\bullet}(\mr{Met}, \operatorname{Dens}^\frac{1}{2}(T^*[-1]\M')$ of total degree $0$ such that 
% \begin{equation}
%     (\delta_g + i\hbar\Delta_{T^*[-1]\M'})\underline{Z}^\mr{glob} = 0. 
% \end{equation}

% {\color{gray}
% OLD STUFF
% \begin{lem}[Results about metric dependence, in progress]
% Let $X^m[A_0,\alpha] = \delta_m Z^m[A_0,\alpha]$.
% \begin{enumerate}

% \item At smooth points $X$ is linear in $\alpha^{(2)}$ (Cattaneo-Mnev) 
% \item We have $X[A_0,\alpha^{(1)},\alpha^{(2)}] = T(\phi^m)^*X[A_0,\alpha^{(1)}=0,\alpha^{(2)}]$. 
% \item $X^m$ is equivariant in the same sense as $Z^m$ and descends as a half-density $[X^m]$ to $T^*[1]\mathcal{M}$, linear in the fiber coordinates. 
% \item $\delta_m [Z^m] = \Delta_{T^*[1]\mathcal{M}}X^m$. 
% \end{enumerate}

% \begin{cor}
% The partition function descends to the shifted cotangent bundle of the smooth part of the moduli space $T^*[-1]\M^{\irr,\sm}$ and defines a half-density $[Z^m]$ there which is independent of the fiber coordinates. 
% \end{cor}

% \begin{thm}[Main result]
% The BV Chern-Simons partition function gives rise to a half-density on $T^*[-1]\M^{\irr,\sm}$. Its BV cohomology class is independent of the gauge-fixing metric.
% \end{thm}
% \end{lem}

% }
\appendix

\section{SDR data and homological perturbation lemma}\label{app: SDR}
%\marginpar{\bl new Appendix, Aug 25}
Here for reader's convenience we review the definition of SDR (strong deformation retraction)  data and the homological perturbation lemma, both well-known in the literature -- see e.g. \cite{gugenheim1989perturbation}, \cite{Crainic2004}.

\begin{defn}
Let $(V^\bt,d_V)$ and $(W^\bt,d_W)$ be a pair of cochain complexes. \emph{SDR data} (or an \emph{$(i,p,K)$ triple}) is a triple of maps 
\begin{equation}
    i\colon W^\bt \ra V^\bt,\quad p\colon V^\bt\ra W^\bt,\quad K\colon V^\bt\ra V^{\bt-1}
\end{equation}
such that:
\begin{itemize}
    \item $i$ and $p$ are chain maps: $d_V i=i d_W$, $d_W p=p d_V$.
    \item $i$ is an inclusion and $p$ a projection, satisfying $pi=\mr{id}_W$.
    \item $K$ is a chain homotopy between $ip$ and $\mr{id}_V$: $d_V K+K d_V=\mr{id}-ip$.
    \item The following side conditions hold: $K^2=Ki=pK=0$.
\end{itemize}
\end{defn}
In particular, existence of SDR data implies that complexes $(V^\bt,d_V)$ and $(W^\bt,d_W)$ are quasi-isomorphic (with $i$ and $p$ quasi-isomorphisms); one calls $W^\bt$ a \emph{deformation retract} of $V^\bt$. 

An important special case is when $(W^\bt,d_W)=(H^\bt(V),0)$ is the cohomology of $V$. 

A choice of SDR data induces a Hodge-like decomposition
\begin{equation}\label{App A Hodge decomp}
    V=i(W)\oplus \Big(\mr{im}(d_V)\cap \mr{ker}(p)\Big) \oplus \mr{im}(K).
\end{equation}
Here $d_V$ acts on the first term and maps the third term to the second isomorphically, with $K$ the inverse.

\begin{lem}[Homological perturbation lemma]\label{lem: HPL}
Let $(V^\bt,d_V)$ and $(W^\bt,d_W)$ be a pair of complexes with SDR data $(i,p,K)$.
Consider a perturbation of the differential on $V$, $d_V \ra \til{d}_V=d_V+\delta$, for some $\delta\colon V^\bt\ra V^{\bt+1}$ such that $(d_W+\delta)^2=0$. Then the perturbed complex $(V^\bt,d_V+\delta)$ is quasi-isomorphic to $(W^\bt,\til{d}_W %=d_W+\til\delta
)$ with SDR data
$(\til{i},\til{p},\til{K})$, where
\begin{eqnarray}
    \til{d}_W&=& d_W+p\delta i-p\delta K \delta i + p\delta K\delta K \delta i-\cdots,\\
    %\til\delta&=& p\delta i-p\delta K \delta i + p\delta K\delta K \delta i-\cdots,\\
    \til{i}&=& i-K\delta i+K\delta K\delta i-\cdots, \\
    \til{p}&=& p-p\delta K + p\delta K\delta K - \cdots, \\
    \til{K}&=& K-K\delta K+K\delta K\delta K-\cdots,
\end{eqnarray}
under the assumption that the geometric progressions above converge.
\end{lem}

\subsection{First-order deformations of SDR data}
Consider a deformation retraction of a cochain complex $(V^\bt,d_V)$ onto its cohomology $(W^\bt=H^\bt(V),0)$ and fix SDR data $(i,p,K)$. The Hodge-like decomposition (\ref{App A Hodge decomp}) in this case is
$V=i(W)\oplus V_{d\mr{-exact}}\oplus V_{K\mr{-exact}}$.
% \begin{lem}\footnote{See \cite{Mnev2008}, \cite{Cattaneo2008}, \cite{cattaneo2020cellular}.}
%     A general infinitesimal deformation of $(i,p,K)$, in the class of SDR data where $p|_{V_{d\mr{-closed}}}$ is the standard projection of closed elements to cohomology classes, has the form
%     \begin{eqnarray}
%         i & \ra & i-\varepsilon d_V \mathbb{I}i , \\
%         p & \ra & p-  \varepsilon p\mathbb{P} d_V, \\
%         K &\ra & K+\varepsilon ([d_V,\Lambda] + ip\mathbb{P} +\mathbb{I} i p),
%     \end{eqnarray}
%     with $\mathbb{I},\mathbb{P},\Lambda$ arbitrary maps
%     \begin{eqnarray}
%         \mathbb{I}\colon V^\bt
%         &\ra& V^{\bt-1}_{K\mr{-exact}},\\
%          \mathbb{P}\colon V^\bt_{d\mr{-exact}}&\ra& V^{\bt-1},\\
%          \Lambda\colon V^\bt_{d\mr{-exact}} &\ra& V^{\bt-2}_{K\mr{-exact}}
%     \end{eqnarray}
%     and $\varepsilon$ the deformation parameter.
% \end{lem}
\begin{lem}\footnote{See \cite{Mnev2008}, \cite{Cattaneo2008}, \cite{cattaneo2020cellular}.} \label{lem: SDR deformations}
    A general infinitesimal deformation of $(i,p,K)$, in the class of SDR data where $p|_{V_{d\mr{-closed}}}$ is the standard projection of closed elements to cohomology classes, has the form
    \begin{eqnarray}
        i & \ra & i-\varepsilon d_V \mathsf{I} , \\
        p & \ra & p-  \varepsilon \mathsf{P} d_V, \\
        K &\ra & K+\varepsilon ([d_V,\Lambda] + i\mathsf{P} +\mathsf{I}  p),
    \end{eqnarray}
    with $\mathsf{I},\mathsf{P},\Lambda$ arbitrary maps
    \begin{eqnarray}
        \mathsf{I}\colon W^\bt
        &\ra& V^{\bt-1}_{K\mr{-exact}},\\
         \mathsf{P}\colon V^\bt_{d\mr{-exact}}&\ra& W^{\bt-1},\\
         \Lambda\colon V^\bt_{d\mr{-exact}} &\ra& V^{\bt-2}_{K\mr{-exact}}
    \end{eqnarray}
    and $\varepsilon$ the deformation parameter.
\end{lem}
For the applications of this paper, we parametrize the maps $\mathsf{I},\mathsf{P}$ above as 
\begin{equation}
    \mathsf{I}=\mathbb{I}i,\quad \mathsf{P}=p\mathbb{P},
\end{equation}
with
\begin{equation}
    \mathbb{I}\colon V^\bt \ra V^{\bt-1}_{K\mr{-exact}}, \quad 
    \mathbb{P}\colon V^\bt_{d\mr{-exact}}\ra V^{\bt-1}
\end{equation}
arbitrary maps.

\section{Variation of desynchronized Hodge SDR data}
%\section{Strong deformation retracts}
%\subsection{Definitions}

%\subsection{Variation of desynchronized SDR data}
%\marginpar{\bl Aug 25: think about sectioning and titles here}
In this section we consider the variation of the SDR data $(i_{A,A'},p_{A,A'},K_{A,A'})$ given by the Hodge decomposition associated to a pair of close flat connections $(A,A')$ and the metric $g$, in the direction of the three parameters $(A,A',g)$.  

Recall that the the metric induces on the complex of $\g$-valued differential forms the pairing
\begin{equation}
\langle \alpha, \beta \rangle_{\Omega^\bullet(M,\g)} = \int_M \langle \alpha, *\beta\rangle_\g
\end{equation}
and associated with it the operator $\dd^*_{A'}$, the formal adjoint of $\dd_{A'}$, the twisted, desynchronized Hodge-de Rham Laplacian $\Delta_{A,A'} := (\dd_{A} + \dd^*_{A'})^2$, the projection $P_{A,A'}$ to $\ker \Delta_{A,A'}$ along $\operatorname{im} \Delta_{A,A'}$, and the Green's operator of the Hodge-de Rham Laplacian, $G_{A,A'} = (\Delta_{A,A'} + P_{A,A'})^{-1}$, satisfying 
$\Delta_{A,A}G_{A,A'} = G_{A,A'}\Delta_{A,A'} = \mathrm{id} - P_{A,A'}$.
%We shall analyze how this SDR data varies when the flat connection $A_0$ is varied. 
%\marginpar{\bl Sep 15: $\iota\ra i$}
Recall that the SDR data  $(i_{A,A'},p_{A,A'},K_{A,A'})$ specified by the Hodge decomposition of the twisted de Rham complex is given by 
\begin{align} 
&i_{A,A'} \colon H^\bullet_A \to \Omega^\bullet,  &i_{A,A'}[\alpha] = P_{A,A'}\alpha \\
&p_{A,A'}\colon \Omega^\bullet \to H^\bullet_{A},  &p_{A,A'}\beta = \left[P_{A,A'}\beta\right] \\ 
&K_{A,A'} \colon \Omega^\bullet \to \Omega^{\bullet-1}, &K_{A,A'}= \dd^*_{A,A'}  G_{A,A'}
\end{align}
for $\alpha$ a $d_{A}$-closed form and $\beta$ any $\g$-valued form. 
\subsection{Changing the kinetic operator} 

\begin{lem}[Changing the kinetic operator] \label{lem:varAt} %\marginpar{need to assume path in smooth locus, then last term in \eqref{eq:varK} vanishes; also 1st term in the rhs should have a minus (put it there in red)}
Let $A_t\colon (-\epsilon,\epsilon) \to \Omega^1(M,\g)$ a path of smooth flat connections such that $(A_t,A')$ is close for all $t$ and $A_0 = A$. Denote $\dot{A_0} = \alpha \in \Omega^1_\mr{cl}(M,\g)$. 
%We denote by $\mathrm{ad}_\alpha^*$ the formal adjoint of $\mathrm{ad}_\alpha$ and $K^*_{A_0}= \dd_{A_0}  G_{A_0}$ the formal adjoint of $K_{A_0}$. 
Then we have
\begin{align}
\restr{\frac{d}{dt}}{t=0}\Delta_{A_t,A'} &= \left\lbrace \dd^*_{A'},\mathrm{ad}_\alpha\right\rbrace , \label{eq:varLaplacekin} \\
\restr{\frac{d}{dt}}{t=0} P_{A_t,A'} &= -  K_{A,A'} \mathrm{ad}_\alpha  P_{A,A'} - P_{A,A'}  \mathrm{ad}_\alpha    K_{A,A'}, \label{eq:varPharmkin} \\
\restr{\frac{d}{dt}}{t=0}{K}_{A_t,A'} &= - K_{A,A'}\mathrm{ad}_\alpha K_{A,A'}.  \label{eq:varKkin}
\end{align}
\end{lem}
\begin{proof}
Since $\dd_{A_t} = d + \mathrm{ad}_{A_t}$ %and $\dd^*_{A_t} = \dd^* + \mathrm{ad}^*_{A_t}$
, we have  $\restr{\frac{d}{dt}}{t=0}{\dd}_{A_t} = \mathrm{ad}_\alpha$. % and $\restr{\frac{d}{dt}}{t=0}\dd^*_{A_t} = \mathrm{ad}^*_\alpha.$
Equation \eqref{eq:varLaplacekin} then follows directly from rewriting the Laplacian as $\Delta_{A,A'} = \left\lbrace \dd^*_{A'},\dd_{A}\right\rbrace$. \\
To prove equation $\eqref{eq:varPharmkin}$, we differentiate the equations $\Delta_{A_t,A'}P_{A_t,A'} = P_{A_t,A'}\Delta_{A_t,A'} = P^2_{A_t,A'}-P_{A_t,A'} =0$. Differentiating the first one we obtain 
$$\Delta_{A,A'}\dot{P}_{A,A'} + \dot{\Delta}_{A,A'}P_{A,A'} = 0,$$
which yields, after composing with $G_{A,A'}$, 
\begin{align}(\mathrm{id} - P_{A,A'})\dot{P}_{A,A'} &= - G_{A,A'} \dot{\Delta}_{A,A'}P_{A,A'}  & \notag \\
 &= -G_{A,A'}\left(\left\lbrace \dd^*_{A'},\mathrm{ad}_\alpha\right\rbrace \right) P_{A,A'} &\text{(using \eqref{eq:varLaplacekin})} \notag\\
 &= -G_{A,A'}\dd^*_{A,A'}\mathrm{ad}_{\alpha}P_{A,A'}  &\text{(since } \dd^*_{A'}P_{A,A'}  = 0) \notag \\
 &= - K_{A,A'}\mathrm{ad}_{\alpha}P_{A,A'}  %K^*_{A_0}\mathrm{ad}^*_{\alpha}P_{A_0} 
 &\text{(since } \dd^*_{A'} %\text{and } \dd_{A'}
 \text{ commutes with } G_{A,A'}).\label{eq:1-PdotP}
\end{align}
Similarly, differentiating $P_{A_t,A'}\Delta_{A_t,A'} = 0$ we obtain 
\begin{equation} \dot{P}_{A,A'}(\mathrm{id} - P_{A,A'}) = - P_{A,A'}\mathrm{ad}_{\alpha}K_{A,A'}.\label{eq:dotP1-P}
\end{equation}
Differentiating $P^2_{A_t,A'} - P_{A_t,A'} = 0$ we obtain 
$$P_{A,A'}\dot{P}_{A,A'} + \dot{P}_{A,A'}P_{A,A'} - \dot{P}_{A,A'} = 0 $$ or 
\begin{equation}P_{A,A'}\dot{P}_{A,A'} = \dot{P}_{A,A'}(\mathrm{id} - P_{A,A'}) .\label{eq:PdotP}
\end{equation}
Finally, we can compute, using \eqref{eq:1-PdotP},\eqref{eq:dotP1-P}, \eqref{eq:PdotP}
$$\dot{P}_{A,A'} = (\mathrm{id} - P_{A,A'})\dot{P}_{A,A'} + P_{A,A'}\dot{P}_{A,A'} = -\left(K_{A,A'}\mathrm{ad}_\alpha %+ K^*_{A_0}\mathrm{ad}^*_\alpha
\right)P_{A,A'} - P_{A,A'}\left(\mathrm{ad}_\alpha K_{A,A'} \right)$$
which proves equation \eqref{eq:varPharmkin}. \\ 
Finally, let us prove \eqref{eq:varKkin}. Remember that we have $K_{A_t,A'} = \dd^*_{A'}   G_{A_t,A'}$ and hence 
$$\restr{\frac{d}{dt}}{t=0}{K}_{A_t,A'} =  \dd^*_{A'}   \left(\restr{\frac{d}{dt}}{t=0}{G_{A_t,A'}}\right).$$
 On the other hand, using that $G_{A_t,A'} = (\Delta_{A_t,A'} + P_{A_t,A'})^{-1}$, we have 
$$\restr{\frac{d}{dt}}{t=0}G_{A_t,A'} = - G_{A_0,A'}   \restr{\frac{d}{dt}}{{t=0}}(\Delta_{A_t,A'} + P_{A_t,A'})   G_{A_0,A'}.$$
Using \eqref{eq:varLaplacekin} and \eqref{eq:varPharmkin} we obtain 
\begin{multline}\label{variation of G wrt A_0}
\dot{G}_{A,A'} = - G_{A,A'}\big( \left\lbrace \dd^*_{A'},\mathrm{ad}_\alpha\right\rbrace  
- K_{A,A'} \mathrm{ad}_\alpha  P_{A,A'} - P_{A ,A'}  \mathrm{ad}_\alpha    K_{A,A'} \big)G_{A,A'}.
\end{multline}
After applying $\dd^*_{A'}$, only the first term survives and yields  
$$\dot{K}_{A,A'} = \dd^*_{A'}\dot{G}_{A,A'}-\dd^*_{A'}G_{A,A'} \left\lbrace \dd^*_{A,A'},\mathrm{ad}_\alpha\right\rbrace G_{A,A'} = - K_{A,A'}\mathrm{ad}_\alpha K_{A,A'}$$
since $\left(\dd^*_{A'}\right)^2 = d^*_{A'}P_{A,A'} = 0$ and $\dd^*_{A'}$ and $G_{A,A'}$ commute. 
\end{proof}
%{\red Remark on HPL (higher order in deformations here?}
\subsection{Changing the gauge-fixing operator}
\begin{prop}[Changing the gauge-fixing operator]\label{prop:varAtprime} %\marginpar{need to assume path in smooth locus, then last term in \eqref{eq:varK} vanishes; also 1st term in the rhs should have a minus (put it there in red)}
Let $A_t'\colon (-\epsilon,\epsilon) \to \Omega^1(M,\g)$ a path of smooth flat connections with $A'_0 = A'$ such that $(A,A'_t)$ is flat for all $t$. Denote $\dot{A_0} = \alpha \in \Omega^1_\mr{cl}(M,\g)$. We denote by $\mathrm{ad}_\alpha^*$ the formal adjoint of $\mathrm{ad}_\alpha$ and $K^*_{A,A'}= \dd_{A}  G_{A,A'}$. 
Then we have
\begin{gather}
\restr{\frac{d}{dt}}{t=0}\Delta_{A,A'_t} =  \left\lbrace \dd_{A},\mathrm{ad}^*_\alpha\right\rbrace, \label{eq:varLaplaceGF} \\
\restr{\frac{d}{dt}}{t=0} P_{A,A'_t} = - K^*_{A,A'} \mathrm{ad}^*_\alpha   P_{A,A'} - P_{A,A'}  \mathrm{ad}^*_\alpha   K^*_{A,A'},  \label{eq:varPharmGF} \\
\restr{\frac{d}{dt}}{t=0}{K}_{A,A'_t} =  \left[d_{A}, K_{A,A'}\mathrm{ad}_\alpha^*G_{A,A'}\right] + P_{A,A'}\mathrm{ad}^*_\alpha G_{A,A'} + G_{A,A'} \mathrm{ad}^*_\alpha P_{A,A'}.  \label{eq:varKGF}
\end{gather}
\end{prop}
\begin{proof}
Again, \eqref{eq:varLaplaceGF} follows directly from writing the Laplacian as $\Delta_{A_t} =\left\lbrace \dd_A,\dd^*_{A'_t}\right\rbrace$. In exactly the same way as above, we then obtain 
$$\dot{P}_{A,A'} = (\mr{id} - P_{A,A'})\dot{P}_{A,A'} + \dot{P}_{A,A'}(\mr{id}-P_{A,A'}) =  - K^*_{A,A'} \mathrm{ad}^*_\alpha   P_{A,A'} - P_{A,A'}  \mathrm{ad}^*_\alpha   K^*_{A,A'}, $$ proving \eqref{eq:varPharmGF}. Using \eqref{eq:varLaplaceGF},\eqref{eq:varPharmGF}, we obtain
\begin{multline}\label{variation of G wrt A prime}
\dot{G}_{A,A'} =  -G_{A,A'}\big( \left\lbrace \dd_{A'},\mathrm{ad}^*_\alpha\right\rbrace  
- K^*_{A,A'} \mathrm{ad}^*_\alpha  P_{A,A'} - P_{A ,A'}  \mathrm{ad}^*_\alpha    K^*_{A,A'} \big)G_{A,A'}.
\end{multline} 
After applying $\dd^*_{A'}$, the third term vanishes. The first one is

\begin{multline*}
\dd^*_{A'}G_{A,A'}\left\lbrace \dd_{A},\mathrm{ad}^*_\alpha \right\rbrace G_{A,A'} = \dd^*_{A'}G_{A,A'}\dd_{A}\mathrm{ad}^*_\alpha G_{A,A'} + K_{A,A'}\mathrm{ad}^*_\alpha G_{A,A'} \dd_{A} \\
= (\mathrm{id} - P_{A,A'} - \dd_{A,A'}\dd^*_{A,A'}G_{A,A'})\mathrm{ad}^*_\alpha  G_{A,A'} + K_{A,A'}\mathrm{ad}^*_\alpha G_{A,A'}\dd_{A,A'}  \\ 
= \mathrm{ad}^*_\alpha G_{A,A'} - P_{A,A'}\mathrm{ad}^*_\alpha G_{A,A'} - \left[\dd_{A,A'},K_{A,A'}\mathrm{ad}^*_\alpha G_{A,A'}\right].
\end{multline*}
The second one is
\begin{multline*}    \dd_{A'}^*G_{A,A'} K^*_{A,A'} \mathrm{ad}^*_\alpha   P_{A,A'} G_{A,A'} = (\mathrm{id} - P_{A,A'} - \dd_{A}\dd_{A'}^*) G_{A,A'}\mathrm{ad}^*_\alpha P_{A,A'}G_{A,A'} 
    \\ = G_{A,A'}\mr{ad}^*_\alpha P_{A,A'} - P_{A,A'} \mr{ad}^*_\alpha P_{A,A'} - \dd_{A}\dd^*_{A'}G_{A,A'}\mr{ad}^*_\alpha P_{A,A'}
\end{multline*} 
where we have used that $P_{A,A'}G_{A,A'} = P_{A,A'}$. 
Using Lemma \ref{lem:ad_alpha} below, the compositions $\dd^*_{A'}\mathrm{ad}^*_\alpha P_{A,A'} =  P_{A,A'}\mathrm{ad}^*_\alpha P_{A,A'} 0$.

The variation of $K_{A,A'}$ is finally given by 
\begin{multline*}\dot{K}_{A,A'} = \mr{ad}^*_\alpha G_{A,A'} + \dd^*_{A'}\dot{G}_{A,A'} \\
= P_{A,A'} \mr{ad}^*_\alpha G_{A,A'} + \left[\dd_{A,A'},K_{A,A'}\mathrm{ad}^*_\alpha G_{A,A'}\right] + G_{A,A'}\mr{ad}^*_\alpha P_{A,A'} 
\end{multline*}
proving \eqref{eq:varKGF}.
\end{proof}

\begin{lem}\label{lem:ad_alpha}
Let $A_0$ be a flat connection and $\alpha$ $d_{A_0}$-closed 1-form. Then 
\begin{itemize}
\item The map $\mathrm{ad}_\alpha$ maps $\dd_{A_0}$-closed forms to $\dd_{A_0}$-closed forms. 
\item If $[A_0]$ defines a smooth point in the moduli space, $\mr{ad}_\alpha$ maps all $\dd_{A_0}$-closed forms to $\dd_{A_0}$-\emph{exact} forms. 
\item Dually, $\mathrm{ad}^*_\alpha$ maps $\dd^*_{A_0}$-closed form to $\dd^*_{A_0}$-closed forms, and $\dd^*_{A_0}$-exact forms if $[A_0] $ is a smooth point. 
\end{itemize}
\end{lem}
\begin{proof}
The first point is obvious since the bracket is compatible with the differential. For the second point, notice that $\mathrm{ad}_\alpha$ always maps exact forms to exact forms. The smoothness assumption implies that for any close $A'$,  $l_2([\alpha],\bullet) = p_{A_0,A'}\alpha$ vanishes on harmonic forms, hence $\mathrm{ad}_\alpha$ maps harmonic forms into exact forms. To prove the last point, let $\beta$ be a coclosed form and $\gamma$ an exact form. Then 
$$\langle\mathrm{ad}^*_\alpha\beta, \gamma \rangle = \langle \beta,\mathrm{ad}_\alpha\gamma\rangle = 0$$ 
since coclosed forms are orthogonal to exact forms. Hence $\mathrm{ad}^*_\alpha\beta$ is also coclosed. If $[A_0]$ is smooth we can let $\gamma$ be any closed form, hence $\langle \mathrm{ad}_\alpha^*\beta\rangle$ is coexact in this case.  
\end{proof}
%\marginpar{Check rest of the section, probably these maps are not equal to B {\bl Aug 26:  for a fixed $A'$ these maps do coincide with $B$ since (227) reproduces $\nabla^\mr{Harm}$}}
%\marginpar{\bl Aug 26: this part is somewhat redundant and repeats the story of cohomology comparison maps}
We are also interested in the variations of $i_{A_t,A'}$ and $p_{A_t,A'}$. However, remember that $i_{A_t,A'} \colon H_{A_t}(M,\g) \to \Omega^\bullet(M,\g)$, so all the $i_{A_t,A'}$ are defined a priori on different spaces. Notice that we have the maps  
   \begin{equation}
 \begin{tikzcd}[every arrow/.append style={shift left}]
 {H^\bullet_{A_0}(M,\mathfrak{g})} \arrow{r}{p_{A_t,A'}i_{A_0,A'}} &{H^\bullet_{A_t}(M,\mathfrak{g})}  \arrow{l}{p_{A_0,A'} i_{A_t,A'}} 
 \end{tikzcd}\label{eq:Hmaps}
   \end{equation}
   
Using these maps to compare the different $i_{A_t,A'}$ and $p_{A_t,A'}$, we have the following result:
   \begin{lem}
   With notation as in Lemma \ref{lem:varAt}, we have 
   \begin{align}
     \restr{\frac{d}{dt}}{t=0} i_{A_t,A'}p_{A_t,A'}i_{A,A'} &= -K_{A,A'}\mathrm{ad}_\alpha  i_{A,A'} \label{eq:var iota}, \\ 
    \restr{\frac{d}{dt}}{t=0} p_{A,A'}i_{A_t,A'}p_{A_t,A'} &= -p_{A,A'}(\mathrm{ad}_\alpha K_{A,A'})\label{eq:var P}.
\end{align}    
   \end{lem}
\begin{proof}
Notice that $i_{A_t ,A'}p_{A_t,A'} = P_{A_t,A'}$. Then the formulae follow immediately from equation \eqref{eq:varPharmkin}. 
\end{proof}
In general, the maps \eqref{eq:Hmaps} are neither injective nor surjective. However, at smooth points, the following is true. 
\begin{prop}
Suppose $[A_0]$ is smooth. Then for small $t$, the maps \eqref{eq:Hmaps} are %inverse 
isomorphisms. 
\end{prop}
\begin{proof}
It is sufficient to show that for small $t$ the restriction of $P_{A_t}$ to $A_0$-harmonic forms is an isomorphism. This follows from Proposition \ref{prop: deformation harmonic}.
\end{proof}

We remark that the maps (\ref{eq:Hmaps}) coincide with the cohomology comparison maps $\mathfrak{B}_{A_t\la A_0,A'}$, $\mathfrak{B}_{A_0\la A_t,A'}$, cf. Section \ref{sss cohomology comparison map}. This follows from comparing (\ref{eq:var iota}) with the connection $\nabla^\mr{Harm}$ (\ref{nabla^harm infinitesimal horizontal transport}).

\begin{rem}
    The formulae (\ref{eq:varPharmkin}), \eqref{eq:varPharmGF} for $\frac{d}{dt}\Big|_{t=0}P_{A_t,A'_t}$ can also be obtained as follows:
    One has 
    \begin{equation}\label{P=lim s 3.1.1}
    P=\lim_{T\ra \infty} e^{-T\Delta_{A_t,A'_t}}, 
    \end{equation}
    hence
    \begin{equation}
         \dot{P}= \lim_{T\ra \infty} \int_0^T dt\, e^{-t \Delta} (-\dot\Delta)
        e^{-(T-t)\Delta}  .
    \end{equation}
The $T\ra \infty$ asymptotics of the integral in the r.h.s. comes from two regions (a) $t\ll T$, (b) $T-t \ll T$ -- neighborhoods of the endpoints of the integration interval $[0,T]$  (the bulk of the interval does not contribute since $P\dot\Delta P=0$):
    \begin{multline}\label{delta P s.3.1.2}
    \dot P
    =\left(\int_0^\infty dt\, e^{-t\Delta}\right)(-\dot\Delta) e^{-\infty\cdot \Delta}+e^{-\infty\cdot \Delta}(-\dot\Delta)\left(\int_0^\infty dt\, e^{-t\Delta}\right) \\=-G \dot\Delta P- P \dot\Delta G .
    %=-G (\dd^*_{A'} \mr{ad}_{\dot{A}_0}+\dd_A \mr{ad}^*_{\dot{A}_0})P-P(\mr{ad}_{\dot{A}_0} \dd^*_{A'}+\mr{ad}^*_{\dot{A}_0} \dd_A) G.
    \end{multline}
    Here we are suppressing the subscripts $A,A'$ for $P$, $\Delta$, $G$; $e^{-\infty \cdot\Delta}$ is a shorthand for the r.h.s. of (\ref{P=lim s 3.1.1}). 
    \begin{figure}[h]
        \centering
        \includegraphics[scale=0.5]{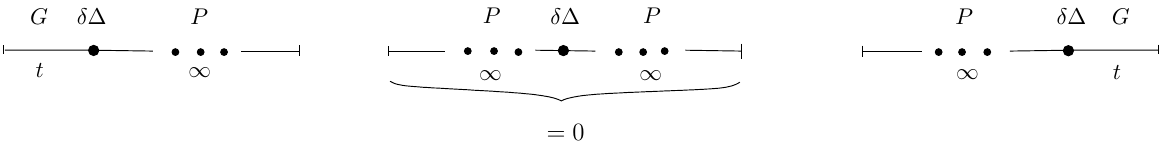}
        \caption{Terms in the formula (\ref{delta P s.3.1.2}) for $\delta P$ correspond to splitting the interval by a point  (a) close to the left endpoint, (c) far from both endpoints (the respective contribution is zero), (b) close to the right endpoint.}
        %\label{fig:enter-label}
    \end{figure}
\end{rem}
% \subsection{Higher order variations}
% We now want to understand what happens to the gauge fixing data when we are considering all orders in $t$, under a deformation of either $d_{A_0}$ or $d^*_{A_0}$. 
% Idea: e.g. for varying $d^*_{A_0}$ to $d^*_{A_0 +\delta}$use again the identity 
% $$(A+B)^{-1} = \sum_{k\geq 0}(-1)^k(A^{-1}B)^k A^{-1}$$ 
% with $A = \Delta_{A_0} + P_{A_0}$ and $B = [d_{A_0},\mathrm{ad}^*_{\delta}] + P_{A_0,d^*_{A_1}} - P_{A_0}$
\subsection{Metric dependence}\label{sec: metric}
%We analyze in detail what happens to the $(i,p,K)$ triple defined by a Riemannian metric under a small variation of the Riemannian metric. To this end, we 
Let $g_t, t \in (-\epsilon,\epsilon)$ be a smooth 1-parameter family of Riemannian metrics and denote $\dot{g} = \frac{d}{dt}\big|_{t=0}g_t \in \Gamma(\Sym^2(T^*M))$. By a partial contraction with $g^{-1} \in \Sym^2(TM)$ we obtain a endomorphism of the tangent bundle, i.e. a vector-field valued 1-form
\begin{equation}
    \mu = g^{-1}\dot{g}\in\Gamma(\mr{End}(TM)) \cong \Gamma(TM \otimes T^*M) = \Omega^1(M,TM).
\end{equation}
\begin{lem}\label{lem: lambda formula}
 Denote $\lambda = \star^{-1} \dot{\star} \colon \Omega^p(M) \to \Omega^p(M)$, then 
 \begin{equation}
     \lambda = \frac12\tr\mu - \iota_\mu.
 \end{equation}
\end{lem}
%\marginpar{This lemma is cool, but do we need it?}
\begin{proof}
%    Pavel :) 
Straightforward computation in local coordinates.
\end{proof}
It is well known (e.g. \cite{Ray1971}) that we have 
    \begin{equation} \dot{\dd}^*_{A'} = [\dd^*_{A'},\lambda].\label{eq: var d star metric}\end{equation}
Analogously to Proposition \ref{prop:varAtprime}, we then have the following statements:
%\marginpar{\bl I corrected a bunch of signs in this proposition and its proof}
\begin{prop}
    Let $g_t$ be a smooth family of Riemannian metrics on $M$, and $\lambda = \star^{-1} \dot{\star}$ as above, extended to act on Lie-algebra valued differential forms by tensoring with the identity on $\g$. Also, let $(A,A')$ be a pair of close flat connections on $M$. Then, we have 
\begin{align}
    \dot{\Delta}_{A,A'} &= [\dd^*_{A'},\lambda]\dd_A + \dd_A[\dd^*_{A'},\lambda], \label{eq: var Delta metric} \\
    \dot{P}_{A,A'} &= - [\dd_A, K_{A,A'}\lambda P_{A,A'} - P_{A,A'}\lambda K_{A,A'}], \label{eq: var P metric} \\
    \dot{K}_{A,A'} &= -[\dd_A,K_{A,A'}\lambda K_{A,A'}] - P_{A,A'}\lambda K_{A,A'} + K_{A,A'} \lambda P_{A,A'}. \label{eq: var K metric}
\end{align}
\end{prop}
\begin{proof}
     Equation \eqref{eq: var Delta metric} follows immediately from \eqref{eq: var d star metric}
    To prove \eqref{eq: var P metric}, we proceed as above in noticing that 
    \begin{equation}
        \dot{P}_{A,A'} = (\mr{id}- P_{A,A'})\dot{P}_{A,A'} + P_{A,A'}\dot{P}_{A,A'} =  - G_{A,A'}\dot{\Delta}_{A,A'} P_{A,A'} - P_{A,A'}\dot{\Delta}_{A,A'}G_{A,A'}.
    \end{equation}
    By using \eqref{eq: var Delta metric}, we obtain 
    \begin{equation} \label{eq: var P metric 2}
        \dot{P}_{A,A'} = - G_{A,A'} \dd_A\dd^*_{A'}\lambda P_{A,A'} + P_{A,A'} \lambda \dd^*_{A'}\dd_AG_{A,A'} = - [\dd_A, K_{A,A'}\lambda P_{A,A'} - P_{A,A'}\lambda K_{A,A'}]
    \end{equation}
    where have also used that $\dd_A$ and $\dd^*_{A'}$ commute with $G_{A,A'}$ and annihilate $P_{A,A'}$. 
    Finally, for the variation of $K_{A,A'}$ \eqref{eq: var K metric}  we obtain 
    \begin{equation}
        \dot{K}_{A,A'} = \underbrace{\dot{d}^*_{A,A'}G_{A,A'}}_{=:I} - \underbrace{\dd^*_{A'}G_{A,A'}\dot{\Delta}_{A,A'}G_{A,A'}}_{=:II} -\underbrace{\dd^*_{A'}G_{A,A'}\dot{P}_{A,A'}G_{A,A'}}_{=:III} = I - II - III. \label{eq: proof var K metric}
    \end{equation}
    Let us look at the three terms separately. From \eqref{eq: var d star metric}, we get 
    \begin{equation}
        I = [\dd^*_{A'},\lambda] G_{A_{0}}. \label{eq: proof var K metric I}
    \end{equation}
    For the second term, we obtain
    \begin{multline}
        II = \dd^*_{A'}G_{A,A'}
        ([\dd^*_{A'},\lambda]\dd_A + \dd_A[\dd^*_{A'},\lambda])
        G_{A,A'} \\
        = -K_{A,A'}\lambda K_{A,A'}\dd_A + \dd^*_{A'}\dd_AG_{A,A'}[\dd^*_{A'},\lambda]G_{A,A'} \label{eq: proof var K metric II}.
    \end{multline}
    Notice that $\Delta_{A,A'}G_{A,A'} = \mr{id} - P_{A,A'}$ implies $\dd^*_{A'}\dd_AG_{A,A'} = - \dd_A\dd^*_{A'}G_{A,A'} + \mr{id} - P_{A,A'}$ and therefore 
    $$ \dd^*_{A'}\dd_AG_{A,A'}[\dd^*_{A'},\lambda]G_{A,A'}  = \dd_AK_{A,A'}\lambda K_{A,A'} + [\dd^*_{A'},\lambda]G_{A,A'} + P_{A,A'}\lambda K_{A,A'},$$ so that 
    \begin{equation}
        II= [\dd_A,K_{A,A'}\lambda K_{A,A'}] + [\dd^*_{A'},\lambda]G_{A,A'} + P_{A,A'}\lambda K_{A,A'}. 
    \end{equation}
    Finally, by using \eqref{eq: var P metric 2} we can rewrite $III$ as 
    \begin{equation}
        III = - K_{A,A'}\dd_AK_{A,A'}\lambda P_{A,A'} = -K_{A,A'}\lambda P_{A,A'} \label{eq: proof var K metric III}
    \end{equation}
    since, suppressing indices, $KdK = K (Kd + \mr{id} - P) = K$ by $K^2 = KP = 0$. Now \eqref{eq: var K metric} follows from \eqref{eq: proof var K metric} by using \eqref{eq: proof var K metric I},\eqref{eq: proof var K metric II},\eqref{eq: proof var K metric III}.
\end{proof}

\section{Construction of extended $(i,p,K)$ triples from families}%\marginpar{Location? Notation?}
\label{sss: extended ipK triples from families}
One can obtain formulae \eqref{ihat, phat, Khat, Phihat} and Lemma \ref{lemma: ipK hat} from homological perturbation theory, as follows. Suppose $Q_q$ is a good gauge fixing operator for $d_A$ for $q \in \mathbb{GF}$, a smooth (but possibly infinite-dimensional) manifold. For fixed $q \in \mathbb{GF}$, one has the SDR data $(i_q,p_q,K_q)$ from \eqref{eq: SDR good gf}. These assemble into SDR data 
%\marginpar{\bl Aug 25: renamed $\wh{K}_0\ra \overline{K}$ etc.}
$(\ib,\pb,\Kb)$ for $d_A$, considered as a differential on $\Omega^\bullet(M\times\mathbb{GF},\g)$: 
\begin{equation}
\begin{aligned}
    \Kb \colon  \Omega^\bullet(M \times \mathbb{GF},\g;d_A) &\to \Omega^\bullet(M \times \mathbb{GF},\g;d_A), \\
    \ib \colon  \Omega^\bullet(\mathbb{GF},H_A(M,\g)) &\to \Omega^\bullet(M \times \mathbb{GF},\g;d_A),  \\
   \pb\colon \Omega^\bullet(M \times \mathbb{GF},\g;d_A) &\to \Omega^\bullet(\mathbb{GF},H_A(M,\g)) .
\end{aligned}
\end{equation}
Similarly, $Q_q$ and the Green's function $G_q$ assemble into operators $\Qb,\Gb$ on $\Omega^\bt(M\times \mathbb{GF},\g)$.

We can now deform the differential $d_A$ to the (twisted) de Rham differential on $M \times \mathbb{GF}$, by the de Rham differential $\delta_q$ in the direction of $\mathbb{GF}$. Note that since $\delta_q$ increases the de Rham degree in $\mathbb{GF}$ by 1, the map 
%\marginpar{\bl not sure about signs here. in HPT f-las, it should be $(1+\delta_q \Kb)^{-1}$ I think}
$1 + \delta_q  \Kb$ is invertible, we denote 
\begin{equation} X := (1 + \delta_q  \Kb)^{-1}  \delta_q = \sum_{k\geq 0}(-\delta_q  \Kb)^k \delta_q = \delta_q - \delta_q\Kb\delta_q + \cdots 
\end{equation} 
(this sum is finite since $\Kb$ decreases the form degree along $M$ by 1). 
We then obtain perturbed SDR data (see Appendix \ref{app: SDR})
%\marginpar{\bl I'd remove the $\circ$ signs (except maybe in some places in (124) where they help understand that $\delta_q$ acts only on the object next to it)}
\begin{equation}\label{eq: deformed induction}
    \begin{aligned}
        \widetilde{i} &=  \ib - \Kb X \ib = \ib - \widetilde{K}  \delta_q  \ib, \\
        \widetilde{p} &=  \pb - \pb X  \Kb = \pb - \pb  \delta_q  \widetilde{K},\\
        \widetilde{K} &=  \Kb - \Kb X  \Kb ,\\
        \widetilde{\delta_q} &= \pb X  \ib = \pb  \delta_q  \ib - \pb  \delta_q  \widetilde{K} \delta_q  \ib.
    \end{aligned}
\end{equation} 
% $Y_n = (\Kb \circ \delta_q)^n \circ \Kb$.
%\sout{By the results in loc. cit.} 
By Lemma \ref{lem: HPL},  $(\widetilde{i},\widetilde{p},\widetilde{K})$ form SDR data between the complexes $\Omega^\bullet(M\times \mathbb{GF},\g;d_A + \delta_q)$ and $\Omega^\bullet(\mathbb{GF}, H_A(M,\g);\widetilde{\delta}_q).$ 
\begin{prop}\label{prop: extended ipK triple, App C}
    One can rewrite formulae \eqref{eq: deformed induction} as follows: 
    \begin{equation}\label{eq: deformed induction 2}
    \begin{aligned}
        \widetilde{i} &=  \sum_{k \geq 0} (-\Gb (\delta_q\Qb))^k \ib, \\
        \widetilde{p} &=  \sum_{k \geq 0}\pb(-(\delta_q\Qb)  \Gb)^k, \\
        \widetilde{K} &=  \sum_{k \geq 0} \Kb(-(\delta_q \Qb) \Gb)^k, \\
        \widetilde{\delta_q} &= \delta_q + \sum_{k \geq 1} \pb(\delta_q\Qb) d_A  \Gb  (-\Gb(\delta_q\Qb))^k \ib.
    \end{aligned}
\end{equation} 
%\marginpar{\bl Only $k=1$ term can survive in $\til\delta_q$ for degree reasons}
\end{prop}
Note that, for degree reasons, only $k=0,1,2$ terms survive in $\til{i},\til{p},\til{K}$ and only $k=1$ term survives in $\til\delta_q$.

In particular, for 
\begin{equation}\label{GF as A' close to A}
\mathbb{GF} = \{A' \in \FC'| (A,A') \text{ close}\},    
\end{equation} 
we have $\widetilde{i} = \wh{i}, \widetilde{p} = \wh{p}, \widetilde{K} =\wh{K}, \widetilde{\delta_q} =\delta_{A'} + \wh{\Theta}$, with $(\wh{i},\wh{p},\wh{K},\wh{\Theta})$ given by \eqref{ihat, phat, Khat, Phihat}. Lemma \ref{lemma: ipK hat} then follows from the fact that the deformed $(i,p,K)$ triple is again an $(i,p,K)$ triple.  \\
\begin{proof}
%\marginpar{This proof could be proof-read  (and moved to an Appendix). {\bl yes to moving to an appendix!}}
One can simplify formulae \eqref{eq: deformed induction} by noticing that \begin{equation} [\delta_q, \Kb] = \delta_q\Kb \in \Omega^1(\mathbb{GF},\operatorname{End}(\Omega^\bullet(M,\g)))
\end{equation} 
where the right hand side acts as a multiplication operator on differential forms in the $\mathbb{GF}$ direction. 
We are using the notations where for an operator $x\in \{\Kb,\ib,\pb,\Qb\}$, $(\delta_q x)$ stands for $[\delta_q,x]$. By induction, one proves that 
\begin{equation}
\widetilde{K} = \sum_{k \geq 0} (-\Kb  \delta_q)^k \Kb = \sum_{k \geq 0}\Kb  (-\delta_q \Kb)^k = \sum_{k \geq 0} (-\delta_q \Kb)^k \Kb
\end{equation}
where in the second equality we have used that $\delta_q \Kb$ and $\Kb$ commute as a consequence of $\Kb^2=0$.
 Using further that  $(\Kb)_q = Q_q G_q$, we have $(\delta_q \Kb)_q = (\delta_q\Qb)_q G_q - Q_q (\delta_q \Gb)_q$, but using that $K_q  Q_q = 0$ we then obtain
 %\marginpar{There should be the sign, need to check where it comes from. {\bl signs should be correct now}} 
 \begin{equation}\label{eq: K tilde}
     \widetilde{K}=\sum_{k \geq 0} \Kb ( -(\delta_q \Qb) \Gb)^k  = \sum_{k \geq 0} \Gb (- (\delta_q \Qb)  \Gb)^k \Qb.
 \end{equation}
 This proves the third equation in \eqref{eq: deformed induction 2}
 We can then also rewrite the first two equations in \eqref{eq: deformed induction} by realizing that 
 \begin{equation}\label{eq: Q delta i}
     \Qb \delta_q \ib = \Qb  (\delta_q\ib) =  (\delta_q \Qb)  \ib
 \end{equation}
 and 
 \begin{equation}\label{eq: p delta Q}
     \pb  \delta_q  \Qb  = -(\delta_q\pb) \Qb = \pb(\delta_q\Qb),
 \end{equation}
 combining \eqref{eq: deformed induction}, \eqref{eq: K tilde},\eqref{eq: Q delta i} we get 
 \begin{multline*}
 \widetilde{i} = \ib - \sum_{k \geq 0} \Gb  (-(\delta_q\Qb)  \Gb)^k  \Qb\delta_q \ib \\= \ib + \sum_{k \geq 0} \Gb  (-(\delta_q\Qb)  \Gb)^k  (-\delta_q\Qb) \ib = \sum_{k \geq 0} (-\Gb  (\delta_q\Qb))^k \ib
 \end{multline*} which proves the first equation in \eqref{eq: deformed induction 2}. Combining \eqref{eq: deformed induction}, \eqref{eq: K tilde} and \eqref{eq: p delta Q}, we obtain 
 \begin{multline*}
      \widetilde{p} = \pb - \pb  (\delta_q  \Qb)  \Gb  \sum_{k \geq 0} (-(\delta_q\Qb)  \Gb)^k  \\
      = \pb - \pb  (\delta_q\Qb)  \Gb  \sum_{k \geq 0}(-(\delta_q\Qb)  \Gb)^k = \pb\sum_{k \geq 0}(-(\delta_q\Qb)  \Gb)^k,
 \end{multline*} which proves the second equation in \eqref{eq: deformed induction 2}. Finally, we focus on the last equation. The first term is simply 
 $$ \pb  \delta_q  \ib = \underbrace{\pb \ib}_{=1_{H^\bullet(M,\g)}}  \delta_q + \underbrace{\pb  (\delta_q\ib)}_{=0} = \delta_q.  $$ 
For the second term, notice that we have $\pb  \widetilde{K} = 0$ and therefore
$$ \pb  \delta_q  \widetilde{K}  \delta_q  \ib = -(\delta_q \pb)  \widetilde{K}  \delta_q \ib = (\delta_q\pb)  \sum_{k\geq 1} (-\Gb (\delta_q \widehat{Q}_0))^k  \ib.$$  
The proof of the last equation in \eqref{eq: deformed induction 2} now follows from 
$$\delta_q\pb = -\pb(\delta_q\Qb) d_A  \Gb,$$ 
which in turn can be proved by deriving the identity $\pb  \Qb = 0$: 
$$(\delta_q\pb)\Qb = -\pb (\delta_q\Qb)$$ and then composing both sides with $d_A\Gb$ on the left: 
$$ -\pb (\delta_q\Qb)  d_A  \Gb = (\delta_q\pb) \Qb  d_A  \Gb = - (\delta_q\pb) (d_A  \Qb\,  \Gb - \mathrm{id} + \ib \pb) = \delta_q\pb,$$
because $(\delta_q\pb) d_A = \delta_q(\pb  d_A) = 0$ and $(\delta_q\pb) \ib = - \pb ( \delta_q\ib) = 0$, because changing $q$ shifts representatives of cohomology by $d_A$-exact terms and $\pb d_A=0$.
%fixes the cohomology classes. \marginpar{Not quite quite sure about this last argument}
 \end{proof}

\section{Some technical proofs}
\subsection{Proof of Propositions \ref{prop: kuranishi} and \ref{prop: bij}}
\subsubsection{Proof of Proposition \ref{prop: kuranishi}}\label{app: proofs 1}
\begin{proof}%\marginpar{\red Move to appendix? {\bl yes!}}
      For point i), notice that because the assumption of boundedness of $K_{A_0}$ in a Banach norm, by the Banach inverse function theorem the inverse exists in a neighborhood of every point $\delta$ where the differential of $\til{\kappa}_{A_0}$, 
\begin{equation}
    (d \til\kappa_{A_0})_\delta = \mathrm{Id} + K_{A_0}\mathrm{ad}_\ddd \colon \Omega^1 \to \Omega^1
\end{equation}
is invertible. In particular, by the triangle inequality this happens when the operator norm of $K_{A_0}\ddd$ is less than one. 

 For point ii), we have to show $\til{\delta}_{A_0}(\til{\kappa}_{A_0}(\alpha)) = \til{\kappa}_{A_0}(\til{\delta}_{A_0}(\alpha) = \alpha.$  To see that $\til\delta_{A_0}(\til\kappa_{A_0}(\alpha)) = \alpha, $ recall that the coefficients $\alpha^{(j)}$ of $\til\delta_{A_0}$ are given by summing over binary trees with $j$ leaves, with prefactor $1/2^{j+1}$ and sign $(-1)^{j+1}$. When evaluating $\til\delta_{A_0}(\til\kappa_{A_0}(\alpha)$ we are placing $\til\kappa_{A_0}(\alpha)$ instead on every leaf. But since $\til\kappa_{A_0}(\alpha) = \alpha + \frac12K_{A_0}[\alpha,\alpha]$, we can express $\til\delta_{A_0}(\til\kappa_{A_0}(\alpha))$ again as a sum over binary trees $T'$ evaluated according to the same rules, but with a different combinatorial factor $c_{T'}$, allowing for the fact that the same tree $T'$ could arise from several different trees $T$.  See Figure \ref{fig: two trees} for an example. 
 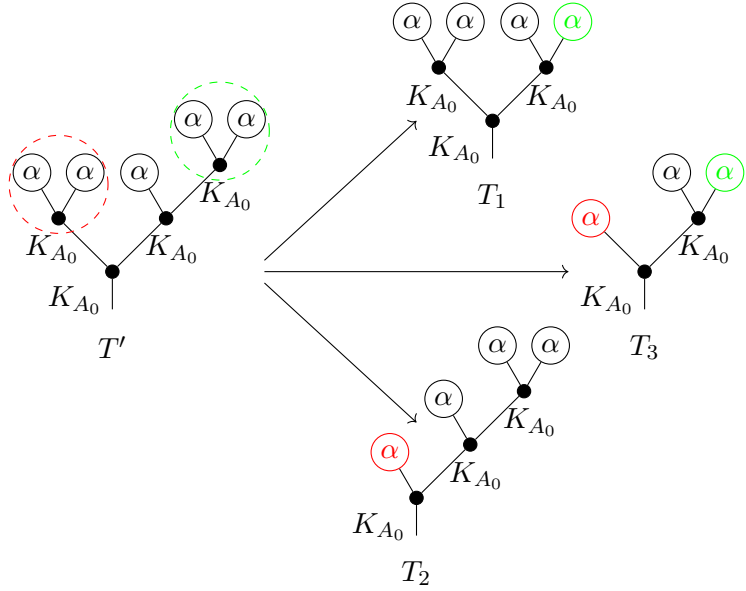
\begin{figure}
     \centering
     \begin{tikzpicture}
     \begin{scope}[shift={(1cm,2cm)}]
     \node at (0,-1) {$T_1$};
        \node[vertex] (r) at (0,0) {}; 
        \draw (r) edge node[pos=0.7, left] {$K_{A_0}$} (0,-.5);
        \node[vertex] (v2) at (45:1) {};
        \draw (r) edge node[pos=0.4, right] {$K_{A_0}$} (v2);
        \node[leaf, green] (l1) at ($(v2) +(60:0.7)$) {$\alpha$};
        \draw (v2) edge node[pos=0.4, right] {} (l1);
        \node[leaf] (l2) at ($(v2)+(120:0.7)$) {$\alpha$};
        \draw (v2) edge node[pos=0.4, left] {} (l2);
        \node[vertex] (v1) at ($(r)+(135:1)$) {};
        \draw (r) edge node[pos=0.4, left] {$K_{A_0}$} (v1);
        \node[leaf] (l3) at ($(v1) +(60:.7)$) {$\alpha$};
        \draw (v1) edge node[pos=0.4, right] {} (l3);
        \node[leaf] (l4) at ($(v1)+(120:.7)$) {$\alpha$};
        \draw (v1) edge node[pos=0.4, left] {} (l4);
        \end{scope}
        \begin{scope}[shift={(3cm,0cm)}]
     \node at (0,-1) {$T_3$};
        \node[vertex] (r) at (0,0) {}; 
        \draw (r) edge node[pos=0.7, left] {$K_{A_0}$} (0,-.5);
        \node[vertex] (v2) at (45:1) {};
        \draw (r) edge node[pos=0.4, right] {$K_{A_0}$} (v2);
        \node[leaf, green] (l1) at ($(v2) +(60:0.7)$) {$\alpha$};
        \draw (v2) edge node[pos=0.4, right] {} (l1);
        \node[leaf] (l2) at ($(v2)+(120:0.7)$) {$\alpha$};
        \draw (v2) edge node[pos=0.4, left] {} (l2);
        \node[leaf, red] (l1) at ($(r)+(135:1)$) {$\alpha$};
        \draw (r) edge node[pos=0.4, left] {} (l1);
        % \node[leaf] (l3) at ($(v1) +(60:.7)$) {$\alpha$};
        % \draw (v1) edge node[pos=0.4, right] {} (l3);
        % \node[leaf] (l4) at ($(v1)+(120:.7)$) {$\alpha$};
        % \draw (v1) edge node[pos=0.4, left] {} (l4);
        \end{scope}
        \begin{scope}[xshift=-4cm]
        \node at (0,-1) {$T'$};
        \node[vertex] (r) at (0,0) {}; 
        \draw (r) edge node[pos=0.7, left] {$K_{A_0}$} (0,-.5);
        \node[vertex] (v1) at ($(r)+(135:1)$) {};
        \draw (r) edge node[pos=0.4, left] {$K_{A_0}$} (v1);
        \node[leaf] (l1) at ($(v1) +(120:0.7)$) {$\alpha$};
        \draw (v1) edge node[pos=0.4, right] {} (l1);
        \node[leaf] (l2) at ($(v1)+(60:0.7)$) {$\alpha$};
        \draw (v1) edge node[pos=0.4, right] {} (l2);
        \draw[dashed,red] ($(v1) + (0,0.45)$) circle (0.65cm);
        \node[vertex] (v2) at (45:1) {};
        \draw (r) edge node[pos=0.4, right] {$K_{A_0}$} (v2);     
        \node[vertex] (v3) at ($(v2) +(45:1)$) {};
        \draw (v2) edge node[pos=0.4, right] {$K_{A_0}$} (v3);
        \node[leaf] (l4) at ($(v3)+(120:.7)$) {$\alpha$};
        \draw (v3) edge node[pos=0.4, left] {} (l4);
        \node[leaf] (l5) at ($(v3)+(60:.7)$) {$\alpha$};
        \draw (v3) edge node[pos=0.4, left] {} (l5);
        \draw[dashed,green] ($(v3) + (0,0.45)$) circle (0.65cm);
        \node[leaf] (l3) at ($(v2)+(120:.7)$) {$\alpha$};
        \draw (v2) edge node[pos=0.4, left] {} (l3);
        \draw[->] (2,0.15) to (4,2);
        \draw[->] (2,0) to (6,0);
        \draw[->] (2,-0.15) to (4,-2);
        \end{scope}
        \begin{scope}[yshift=-3cm]
        \node at (0,-1) {$T_2$};
        \node[vertex] (r) at (0,0) {}; 
        \draw (r) edge node[pos=0.7, left] {$K_{A_0}$} (0,-.5);
        \node[leaf,red] (l1) at ($(r) +(120:0.7)$) {$\alpha$};
        \draw (r) edge node[pos=0.4, right] {} (l1);
        %\node[leaf] (l2) at ($(v1)+(60:0.7)$) {$\alpha$};
        %\draw (v1) edge node[pos=0.4, right] {} (l2);
        \node[vertex] (v2) at (45:1) {};
        \draw (r) edge node[pos=0.4, right] {$K_{A_0}$} (v2);     
        \node[vertex] (v3) at ($(v2) +(45:1)$) {};
        \draw (v2) edge node[pos=0.4, right] {$K_{A_0}$} (v3);
        \node[leaf] (l4) at ($(v3)+(120:.7)$) {$\alpha$};
        \draw (v3) edge node[pos=0.4, left] {} (l4);
        \node[leaf] (l5) at ($(v3)+(60:.7)$) {$\alpha$};
        \draw (v3) edge node[pos=0.4, left] {} (l5);
        \node[leaf] (l3) at ($(v2)+(120:.7)$) {$\alpha$};
        \draw (v2) edge node[pos=0.4, left] {} (l3);
        \end{scope}
\end{tikzpicture}     \caption{The tree $T'$ has $n_{T'} = 2 $ corollas, collapsing the green one yields $T_1$, collapsing the red one yields $T_2$, collapsing both yields $T_3$. Thus it will appear in $\til\delta_{A_0}(\til{\kappa}_{A_0}(\alpha))$ four times, with total combinatorial coefficient $1 + (-1) + (-1) + 1 = 0$.}
     \label{fig: two trees}
 \end{figure}We claim that $c_{T'} = 0$ for all trees with at least two leaves. Indeed, for a tree $T'$ let $n_{T'}$ denote the number of internal vertices connected to exactly two leaves. Note that $n_{T'} = 0$ if and only if $T'$ is the tree with a single leaf at the root and no internal vertex. If $v$ is such an internal vertex, then we call $v$ together with the two adjacent leaves a corolla. Then $T'$ could be obtained from the tree $T$ where we collapse the corolla of $v$ into a leaf $\alpha$. Note that this operation changes the sign. In total, we will obtain the tree $T'$ exactly $2^{n_{T'}}$ times, but with different signs: If we collapse $k$ corollas then there is a sign $(-1)^k$. Therefore the combinatorial coefficient of $T'$ is $c_{T'} = \sum_{k \geq 0}^{n_{T'}}(-1)^k{k \choose {n_{T'}}} = 0$. The other direction $\til\kappa_{A_0}(\til{\delta}_{A_0}(\alpha)) = \alpha$ is proven similarly. 

 As for point iii), we simply compute 
 $$\dd_{A_0}{\til\kappa_{A_0}}(\alpha) = d_{A_0}\alpha + \frac12 \dd_{A_0}K_{A_0}[\alpha,\alpha] = \dd_{A_0}\alpha + \frac12[\alpha,\alpha] - K_{A_0}\dd_{A_0}[\alpha,\alpha] - p_{A_0}[\alpha,\alpha].$$
 If $\dd_{A_0}\alpha = -\frac12[\alpha,\alpha]$, then the first two terms cancel and the latter two terms vanish because $[\alpha,\alpha] = -2d_{A_0}\alpha$ is exact. 
 Point iv) follows immediately from the definition of $\til\varphi_{A_0} = A_0 + \til\delta_{A_0}$ and points i) and iii). 
\end{proof}
\subsubsection{Proof of Proposition \ref{prop: bij}}\label{app: proofs 2}

  Let $A_1$ be a different flat connection, and $m = A_1 - A_0$. Then $m \in MC_{A_0}$, and therefore $\dd_{A_0}\til\kappa_{A_0}(m)=0$, so it defines a tangent vector to the space of flat connections at $A_0$ (and moreover, if $K_{A_0}m = 0$, then $K_{A_0}^2 = 0$ implies $K_{A_0}\til{\kappa}_{A_0} = 0$). In this case, $A(t) = A_0+\til{\delta}_{A_0}(t \til{\kappa}_{A_0}(m))$ defines a curve of flat connections with $A(1) = A_1$. %\marginpar{\bl I'd rephrase to ``..might fail to contain any neighborhood'' - since later the statement is that it actually does contain a neighborhood, for $A_0$ smooth.} In general, the image of the map $\til{\kappa}\colon MC_{A_0} \to \Omega^1_{A_0-\mr{cl}}$ does not contain any neighborhood of $0$, in other words, in general there are arbitrarily small closed forms $\til{\kappa}(\alpha)$ for which $ \alpha = \til{\delta}\til{\kappa}(\alpha) \notin MC_{A_0}$. However, if $MC_{A_0,B}$ (or equivalently $FC_B$) is a smooth (Banach) manifold near $0$ (resp. $A_0$), then its tangent space at 0 is the space $\Omega^1_{A_0-\mr{cl},B}$ and we can apply the inverse function theorem to the restriction of $\til{\kappa}$ to $MC_{A_0,B}$ and conclude it is a local diffeomorphism. We will now show that this is the case if $A_0$ is smooth in the sense of Definition \ref{def: smooth}. The strategy is to first show that the restriction of $\til{\delta}_{A_0}$ to some open neighborhood in $\operatorname{Im}{\red i}_{A_0}$ is bijective, and then show that $\psi\colon\Omega^1_{A_0-\mr{cl}}\subset U \to V \subset MC_{A_0}$ is a bijection by using gauge transformations. 
  \begin{lem}Suppose $A_0$ is a smooth point. Then there are neighborhoods  $0 \in U \subset \mathrm{Im}\, i_{A_0}$ and $0 \in V \subset MC_{A_0} \cap \ker K_{A_0}$ such that 
      $ \til{\delta}_{A_0}\colon U \to V$ is bijective with inverse $\til{\kappa}_{A_0}$. 
  \end{lem}
  \begin{proof}
      We know that $\til{\delta}_{A_0}$ converges in a neighborhood $U_1$ of $0 \in \Omega^1_l$, therefore $\til{\delta}_{A_0}$ is defined on $U = U_1 \cap \operatorname{Im}i_{A_0}$. On $\operatorname{Im} i_{A_0}$, we know that $\til{\delta}_{A_0} \in MC_{A_0}$ by Proposition \ref{prop: deformation of A0}. Since $\til{\delta}_{A_0}(\alpha) = \alpha + K_{A_0}(\ldots)$, we have $K_{A_0}\til\delta_{A_0}(\alpha) = 0$. Therefore $\til\delta_{A_0}(U) \subset  MC_{A_0} \cap \ker K_{A_0}$. On the other hand, if $m \in MC_{A_0} \cap \ker K_{A_0}$, then $\dd_{A_0}\til{\kappa}(m) = K_{A_0}\til{\kappa}(m) = 0$, i.e $\til{\kappa}_{A_0}(m) \in \operatorname{Im}i_{A_0}$. For $m$ small enough, we therefore have $\til\kappa_{A_0}m \in U$, and since we already know that $\til{\delta}_{A_0}$ and $\til\kappa_{A_0}$ are inverse to each other, we conclude the statement. 
  \end{proof}
\begin{cor}
The restriction of the map $\kappa_{A_0}\colon\Omega^1 \to H^1_{A_0}$ given by
\begin{equation}
\kappa_{A_0}(\ddd) = p_{A_0}(\ddd + \frac{1}{2}K_{A_0}[\ddd,\ddd]) = p_{A_0}(\ddd)\in H^1_{A_0}
\end{equation}
to the image of $\delta_{A_0}\colon H^1_{A_0}\supset U \to {\Omega^1(M,\g)}$, is a compositional inverse to $\delta_{A_0}$. % i.e. $\delta_{A_0}(\kappa_{A_0}(\alpha)) = \kappa_{A_0}(\delta_{A_0}(\alpha))= \alpha$. 
\end{cor} 
%\marginpar{\bl Notations: in $\delta\circ\kappa$, $\alpha$ is an MC element; in $\kappa\circ\delta$, $\alpha$ is a cohomology class.}
%\marginpar{\bl What's the assumption on $\alpha$ to have $\delta(\kappa(\alpha))=\alpha$? $\alpha$ should be a $K$-closed MC element perhaps? (it should be restricted to a fin. dim. manifold, since we are factoring through $H^1$.)}
Finally, let us prove Proposition \ref{prop: bij}. 
\begin{proof}[Proof of Proposition \ref{prop: bij}]
We begin with point i). 
   To show surjectivity onto a small neighborhood of 
   %\marginpar{\bl what is $\FC_0$?}
   %\marginpar{\bl I switched to the left gauge group action in the proof}
   $A_0 \in \FC_l$, in the first step, we construct a gauge transformation that takes an arbitrary connection $A_1$ close enough to $A_0$ to a 
%\sout{$K_{A_0}$-closed connection $A_1'$. }
   connection $A_1'$ satisfying $K_{A_0}(A_1'-A_0)=0$.
   To this end, consider the map
    $F \colon\Omega^0_{K-\mr{ex}, l} \times \Omega^1_l \to \Omega^0_{K-\mr{ex},l}$ given by 
    $$ (\beta, \ddd) \to K_{A_0}(\, {}^{\exp(-\beta)}(A_0 + \ddd) - A_0 ). $$
    We want to solve for $\beta = \beta(\ddd)$ such that 
    %\marginpar{\bl Apr19 changed $F(\beta,\beta(\ddd)) = 0$ to $F(\beta(\ddd),\ddd) = 0$}
    $F(\beta(\ddd),\ddd) = 0$,
    %$F(\beta,\beta(\ddd)) = 0$, 
    this is the desired gauge transformation. Existence of $\beta$, for small enough $\ddd$, is then guaranteed by the implicit function theorem for Banach spaces, since the derivative of $F$ at $(0,0)$ in direction of $\beta$ is $(dF/d\beta)(0,0) = K_{A_0}d_{A_0} = \mr{id}_{\Omega^0_{K-\mr{ex}}}$. 
    This means that for $A_1$ close enough to $A_0$, there is $\beta$ such that $K_{A_0}(\,{}^{\exp(-\beta)}A_1-A_0) = 0$. For such connections, we know that $\til{\kappa}_{A_0}(\,{}^{\exp(-\beta)}A_1-A_0)$ is a $d_{A_0}$- and $K_{A_0}$-closed 1-form. Therefore, if $A_1$ is a flat connection close to $A_0$, then 
    $\alpha = \til\kappa_{A_0}(\,{}^{\exp(-\beta)}A_1-A_0) -  d_{A_0}\beta$ 
    is a $d_{A_0}$-closed form such that $\til{\psi}_{A_0}(\alpha) = A_1$.
    %\marginpar{\bl changed $\til\varphi$ to $\til\psi$}
    
As for point ii), we notice that the implicit function theorem also guarantees smoothness of the map $\ddd \mapsto \beta(\ddd)$. It follows that $\til\psi_{A_1}^{-1} \circ \til\psi_{A_0}$ is a composition of smooth maps between differential forms (projections, gauge transformations, $\beta$, and the maps $\til\varphi$ and $\til\kappa$) and hence smooth. 
\end{proof}

\subsection{Proof of Proposition \ref{prop: tau Grothendieck horizontality}: ``horizontality'' of Ray-Singer torsion}
\label{Appendix: proof of Prop 4.4}

%We break the proof into several lemmata.
We first need some auxiliary results.

\begin{lem}\label{lemma: dependence of torsion on A}
Given a path of flat connections $A_t$, one has the following formula for the infinitesimal change of the Ray-Singer torsion $\tau_{A_t}$:
    \begin{equation}\label{tau dependence on A}
        \left.\frac{d}{dt}\right|_{t=0}\det(\BB^\mr{diag}_{A_0\la A_t})\tau_{A_t}=\tau_{A_0}\,\mr{Str}_{\Omega^\bt}(K_{A_0}\ad_{\dot{A}_0}).
    \end{equation}
    Here $\BB^\mr{diag}_{A_0\la A_t}\colon H_{A_t}\ra H_{A_0}$ is the projection to cohomology of the parallel transport of the connection $\nabla^\mr{Harm}$ along the path $(A_{t-\tau},A_{t-\tau})$, $0\leq \tau\leq t$ in $\FC'\times\FC'$.
\end{lem}
%\subsection{Proof of Lemma \ref{lemma: dependence of torsion on A}: dependence of Ray-Singer torsion on the flat connection}\label{appendix: proof of Lemma on dependence of tau on A}
\begin{proof}
    By definition of Ray-Singer torsion, 
    \begin{equation}
    \tau_{A_t}=\mu_{A_t} \prod_{p=0}^3 \left({\det}'_{\Omega^p} \Delta_{A_t}\right)^{\frac{-(-1)^pp}{2}}
    \end{equation}
with $\mu_{A_t}$ the volume element in $\mr{Det}H^\bt_{A_t}$ corresponding to Hodge inner product. Note that in this lemma we are using the synchronized $(A_t,A_t)$ Hodge decompositon. Since $\BB_{A_t\la A_0}^\mr{diag}$ is an isometry (Proposition \ref{prop 3.5} (\ref{prop 3.5 (b)})), we have $\det(\BB_{A_t\la A_0}) \mu_{A_t}=\mu_{A_0}$. Hence,
\begin{multline}\label{lem 4.12 computation}
     \tau_{A_0}^{-1}\left.\frac{d}{dt}\right|_{t=0}\det(\BB_{A_0\la A_t})\tau_{A_t}=
     \left.\frac{d}{dt}\right|_{t=0} \log \prod_{p=0}^3 \left({\det}'_{\Omega^p} \Delta_{A_t}\right)^{\frac{-(-1)^pp}{2}} \\
     =
     \sum_{p=0}^3\frac{-(-1)^p p}{2}   %\left.\frac{d}{dt}\right|_{t=0} \+
     \mr{tr}_{\Omega^p} (\dot{\Delta} G), 
\end{multline}
%\marginpar{Formally, $\mr{tr}$ should be replaced with $\zeta'(0)$.}
where 
$$\dot{\Delta}\colon=\frac{d}{dt}\Big|_{t=0}\Delta_{A_t}=[\ad_{\dot{A}_0},d^*]_+ +
[d,\ad^*_{\dot{A}_0}]_+ .$$
Here we suppress the subscript $A_0$ in $G,d,d^*$.  Continuing the computation (\ref{lem 4.12 computation}) we have
\begin{multline*}
    \cdots = \sum_{p=0}^3 \frac{-(-1)^p p}{2}\Big(\tr_{\Omega^{p-1}} \underbrace{d^*G}_K \ad_{\dot{A}_0} +\tr_{\Omega^{p}} \underbrace{d^*G}_K \ad_{\dot{A}_0} 
    +\tr_{\Omega^p} \underbrace{dG}_{K^*} \ad^*_{\dot{A}_0} + \tr_{\Omega^{p+1}} \underbrace{dG}_{K^*} \ad^*_{\dot{A}_0}
    \Big)
\\ =
\sum_{p=0}^3 \underbrace{\frac{-(-1)^p p-(-1)^{p+1}(p+1)}{2}}_{\frac{(-1)^p}{2}} \tr_{\Omega^p} K \ad_{\dot{A}_0} +
\underbrace{\frac{-(-1)^p p-(-1)^{p-1}(p-1)}{2}}_{\frac{-(-1)^p}{2}} \tr_{\Omega^p} K^* \ad^*_{\dot{A}_0}\\
=\frac12 \mr{Str}_{\Omega^\bt} K\ad_{\dot{A}_0}- \frac12 \mr{Str}_{\Omega^\bt} \underbrace{K^*\ad^*_{\dot{A}_0}}_{*K\ad_{\dot{A}_0}*} = \mr{Str}_{\Omega^\bt} K\ad_{\dot{A}_0}.
\end{multline*}
This proves (\ref{tau dependence on A}).
\end{proof}

\begin{rem}\label{rem: regularization of supertraces}
    Traces in the proof above should be understood as zeta-regularized traces. For instance, $\mr{Str}_{\Omega^\bt}K\ad_{\dot{A}_0}$ should be understood as
    \begin{equation}
        \mr{Str}_{\Omega^\bt}K\ad_{\dot{A}_0}\colon=\lim_{s\ra 0} \int_0^\infty du\, u^s\, \mr{Str}_{\Omega^\bt} d^* e^{-u \Delta_{A_0}} \ad_{\dot{A}_0}.
    \end{equation}
    However, by the results of Axelrod-Singer \cite{Axelrod1991}, the singular terms of the heat kernel expansion are proportional to $\mr{id}\in \mr{End}(\g)$ and hence vanish under the trace with $\ad_{\dot{A}_0}$ by unimodularity of $\g$. Therefore, the zeta-regularized supertrace coincides with the point-splitting regularized supertrace that we use to define tadpoles in Feynman diagrams, cf. footnote \ref{footnote: tadpoles}.
\end{rem}

\begin{lem}\label{lemma D.2}
    Given a path of flat connections $A_t$ and $A'$ close $A_0$, we have
    \begin{equation}
    \left. \frac{d}{dt}\right|_{t=0} \det(\BB_{A_0\la A_t;A'})\tau_{A_t} = \tau_{A_0} \,\mr{Str}_{\Omega^\bt} (K_{A_0,A'} \ad_{\dot{A}_0}).
    \end{equation}
\end{lem}

\begin{proof}
    Consider a path $A'_s$ in $\FC'$ starting at $A'_0=A_0$ and ending at $A'_1=A'$ (and staying close to $A_0$). Denote
    $$ f_s\colon = \tau_{A_0}^{-1}\left. \frac{d}{dt}\right|_{t=0} \det(\BB_{A_0\la A_t;A'_s})\tau_{A_t},\quad
    h_s\colon= \mr{Str}_{\Omega^\bt} (K_{A_0,A'_s} \ad_{\dot{A}_0}).
    $$
Note that Lemma \ref{lemma: dependence of torsion on A} implies that $f_0=h_0$. To prove the result it suffices to show that $\frac{d}{ds} f_s=\frac{d}{ds} h_s$.
For the derivative of $h_s$ we find
\begin{multline}\label{lemma D.2 proof eq1}
    \frac{d}{ds} h_s= \mr{Str}(\frac{d}{ds}K_{A_0,A'_s} \ad_{\dot{A}_0})=
    \mr{Str}(\cancel{[d,K\ad^*_{\partial_s A'_s} G]}+P\ad^*_{\partial_s A'_s}G+G\ad^*_{\partial_s A'_s} P)\ad_{\dot{A}_0}\\
    =\mr{Str}\,P(\ad^*_{\partial_s A'_s}G \ad_{\dot{A}_0}-\ad_{\dot{A}_0} G\ad^*_{\partial_s A'_s}).
\end{multline}
For $f_s$ we have
\begin{multline}\label{f_s via hol around rectangle}
    f_s=\tau_{A_0}^{-1}\left.\frac{d}{dt}\right|_{t=0} \det (\mr{Hol}_{\nabla^\mr{Harm}}(R_{s,t}))\cdot \det (\BB_{A_0\la A_t;A_0}) \tau_{A_t} \\= f_0+\left.\frac{d}{dt}\right|_{t=0}\det \mr{Hol}_{\nabla^\mr{Harm}}(R_{s,t}).
\end{multline}
Here $R_{s,t}$ is the (curved) rectangle in $\mc{U}$ with sides (i) $(A_\tau,A_0)$ with $0<\tau<t$, (ii) $(A_t,A'_\sigma)$ with $0<\sigma <s$, (iii) $(A_{t-\tau},A'_s)$ with $0<\tau<t$, (iv) $(A_0,A'_{s-\sigma})$ with $0<\sigma < s$; $\mr{Hol}_{\nabla^\mr{Harm}}(R_{s,t})\in \mr{End}(\mr{Harm}_{A_0,A_0})$ stands for the holonomy of $\nabla^\mr{Harm}$ around the rectangle. 

Denote $\rho_{s,t,\epsilon}=R_{s+\epsilon,t}-R_{s,t}$ (here difference is an operation on singular 1-chains) -- a small rectangle with vertices at $(A_0,A'_s)$, $(A_t,A'_s)$, $(A_t,A'_{s+\epsilon})$, $(A_0,A'_{s+\epsilon})$. 
Next, (\ref{f_s via hol around rectangle}) implies
\begin{multline}\label{lemma D.2 proof eq2}
    \frac{d}{ds}f_s= \left.\frac{\partial^2}{\partial \epsilon\, \partial t}\right|_{\epsilon=t=0} \det \mr{Hol}_{\nabla^\mr{Harm}}(\rho_{s,t,\epsilon})\\
    =-\mr{Str}\, \iota_{\partial_s A'_s}\iota_{\dot{A}_0}F_{\nabla^\mr{Harm}}\Big|_{(A_0,A'_s)} =\mr{Str}\,P (\ad^*_{\partial_s A'_s} G \ad_{\dot{A}_0}-\ad_{\dot{A}_0} G \ad^*_{\partial_s A'_s}).
\end{multline}
Here in the last step we used the result (\ref{nabla^Harm curvature}) for the curvature of $\nabla^\mr{Harm}$. Comparing with (\ref{lemma D.2 proof eq1}), we see that we have $\partial_s f_s=\partial_s h_s$ which, together with the initial condition $f_0=h_0$ implies the desired result $f_1=h_1$. 
\end{proof}

\begin{proof}[Proof of Proposition \ref{prop: tau Grothendieck horizontality}]
 Let $A_t=\varphi(A,A',t\alpha)$  -- a path of flat connections from $A$ at $t=0$ to $\til{A}$ at $t=1$. We want to show that
 \begin{equation}\label{Prop 4.4 proof eq1}
     \det(\BB_{A\la A_t;A'})\circ \tau_{A_t} \stackrel{!}{=} \tau_A \exp\sum_\gamma \frac{2}{|\mr{Aut}(\gamma)|}\Phi_{\gamma,A,A'}(t\alpha).
 \end{equation}
For $t=1$, this is the desired relation (\ref{tau G horizontality}). Denote the l.h.s. of (\ref{Prop 4.4 proof eq1}) by $\lambda_t$ and the r.h.s. by $\mu_t$. We have $\lambda_0=\mu_0$, so it suffices to prove $\lambda_t^{-1}\partial_t \lambda_t=\mu_t^{-1}\partial_t\mu_t$.

We have
\begin{multline}\label{Prop 4.4 proof eq4}
    \frac{d}{dt}\lambda_t=\left.\frac{d}{d\epsilon}\right|_{\epsilon=0} \det\BB_{A\la A_t;A'}\circ(\det\BB_{A_t\la A_{t+\epsilon};A'}\circ \tau_{A_t+\epsilon})\\
    \underset{\mr{Lemma}\; \ref{lemma D.2}}{=}
    \det\BB_{A\la A_t;A'}\circ\tau_{A_t}\,\mr{Str}\, K_{A_t,A'}\ad_{\dot{A}_t}=
    \lambda_t \,\mr{Str}\, K_{A_t,A'}\ad_{\dot{A}_t}.
\end{multline}

To analyze $\mu_t$, we first remark that 
\begin{equation}\label{Prop 4.4 proof eq2}
\exp\sum_\gamma \frac{2}{|\mr{Aut}(\gamma)|}\Phi_{\gamma,A,A'}(t\alpha)= \mr{Sdet}_{\Omega^\bt}(1+K_{A,A'}\ad_{A_t-A}). 
\end{equation}
Indeed, log of the r.h.s. here is 
$$ \mr{Str}\,\log (1+K_{A,A'}\ad_{A_t-A})=\sum_{n\geq 1}\frac{-1}{n}\mr{Str}(-K_{A,A'}\ad_{A_t-A})^n $$
-- twice the sum of  one-loop graphs, with $n\geq 1$ trees plugged into the cycle.

From (\ref{Prop 4.4 proof eq2}) we find
\begin{equation}\label{Prop 4.4 proof eq3}
    \frac{d}{dt} \mu_t=\mu_t \mr{Str} \underbrace{(1+K_{A,A'}\ad_{A_t-A})^{-1}K_{A,A'}}_{K_{A_t,A'}}\ad_{\dot{A}_t}.
\end{equation}
Comparing with (\ref{Prop 4.4 proof eq4}), we see that $\lambda_t^{-1}\dot\lambda_t=\mu_t^{-1}\dot\mu_t$.\footnote{The r.h.s. here is regularized via point-splitting and the l.h.s. is defined via zeta-regularization. By Remark \ref{rem: regularization of supertraces}, they coincide.
} This finishes the proof.
\end{proof}

\subsection{Vanishing of contributions of hidden boundary strata to $\delta_{A'}Z$ in Proposition \ref{prop: variation of Z wrt A'}}
%\footnote{We refer the reader to \cite[Appendix A]{Cattaneo2008} and references therein for details on vanishing of hidden boundary strata in Chern-Simons theory (the reference considers a change of metric but the technology carries over).}
\label{Appendix: vanishing of hidden boundary strata}
Let $\Gamma$ be a %connected 
trivalent graph, possibly with leaves. Let $\Gamma^\mm$ be $\Gamma$ with one edge marked by $\Lambda$ (instead of $K$) or one leaf marked by $\mathbb{I} i(\sfa)$ (instead of $i(\sfa)$). Let $\omega_{\Gamma^\mm}$ be the corresponding form on the configuration space $\overline{\mr{Conf}}_V(M)$ -- the integrand of (\ref{eq: def Phi_Gamma}) in the $(A,A')$-gauge, with the factor for the marked edge or leaf replaced by $\Lambda$ or $\mathbb{I}i(\sfa)$.

Let $\Gamma'$ be a full subgraph of $\Gamma^\mm$ on a subset $V'$ of vertices with $|V'|\geq 3$ %vertices $V'$ 
%(we also allow the case of two vertices if they are connected by at least two edges in $\Gamma'$)\marginpar{Case of two vertices - correct? want to take care of Theta graph}
and let $\partial_{V'} \overline{\mr{Conf}}_V(M)$ be the codimension $1$ boundary stratum of the compactified configuration space corresponding to vertices $V'$ %of $\Gamma'$ 
collapsing  at a point of $M$, see \cite{Axelrod1994}.

\begin{prop}\label{prop: hidden strata}
    The contribution of the hidden boundary stratum $\partial_{V'} \overline{\mr{Conf}}_V(M)$ to the variation of the partition function $\delta_{A'}Z_{A,A'}(\sfa)$ vanishes:
    \begin{equation}\label{hidden stratum integral}
        \int_{\partial_{V'} \overline{\mr{Conf}}_V(M)} \omega_{\Gamma^\mm} =0
    \end{equation}
\end{prop}

Before the proving the proposition, we need the following lemma. %(and its corollary).

    Let $\eta_\Lambda\in \Omega^1(\overline{\mr{Conf}}_2(M),\g\otimes \g)$ be the integral kernel of the operator $\Lambda$ (\ref{Lambda, II, PP}).

We will call a form $\alpha\in \Omega^\bt(\overline{\mr{Conf}}_2(M))$ 
%on $\overline{\mr{Conf}}_2(M)$ 
\emph{regular}\footnote{We think of ``regularity'' as a slightly weakened notion of continuity of a form across the diagonal in $M\times M$. A toy example is the 1-form $r d\phi $ on the punctured plane $\mathbb{R}^2\setminus \{0\}$. It does not continuously extend to the origin, but it extends to the differential geometric blow-up of $\mathbb{R}^2$ at $0$ and vanishes on the boundary circle of the blow-up.
} if 
its restriction to the boundary of the configuration space\footnote{
Recall that the boundary of the compactified configuration space of two points on $M$ can be identified with the sphere tangent bundle of $M$ (seen as the diagonal of $M\times M$): 
$\partial \overline{\mr{Conf}}_2(M)\cong  STM$.
}
is the pullback of a smooth form on $M$:
\begin{equation*}
\alpha|_{\partial \overline{\mr{Conf}}_2(M)} = (\rho|_\partial)^* \beta
\end{equation*}
for some $\beta\in \Omega^1(M,\g\otimes \g)$. Here $\rho\colon  \overline{\mr{Conf}}_2(M)\ra M\times M$ is the blowdown map and $\rho|_\partial\colon \partial \overline{\mr{Conf}}_2(M)\ra \mr{Diag}\cong M$  is its restriction to the boundary ($\mr{Diag}\subset M\times M$ is the diagonal).

%it is the pullback by the blowdown map $\rho\colon \overline{\mr{Conf}}_2(M)\ra M\times M$ of a continuous form on $M\times M$, smooth on the complement of the diagonal. %Note that a regular form 

\begin{lem}\label{lemma: eta_Lambda behavior on boundary}
The form $\eta_\Lambda$ is regular.
% The restriction of the form $\eta_\Lambda$ to the boundary of the configuration space\footnote{
% Recall that the boundary of the compactified configuration space of two points on $M$ can be identified with the sphere tangent bundle of $M$ (seen as the diagonal of $M\times M$): 
% $\partial \overline{\mr{Conf}}_2(M)\cong  STM$.
% } is the pullback of a smooth form on $M$:
% \begin{equation}
% \eta_\Lambda|_{\partial \overline{\mr{Conf}}_2(M)} = (\rho|_\partial)^* \alpha
% \end{equation}
% for some $\alpha\in \Omega^1(M,\g\otimes \g)$. Here $\rho\colon  \overline{\mr{Conf}}_2(M)\ra M\times M$ is the blowdown map and $\rho|_\partial\colon \partial \overline{\mr{Conf}}_2(M)\ra \mr{Diag}\cong M$  is its restriction to the boundary ($\mr{Diag}\subset M\times M$ is the diagonal).
%a pullback by the blowdown map $\rho\colon \overline{\mr{Conf}}_2(M)\ra M\times M$ of a continuous form on $M\times M$, smooth on the complement of the diagonal.
\end{lem}

\begin{proof}
    The leading singularity of $\eta_\Lambda$ on the diagonal is 
    \begin{equation}\label{eta_Lambda singularity}
        \eta_\Lambda(x,y) \sim_{x\ra y}  \frac{1}{4\pi}(\det g)^{\frac12} g^{\alpha \zeta} f^{abc} \epsilon_{\alpha\beta\gamma} \delta A'^{b}_\zeta\frac{u^\beta}{|| u||} du^\gamma + \cdots
    \end{equation}
    where $\cdots$ stands for less singular terms; $u=x-y$ (in a local chart at $y$); $g$ and $\delta A'$ are evaluated at $y$. This result is obtained by Hadamard parametrix method, by performing a model computation of the integral kernel of $\Lambda=K \ad^*_{\delta A'}G$ on $\mathbb{R}^3$ with $A=A'=0$ and constant $\delta A'$.\footnote{
    On $\mathbb{R}^3$, fixing $y=0$, we have
    $\displaystyle \eta_\Lambda(x,0)=[K\ad^*_{\delta A'} G](x,0)= \int_{\mathbb{R}^3\ni z} \epsilon_{\alpha\beta\gamma} \frac{v^\alpha dv^\beta dv^\gamma}{4\pi ||v ||^3} f^{abc}\iota_{\delta A'^{b\mu}\frac{\partial}{\partial z^\mu}} \epsilon_{\delta \epsilon\zeta}\frac{ dz^\delta dz^\epsilon dz^\zeta}{24 \pi ||z||}=
    \frac{1}{16\pi^2} f^{abc}A'^{b\mu} \epsilon_{\alpha\beta\gamma}\epsilon_{\delta\epsilon\mu}dx^\beta\int_{\mathbb{R}^3\ni z}\frac{v^\alpha dz^\gamma dz^\delta dz^\epsilon}{||v||^3 ||z||} 
    =\frac{1}{4\pi}\epsilon_{\alpha\beta\gamma}f^{abc} \delta A'^{b\alpha} \frac{x^\beta}{||x ||} dx^\gamma $, 
    where $v=x-z$; $[\cdots]$ stands for the integral kernel of the operator in brackets. Here the model integral is $\int_{\mathbb{R}^3\ni z}d^3 z\frac{(x-z)^\alpha}{||x-z||^3 ||z||}=\frac{2\pi x^\alpha}{|| x||}$.
    }  %\marginpar{outline the computation in $\mathbb{R}^3$ in a footnote?}
    The ansatz (\ref{eta_Lambda singularity}) immediately implies the statement of the lemma.
    %Observe that (\ref{eta_Lambda singularity}) is continuous across the diagonal. 
    %and restricts to zero 
\end{proof}

\begin{cor}\label{cor: eta_A,A' boundary behavior}
The integral kernel $\eta_{A,A'}$ of the chain homotopy $K_{A,A'}$ considered modulo regular forms on $\overline{\mr{Conf}}_2(M)$ does not depend on $A'$. In particular, it coincides with the result of Axelrod-Singer \cite[PL5]{Axelrod1991} at $A'=A$.
%The restriction of the integral kernel $\eta_{A,A'}$ of the chain homotopy $K_{A,A'}$ to the boundary of the configuration space  $\partial \overline{\mr{Conf}}_2(M)$, considered modulo pullbacks of smooth forms on $M$,  does not depend on $A'$. In particular, it coincides with the result of Axelrod-Singer \cite[PL5]{Axelrod1991} at $A'=A$.
\end{cor}
\begin{proof}
    By (\ref{variation of ipK in A'}), we have $\delta_{A'}\eta_{A,A'}=d_A \eta_\Lambda+ \rho^*(\cdots)$ with $\cdots$ a smooth form on $M\times M$. Here $d_A$ acts on both arguments of $\eta_\Lambda$. Restricting to $\partial \overline{\mr{Conf}}_2(M)$ and using Lemma \ref{lemma: eta_Lambda behavior on boundary}, we get the statement of the corollary.
\end{proof}

Thus, according to the result of \cite[PL5]{Axelrod1991}, the propagator near the diagonal in $M\times M$ splits into singular, bounded and regular pieces:
$\eta=\eta^\mr{sing}+\eta^\mr{bd}+\eta^\mr{reg}$. Note that the singular and bounded pieces do not depend on $A,A'$ and are of the form $\delta^{ab}$ times a form on $\overline{\mr{Conf}}_2(M)$ with real coefficients.
\begin{proof}[Proof of Proposition \ref{prop: hidden strata}]
    Consider the integral in the l.h.s. of (\ref{hidden stratum integral}). Consider a decoration of edges of $\Gamma'$ by pieces $\eta^\mr{sing}$, $\eta^\mr{bd}$, $\eta^\mr{reg}$. Let $\overline{\Gamma}'$ be the subgraph of $\Gamma'$ where we only retain edges decorated by $\eta^\mr{sing}$ or $\eta^\mr{bd}$  (edges decorated by $\eta^\mr{reg}$, the possible edge decorated by $\eta_\Lambda=\eta_\Lambda^\mr{reg}$, and the leaves are all removed in $\overline\Gamma'$). We have the following:
    \begin{enumerate}[(a)]
        \item If $\bar\Gamma'$ contains a univalent vertex $v$, the integral (\ref{hidden stratum integral}) vanishes by degree reason (one has an integral over 3-dimensional space of positions of $v$, of a form of degree $\leq 2$ in the position of $v$).\footnote{Here 3 is the dimension of the fiber of $\partial_{V'}\overline{\mr{Conf}}_V(M)\ra \partial_{V'\setminus v}\overline{\mr{Conf}}_V(M)$. This argument fails in the special case when $\bar\Gamma'$ consists of two vertices connected by an edge (in this case the dimension of the fiber is 2 instead). This is the case of a \emph{principal} (as opposed to \emph{hidden}) boundary stratum.}
        Also, if $\bar\Gamma'$ contains a 0-valent vertex, the integral (\ref{hidden stratum integral}) vanishes by a similar degree reason.
        \item If $\bar\Gamma'$ contains a bivalent vertex, the integral vanishes by Kontsevich's vanishing lemma \cite[Lemma 2.2]{Kontsevich1994}.\footnote{Here we are using the fact that $T^*\eta=-\eta$, $T^*\eta_\Lambda=-\eta_\Lambda$ for $T\colon M\times M\ra M\times M$, $(x,y)\mapsto (y,x)$; we understand that $T^*$ also swaps the two $\g$ factors. This symmetry property of propagators follows from skew-self-adjointness of operators $K,\Lambda$ with respect to Poincar\'e inner product $\int \langle \alpha \stackrel{\wedge}{,} \beta \rangle$ (not the Hodge inner product, where one does not have self-adjointness in the desynchronized setting).
        }
        \item %\marginpar{\bl Apr 21 edit}
        If $\overline\Gamma'=\overline\Gamma_1\sqcup \overline\Gamma_2$ is disconnected, the integral (\ref{hidden stratum integral}) vanishes: fixing a point $x\in M$, the integral over the fiber of $\partial_{V'}\Conf\ra M$ over $x$ vanishes due to invariance of the integrand under translation of vertices of $\overline\Gamma_1$ by any vector $w\in T_xM$. 
        %by degree reason. 
        %[ELABORATE?]
    \end{enumerate}
    %\marginpar{Update for possibly disjoint $\Gamma$}
    Starting with a %connected 
    trivalent graph $\Gamma$, possibly with leaves, with at least one marked edge or one marked leaf, select a full subgraph $\Gamma'$ and construct the associated graph $\overline\Gamma'$ -- in particular the marked edge/leaf is removed in $\overline\Gamma'$. The integral (\ref{hidden stratum integral}) vanishes by (a), (b), (c) above unless $\overline\Gamma'$ is a connected trivalent subgraph of $\Gamma$ with all edges decorated by $\eta^\mr{sing}$ or $\eta^\mr{bd}$. But the latter integral also vanishes by degree reason: the form degree of the integrand is 1 higher than the dimension of the domain of integration.
    % It either contains vertices of valence $\leq 2$ (and then (\ref{hidden stratum integral}) vanishes by (a), (b) above), or we have $\Gamma'=\Gamma$. In the latter case, $\Gamma'$ contains either a marked edge or a leaf, and so $\overline\Gamma'$ must contain a vertex of valence $\leq 2$, and then again (\ref{hidden stratum integral}) vanishes by (a), (b) above.
\end{proof}

%\marginpar{One can also call this situation a ``hidden stratum,'' though it would be a definition depending on which pieces of propagators we decorate the edges by.}
\begin{rem}
    The proof of vanishing above also works in the case when $\Gamma'$ contains two vertices $u,v$, with $\overline\Gamma'$ containing 2 or 3 edges connecting $u$ and $v$.
\end{rem}

\begin{rem}
    The proof above is an application of technology of \cite{Axelrod1994}, \cite[Appendix A]{Cattaneo2008}. These references consider the contribution of hidden strata of the configuration space on the variation of Chern-Simons partition function with respect to a change of metric. In that case one actually does get an anomaly (the framing anomaly) from the strata corresponding to the collapse of all vertices of a connected component (with no leaves) of a graph. In the case of variation of $A'$ such an anomaly does not arise, ultimately because the integral kernel of $\Lambda_{\delta A'}$ has a milder singularity on the diagonal (Lemma \ref{lemma: eta_Lambda behavior on boundary}) than $\Lambda_{\delta g}$ (see \cite[PL6]{Axelrod1991}): 
    \begin{equation*}
        \eta_{\Lambda_{\delta g}}(x,y)\sim_{x\ra y} \frac{1}{8\pi} (\det g)^{\frac12} \epsilon_{\mu\nu\rho} \frac{u^\mu}{|| u||^3} du^\nu (g^{-1}\delta g)^\rho_{\tau} u^\tau \delta^{ab}+\cdots
    \end{equation*}
    
    %Note that under scaling $u\ra \lambda u$, with $\lambda\ra 0$, the leading singularity of $\eta_{\Lambda_{\delta A'}}$ is $O(\lambda)$ whereas the leading singularity of $\eta_{\Lambda_{\delta g}}$ is $O(\lambda^0)$.

This difference of singular behavior can be traced to the fact that $\Lambda_{\delta g}$ is a pseudodifferential operator of order $-2$ while $\Lambda_{\delta A'}$ is a pseudodifferential operator of order $-3$.

    %\marginpar{edit. maybe throw out (b)}
    %This difference of singular behavior arises from two mechanisms working in conjunction: (i) $\Lambda_{\delta g}$ is a pseudodifferential operator of order $-2$ while $\Lambda_{\delta A'}$ is a pseudodifferential operator of order $-3$, (ii) additionally, $\Lambda_{\delta A'}$ is parametrized by a $\delta A'$ -- a 1-form on $M$, whereas $\Lambda_{\delta g}$ is parametrized by $g^{-1}\delta g$ -- an endomorphism of the tangent bundle of $M$.
\end{rem}
 
\printbibliography
%\newpage
%\printunsrtglossary[type=main,title=Dramatis personae,nonumberlist,style=myNoHeaderStyle]
%\newpage

\textbf{Data Availability Statement:}  No data have been used or created for this work. 

\textbf{Conflict of Interest Statement:}  The authors declare no conflict of interest. 

\end{document}